\documentclass[12pt]{article}

\usepackage[top=1in, bottom=1in, left=2cm, right=2cm]{geometry} 

\usepackage{amsthm,amsmath,amsfonts,amssymb}
\usepackage{natbib,comment}
\usepackage{multibib}
\usepackage{multibib}
\newcites{SM}{Supplement References}
\RequirePackage{graphicx}
\usepackage{wrapfig}
\usepackage{bigints}
\usepackage[doublespacing]{setspace}
\usepackage{relsize}
\usepackage{catchfilebetweentags}
\usepackage{layout} 
\usepackage{bm} 
\usepackage[shortlabels]{enumitem} 
\usepackage{caption} 
\usepackage{float}
\usepackage{multirow} 
\usepackage{multicol} 
\usepackage{xcolor} 
\usepackage{rotating}
\usepackage{lscape} 
\usepackage{pdflscape}
\usepackage{subfigure}
\usepackage{multicol,lipsum}
\allowdisplaybreaks 
\usepackage{mathtools}

\numberwithin{equation}{section}
\newtheorem{theorem}{Theorem}[section]
\newtheorem{lemma}{Lemma}[section]
\newtheorem{proposition}{Proposition}[section]
\newtheorem{corollary}{Corollary}[section]

\newtheorem{assumption}{Assumption}[section]

\newtheorem{assumptionalt}{Assumption}[section]
\newenvironment{assumptionp}[1]{
  \renewcommand\theassumptionalt{#1}
  \assumptionalt
}{\endassumptionalt}

\newtheorem{definition}{Definition}[section]
\newtheorem{remark}{Remark}[section]
\newtheorem{example}{Example}[section]

  \newtheoremstyle{example_contd}
{\topsep} {\topsep}%
{\upshape}
{}
{\bfseries\scshape}
{}
{3mm}
{\thmname{#1}\thmnumber{ #2}\thmnote{ #3.}\enspace(continued)}

\theoremstyle{example_contd}
\newtheorem*{example_contd}{Example}

\newcommand{\m}[1]{\mathcal{#1}}

\newcommand{\boup}[1]{\boldsymbol{\mathrm{#1}}}

\newcommand{\ha}[1]{\widehat{#1}}
\newcommand{\ti}[1]{\widetilde{#1}}
\newcommand{\lbar}[1]{\underline{#1}}
\newcommand{\ubar}[1]{\overline{#1}}
\newcommand{\mmc}{\mathbb{C}}
\newcommand{\mme}{\mathbb{E}}
\newcommand{\mmi}{\mathbb{I}}

\newcommand{\mmp}{\mathbb{P}}
\newcommand{\mmr}{\mathbb{R}}
\newcommand{\mmv}{\mathbb{V}}

\newcommand{\mmx}{\mathbb{X}}

\newcommand{\mmq}{\mathbb{Q}}

\newcommand{\be}{\boup{e}}

\newcommand{\bphi}{\boup{\varphi}}
\newcommand{\bgam}{\boup{\gamma}}

\newcommand{\bG}{\boup{G}}
\newcommand{\bHt}{\ti{\boup{H}}}
\newcommand{\bH}{\boup{H}}

\newcommand{\bQ}{\boup{Q}}

\newcommand{\bR}{\boup{R}}
\newcommand{\bV}{\boup{V}}

\newcommand{\bZ}{\boup{Z}}

\newcommand{\bGam}{\boup{\Gamma}}
\newcommand{\bDel}{\boup{\Delta}}

\newcommand{\eps}{\varepsilon}
\newcommand{\pto}{\overset{p}{\to}}
\newcommand{\dto}{\overset{d}{\to}}
\newcommand{\asto}{\overset{as}{\to}}
\newcommand{\dd}{\overset{d}{=}}
\newcommand{\overbar}[1]{\mkern 1.5mu\overline{\mkern-1.5mu#1\mkern-1.5mu}\mkern 1.5mu}

\def\Exp{\mme}
\def\Var{\mmv}

\def\Pr{\mmp}
\def \argmin{\mathop{\hbox{\rm argmin}}}

\newif\ifblind

\blindfalse 

\makeatletter
\def\old@comma{,}
\catcode`\,=13
\def,{%
  \ifmmode%
    \old@comma\discretionary{}{}{}%
  \else%
    \old@comma%
  \fi%
}
\makeatother

\begin{document}

\title{Permutation Tests at Nonparametric Rates}

\ifblind
\else
{\author{Marinho Bertanha\footnote{
  Department of Economics, University of Notre Dame.
  Address: 3060 Jenkins Nanovic Halls, Notre Dame, IN 46556.
  Email: mbertanha@nd.edu. 
  Website: www.nd.edu/$\sim$mbertanh.
  }
  \and
  EunYi Chung\footnote{
  Department of Economics, University of Illinois at Urbana-Champaign.
  Address: 1407 W Gregory Dr, Urbana, IL 61802.
  Email: eunyi@illinois.edu.
  Website: economics.illinois.edu/profile/eunyi.
  }
}}
\fi

\date{
April 20th, 2022
}

\maketitle

\vspace{-1cm}

\begin{abstract}
Classical two-sample permutation tests for equality of distributions have exact size in finite samples,
but they fail to control size for testing equality of parameters that summarize each distribution.
This paper proposes permutation tests for equality of parameters that are estimated at root-$n$ or slower rates.
Our general framework applies to both parametric and nonparametric models, with two samples or one sample split into two subsamples.
Our tests have correct size asymptotically while preserving exact size in finite samples when distributions are equal.
They have no loss in local asymptotic power compared to tests that use asymptotic critical values.
We propose confidence sets with correct coverage in large samples that also have exact coverage in finite samples if distributions are equal up to a transformation.
We apply our theory to four commonly-used hypothesis tests of nonparametric functions evaluated at a point.
Lastly, simulations show good finite sample properties, and two empirical examples illustrate our tests in practice.

\end{abstract}


\vspace{-6mm}

\section{Introduction}
\indent

\vspace{-2mm}

Applications of permutation tests have gained widespread popularity in empirical analyses in the social and natural sciences.
Classical two-sample permutation tests appeal to applied researchers because they are easy to implement and have exact size in finite samples under the so-called ``sharp null hypothesis''.
The sharp null hypothesis states that the two population distributions are equal.
However, researchers are often interested in testing equality of parameters that summarize the distributions.
For example, one may want to test equality of average outcomes between treatment and control groups while nonparametrically controlling for age and income.
Classical permutation tests fail to control size under such nulls, in both finite and large samples.\footnote{The lack of size control of the classical permutation test
outside of the sharp null has been studied for a long time (e.g., \cite{romano1990}).
Theorem \ref{theo1} below confirms the lack of size control in our setting.
}

This paper proposes robust two-sample permutation tests for equality of parameters that are estimated at root-$n$ or slower rates.
The tests are robust in the sense that they control size asymptotically while preserving finite-sample exactness under the sharp null.
Our general framework covers both parametric and nonparametric models, in cases with two samples from two populations or one sample from a union of two populations.
In addition, the paper makes three further contributions.
First, we derive the asymptotic permutation distribution in our general framework under both null and alternative hypotheses, which requires novel technical arguments. 
Second, we provide four examples of tests in widely-used nonparametric models, and we prove that they satisfy the conditions  of our framework.
Third, we construct robust confidence sets for differences between the two populations.
The confidence sets are robust meaning that they have correct coverage asymptotically and exact coverage in finite samples if the sharp null holds under a class of transformations.

Our framework considers a summary parameter that can be consistently estimated using an asymptotically linear statistic.
The influence function depends on the data, the population distribution, and the sample size.
There may be two \textit{iid} samples from two populations, or one \textit{iid} sample from a union of two populations.
In the case of one sample, there is a variable in the data that identifies the population of each observation, and the sample is split into two.
The researcher applies the estimator to each of the two samples, computes the difference, and tests  whether the two parameters are equal.
The classical permutation test compares the estimated difference to critical values from the permutation distribution,
that is, the distribution of estimates over all permutations of observations across the two samples.

We derive the asymptotic permutation distribution of the estimated difference, both with and without the null hypothesis, and find it to be generally different from its asymptotic sampling distribution.
This leads the classical permutation test to have incorrect size, as shown in a variety of other settings (see related literature below). 
The derivation has two key technical features. 
First, we borrow the coupling approximation from \cite{chung2013} (henceforth CR) and apply it to our setting.
The approximation consists of two steps:
(a)  couple the original sample with a new sample from a particular mixture of the two populations;
and
(b) prove that replacing the original sample with the new coupled sample renders no change in the asymptotic permutation distribution of the test statistic.
We utilize the coupling approximation because it is more tractable to study the limiting behavior of permutation tests when both samples are drawn \textit{iid} from the same mixture distribution, rather than different distributions. 
Our coupling approximation differs from that of CR in that our proof does not assume the null hypothesis.
This allows us to study both size and power of the permutation test.
Our proof requires a new argument to bound the variance of the approximation error 
(Section B.2.3 in the supplement).
The second key technical feature is that our theory allows for random sample sizes.
Sample sizes are random when the researcher splits one sample into two as a function of the data.
Thus, the derivation of the limiting distribution must be valid conditional on any sequence of sample splits that occurs with probability one and requires novel arguments:
(a) a conditional central limit theorem of weighted sums of triangular arrays,
where the weights are random 
(Lemmas F.3 and F.4 in the supplement);
and
(b) an asymptotic linear representation for estimators that must hold uniformly over convex combinations of the two populations (Assumption \ref{aspt:asy_linear}).


Our proposed permutation test uses a studentized test statistic, which is the estimated difference of parameters divided by a consistent estimator of its standard deviation.
We then show that both the asymptotic permutation and sampling distributions are standard normal.
It follows that our permutation test has correct size in large samples, and its asymptotic power against local alternatives is identical to that of the test 
that relies on critical values from a standard normal.
Finally, we construct a confidence set by inverting our test, which requires testing null hypotheses that are more general than simple equality of parameters.
We propose ways to transform the data in order to test more general hypotheses and preserve finite sample exactness when populations are equal up to the data transformations.

Examples of applications of permutation tests abound.
Table 3 
in the supplement lists top-publications from the last decade in a variety of disciplines that use permutation tests.
This broad applicability motivates our extension of the theory.
We illustrate our framework using four nonparametric examples of hypothesis tests that are often used in empirical studies.
The first and second examples test equality at a point of  nonparametric conditional mean and quantile functions, respectively. 
The third and fourth examples test continuity at a point of nonparametric conditional mean or probability density functions (PDFs), respectively.
We explain how to implement the permutation test in each case and give sufficient conditions to derive the limiting permutation distribution. 
Implementation requires a sign change and sample splitting in the third and fourth examples.
We find that asymptotic size control requires studentization, except in the fourth example.

\vspace{.3cm}

\noindent {\it Related Literature}

Randomization inference has recently received keen attention  \citep{canay2017sym,shaikh2021}.
For the case of permutation inference, 
the insight of robustness through studentization has been proposed before in specific settings that differ from our general framework: 
\cite{neuhaus1993}, \cite{janssen1997}, \cite{janssen2005}, \cite{neubert2007}, \citet[CR]{chung2013}, \cite{pauly2015}, \cite{chung2016}, and \cite{diciccio2017}.
In particular, the framework of \cite{janssen1997} applies to the problem of testing difference of means, 
while CR study the more general problem of difference of parameters that are estimable at root-$n$. 
It is important to emphasize that this paper is not a straightforward generalization of their work, and none of our applications fits CR's framework.
For example, CR handles the difference of parametric quantiles for two independent samples; in contrast, our framework covers the cases of parametric or nonparametric quantiles, for two independent samples or one sample randomly split into two. 
Moreover, our verification of the coupling approximation differs from that of CR because we do not assume the null hypothesis. 
All these features make many of our proofs substantially different from theirs.

Studentization is not the only way to achieve robustness. 
Robustness is also obtained through prepivoting (\cite{chung2016multivariate}; \cite{fogarty2021}) and the Khmaladze transformation (\cite{chung2021}).
Prepivoting obtains robustness in examples where studentization does not help.
For instance, consider a vector of differences in means where one is interested in the maximum absolute difference.
The asymptotic distribution cannot be made pivotal by simply studentizing the maximum statistic. 
Likewise, the Khmaladze transformation is applied to empirical processes to make them asymptotically pivotal.
 
Previous works have also considered randomization tests for continuity of nonparametric models at a point.
\cite{cattaneo2015} propose local randomization inference procedures for a sharp null hypothesis,
while \cite{canay2016rd} provide permutation tests for continuity of the whole distribution of an outcome variable conditional on a control variable at a point.
In contrast, our permutation test applies to testing continuity of summary statistics of the conditional distribution such as mean, quantile, variance, etc. 
Our fourth example is related to \cite{bugni2018}, who propose a sign-change test for continuity of PDFs at a point, where critical values come from maximizing a function of a binomial distribution. 
We show how the same null hypothesis fits into our framework and is simply testable using permutations.
The last two papers use the insight that non-\textit{iid} order statistics converge in distribution to \textit{iid} variables, which is 
technically distinct from our coupling approximation.
Finally, permutation-based confidence sets have previously been proposed only in specific settings. 
For example, the confidence sets of \cite{imbens2005} assume that treatment effects divided by treatment doses are constant across individuals, and that the distribution of treatment eligibility is known.
To the best of our knowledge, ours is the first paper to provide valid two-sample permutation tests and confidence sets for scalar parameters estimated at nonparametric rates.

The rest of this paper is outlined as follows. 
Section \ref{sec:theory} presents
the general framework, assumptions, and asymptotic distributions of the classical and robust permutation tests.
Section \ref{sec:applications} studies how our theory applies to four nonparametric examples.
Section \ref{sec:setup:cset} explains how to invert permutation tests to build robust confidence sets.
Section \ref{sec:mc} displays a simulation study that confirms our theory and illustrates good finite sample properties of the robust permutation test.
Section \ref{sec:empirical} illustrates the practical relevance of our procedures to test for continuity of conditional mean functions using real-world data.
The supplement contains all proofs.

\vspace{-6mm}

\section{Theory}
\label{sec:theory}

\vspace{-2mm}

\indent 

Consider two populations $P_1$ and $P_2$, and a real-valued parameter $\theta(P_k)$ summarizing distribution $P_k$, $k=1,2$.
The null hypothesis is stated as
$H_0 : \theta(P_1) = \theta(P_2)$.
Note that our null hypothesis is a bigger set of distributions compared to the set of distributions in the so-called ``sharp null hypothesis,'' respectively, 
$\{ (P_1,P_2) : \theta(P_1) = \theta(P_2)\}$ vs. $\{ (P_1,P_2) : P_1=P_2 \}$.
For each population $k$, there are $n_k$ \textit{iid} observations $\bZ_{k,i} \in \mmr^q$ from distribution $P_k$, $i=1,\ldots, n_k$. 
Observations  are independent across $k$, and the total number of observations is $n=n_1+n_2$.
We define $\m{P}$ to be the convex hull of $\{P_1, P_2 \}$,
that is, $\m{P} = \{P : P = \eta P_1 + (1-\eta) P_2, ~ \eta \in [0,1] \}$.
Throughout this paper, random variables with subscript ``$k$" indicate they have distribution $P_k$, 
e.g., $\bZ_k$. 
For any other distribution in $\m{P}$, the random variable is denoted $\bV\in \mmr^q$. 
Operators such as $\mmp$ (probability), $\mme$ (expectation),  or $\mmv$ (variance) applied to $\bZ_k$ do not carry the subscript $P_k$, but operators applied to $\bV$ carry the subscript $P$,
e.g., $\mme[\bZ_k]$ vs. $\mme_P[\bV]$.
The parameter $\theta(P_k)$ is consistently 
estimated by 
$\ha{\theta}_{k} = \theta_{n_k,n}(\bZ_{k,1}, \ldots, \bZ_{k,n_k})$,
where  the functions
$\theta_{n_1,n}$
and
$\theta_{n_2,n}$ 
satisfy the following assumption.

\begin{assumption}\label{aspt:asy_linear} Let $\bV_1, \ldots, \bV_m$ be an \textit{iid} sample from a distribution $P \in \m{P}$.
Let $m$ grow with $n$ such that $ m/n \to \gamma$, for some $\gamma \in  (0,1)$. 
Use these observations to construct the estimator $\ha\theta = \theta_{m,n}(\bV_1, \ldots, \bV_m)$.
Assume the following objects exist: a sequence of functions $\psi_n: \mmr^q \times \m{P} \to \mmr$,
a function $\xi: \m{P} \to \mmr$,
and 
a non-increasing sequence $ h_n$ such that $n h_n \to \infty$. Further assume  
\begin{align} 
& \forall \eps>0 : \sup\limits_{P \in \m{P} } \mmp_P \left\{ \left| \sqrt{m h_n } \left( \ha{\theta} - \theta(P)  \right) - \left( \frac{1}{\sqrt{m}} \sum_{i=1}^{m} \psi_n(\bV_{i}, {P} )\right) \right| >\eps \right\} \to 0,
\label{eq:aspt:asy_linear}
\\
 & \mme_P[  \psi_n(\bV_{i}, {P} ) ]  =  0 ~~ \forall P \in \m{P},
 \label{eq:aspt:asy_linear_zeromean} 
 \\
 & \sup\limits_{P \in \m{P} } \left| \mmv_P[  \psi_n(\bV_{i}, {P} ) ] - \xi^2(P) \right| \to 0,
 \label{eq:aspt:asy_linear_var} 
 \\
 & \sup\limits_{P \in \m{P} }  \mme \left[ \psi_n^2(\bZ_k, {P} ) \right]   < \infty , ~~\mbox{for} ~ k\in\{1,2\},
 \label{eq:aspt:asy_linear_mom} 
\\
& \exists \zeta>0 :   n^{-\zeta/2} \sup\limits_{P \in \m{P} }   \mme_P \left| \frac{\psi_n(\bV_{i}, {P} )}{ \mmv_P[  \psi_n(\bV_{i}, {P} ) ]^{1/2} } \right|^{2+\zeta} \to 0 \text{, and }
 \label{eq:aspt:asy_linear_lind} 
\\
& \xi^2\left( \frac{m}{n} P_1 + \frac{n-m}{n}  P_2 \right) \to \xi^2 \left( \gamma P_1 + (1-\gamma ) P_2  \right) .
 \label{eq:aspt:asy_linear_xicont} 
\end{align}
\end{assumption}

Assumption \ref{aspt:asy_linear} requires the estimator $\ha \theta_k$ to have an asymptotically linear expansion at rate $\sqrt{n_k h}$ (Eq. \ref{eq:aspt:asy_linear}) with zero-mean influence function $\psi_n$ (Eq. \ref{eq:aspt:asy_linear_zeromean}) and asymptotic variance $\xi^2$ (Eq. \ref{eq:aspt:asy_linear_var}).
The expansion must hold for data drawn from any combination of $P_1$ and $P_2$ in $\m{P}$ because the limiting permutation distribution of $\ha{\theta}_1 - \ha{\theta}_2$ behaves as if the data were drawn from a sequence of mixtures of $P_1$ and $P_2$ (Theorems \ref{theo1} and \ref{theo2} below).
Assumption \ref{aspt:asy_linear} also imposes bounds on moments of $\psi_n$ (Eqs. \ref{eq:aspt:asy_linear_mom}--\ref{eq:aspt:asy_linear_lind}) to enable us to apply a central limit theorem and derive the limiting permutation distribution. 
This requires $\xi^2(P)$ to be smooth with respect to $P$ (Eq. \ref{eq:aspt:asy_linear_xicont}) in order for 
the limit of the variance $\xi^2\left( \frac{m}{n} P_1 + \frac{n-m}{n}  P_2 \right)$ to be well-defined.

Situations arise where the number of observations $n_k$ is random rather than deterministic.
For example, suppose the researcher desires to compare the female and male subpopulations of a country but only has one \textit{iid} sample with $n$ individuals from that country.
The researcher splits the sample into two subsamples based on the gender of each observation and sample sizes are random.
In order to accommodate both deterministic and random sample sizes,
we consider a sampling scheme which is dictated by a vector of indicator variables 
$\lbar{W}_n = (W_{1}, \ldots, W_{n})$, $W_{i} \in \{1,2\}$ for $i=1,\ldots,n$, where $\lbar{W}_n$ has distribution $\bQ_n$.
Conditional on $\lbar{W}_n$, the sample $\lbar{\bZ}_n = (\bZ_1, \ldots, \bZ_n)$ has $\bZ_i$ drawn from distribution $P_1$ if $W_{i}=1$ or from distribution $P_2$ if $W_{i}=2$, with observations independent across $i$. 
This accommodates the standard two-population sampling by making $\lbar{W}_n$ non-random with $n_1$ entries equal to $1$, $n_2$ entries equal to $2$, and $n_1+n_2=n$.
It also accommodates the example above of male and female subpopulations by making $W_{i}$ \textit{iid} and $\Pr(W_{i}=1)$ equal  to the probability of being female.
Conditional on $\lbar{W}_n$, there are $n_k$ \textit{iid} observations $\bZ_{k,i} \in \mmr^q$ from distribution $P_k$, $k=1,2$, 
and observations  are independent across $k$.
As before, $\ha{\theta}_{k} = \theta_{n_k,n}(\bZ_{k,1}, \ldots, \bZ_{k,n_k})$.

\begin{assumption}\label{aspt:sampling} There exists  $\lambda \in (0,1)$
 such that  the sequence of distributions  $\bQ_n$ satisfies ${n_1}/{n}  \pto   \lambda $  as $n\to \infty$.
 Moreover, Assumption \ref{aspt:asy_linear} holds for all sequences of sample sizes $m$ such that  $ m/n  \to   \lambda $
 or
 $ m/n  \to   1-\lambda $.

\end{assumption}

The test statistic $T_n$  is a function of the data $(\lbar{W}_n, \lbar{\bZ}_n)$ 
as follows:

\vspace{-1cm}
\begin{gather}
T_n(\lbar{W}_n, \lbar{\bZ}_n)  \doteq \sqrt{n h}   \left( \ha\theta_1 -  \ha\theta_2  \right), 
\label{eq:setup:teststat}
\end{gather}
\vspace{-1cm}

\noindent where we omit the subscript $n$ from the sequence $h_n$ of  Assumption \ref{aspt:asy_linear} to simplify notation.
 
The permutation test is constructed by permuting the order of observations in $\lbar{\bZ}_n$, while keeping the indicator variables $\lbar{W}_n$ unchanged,
and recomputing the test statistic.
A permutation is a one-to-one function $\pi: \{1,\ldots,n\} \to \{1, \ldots, n\}$, where $\pi(i) = j $ 
says that the $j$-th observation becomes the $i$-th observation once permutation $\pi$ is applied.
Given permutation $\pi$, the permuted sample becomes 
$(\lbar{W}_n,\lbar{\bZ}_n^{\pi}) = (W_{1}, \ldots, W_{n},\bZ_{\pi(1)}, \ldots, \bZ_{\pi(n)})$,
and the re-computed value of the test statistic is $T_n^{\pi} =  T_n(\lbar{W}_n, \lbar{\bZ}_n^\pi)$.
In other words, permutations swap individuals across the two samples to which they originally belonged according to $\lbar{W}_n$, which remains fixed. 
The set $\boup{G}_n$ is the set of all possible permutations $\pi$.
The number of elements in $\boup{G}_n$ is $n!$.

The two-sided permutation  test with nominal level $\alpha \in (0,1)$ is constructed as follows.
First, re-compute the test statistic $T_n(\lbar{W}_n, \lbar{\bZ}_n)$ for every $\pi \in \bG_n$. 
Rank the values of $T^{\pi}_n$ across $\pi$:
$T^{(1)}_{n} \leq T^{(2)}_{n} \leq \ldots \leq T^{(n!)}_{n}$.
Second, fix a nominal level $\alpha \in (0,1)$ and 
let $k^- = \lfloor n! \alpha/2  \rfloor$, that is, the largest integer less than or equal to $n!  \alpha/2 $,
and 
$k^+ = n! - k^- $. 
Third, compute the following quantities:
(i) $M^+$, the number of values $T^{(j)}_{n}$, $j = 1, \ldots, n!$, that are strictly greater than $T^{(k^+)}_{n}$;
(ii) $M^-$, the number of values $T^{(j)}_{n}$ that are strictly smaller than $T^{(k^-)}_{n}$;
(iii) $M^0$, the number of values $T^{(j)}_{n}$ that are equal to either $T^{(k^+)}_{n}$ or $T^{(k^-)}_{n}$;
and (iv) $a = \left( \alpha n! - M^+ - M^- \right)/M^0 $.
Finally, the outcome of the test is based on the test function $\phi$:

\vspace{-1cm}
\begin{align}
\phi (\lbar{W}_n, \lbar{\bZ}_n) 
= \left\{ 
	\begin{array}{rcl}
		1& \mbox{if} & T_{n} > T^{(k^+)}_{n} \text{ or } T_{n} < T^{(k^-)}_{n},
		\\
		a & \mbox{if} & T_{n} = T^{(k^+)}_{n} \text{ or } T_{n} = T^{(k^-)}_{n}, \\
		0 & \mbox{if} & T^{(k^-)}_{n} < T_{n} < T^{(k^+)}_{n}.
	\end{array}\right.
\label{eq:setup:testfun}
\end{align}
\vspace{-1cm}

For a given sample, if $\phi=1$, we reject the  null hypothesis;
if $\phi=a$, we randomly reject the null hypothesis with probability $a$;
otherwise, if $\phi=0$, we fail to reject the null.
A classic property of permutation tests  is exact size in finite samples under the ``sharp null,'' that is, 
the null hypothesis stating $P_1 = P_2$.

\begin{lemma}
\label{lemma:exact}
For any $n$, $\bQ_n$, $P_1$, and $P_2$, if $P_1=P_2$, then $\mme[\phi (\lbar{W}_n, \lbar{\bZ}_n)] = \alpha$.
\end{lemma}

\begin{remark}
Re-computing the test statistic for all $n!$ permutations is costly, even for  small $n$. 
Lemma \ref{lemma:exact} and remaining results in this section are unchanged if, instead of $\bG_n$, we use a random sample of $\bG_n$ with or without replacement (\cite{lehmann2005}, page 636). 
\end{remark}
\begin{remark}
The randomized outcome in the case of ties is important for exact size in finite samples if $P_1=P_2$.
However, it may be desirable to have a deterministic answer to a hypothesis test after observing a sample of data. 
An easy way to fix that is to set $\phi=0$ in the case of ties.
This makes the test conservative, that is, the size becomes less than or equal to $\alpha$.
\end{remark}

The set of distributions that satisfy the null hypothesis $\theta(P_1)=\theta(P_2)$ is in general larger than the set of distributions that satisfy the sharp null $P_1=P_2$.
Thus, there is no finite sample size control in general. 
To investigate the asymptotic properties of the test in \eqref{eq:setup:testfun}, we derive the probability limit of the permutation distribution,

\vspace{-1cm}
\begin{gather}\label{eq:cdfRT0}
\ha{R}_{T_n}(t) = \frac{1}{n!} \sum_{\pi \in \boup{G}_n} \mmi\{ T_{n}( \lbar{W}_n, \lbar{\bZ}_n^{\pi} ) \leq t \}.
\end{gather}
\vspace{-1cm}

The hypothesis test \eqref{eq:setup:testfun} utilizes critical values from  $\ha{R}_{T_n}$.
The  test has asymptotic size control if, under the null hypothesis, the probability limit of $\ha{R}_{T_n}$ equals the cumulative distribution function (CDF) of the limiting distribution of  $T_n$.
In order to study both size and power, we derive these limiting distributions without imposing the null hypothesis in the following theorems.

\begin{theorem}
\label{theo1}
Under Assumptions \ref{aspt:asy_linear}--\ref{aspt:sampling},
the permutation distribution $\ha{R}_{T_n}$ converges uniformly in probability to the CDF of a $N \left( 0, \tau^2 \right)$, i.e., 
$\sup_t \left| \ha{R}_{T_n}(t) - \Phi\left( t / \tau \right) \right| \pto 0,$
 where $\tau^2 \doteq {\xi^2(\ubar{P})} / {\lambda(1- \lambda)}$ and $\ubar{P} \doteq \lambda P_1 + (1-\lambda) P_2$.
 Moreover, 
$T_n - \sqrt{nh}\left( \theta(P_1) - \theta(P_2)  \right)  \dto N \left( 0, \sigma^2 \right),$
 where $\sigma^2 \doteq { \xi^2(P_1)}/{\lambda} + { \xi^2(P_2) }/{(1-\lambda)}$.
\end{theorem}

The permutation distribution fails to control size asymptotically because the asymptotic variance of the permutation distribution 
$\tau^2$ generally differs from $\sigma^2$.
To appreciate the implications of Theorem \ref{theo1}, we present the parametric example below and four nonparametric examples in Section \ref{sec:applications} detailing cases where $\tau^2 = \sigma^2$ and $\tau^2 \neq \sigma^2$.

\begin{example} \label{example:param}\rm (Parametric Model)
 For $k=1,2$, $\bZ_k=(X_k,Y_k) \sim P_k$, $X_k \sim U[0,1]$, $\Exp{[Y_k|X_k]}=\theta(P_k) + \beta X_k$, and $\Var{[Y_k|X_k]} = v_k$.
There are two independent samples with $n_1/n \to \lambda$. 
The ordinary least squares estimator is $\ha{\theta}_k = \frac{\overbar{X_k^2}\overbar{Y_k} - \overbar{X_k}\overbar{X_kY_k}}{\overbar{X_k^2} - (\overbar{X_k})^2},$
where $\overbar{X_k^s} = \frac{1}{n_k}\sum_{i=1}^{n_k} X_{k,i}^s$ for $s=1, 2$, and $\overbar{X_kY_k} = \frac{1}{n_k}\sum_{i=1}^{n_k} X_{k,i}Y_{k,i}$.
For $\bV = (R, S) \sim P$, the setting satisfies Assumption $\ref{aspt:asy_linear}$
with $h=1$; 
$\psi_n(\bV, P)=\psi(\bV, P)=(S-\Exp_P{(S|R)})(4-6R)$; 
and $\xi^2(P) = \Exp_P\left[ (S-\Exp_P(S|R))^2 (4-6R)^2\right].$
Under the null hypothesis, the asymptotic variance of 
$T_n = \sqrt{n}(\ha{\theta}_1 - \ha{\theta}_2)$ and of the permutation distribution are, respectively,
    $\sigma^2 = 4 \left[ \frac{ v_1 }{ \lambda } + \frac{ v_2 }{1-\lambda } \right]$ and 
    $\tau^2  = 4 \left[ \frac{ v_1 }{ 1- \lambda } + \frac{ v_2 }{\lambda } \right].$
Unless $\lambda=1/2$ or $v_1=v_2$, $\sigma^2$ and $\tau^2$ do not agree in general, distorting the rejection probability under the null for the permutation test. 
Section C 
in the supplement displays a graph that illustrates this distortion using simulated data. 
\end{example}

Since $\tau^2$ and $\sigma^2$ are generally different, the test statistic $T_{n}$ must be transformed to become asymptotically pivotal.
Thus, we divide $T_n$ by the square root of a consistent estimator for its asymptotic variance.
For each population $k \in \{ 1, 2\}$, let $\hat \xi_{k}^2 = \xi_{n_k,n}^2(\bZ_{k,1}, \ldots, \bZ_{k,n_k})$ be a consistent
estimator for  $\xi^2(P_k)$ and assume the functions $\xi_{n_1,n}^2$ and $\xi_{n_2,n}^2$ satisfy the following assumption.
\begin{assumption}
\label{aspt:var:est}
Let $\bV_1, \ldots, \bV_m$ be an \textit{iid} sample from a distribution $P \in \m{P}$.
Use these observations to construct the estimator $\ha\xi^2 = \xi^2_{m,n}(\bV_1, \ldots, \bV_m)$.
Assume that, for any  sequence of sample sizes $m$ such that  $ m/n  \to   \lambda $
or
$ m/n  \to   1-\lambda $,
$\ha\xi^2 - \xi^2(P)$ converges in probability to $0$ uniformly over $P \in \m{P}$.
\end{assumption}

Then, the studentized test statistic $S_n$ is

\vspace{-1cm}
\begin{align}
S_n (\lbar W_n, \lbar \bZ_n)= \frac{ T_n (\lbar W_n, \lbar \bZ_n) } { \ha \sigma_n }
\label{eq:setup:teststat:stu}
\end{align}
\vspace{-1cm}

\noindent where $\ha \sigma_n$ is the square root of the consistent estimator for the asymptotic variance of $T_n$, that is, 
$\ha\sigma_n^2 =   \frac{n}{n_1}  \ha\xi_{1}^{2} +  \frac{n}{n_2} \ha\xi_{2}^2.$

\begin{theorem}
\label{theo2}
Let $\ha{R}_{S_n}$ be the permutation CDF defined in (\ref{eq:cdfRT0}) with $T_n$ replaced by $S_n$.
Under Assumptions \ref{aspt:asy_linear}--\ref{aspt:var:est},
$\ha{R}_{S_n}$ converges uniformly in probability to the CDF of a $N \left( 0, 1 \right)$, i.e.,
$\sup_t \left| \ha{R}_{S_n}(t) - \Phi\left( t \right) \right| \pto 0.$
Moreover, 
$S_n - \frac{\sqrt{nh}\left( \theta(P_1) - \theta(P_2)  \right)}{\ha \sigma_n}  \dto N \left( 0, 1 \right).$
\end{theorem}

\begin{example_contd}[\ref{example:param}]\rm If the test statistic $T_n$ is appropriately studentized by the usual ordinary least squares formula for standard errors, both permutation and sampling distributions of $S_n$ coincide asymptotically. 
Section C in the supplement simulates data for this example and presents these two distributions graphically. 
\end{example_contd}

Note that the standard deviation $\hat \sigma_n$ that divides $T_n$ must be consistent for $\sigma$, as opposed to $\tau$.
However, when $\hat \sigma_n$ is evaluated using permuted samples, it converges in probability to $\tau$. 
Under the null hypothesis, both the permutation distribution and the test statistic $S_n$ are asymptotically standard normal. 
Therefore, our robust permutation test in \eqref{eq:setup:testfun} with $T_n$ replaced by $S_n$ has asymptotic size equal to the nominal level $\alpha$, even if $P_1 \neq P_2$.
In case $P_1 = P_2$, this test has exact size in finite samples.
Since Theorems \ref{theo1} and \ref{theo2} are true regardless of whether the null hypothesis holds or not, we can now study the power properties of the permutation test.

\begin{corollary}\label{coro1} Let $\phi(\lbar{W}_n, \bZ_n)$ be the permutation test in \eqref{eq:setup:testfun} with $T_n$ replaced by $S_n$, and suppose Assumptions \ref{aspt:asy_linear}--\ref{aspt:var:est} hold.
If the null hypothesis holds, then $\mme[\phi(\lbar{W}_n, \bZ_n)] \to \alpha$; otherwise, $\mme[\phi(\lbar{W}_n, \bZ_n)] \to 1$.
Moreover, assume $S_n$ has a limiting distribution under a sequence of local alternatives contiguous to the null.
Then, the asymptotic power of the robust permutation test against local alternatives is the same as that of the test that uses critical
values from the limiting null distribution of $S_n$.\footnote{Alternative resampling methods such as the bootstrap and subsampling also share the same asymptotic local power as the permutation test because they produce critical values that are consistent for standard Gaussian critical values under the null hypothesis.}$^,$\footnote{Suppose one uses an estimator  $\ha{\sigma}^2_n$ that assumes the null hypothesis is true; that is, $\ha{\sigma}^2_n$ is consistent for $\sigma^2$ under the null hypothesis 
but has a different probability limit under the alternative.
Regardless of whether the null is true, such an estimator applied to a random permutation of the data is generally consistent for $\tau^2$, and  Corollary \ref{coro1} remains true.
Consistency for $\tau^2$ comes from the fact that an estimator applied to a random permutation behaves as if it were applied to data from a mixture distribution, where the null is always true 
(Section B.3.2 in the supplement).
}
\end{corollary}

\vspace{-6mm}

\section{Applications}
\label{sec:applications}

\vspace{-2mm}

\indent 

In this section, we apply our theory  to four different nonparametric problems:
testing for equality of conditional mean and quantile functions evaluated at a point,
and 
testing for continuity of conditional mean and PDF 
at a point. 
For a simple and intuitive presentation,  we use the Nadaraya--Watson (NW) type of kernel estimators throughout, but proofs in this section generalize to other types of estimators, 
e.g., local-polynomial regression (LPR), sieves, etc.
In particular, we demonstrate the generalization of our third application in Section \ref{sec:applications:rdd} to the LPR estimator in 
Section D.5 
of the supplement. 
We obtain the usual rate restrictions on the bandwidth tuning parameter $h$, which allows researchers to choose among standard options available in the literature. 
We fit the four examples into the general framework of Section \ref{sec:theory}, 
demonstrate the validity of Assumptions  \ref{aspt:asy_linear} and \ref{aspt:sampling},
and show that studentization is generally required for asymptotic size control, except for the PDF continuity test.
The third and fourth examples (testing for continuity of conditional mean and PDF) illustrate the sample-splitting feature of our general framework.

Throughout this section, $K(u)$ denotes a symmetric, non-negative, bounded kernel density function  with $\int_{-\infty}^{\infty} |u|^3 K(u) ~ du <\infty$. 
We denote $\kappa_{s,t} = \int_{-\infty}^{\infty} u^s K^t(u) ~ du$
and 
$\kappa_{s,t}^+ = \int_{0}^{\infty} u^s K^t(u) ~ du$.

\subsection{Controlled Means}
\label{sec:applications:control_mean}

\indent 

Researchers are often interested in comparing conditional mean functions between two different populations.
For example, in randomized controlled trials, $P_1$ and $P_2$ are the populations of control and treatment individuals, respectively. 
Of interest is the average outcome $Y$ after controlling for an individual characteristic $X$.
For instance, outcomes of a professional training program may differ between rich and poor individuals, and we would like to condition on income $X$.

There are two independent samples of bivariate variables: $n_1$ observations with $\bZ_{1,i}=(X_{1,i}, Y_{1,i})$ \textit{iid} $P_1$ and $n_2$ observations with $\bZ_{2,i}=(X_{2,i}, Y_{2,i})$ \textit{iid} $P_2$.
The vector $\lbar{W}_n$ is non-random with $n_1$ entries equal to $1$ and $n_2$ entries equal to $2$, where $n=n_1+n_2$.
For a given interior point $x$,
$\theta(P_k) = \mme[Y_{k,i} | X_{k,i}=x]$, $k=1,2$, and the null hypothesis is $\theta(P_1) = \theta(P_2)$.

A common estimator for conditional mean functions is the NW kernel estimator, which computes a weighted average of $Y$ local to $X=x$ for each population $k=1,2$. 
For a bandwidth $h > 0$ and a kernel density function $K$,
the NW estimator for $\theta(P_k)$ is written as

\vspace{-1cm}
\begin{gather}
\ha{\theta}_{k}^b = \theta_{n_k,n}^b(\bZ_{k,1}, \ldots, \bZ_{k,n_k}) \doteq \frac{\sum\limits_{i=1}^{n_k} K \left(  \frac{ X_{k,i}-x }{ h }  \right)  Y_{k,i} }
{\sum\limits_{i=1}^{n_k} K \left(  \frac{ X_{k,i}-x }{ h }  \right) }.
\label{eq:nw_estimator}
\end{gather}
\vspace{-1cm}

The superscript $b$ indicates that there is bias in the asymptotic distribution of   
$\sqrt{n h} (\ha{\theta}_{k}^b - \theta(P_k))$ 
whenever the bandwidth choice converges to zero at the slowest possible rate, i.e., $h=O(n^{-1/5})$ (Proposition \ref{prop:control_mean}).
This is the case of mean squared error (MSE) optimal bandwidths, and inference requires bias correction in this case.
A conventional solution is to subtract  a first-order bias term $h^2 B(P_k)$ from $\ha{\theta}_{k}^b$, 
where $B(P_k)$ is nonparametrically estimated by $\ha{B}_k = B_{n_k,n}(\bZ_{k,1}, \ldots, \bZ_{k,n_k})$.
We give the analytical formulas for $B(P_k)$ and $\ha{B}_k$ in 
Section D.1 (Equation D.6)
of the supplemental appendix.
Our permutation tests utilize the bias-corrected NW estimator 
$\ha{\theta}_{k} = \theta_{n_k,n}(\bZ_{k,1}, \ldots, \bZ_{k,n_k}) \doteq 
 \theta_{n_k,n}^b(\bZ_{k,1}, \ldots, \bZ_{k,n_k}) - h^2 B_{n_k,n}(\bZ_{k,1}, \ldots, \bZ_{k,n_k})
=\ha{\theta}_{k}^b - h^2 \ha{B}_k$.
Note that no bias correction is needed if $h=o(n^{-1/5})$ because
$\sqrt{n h} (\ha{\theta}_{k}^b - \theta(P_k) - h^2 B(P_k)) 
= \sqrt{n h} (\ha{\theta}_{k}^b - \theta(P_k))  +o(1)$.

An alternative solution to the bias issue consists of replacing the NW estimator $\ha{\theta}_{k}^b$, which is the LPR estimator of order zero, with the LPR estimator of order two.
For LPR estimation at an interior point $x$, if $\ha{\theta}_{k}^b$ is LPR of order $\rho$,  
the asymptotic bias of $\sqrt{n h} (\ha{\theta}_{k}^b - \theta(P_k))$
is $O\left( \sqrt{n h} h^{\rho+2} \right)$ if $\rho$ is even, 
or $O\left( \sqrt{n h} h^{\rho+1} \right)$ if $\rho$ is odd (Theorem 3.1 by \cite{fan1996}).
Thus, if an LPR of order $\rho$ has asymptotic bias, that bias vanishes if we increase the order of the polynomial by two if $\rho$ is even, or by one if $\rho$ is odd.
We keep the NW estimator for a simple and intuitive presentation of our theory in this section, but our permutation tests also apply to LPR estimators.
We demonstrate this in the context of our third application (Section \ref{sec:applications:rdd}) in the supplemental appendix 
(Section D.5).
In the same spirit, our simulations in Section \ref{sec:mc}
and empirical examples in Section \ref{sec:empirical} utilize LPR of order one with MSE-optimal bandwidth and bias-correct it using LPR of order two.
Other options of bias correction include higher-order kernels (\cite{liracine2007}, Section 1.11) and the bootstrap (\cite{racine2001}).

When the distribution of  $(X_1,Y_1)$ equals that of $(X_2,Y_2)$, the permutation test in \eqref{eq:setup:testfun} has exact size in finite samples.
For other cases, we rely on asymptotic size control, which depends on Assumptions \ref{aspt:asy_linear} and \ref{aspt:sampling}.
Below, we describe regularity conditions such as continuous differentiability of conditional moments, and Proposition \ref{prop:control_mean} proves the asymptotic linear representation of the bias-corrected NW estimator.

\begin{assumption}\label{aspt:prop:rate}
As $n\to \infty$, 
$  n_1/n \to \lambda \in  (0,1) $, 
$h \to 0$,  
$n h \to \infty$, and 
$\sqrt{n h} h^2 \to c \in [0,\infty)$.
\end{assumption}

\vspace{-4mm}

\begin{assumption}\label{aspt:prop:pdf}
For $k=1,2$, the distribution of $X_k$ has PDF $f_{X_k}(x_k)$ that is bounded, bounded away from zero, and three times differentiable with bounded derivatives. 
\end{assumption}

\vspace{-4mm}

\begin{assumption}\label{aspt:prop:condmean}
For $k=1,2$, $m_{Y_k|X_k}(x_k) \doteq \Exp[Y_k|X_k=x_k]$  is bounded, three times differentiable, and its derivatives are bounded;
there exists $\zeta>0$ such that $\Exp[|Y_k |^{2+\zeta} | X_k]$ is almost surely bounded.
\end{assumption}

\vspace{-4mm}

\begin{assumption}\label{aspt:prop:condvar}
For $k=1,2$, $v_{Y_k|X_k}(x_k) \doteq \Var[Y_k|X_k=x_k]$ is bounded, differentiable, its derivative is bounded,
and $v_{Y_k|X_k}(x) >0$.
\end{assumption}

\vspace{-4mm}

\begin{assumption}\label{aspt:prop:bias}
Let $\bV_1, \ldots, \bV_m$ be an \textit{iid} sample from a distribution $P \in \m{P}$,
where $m$ grows with $n$.
Let $\ha{B} = B_{m,n}(\bV_{1}, \ldots, \bV_{m})$ be a consistent estimator for the first-order bias term $B(P)$. 
Assume that $\left( \ha{B} - B(P) \right) \pto 0$
uniformly over $P \in \mathcal{P}$ for any  sequence $m$ such that  $ m/n  \to   \lambda$ 
or
$ m/n  \to   1-\lambda $.
\end{assumption}

\vspace{-4mm}
 
\begin{proposition}
\label{prop:control_mean}
Suppose Assumptions \ref{aspt:prop:rate}--\ref{aspt:prop:bias} hold.
Let $\bV_1=(R_1,S_1), \ldots, \bV_m=(R_m,S_m)$ be an \textit{iid} sample from a distribution $P \in \m{P}$,
where $m$ grows with $n$, and consider the bias-corrected estimator 
$\ha{\theta} = \theta_{m,n}(\bV_{1}, \ldots, \bV_{m})$ as described in the text.
Then, Assumptions \ref{aspt:asy_linear} and \ref{aspt:sampling} hold for $\ha{\theta}$ with 
$\psi_{n}(\bV, P)  =  
		K\left( \frac{ R - x }{h_n} \right) 
		\left( \frac{S - m_{S|R}(R;P)}{\sqrt{h_n} f_{R}(x;P) } \right)$ and 
$\xi^2(P)  = 	\frac{ v_{S|R}(x;P)  \kappa_{0,2} }{ f_{R}(x;P) },$
where $m_{S|R}(x;P)$ is the conditional mean of $S$ given $R=x$,
$v_{S|R}(x;P)$ is the conditional variance of $S$ given $R=x$,
and 
$f_{R}(x;P)$ is the PDF of $R$ at $x$, 
all three assuming  $\bV=(R,S) \sim P \in \m{P}$.
\end{proposition}

The proof of Proposition \ref{prop:control_mean} adapts conventional arguments for nonparametric asymptotics (e.g., Theorem 2.2 by \cite{liracine2007})  and is found in 
Section D.1 of the supplement.
Under the null hypothesis, the asymptotic variance of $T_n$ and of the permutation distribution are, respectively, 
$\sigma^2  =  \kappa_{0,2} \left( \frac{ v_{Y_1|X_1}(x)   }{ \lambda f_{X_1}(x) } + \frac{ v_{Y_2|X_2}(x)   }{ (1-\lambda) f_{X_2}(x) }\right)$ and 
$\tau^2  =  \kappa_{0,2} \left( \frac{ f_{X_1}(x) v_{Y_1|X_1}(x)   }{ (1-\lambda) f_{R}^2(x;\ubar{P}) } + \frac{ f_{X_2}(x) v_{Y_2|X_2}(x)   }{ \lambda f_{R}^2(x;\ubar{P} ) } \right),$
where 
$f_{R}^2(x;\ubar{P})$ is the square of the PDF of $R$ evaluated at $x$ under $\ubar{P} = \lambda P_1 + (1-\lambda) P_2 \in \m{P}$.
These variances are generally different, except in special cases, e.g., when $f_{X_1}(x)=f_{X_2}(x)$ and $\lambda=1/2$ 
or when $f_{X_1}(x)=f_{X_2}(x)$ and $v_{Y_1|X_1}(x) = v_{Y_2|X_2}(x)$. 
Thus, in general, the researcher must use the studentized test statistic for the permutation test to have asymptotic size control.

\subsection{Controlled Quantiles}
\label{sec:applications:control_quan}

\indent 

In this subsection, we examine equality of conditional quantile functions for two populations.
For example, a researcher may wish to compare not only averages (Section \ref{sec:applications:control_mean}) but also other features of a conditional distribution between $P_1$ and $P_2$, such as the median, tails, interquartile range, etc.
The goal is to test the difference of the $\chi$-th quantile of the outcomes $Y$ between the two populations, after controlling for a given value of the variable $X$.
For instance, the immune response $Y$ of a certain treatment  conditional on age $X$ may differ for individuals at the bottom, median, or top of the immunity distribution.

As in Section \ref{sec:applications:control_mean}, there are two independent samples, $\bZ_{1,i}=(X_{1,i}, Y_{1,i})$, $i=1, \ldots, n_1$, and $\bZ_{2,i}=(X_{2,i}, Y_{2,i})$,  $i=1, \ldots, n_2$,
and the vector $\lbar{W}_n$ is non-random. 
For a given interior point $x$,
the parameter of interest is the $\chi$-th conditional quantile, that is, 
$\theta(P_k) = \mmq_{\chi}[Y_{k,i} | X_{k,i}=x] = \arg\min_a \mme[\rho_{\chi}(Y_{k,i}-a)|X_{k,i}=x]~,$
where $\rho_{\chi}(u) = \left(\chi - \mmi(u < 0)\right)u$.

For a bandwidth $h > 0$ and a  kernel density function $K$,
a consistent estimator of the NW style is
$\ha{\theta}_{k}^b 
=\theta_{n_k,n}^b(\bZ_{k,1}, \ldots, \bZ_{k,n_k})
\doteq 
\argmin_a \sum\limits_{i=1}^{n_k} \rho_{\chi}(Y_{k,i}- a)K \left(  \frac{ X_{k,i}-x }{ h }  \right) 
\text{, for }k=1,2.$

The superscript $b$ indicates that there is bias in the asymptotic distribution of   
$\sqrt{n h} (\ha{\theta}_{k}^b - \theta(P_k))$ 
whenever the bandwidth choice converges to zero at the slowest possible rate, i.e., $h=O(n^{-1/5})$ (Proposition \ref{prop:control_quan}).
This is the case of mean squared error (MSE) optimal bandwidths, and inference requires bias correction in this case.
A conventional solution is to subtract  a first-order bias term $h^2 B(P_k)$ from $\ha{\theta}_{k}^b$, 
where $B(P_k)$ is nonparametrically estimated by $\ha{B}_k = B_{n_k,n}(\bZ_{k,1}, \ldots, \bZ_{k,n_k})$.
We give the analytical formulas for $B(P_k)$ and $\ha{B}_k$ in 
Section D.2 (Equation D.14) 
of the supplemental appendix.
Our permutation tests utilize the bias-corrected estimator 
$\ha{\theta}_{k} = \theta_{n_k,n}(\bZ_{k,1}, \ldots, \bZ_{k,n_k}) \doteq 
 \theta_{n_k,n}^b(\bZ_{k,1}, \ldots, \bZ_{k,n_k}) - h^2 B_{n_k,n}(\bZ_{k,1}, \ldots, \bZ_{k,n_k})
=\ha{\theta}_{k}^b - h^2 \ha{B}_k$.
Note that no bias correction is needed if $h=o(n^{-1/5})$ because
$\sqrt{n h} (\ha{\theta}_{k}^b - \theta(P_k) - h^2 B(P_k)) 
= \sqrt{n h} (\ha{\theta}_{k}^b - \theta(P_k))  +o(1)$.
An alternative solution to the bias issue consists of replacing $\ha{\theta}_{k}^b$ with the LPR estimator, while keeping the MSE-optimal bandwidth choice for the NW estimator.
See 
Sections \ref{sec:applications:control_mean}
and
\ref{sec:applications:rdd} for details.

Proposition \ref{prop:control_quan} below demonstrates validity of Assumptions \ref{aspt:asy_linear}--\ref{aspt:sampling} for the NW quantile regression estimator. It relies on some of the same assumptions as Section \ref{sec:applications:control_mean} plus a couple of assumptions on the distribution of $Y_k$ conditional on $X_k$, $k=1,2$.

\begin{assumptionp}{\ref{aspt:prop:condmean}'}\label{aspt:prop:condpdf}
For $k=1,2$, the distribution of $Y_k$ conditional on $X_k$ has PDF $f_{Y_k | X_k}(y_k | x_k)$ that is a bounded and differentiable function of $(x_k,y_k)$,
has bounded partial derivatives with respect to $x_k$ and $y_k$, and is bounded away from zero over $y_k$ for $x_k=x$.
\end{assumptionp}

\vspace{-4mm}

\begin{assumptionp}{\ref{aspt:prop:condvar}'}\label{aspt:prop:condcdf}
For $k=1,2$, the distribution of $Y_k$ conditional on $X_k$ has CDF $F_{Y_k | X_k}(y_k | x_k)$ that is three times partially differentiable with respect to $x_k$ with bounded partial derivatives.
\end{assumptionp}

\vspace{-4mm}

\begin{proposition}
\label{prop:control_quan}
Suppose Assumptions \ref{aspt:prop:rate}, \ref{aspt:prop:pdf}, \ref{aspt:prop:condpdf}, \ref{aspt:prop:condcdf}, and \ref{aspt:prop:bias} hold.
Let $\bV_1=(R_1,S_1), \ldots, \bV_m=(R_m,S_m)$ be an \textit{iid} sample from a distribution $P \in \m{P}$,
where $m$ grows with $n$,
and consider the bias-corrected estimator 
$\ha{\theta} = \theta_{m,n}(\bV_{1}, \ldots, \bV_{m})$ as described in the text.
Then, $\ha{\theta}$ satisfies Assumptions \ref{aspt:asy_linear} and \ref{aspt:sampling} with 
$\psi_{n}(\bV, P)
 = -K\left(\frac{R - x }{ h_n }\right) 
 \frac{\left( \mmi\{S  < \theta(P) \} - F_{S|R}(\theta(P)| R ;P) \right)}
 {f_{S|R}(\theta(P) | x ;  P) f_R(x;P) \sqrt{h_n} }
 $
and 
$\xi^2(P) = 
   \frac{\kappa_{0,2}\chi (1-\chi) }{f_{S|R}^2(\theta(P) | x ;  P)  f_R(x;P)}, $
where 
$f_{S|R}(s | r ;  P)$ is the conditional PDF of $S$ given $R=r$,
$f_R(r;P)$ is the PDF of $R$, and
$F_{S|R}(s| r ;P) $ is the conditional CDF of $S$ given $R=r$,
all three assuming $\bV=(R,S) \sim P \in \m{P}$.
\end{proposition}

The proof of Proposition \ref{prop:control_quan} adapts arguments by \cite{pollard1991}, \cite{chaudhuri1991}, and \cite{fan1994} and is found in 
Section D.2
of the supplement.
The asymptotic variance of $T_n$ and of the permutation distribution are, respectively, 
$\sigma^2  =  
\frac{\kappa_{0,2} \chi (1- \chi)    }{ \lambda f_{Y_1 | X_1}^2(\theta(P_1)|x) f_{X_1}(x) } 
+ \frac{ \kappa_{0,2} \chi (1- \chi)    }{ (1-\lambda) f_{Y_2|X_2}^2(\theta(P_2)|x ) f_{X_2}(x ) }$
and 
$\tau^2  = \frac{ \kappa_{0,2} \chi (1- \chi)  }{ \lambda (1-\lambda) f_{S|R}^2(\theta(\bar P) | x; \bar P) f_{R}(x; \bar P) }.$
These variances are generally different except in special cases.
For example, the null hypothesis implies $\theta(P_1)=\theta(P_2) = \theta(\bar{P}) $.
If $f_{Y_1|X_1}(\theta(P_1)|x) = f_{Y_2|X_2}(\theta(P_2)|x)$ and $f_{X_1}(x)=f_{X_2}(x)$,
then $\sigma^2 = \tau^2$.
Thus, in general, the researcher must use the studentized test statistic for the permutation test to have asymptotic size control.


\subsection{Discontinuity of Conditional Mean}
\label{sec:applications:rdd}

\indent

Numerous empirical studies in the social sciences have relied on estimation  and inference on the size of a discontinuity in a conditional mean function at a certain point.
In the so-called regression discontinuity design (RDD), an individual $i$ receives treatment  if and only if a running variable $X_i$ is above a fixed threshold.
If individuals do not know the threshold or do not have perfect manipulation over $X$, untreated individuals who barely missed the cutoff serve
as a control group for treated individuals who barely made it across the cutoff.
Assume the threshold for treatment is $0$ without loss of generality.
The difference in side limits $\mme[Y | X=0^+] - \mme[Y | X=0^-]$ identifies the causal effect of treatment on an outcome variable $Y$.
Thus, the null hypothesis of zero causal effect is equivalent to continuity of the conditional mean function $\mme[Y | X=x]$ at $x=0$.

This idea first appeared in psychology \citep{thistlethwaite1960}, made its way to economics \citep{hahn2001id}, 
and has a growing number of applications in the social sciences.
Examples include
\cite{agarwal2016} in economics,
\cite{valentine2017} in education, 
\cite{abou2020} in political science,
and 
\cite{zoorob2020} in sociology.
We focus on the conditional mean function in this subsection, but the whole argument goes through if one desires to use the conditional quantile function.

The researcher has a sample of $n$ \textit{iid} observations $(X_i,Y_i)$ and the NW estimator for the discontinuity is

\vspace{-1.5cm}
\begin{align*}
	\frac{ \sum_{i=1}^{n} K\left( \frac{X_i}{h} \right) \mmi\{X_i \geq 0 \} Y_{i}  }{  \sum_{i=1}^{n} K\left( \frac{X_i}{h} \right) \mmi\{X_i \geq 0 \} }
	-
	\frac{ \sum_{i=1}^{n} K\left( \frac{X_i}{h} \right) \mmi\{X_i <0  \} Y_{i} }{ \sum_{i=1}^{n}  K\left( \frac{X_i}{h} \right) \mmi\{X_i < 0 \} }.
\end{align*}

\vspace{-3mm}

RDD fits in our two-population framework by splitting the sample based on $X$ being  above or below the cutoff.
Construct $W_{i}  = 2 - \mmi\{X_i \geq 0\}$ so that $n_1 / n \pto \lambda = \mmp[W_i=1]$.
Re-order the sample such that the $n_1$ observations with $W_i=1$ come first, and the $n_2$ observations with $W_i=2$ come second.
Define $\bZ_{1,i} = (X_{1,i}, Y_{1,i}) = (X_i, Y_i)$ for $i=1,\ldots, n_1$
and 
$\bZ_{2,i} = (X_{2,i}, Y_{2,i}) = (-X_{n_1+i}, Y_{n_1+i})$ for $i=1,\ldots, n_2$.
We have $\lbar{\bZ}_n = ( \bZ_{1,1}, \ldots, \bZ_{1,n_1}, \bZ_{2,1}, \ldots, \bZ_{2,n_2})$.
Conditional on $\lbar{W}_n$, the distribution of $Z_{1,i}$ is $P_1$, which equals the distribution of $(X,Y)$ conditional on $X \geq 0$.
Likewise for $Z_{2,i}$, $P_2$ is the distribution of $(-X,Y)$ conditional on $X < 0$.
Permutations re-order observations in $\lbar{\bZ}_n$ but keep $\lbar{W}_n$ unchanged.
The RDD parameter becomes $\theta(P_1)-\theta(P_2)$, where $\theta(P_k) = \mme[Y_k|X_k=0^+]$ for $k=1,2$.

The NW estimator for the RDD parameter is $\ha\theta_1^b - \ha\theta_2^b$, where $\ha \theta_k^b$ is defined in Equation \ref{eq:nw_estimator} for $k=1,2$ and evaluation point $x$ set to zero.
The superscript $b$ indicates that there is bias in the asymptotic distribution of   
$\sqrt{n h} (\ha{\theta}_{k}^b - \theta(P_k))$ 
whenever the bandwidth choice converges to zero at the slowest possible rate, i.e., $h=O(n^{-1/3})$
(Proposition \ref{prop:rdd}).
This is the case of MSE-optimal bandwidths, and inference requires bias correction in this case.
A conventional solution is to subtract  a first-order bias term $h B(P_k)$ from $\ha{\theta}_{k}^b$, 
where $B(P_k)$ is nonparametrically estimated by $\ha{B}_k = B_{n_k,n}(\bZ_{k,1}, \ldots, \bZ_{k,n_k})$.
We give the analytical formulas for $B(P_k)$ and $\ha{B}_k$ in 
Section D.3 (Equation D.21) 
of the supplemental appendix.
Our permutation tests utilize the bias-corrected NW estimator 
$\ha{\theta}_{k} = \theta_{n_k,n}(\bZ_{k,1}, \ldots, \bZ_{k,n_k}) \doteq 
 \theta_{n_k,n}^b(\bZ_{k,1}, \ldots, \bZ_{k,n_k}) - h  B_{n_k,n}(\bZ_{k,1}, \ldots, \bZ_{k,n_k})
=\ha{\theta}_{k}^b - h  \ha{B}_k$.
Note that no bias correction is needed if $h=o(n^{-1/3})$ because
$\sqrt{n h} (\ha{\theta}_{k}^b - \theta(P_k) - h  B(P_k)) 
= \sqrt{n h} (\ha{\theta}_{k}^b - \theta(P_k))  +o(1)$.

An alternative solution to the bias issue consists of replacing the NW estimator $\ha{\theta}_{k}^b$, which is the LPR estimator of order zero, with the  LPR estimator of order one, while keeping the MSE-optimal bandwidth choice for the NW estimator.
For LPR estimation at a boundary point, if $\ha{\theta}_{k}^b$ is LPR of order $\rho$,  
the asymptotic bias of $\sqrt{n h} (\ha{\theta}_{k}^b - \theta(P_k))$
is $O\left( \sqrt{n h} h^{\rho} \right)$ (Theorem 3.2 by \cite{fan1996}).
Thus, if an LPR of order $\rho$ has asymptotic bias, that bias vanishes if we increase the order of the polynomial by one. 
We keep the NW estimator for a simple and intuitive presentation of our theory in this section, but we demonstrate that our permutation tests also apply to LPR estimators of any order in 
Section D.5
of the supplement.
Section \ref{sec:empirical} illustrates our procedures with two empirical examples of RDD using LPR estimators and MSE-optimal bandwidths, which is standard practice in applied literature.
A third option for bias correction is the method proposed by \cite{armstrong2018}.

In case the distribution of  $(Y,X)$ equals that of $(Y,-X)$, then $P_1=P_2$ and the permutation test in \eqref{eq:setup:testfun} has exact size in finite samples.
Note that this X-symmetry restriction eliminates the impossibility problem in RDD tests (\cite{kamat2018} and \cite{bertanha2019}) because there is no bias in estimation. 
For other cases, we rely on asymptotic size control, which depends on Assumptions \ref{aspt:asy_linear} and \ref{aspt:sampling}.
Proposition \ref{prop:rdd} below builds on assumptions similar to those of Section \ref{sec:applications:control_mean}, which we rewrite below 
 in terms of the originally sampled variables $(X,Y)$, that is, before the variables are transformed into $(X_k,Y_k)$, $k=1,2$.

\begin{assumptionp}{\ref{aspt:prop:rate}$^*$}\label{aspt:prop:rate2}
As $n\to \infty$, $h \to 0$,  $n h \to \infty$, and $\sqrt{n h} h \to c \in [0,\infty)$.
\end{assumptionp}

\vspace{-4mm}

\begin{assumptionp}{\ref{aspt:prop:pdf}$^*$}\label{aspt:prop:pdf2}
The distribution of $X$ has PDF $f_{X}(x)$ that is bounded, bounded away from zero,
twice differentiable except at $x=0$, and has bounded derivatives.
\end{assumptionp}

\vspace{-4mm}

\begin{assumptionp}{\ref{aspt:prop:condmean}$^*$}\label{aspt:prop:condmean2}
$\Exp[Y | X = x]$ is bounded, twice differentiable except at $x=0$, and has bounded derivatives.
There exists $\zeta>0$ such that $\Exp[|Y|^{2+\zeta} | X]$ is almost surely bounded.
\end{assumptionp}

\vspace{-4mm}

\begin{assumptionp}{\ref{aspt:prop:condvar}$^*$}\label{aspt:prop:condvar2}
$\Var[Y | X = x]$ is bounded, differentiable except at $x=0$, has bounded derivative,
$\Var[Y|X=0^+]>0$, and $\Var[Y|X=0^-]>0$, where $0^+$ and $0^-$ denote side limits as $x\to 0$.
\end{assumptionp}

\vspace{-4mm}

\begin{proposition}
\label{prop:rdd}
Suppose Assumptions \ref{aspt:prop:rate2}, \ref{aspt:prop:pdf2}, \ref{aspt:prop:condmean2}, \ref{aspt:prop:condvar2}, and \ref{aspt:prop:bias} hold.
Let $\bV_1=(R_1,S_1), \ldots, \bV_m=(R_m,S_m)$ be an \textit{iid} sample from a distribution $P \in \m{P}$,
where $m$ grows with $n$, and consider the bias-corrected estimator 
$\ha{\theta} = \theta_{m,n}(\bV_{1}, \ldots, \bV_{m})$ as described in the text.
Then, $\ha{\theta}$ satisfies Assumptions \ref{aspt:asy_linear} and \ref{aspt:sampling}  with
$\psi_{n}(\bV, P)  =  
		K\left( \frac{ R  }{h_n} \right) 
		\left( \frac{S - m_{S|R}(R;P)}{\sqrt{h_n} f_{R}(0^+;P)/2 } \right)$
and 
$\xi^2(P)  = 	\frac{ v_{S|R}(0^+;P) \kappa_{0,2}^+ }{ f_{R}(0^+;P)/4 },$
where 
$m_{S|R}(r;P)$ is the conditional mean of $S$ given $R=r$,
$v_{S|R}(r;P)$ is the conditional variance of $S$ given $R=r$,
and 
$f_{R}(r;P)$ is the PDF of $R$ at $r$, 
all three assuming $\bV=(R,S) \sim P \in \m{P}$.
\end{proposition}

The proof is in 
Section D.3
of the supplement.
The conditions and proof of Proposition \ref{prop:rdd} follow along the lines of the conditional mean case (Proposition \ref{prop:control_mean}).
Unlike Section \ref{sec:applications:control_mean}, the evaluation point $x$ lies at the boundary of the support of $X_k$.
As a result, $h$ has to converge faster to zero to bound the
asymptotic bias of $T_n$ 
(i.e., $\sqrt{nh}h \to c$ in Assumption \ref{aspt:prop:rate2}  
vs. $\sqrt{nh}h^2 \to c$ in Assumption \ref{aspt:prop:rate}).
Proposition \ref{prop:rdd} is extended to LPR estimators of any order $\rho$ in 
Section D.5 
of the supplement.
Common practice in RDD uses local-linear regression ($\rho=1$) with a MSE-optimal bandwidth, and the resulting asymptotic bias is often corrected by using a local-quadratic regression ($\rho=2$).

If agents do not manipulate $X$ to change their treatment status, which is a key assumption in RDD, then the PDF of $X$ should be continuous at the cutoff. 
This implies that $f_X(0^+) = f_{X_1}(0^+)\lambda = f_{X_2}(0^+)(1-\lambda) = f_X(0^-)=f_X(0)$.
In this case, under the null hypothesis, the asymptotic variance of $T_n$ and of the permutation distribution are, respectively, 
$\sigma^2  =  \frac{4 \kappa_{0,2}^+}{f_X(0)} (  v_{Y_1|X_1}(0^+)    + v_{Y_2|X_2}(0^+)    )$ and 
$\tau^2  = \frac{ \kappa_{0,2}^+ }{ \lambda (1-\lambda) f_X(0)}  \left(  v_{Y_1|X_1}(0^+)    + v_{Y_2|X_2}(0^+)    \right).$
These are generally different, except when $\lambda = 1/2$.
Thus, in general, the researcher must use the studentized test statistic for the permutation test to have asymptotic size control.

\subsection{Discontinuity of Density}
\label{sec:applications:bunching}

\indent 

In many settings, the distribution of a random variable may exhibit a discontinuity at a given point if a certain phenomenon of interest occurs. 
For example, estimating agents' responses to incentives is a central objective in the social sciences.
A continuous distribution of agents who face a discontinuous schedule of incentives results in a distribution of responses with a discontinuity at a known point.
For example, \cite{saez2010} looks for evidence of a mass point in the distribution of reported income at tax brackets as evidence of agents' responses to tax rates.
\cite{caetano2015} proposes an exogeneity test in nonparametric regression models, where the distribution of the potentially endogenous regressor may have a mass point.
Identification of causal effects with RDD depends heavily on continuity assumptions, and these imply that the PDF of the control variable is continuous at the cutoff.

Consider a scalar random variable $X$ with PDF $f$ that is continuous, except for point $x=0$.
We want to test the null hypothesis of continuity of the 
PDF at $x=0$.
For a sample with $n$ \textit{iid} observations $X_i$, a kernel density estimator for the size of the discontinuity is

\vspace{-1cm}
\begin{align*}
	\frac{2}{nh}  \sum_{i=1}^{n} K\left( \frac{X_i}{h} \right) 
	\left[ \mmi\{X_i \geq 0 \} - \mmi\{X_i < 0\} \right].
\end{align*}
\vspace{-1cm}

The problem fits in our two-population framework by randomly splitting the sample as follows.
Make $n_1 \doteq  \lfloor n/2 \rfloor $ and $n_2 \doteq n-n_1$.
For observations $1 \leq i  \leq n_1$, set $W_i=1$ and let $\bZ_{1,i} = X_{1,i} = X_i$;
for observations $n_1 < i  \leq n$, set $W_i=2$ and let 
$\bZ_{2,i-n_1} = X_{2,i-n_1} = -X_{  i}$.
This implies that $n_1 / n \to 1/2.$
Permutations re-order observations in $\lbar{\bZ}_n = (X_{1,1},\ldots, X_{1,n_1}, X_{2,1},\ldots, X_{2,n_2})$ 
but keep $\lbar{W}_n$ unchanged.
Conditional on $\lbar{W}_n$, the distribution of $X_{1,i}$ is $P_1$, which equals the distribution of $X$.
Likewise for $X_{2,i}$, $P_2$ is the distribution of $-X$.\footnote{We cannot split the sample based on $X$ being above or below $0$ as we do in Section \ref{sec:applications:rdd}.
If we split the sample based on $X$ and the distribution of $X$ is asymmetric, it becomes impossible to identify the side limit of $f$ at 0 using only data from either sample, as required by Assumption \ref{aspt:asy_linear}.} 

Let $\bV$ be a scalar variable $ R \sim P \in \m{P}$.
The parameter of interest is
$\theta(P) = 1/2 (f_{R}(0^+;P) - f_{R}(0^-;P))$,
where $f_{R}(0^+;P)$ and $f_{R}(0^-;P)$
are the side-limits at zero of the PDF of $R$ under the $P$ distribution.
The discontinuity parameter is $\theta(P_1)-\theta(P_2)$, which equals $f_X(0^+)-f_X(0^-)$ in terms of the PDF of $X$.
The kernel estimator for the density discontinuity becomes $\ha\theta_1^b - \ha\theta_2^b$, where $\ha \theta_k^b$ is defined as 
$\ha{\theta}_{k}^b = \frac{1}{ n_k h} \sum\limits_{i=1}^{n_k} K \left(  \frac{ X_{k,i} }{ h }  \right)  \left(\mmi\{X_{k,i} \ge 0\} - \mmi\{X_{k,i} < 0\}  \right)
\text{, for }k=1,2.$

The $b$ superscript denotes asymptotic bias in the limiting distribution
of $\sqrt{n_k h}  ( \ha{\theta}_{k}^b - \theta(P_k) )$ whenever the bandwidth choice converges to zero at the slowest possible rate, i.e., $h=O(n^{-1/3})$
(Proposition \ref{prop:bunching}).
This is the case of MSE-optimal bandwidths, and inference requires bias correction in this case.
A conventional solution is to subtract  a first-order bias term $h B(P_k)$ from $\ha{\theta}_{k}^b$, 
where $B(P_k)$ is nonparametrically estimated by $\ha{B}_k = B_{n_k,n}(\bZ_{k,1}, \ldots, \bZ_{k,n_k})$.
We give the analytical formulas for $B(P_k)$ and $\ha{B}_k$ in 
Section D.4 (Equation D.24) 
of the supplemental appendix.
Our permutation tests utilize the bias-corrected NW estimator 
$\ha{\theta}_{k} = \theta_{n_k,n}(\bZ_{k,1}, \ldots, \bZ_{k,n_k}) \doteq 
 \theta_{n_k,n}^b(\bZ_{k,1}, \ldots, \bZ_{k,n_k}) - h  B_{n_k,n}(\bZ_{k,1}, \ldots, \bZ_{k,n_k})
=\ha{\theta}_{k}^b - h  \ha{B}_k$.
Note that no bias correction is needed if $h=o(n^{-1/3})$ because
$\sqrt{n h} (\ha{\theta}_{k}^b - \theta(P_k) - h  B(P_k)) 
= \sqrt{n h} (\ha{\theta}_{k}^b - \theta(P_k))  +o(1)$.
Alternative solutions to the bias issue include using transformed-kernel estimators 
(\cite{marron1994}) or LPR density estimators (\cite{cattaneo2020}).

It is important to emphasize that if the sample size $n$ is even and split in half, then the test statistic $T_n=\sqrt{nh}\left(\ha \theta_1 - \ha \theta_2 \right)$ is invariant to the way the original sample is split.
In case the distribution of  $X$ is symmetric at $0$, then $P_1=P_2$ and the permutation test in \eqref{eq:setup:testfun} has exact size in finite samples.
For other cases, we rely on asymptotic size control, and thus need to verify Assumptions \ref{aspt:asy_linear} and \ref{aspt:sampling}.

\begin{proposition}
\label{prop:bunching}
Suppose Assumptions \ref{aspt:prop:rate2}, \ref{aspt:prop:pdf2},  and \ref{aspt:prop:bias} hold.
Let $\bV_1=(R_1,S_1), \ldots, \bV_m=(R_m,S_m)$ be an \textit{iid} sample from a distribution $P \in \m{P}$,
where $m$ grows with $n$, and consider the bias-corrected estimator 
$\ha{\theta} = \theta_{m,n}(\bV_{1}, \ldots, \bV_{m})$ as described in the text.
Then, 
Assumptions \ref{aspt:asy_linear} and \ref{aspt:sampling} hold for $\ha{\theta}$
with 
$\psi_{n}(\bV, P)  =  h_n^{-1/2} K\left( \frac{ R }{h_n} \right) \left(\mmi\{ R \ge 0\} - \mmi\{R < 0\}  \right)
- h_n^{-1/2} \mme_P[ K\left( \frac{ R }{h} \right)(\mmi\{ R \ge 0\} - \mmi\{ R < 0\}  )]$
and
$\xi^2(P)  = \kappa_{0,2}^+\left( f_{R}(0^+;P)  + f_{R}(0^-;P) \right)  $.
\end{proposition}

The proof of Proposition \ref{prop:bunching} is in 
Section D.4 
in the supplement.
The asymptotic variance of $T_n$ and of the permutation distribution are, respectively,

\vspace{-1.5cm}
\begin{align*}
\sigma^2  & = 2 \kappa_{0,2}^+ \left[ f_{X_1}(0^+) + f_{X_1}(0^-) + f_{X_2}(0^+) +  f_{X_2}(0^-)   \right]
\\
& = 4 \kappa_{0,2}^+  \left[ f_{R}(0^+;\bar P) + f_{R}(0^-;\bar P) \right] = \tau^2.
\end{align*}
\vspace{-1.5cm}

These are the same regardless if the null hypothesis is true or not.
Thus, unlike the previous examples, we do not need to studentize the test 
statistic 
for asymptotic validity of the permutation test.

\vspace{-6mm}

\section{Confidence Sets}
\label{sec:setup:cset}

\vspace{-2mm}

\indent

This section constructs  robust confidence sets for a discrepancy measure $\Psi (P_1,P_2)$ between the two populations by ``inverting'' the permutation test for the null hypothesis $\Psi (P_1,P_2) = \delta$, $\delta \in \mmr$.
The discrepancy measure satisfies two requirements.
First, there exists a unique $\delta_0 \in \mmr$ such that $\Psi (P_1,P_2)=\delta_0$
is equivalent to $\theta(P_1) = \theta(P_2)$.
Second, for every $\delta \neq \delta_0$, there exists a known data transformation $\psi_{\delta}$ that applies to observations from the first sample such that 
the distribution $\ti P_1$ of $\psi_{\delta}(\bZ_{1,i})$ satisfies $\theta(\ti P_1)=\theta(P_2)$.

In terms of the examples of Sections \ref{sec:applications:control_mean}--\ref{sec:applications:rdd}, we set 
$\Psi(P_1,P_2) = \theta(P_1) - \theta(P_2)$, and it follows that $\delta_0=0$ and
$\psi_{\delta}(\bZ_{1,i}) = (X_{1,i}, Y_{1,i} - \delta)$ satisfy the two requirements for the discrepancy measure.
Note that the null hypothesis $\Psi(P_1,P_2) = \delta$
is equivalent to  $\mme[Y_1|X_1=x] - \mme[Y_2|X_2=x] = \delta$ in the conditional mean case,
to
$\mmq_{\chi}[Y_1|X_1=x] - \mmq_{\chi}[Y_2|X_2=x] = \delta$ in the conditional quantile case,
and to
$\mme[Y|X=0^+] - \mme[Y|X=0^-] = \delta$ in the discontinuity of conditional mean case.
For the discontinuity of PDF example of Section \ref{sec:applications:bunching}, we make
$\Psi(P_1,P_2) = f_{X_1}(0^+) /  f_{X_2}(0^+)$, and we have that $\delta_0=1$
and
$\psi_{\delta}(\bZ_{1,i}) = X_{1,i} \left( \delta \mmi\{X_{1,i} \geq 0 \}  + (1/\delta) \mmi\{X_{1,i} < 0 \} \right)$
satisfy the two requirements.
In this example, $\Psi(P_1,P_2) =\delta$ is equivalent to $f_{X}(0^+) / f_{X}(0^-) = \delta$.

Define $\phi_{\delta_0}(\lbar{W}_n, \lbar{\bZ}_n)$ to be the test described in Equation \ref{eq:setup:testfun} with the studentized test statistic $S_n$ of Equation \ref{eq:setup:teststat:stu}
replacing $T_n$.
This test applies to the null hypothesis $\Psi(P_1,P_2)=\delta_0$. 
Next, for $\delta \neq \delta_0$, we first transform the data  
$\lbar{\bZ}_n = (\bZ_{1,1}, \ldots, \bZ_{1,n_1}, \bZ_{2,1}, \ldots, \bZ_{2,n_2})$ 
to $\ti{ \lbar{\bZ}}_n = (\ti{\bZ}_{1,1}, \ldots, \ti{\bZ}_{1,n_1}, \bZ_{2,1}, \ldots, \bZ_{2,n_2})$, where $\ti \bZ_{1,i} = \psi_{\delta}(\bZ_{1,i})$ for $i=1, \ldots, n_1$.
The robust permutation test for
the null hypothesis 
$\Psi(P_1,P_2)=\delta$
is defined as 
$\phi_{\delta}(\lbar{W}_n,  \lbar{\bZ}_n) = \phi_{\delta_0}(\lbar{W}_n,  \ti{ \lbar{\bZ}}_n)$.

Let $U$ be a uniform random variable in $[0,1]$ and independent of the data. 
The confidence set  with $1-\alpha$ nominal coverage is
$C(\lbar{W}_n, \lbar{\bZ}_n)   \doteq \{\delta  : U >  \phi_\delta(\lbar{W}_n, \lbar{\bZ}_n) \}.$
\label{eq:setup:cset}
The set almost-surely includes all values of $\delta$ for which the test fails to reject, and it  excludes the ones the test rejects.
For those values of $\delta$ for which the test outcome is randomized with rejection probability $a$, the inclusion in the confidence set occurs with probability $1-a$.
The purpose of a randomized confidence set is to guarantee exact coverage whenever the test $\phi_\delta$ has exact size.
A non-randomized confidence set is  $\ti C(\lbar{W}_n, \lbar{\bZ}_n)   \doteq \{\delta \in \mmr: \phi_\delta(\lbar{W}_n, \lbar{\bZ}_n)  < 1 \}$, but its coverage is conservative, especially in small samples.

Lemma \ref{lemma:exact} implies that
$\phi_\delta(\lbar{W}_n,  \lbar{\bZ}_n)$ has exact size $\alpha$ in finite samples
for any $P_1$ and $P_2$ such that $\Psi(P_1,P_2) = \delta$ and $\ti{P}_1 = \psi_\delta P_1 = P_2$.
This implies that the confidence set $C(\lbar{W}_n, \lbar{\bZ}_n)$ has exact coverage in finite samples if distributions are equal up to a transformation $\psi_{\delta}$.
In the examples of Sections \ref{sec:applications:control_mean}--\ref{sec:applications:control_quan},
exactness occurs when the distributions of $P_1$ and $P_2$
are such that there exists $\delta \in \mmr$ for which the distribution of $(Y_{1,i}-\delta,X_{1,i})$ equals that of $ (Y_{2,i} ,X_{2,i})$;
in Section \ref{sec:applications:rdd}, when there exists $\delta \in \mmr$ for which the distribution of $(Y-\delta,X)|X \geq 0$ equals that of $(Y , - X) | X<0$;
and
in Section \ref{sec:applications:bunching}, when there exists $\delta \in (0,\infty)$ such that
$\mmi\{x \geq 0\} f_X(x/\delta)/ \delta + \mmi\{x < 0\} f_X(x \delta) \delta$ is symmetric around $x=0$.
If these restrictions do not apply, then the confidence set has correct coverage asymptotically.

\begin{corollary}\label{coro2}
Consider the discrepancy measure $\Psi$ and class of data transformations $\psi_\delta$ discussed above.
For any $n$, $\bQ_n$, $P_1$, and $P_2$,
if $\ti{P}_1 = \psi_\delta P_1 = P_2$ for $\delta = \Psi(  P_1 , P_2 )$,
then \\
$\mmp \left[ \Psi(  P_1 , P_2 ) \in C(\lbar{W}_n, \lbar{\bZ}_n) \right]  = 1 -\alpha.$
Assume instead that $\ti{P}_1 = \psi_\delta P_1 \neq  P_2$ and Assumptions \ref{aspt:asy_linear}--\ref{aspt:var:est} hold. 
Then, as $n \to \infty$,
$\mmp \left[ \Psi( P_1 ,   P_2 ) \in C(\lbar{W}_n, \lbar{\bZ}_n) \right]  \to 1-\alpha.$\footnote{The data transformation $\psi_\delta$ is used to obtain finite sample exactness when distributions are equal up to a transformation $\psi_\delta$, but it is not necessary for correct asymptotic coverage.
In fact, for the null hypothesis $\theta(P_1)-\theta(P_2) =\delta$, one may construct a permutation test 
$ \phi^*_\delta$ that compares the value of $S_n - \sqrt{nh}\delta/\ha{\sigma}_n$
with the critical values from $\ha{R}_{S_n}$. 
The test $ \phi^*_\delta$ has correct size asymptotically,
and the confidence set constructed by inverting $ \phi^*_\delta$ has correct coverage in large samples.
}
\end{corollary}

\vspace{-6mm}

\section{Monte Carlo Simulations}
\label{sec:mc}

\vspace{-2mm}

\indent 

We conduct Monte Carlo simulations to compare the finite sample performance of our permutation test to the conventional $t$-test, that is, the test that rejects the null if $|S_n| > 1 - \Phi(1-\alpha/2)$. 
The goal is not to show that the permutation test dominates the $t$-test in all cases;
instead,
the goal is to verify the theoretical predictions of Section \ref{sec:theory} and explore DGP variations that illustrate pros and cons of our methods.
The exercise confirms the theoretical findings of size control in large samples and in finite samples under the sharp null; it also shows similar power curves between permutation and $t$-tests in large samples.
Moreover, we find several cases where the permutation test performs significantly better than the $t$-test, both in power and in size control outside of the sharp null. 
We also compare our permutation test to two other popular resampling procedures, namely, the bootstrap and subsample, and the permutation test compares favorably to them. 

The basic setup of our DGP is as follows.
For $k=1,2$, $X_k \sim U[0,1]$,
$\eps_k \sim N(0,\sigma_k^2)$, where $X_k$ is independent of $\eps_k$, 
$Y_k = m_{y,k}(X_k) + \eps_k$,
and the variances are $(\sigma_1^2,\sigma_2^2) \in \{ (1,1),  (1,5) \}$.
The experiments simulate  \textit{iid} samples from the following designs:

\textbf{Design 1: } $m_{y,k}(x) = g_k(x)$, where
$g_{1}(x) \doteq 5(x-0.2)(x-0.8) \mmi\{|x-0.5|>0.3\}$
and
$g_{2}(x) \doteq  -15 (x-0.2)(x-0.8)  \mmi\{ |x-0.5|>0.3 \} $;
 the sample sizes are $(n_1,n_2) \in \{(100,1900), (250,4750), (500,9500), (40,1960), (100,4900), (200, 9800)\}$;
 and
the null hypothesis is $H_0: \theta(P_1) = m_{y,1}(0.5) = m_{y,2}(0.5) = \theta(P_2)$;

 \textbf{Design 2: } 
$m_{y,k}(x)=1+ g_k(x)$ for $g_k(x)$ of Design 1;
$D_k = m_{d,k}(X_k) + \eps_{d,k}$ for mutually independent $(X_k, \eps_k, \eps_{d,k})$ and $\eps_{d,k} \sim N(0,1)$;
$m_{d,k}(x) =  \mu \in \{1,10\}$;
 the sample sizes are $(n_1,n_2) \in \{ (75,75), (150,150), (1000, 1000)\}$;
and the null hypothesis is $H_0: \theta(P_1) = m_{y,1}(0.5)/m_{d,1}(0.5) = m_{y,2}(0.5)/m_{d,2}(0.5) = \theta(P_2)$.

Design 1 is an example of the controlled means case studied in Section \ref{sec:applications:control_mean}, and 
Design 2 is a variation of that case.
Design 2 has controlled means of two different variables with the ratio of means being of interest.
The asymptotic behavior of the ratio estimator can be obtained via the Delta method, and Assumptions \ref{aspt:asy_linear} and \ref{aspt:sampling} can be verified using arguments similar to those in Section \ref{sec:applications:control_mean}.

These designs represent practical situations in which the $t$-test is known to perform poorly in small samples.
Design 1 corresponds to cases of sample imbalance, that is, cases where the sample sizes are very different.  
For example, a researcher has a much larger sample of women than men and is interested in comparing the average outcome from a professional training between men and women, after conditioning on a test score.
Design 2 encompasses scenarios where the ratio of conditional mean functions is of interest, and the denominator may be small.
An example is estimating the efficacy rate of a vaccine conditional on blood pressure, and comparing this rate between men and women.
The efficacy rate in vaccine trials is the difference in proportions of infected individuals between treatment and control groups, divided by the proportion in the control group.
Both designs have  conditional mean functions for $Y$ given $X$ that are equal and flat for $|X-0.5| \leq  0.3$, but different and non-linear
otherwise 
(Figure 4 in Section E of the supplement).
This shape of conditional mean function 
allows us to experiment with scenarios with or without estimation bias, and inside or outside the sharp null, depending on the choice of bandwidth.
Both designs fall under the null hypothesis; they fall under the sharp null if we further set  $\sigma_1^2 = \sigma_2^2$
and restrict the sample to $|X-0.5| \leq  0.3$.

The test statistic $T_n$ is the difference of consistent estimators $\ha{\theta}_1 - \ha{\theta}_2$ multiplied by $\sqrt{nh}$ (Equation \ref{eq:setup:teststat}).
The studentized 
statistic  
$S_n$ equals the difference of consistent estimators divided by the standard error of the difference (Equation \ref{eq:setup:teststat:stu}).
The conditional mean functions at point $0.5$ are consistently estimated by local linear regressions with triangular kernel and a bandwidth choice $h$ that shrinks to zero as $n$ increases.
A practical choice for $h$ is the estimated MSE-optimal bandwidth for local-linear regression (LLR), which decreases at rate $n^{-1/5}$.
In particular, we adapt the algorithm of \cite{imbens2012}, denoted IK bandwidth, to our setting.
This choice of bandwidth implies that $T_n$ and $S_n$ have asymptotic distributions not centered at zero.
Thus, we employ local quadratic regressions, using the same kernel and bandwidth as before, to construct the test statistics and avoid the 
asymptotic bias.\footnote{For more details on bias correction see discussion in Section 3.1. 
Section D.5 
of the supplement demonstrates validity of our permutation tests with the LPR estimator.}
We use White's robust formula for  local quadratic regressions to compute standard errors, 
where the the squared residuals are obtained by the nearest-neighbor matching estimator using three neighbors (\cite{abadie2006}).

We consider 10,000 simulated samples and 1,000 random permutations for each variation of the two designs.
We compare the null rejection probability of three different tests at 5\% nominal size:
the non-studentized permutation test (NSP), the studentized permutation test (SP),  and the $t$-test ($t$).
All tests use the same bandwidth choice, and we experiment with four possibilities: three fixed choices of $h$, $0.1,$ $0.3,$ and $0.5$, and the IK data-driven MSE-optimal bandwidth $\ha{h}_{mse}$.

Table \ref{table:size} displays simulated rejection rates under the null hypothesis.
DGPs with $\sigma_1^2=\sigma_2^2$ and $h=0.1$ or $0.3$ fall under the sharp null hypothesis.
As the theory predicts, both NSP and SP control size in these cases, but $t$ fails to do so, most notably in Design 1 with  $n_1=40$  and Design 2 with $\mu=1$.
Models with $\sigma_1^2=\sigma_2^2$ and $h=0.5$ fall outside the sharp null, and the local quadratic estimators are biased.
Bias makes the mean of $T_n$ diverge to infinity as the sample size increases, which explains the increasing size distortion of all tests.
In these cases, all tests fail to control size, although the distortions are smaller for SP than for $t$.
 
The rows of Table \ref{table:size} with $\sigma_1^2 \neq \sigma_2^2$ fall outside of the sharp null.
Cases with $\sigma_1^2 \neq \sigma_2^2$ and $h\leq 0.3$  violate  the sharp null, but the distribution of $T_n$ does not diverge as in the case of $h=0.5$.
Design 1 with $h=0.1$ or $0.3$ has a large size distortion of NSP and a small size distortion of SP that decreases with $n$, as predicted by our theory.
The size distortion of SP is much smaller than that of $t$, especially for smaller samples.
For Design 2 with $h=0.1$ or $0.3$, the size distortions of the permutation test are again smaller than $t$ and decrease with $n$.
Finally, the IK MSE-optimal bandwidth $\ha{h}_{mse}$ balances bias and variance, and does a good job keeping low size distortions of SP in all cases.

Since \cite{hall1990}, many authors have proposed resampling procedures to test nonparametric hypotheses.
It is natural to ask how some of these procedures compare to our robust permutation test.
We implement the tests of Designs 1 and 2 using critical values for $S_n$ generated by the wild bootstrap of \cite{cao1991} and the subsample of \cite*{politis1999subsampling}.
The last two columns of Table \ref{table:size} report the simulated rejection rates of studentized bootstrap (SB) and studentized subsample (SS) using the IK MSE-optimal bandwidth 
(Table 4 in Section E of the supplement reports the rates for fixed $h$).
The performance of SB and SS are generally similar to that of the $t$-test, with SS having worse size control in some cases.\footnote{
It is worth noting that  SB takes longer than SP to compute because it requires re-estimation of $\ha{m}_k(X_{k,i})$ for two different bandwidths and multiple observations $i$ (\cite{cao1991}, page 2227).
We compare the computation time of all tests with two empirical illustrations in the next section (Table \ref{table:applications}).}

\begin{sidewaystable}[htbp]
\caption{Simulated Rejection Rates - 5\% Nominal Size}
\label{table:size}
\small
\begin{minipage}{18cm}
\centering

\begin{tabular}{c c c c|    c c c|     c c c|     c c c|     c c c  c c}  
\hline \hline 
\multicolumn{4}{c|}{\textbf{Design 1}}   &  \multicolumn{3}{c|}{$h=0.1$}  & \multicolumn{3}{c|}{$h=0.3$}   &  \multicolumn{3}{c|}{$h=0.5$} &  \multicolumn{5}{c}{$\ha{h}_{mse}$} 
\\ 
$\sigma_1^2$ & $\sigma_2^2$ & $n_1$ & $n_2$ & NSP & SP & $t$ & NSP & SP & $t$ & NSP & SP & $t$ &  NSP & SP & $t$ & SB & SS 
\\ 
\hline 
 1 &  1 &  100 &  1900 &  0.048 &  0.049 &  0.093 &  0.047 &  0.046 &  0.062 &  0.158 &  0.147 &  0.177 &  0.047 &  0.045 &  0.065 &  0.067 &  0.090 
\\ 
 1 &  1 &  250 &  4750 &  0.048 &  0.047 &  0.065 &  0.049 &  0.048 &  0.055 &  0.327 &  0.321 &  0.340 &  0.048 &  0.046 &  0.056 &  0.057 &  0.056 
\\ 
 1 &  1 &  500 &  9500 &  0.051 &  0.048 &  0.059 &  0.048 &  0.048 &  0.051 &  0.577 &  0.569 &  0.589 &  0.050 &  0.048 &  0.054 &  0.055 &  0.048 
\\ 
\hline 
 1 &  1 &  40 &  1960 &  0.050 &  0.050 &  0.124 &  0.049 &  0.051 &  0.090 &  0.089 &  0.082 &  0.121 &  0.047 &  0.051 &  0.101 &  0.105 &  0.258 
\\ 
 1 &  1 &  100 &  4900 &  0.051 &  0.052 &  0.099 &  0.050 &  0.048 &  0.065 &  0.166 &  0.153 &  0.181 &  0.049 &  0.050 &  0.073 &  0.078 &  0.097 
\\ 
 1 &  1 &  200 &  9800 &  0.050 &  0.047 &  0.075 &  0.045 &  0.046 &  0.055 &  0.274 &  0.266 &  0.292 &  0.046 &  0.045 &  0.059 &  0.063 &  0.060 
\\ 
\hline 
 5 &  1 &  100 &  1900 &  0.298 &  0.068 &  0.098 &  0.310 &  0.056 &  0.064 &  0.360 &  0.073 &  0.081 &  0.310 &  0.056 &  0.065 &  0.065 &  0.095 
\\ 
 5 &  1 &  250 &  4750 &  0.314 &  0.059 &  0.069 &  0.314 &  0.053 &  0.056 &  0.426 &  0.104 &  0.112 &  0.316 &  0.053 &  0.055 &  0.052 &  0.056 
\\ 
 5 &  1 &  500 &  9500 &  0.322 &  0.054 &  0.058 &  0.323 &  0.051 &  0.053 &  0.521 &  0.164 &  0.172 &  0.326 &  0.052 &  0.052 &  0.050 &  0.046 
\\ 
\hline 
 5 &  1 &  40 &  1960 &  0.264 &  0.060 &  0.129 &  0.344 &  0.057 &  0.091 &  0.365 &  0.060 &  0.084 &  0.338 &  0.057 &  0.099 &  0.103 &  0.308 
\\ 
 5 &  1 &  100 &  4900 &  0.335 &  0.058 &  0.100 &  0.351 &  0.053 &  0.066 &  0.392 &  0.072 &  0.085 &  0.347 &  0.051 &  0.069 &  0.069 &  0.096 
\\ 
 5 &  1 &  200 &  9800 &  0.347 &  0.054 &  0.075 &  0.349 &  0.047 &  0.054 &  0.447 &  0.089 &  0.100 &  0.345 &  0.047 &  0.055 &  0.055 &  0.056 
\\ 
\hline \hline 
\end{tabular}

\end{minipage}

\vspace{0.25cm}

\begin{minipage}{18cm}
\centering

\begin{tabular}{c c c c|    c c c|     c c c|     c c c|     c c c  c c}  
\hline \hline 
\multicolumn{4}{c|}{\textbf{Design 2}}   &  \multicolumn{3}{c|}{$h=0.1$}  & \multicolumn{3}{c|}{$h=0.3$}   &  \multicolumn{3}{c|}{$h=0.5$} &  \multicolumn{5}{c}{$\ha{h}_{mse}$} 
\\ 
$\sigma_1^2$ & $\sigma_2^2$ & $\mu$ & $n_1=n_2$ & NSP & SP & $t$ & NSP & SP & $t$ & NSP & SP & $t$ &  NSP & SP & $t$ & SB & SS 
\\ 
\hline 
 1 &  1 &  10 &  75 &  0.050 &  0.051 &  0.087 &  0.048 &  0.049 &  0.061 &  0.085 &  0.082 &  0.094 &  0.039 &  0.039 &  0.050 &  0.044 &  0.128 
\\ 
 1 &  1 &  10 &  150 &  0.051 &  0.051 &  0.073 &  0.051 &  0.050 &  0.057 &  0.138 &  0.132 &  0.144 &  0.047 &  0.045 &  0.052 &  0.050 &  0.084 
\\ 
 1 &  1 &  10 &  1000 &  0.047 &  0.049 &  0.053 &  0.050 &  0.049 &  0.052 &  0.605 &  0.606 &  0.616 &  0.048 &  0.047 &  0.051 &  0.049 &  0.055 
\\ 
\hline 
 1 &  1 &  1 &  75 &  0.054 &  0.053 &  0.219 &  0.052 &  0.050 &  0.179 &  0.069 &  0.070 &  0.213 &  0.046 &  0.045 &  0.168 &  0.336 &  0.429 
\\ 
 1 &  1 &  1 &  150 &  0.049 &  0.051 &  0.192 &  0.052 &  0.050 &  0.168 &  0.088 &  0.091 &  0.249 &  0.050 &  0.050 &  0.162 &  0.266 &  0.360 
\\ 
 1 &  1 &  1 &  1000 &  0.047 &  0.045 &  0.167 &  0.050 &  0.049 &  0.161 &  0.333 &  0.339 &  0.582 &  0.048 &  0.048 &  0.163 &  0.184 &  0.273 
\\ 
\hline 
 5 &  1 &  10 &  75 &  0.068 &  0.060 &  0.095 &  0.056 &  0.051 &  0.063 &  0.062 &  0.061 &  0.069 &  0.044 &  0.044 &  0.053 &  0.040 &  0.129 
\\ 
 5 &  1 &  10 &  150 &  0.059 &  0.058 &  0.076 &  0.055 &  0.054 &  0.060 &  0.079 &  0.076 &  0.080 &  0.046 &  0.045 &  0.051 &  0.044 &  0.076 
\\ 
 5 &  1 &  10 &  1000 &  0.050 &  0.048 &  0.053 &  0.047 &  0.048 &  0.050 &  0.253 &  0.254 &  0.260 &  0.048 &  0.049 &  0.051 &  0.049 &  0.053 
\\ 
\hline 
 5 &  1 &  1 &  75 &  0.057 &  0.066 &  0.159 &  0.055 &  0.058 &  0.111 &  0.063 &  0.070 &  0.124 &  0.048 &  0.054 &  0.104 &  0.228 &  0.309 
\\ 
 5 &  1 &  1 &  150 &  0.055 &  0.060 &  0.128 &  0.053 &  0.053 &  0.098 &  0.073 &  0.080 &  0.135 &  0.047 &  0.051 &  0.091 &  0.157 &  0.220 
\\ 
 5 &  1 &  1 &  1000 &  0.050 &  0.049 &  0.092 &  0.046 &  0.048 &  0.087 &  0.201 &  0.208 &  0.298 &  0.047 &  0.048 &  0.087 &  0.096 &  0.145 
\\ 
\hline \hline 
\end{tabular} 

\end{minipage} 
 \caption*{\footnotesize
Notes:  The table displays  simulated rejection rates under the null hypothesis for the non-studentized permutation test (NSP), 
studentized permutation test (SP), and $t$-test ($t$).
Rows correspond to variations of Designs 1 and 2, as explained in the text.
The bandwidth choice $h$ equals one of three  fixed values (0.1, 0.3, and 0.5) or the IK data-driven MSE-optimal bandwidth $\ha{h}_{mse}.$
The last two columns display simulated rejection rates under the null hypothesis for the studentized bootstrap (SB) and studentized subsample (SS) tests, both using $\ha{h}_{mse}.$
}
\end{sidewaystable}

\begin{sidewaysfigure}[htbp]
    \centering
    \begin{minipage}{0.49\textwidth}
        \centering
        \caption{\label{figure:power_d1}Design 1 - Simulated Power Curves}
\begin{center}
  \begin{minipage}{.5 \textwidth}
    \centering
    (a) {\small $\sigma_1^2=5$, $\sigma_2^2 = 1$ \\ $n_1 / (n_1+n_2) = 0.05$ \\ SP(---) and $t$(- -) }

    \includegraphics[height=1.8in]{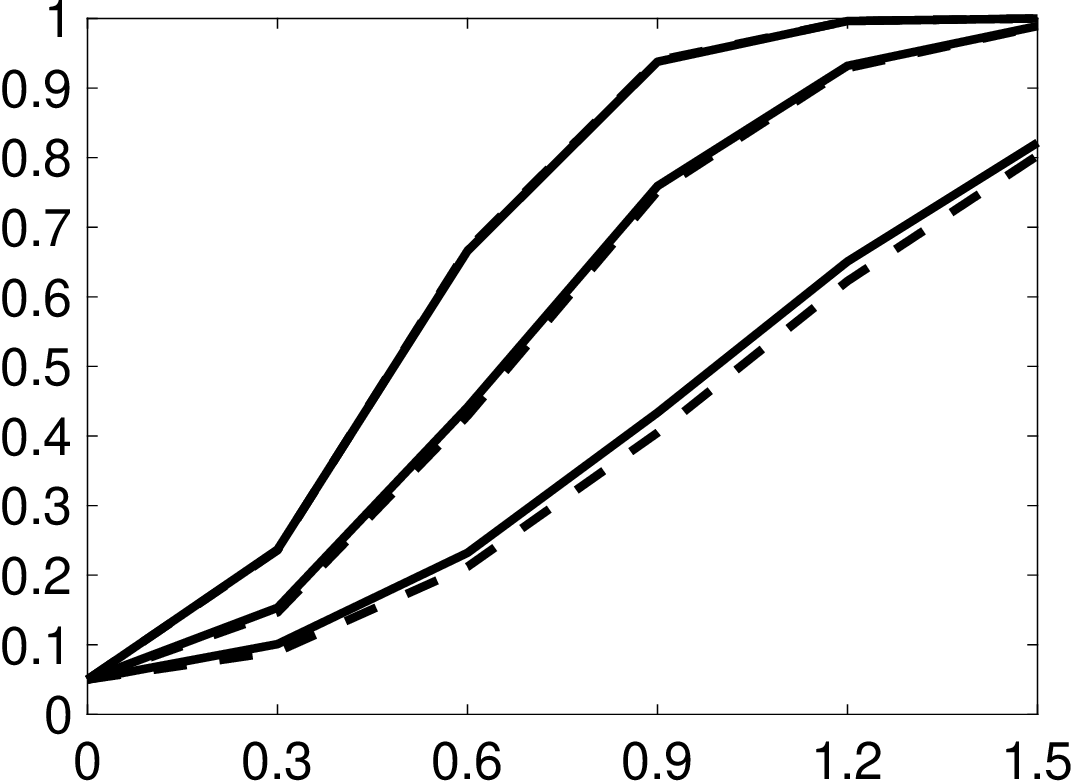}
  \end{minipage}%
  \begin{minipage}{.5 \textwidth}
    \centering
    (b) {\small $\sigma_1^2=5$, $\sigma_2^2 = 1$ \\ $n_1 / (n_1+n_2) = 0.02$ \\SP(---) and $t$(- -) }

    \includegraphics[height=1.8in]{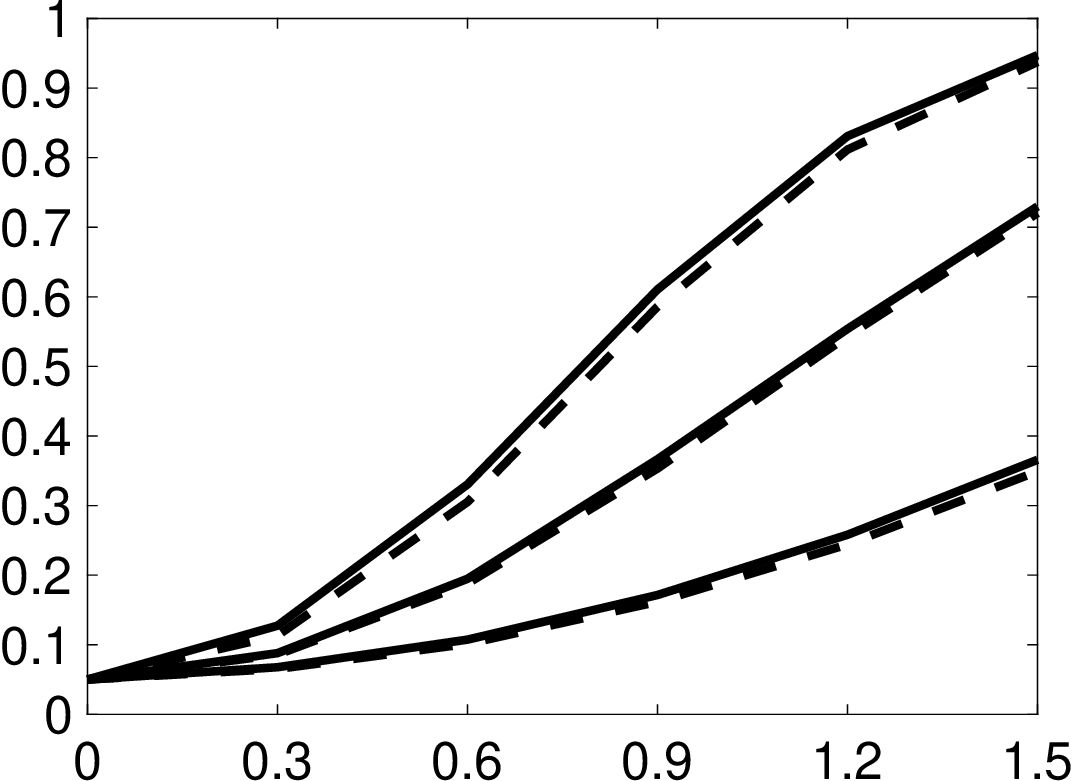}
  \end{minipage}%
\end{center}

\begin{center}
  \begin{minipage}{.5 \textwidth}
    \centering
    (c) {\small $\sigma_1^2= 5$, $\sigma_2^2 = 1$ \\ $n_1 / (n_1+n_2) = 0.05$ \\ SP(---), SB(- -), and SS($\cdot\hspace{-.8mm}\cdot\hspace{-.8mm}\cdot$) } 

    \includegraphics[height=1.8in]{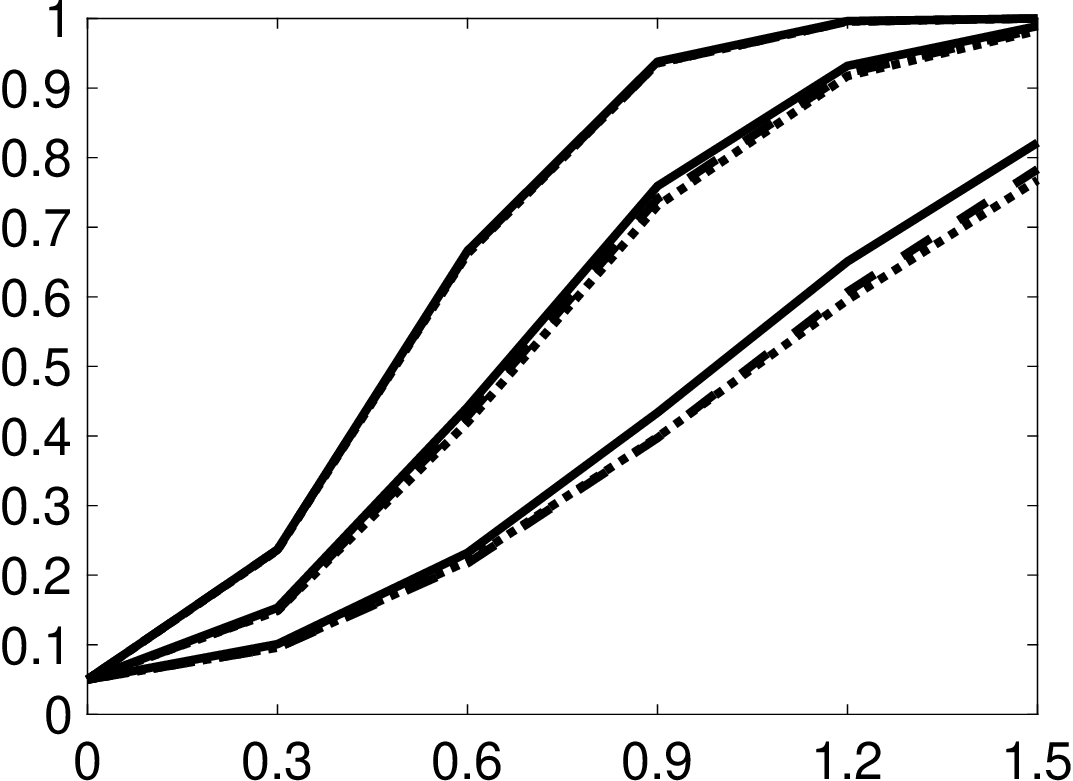}
  \end{minipage}%
  \begin{minipage}{.5 \textwidth}
    \centering
    (d) {\small $\sigma_1^2= 5$, $\sigma_2^2 = 1$ \\ $n_1 / (n_1+n_2) = 0.02$ \\ SP(---), SB(- -), and SS($\cdot\hspace{-.8mm}\cdot\hspace{-.8mm}\cdot$)}

    \includegraphics[height=1.8in]{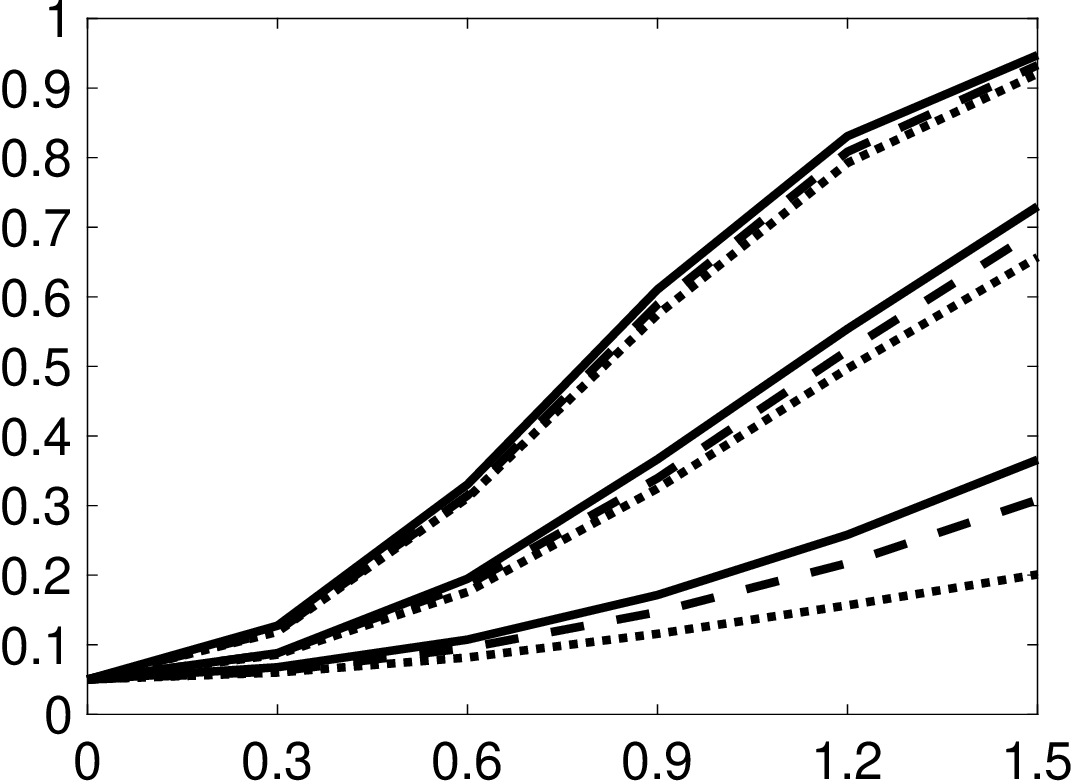}
  \end{minipage}%
\end{center}
\caption*{\footnotesize
Notes: Panels (a)--(b) compare power of the studentized permutation test (SP) with the $t$-test ($t$).
Panels (c)--(d) compare SP with the studentized bootstrap (SB) and studentized subsample (SS).
Darker lines correspond to bigger sample sizes, i.e., 
Panels (a) and (c): black $(n_1,n_2) = (500,9500)$,  dark gray $(n_1,n_2) = (250,4750)$, and light gray $(n_1,n_2) = (100,1900)$; 
Panels (b) and (d): black $(n_1,n_2) = (200,9800)$,  dark gray $(n_1,n_2) = (100,4900)$, and light gray $(n_1,n_2) = (40,1960)$.
The x-axis represents the difference $\theta(P_1)-\theta(P_2)$ and the y-axis, the simulated probability of rejection.
Sizes of all tests 
are artificially adjusted such that the simulated rejection under the null is always equal to 5\%.
All estimates use the IK MSE-optimal bandwidth $\ha{h}_{mse}$.
}
    \end{minipage} \hfill
    \begin{minipage}{0.49\textwidth}
        \centering
        \caption{\label{figure:power_d2}Design 2 - Simulated Power Curves}
\begin{center}
  \begin{minipage}{.5 \textwidth}
    \centering
    (a) {\small $\sigma_1^2=5$, $ \sigma_2^2 = 1$ \\ $\mu=10$ \\ SP(---) and $t$(- -)  }

    \includegraphics[height=1.8in]{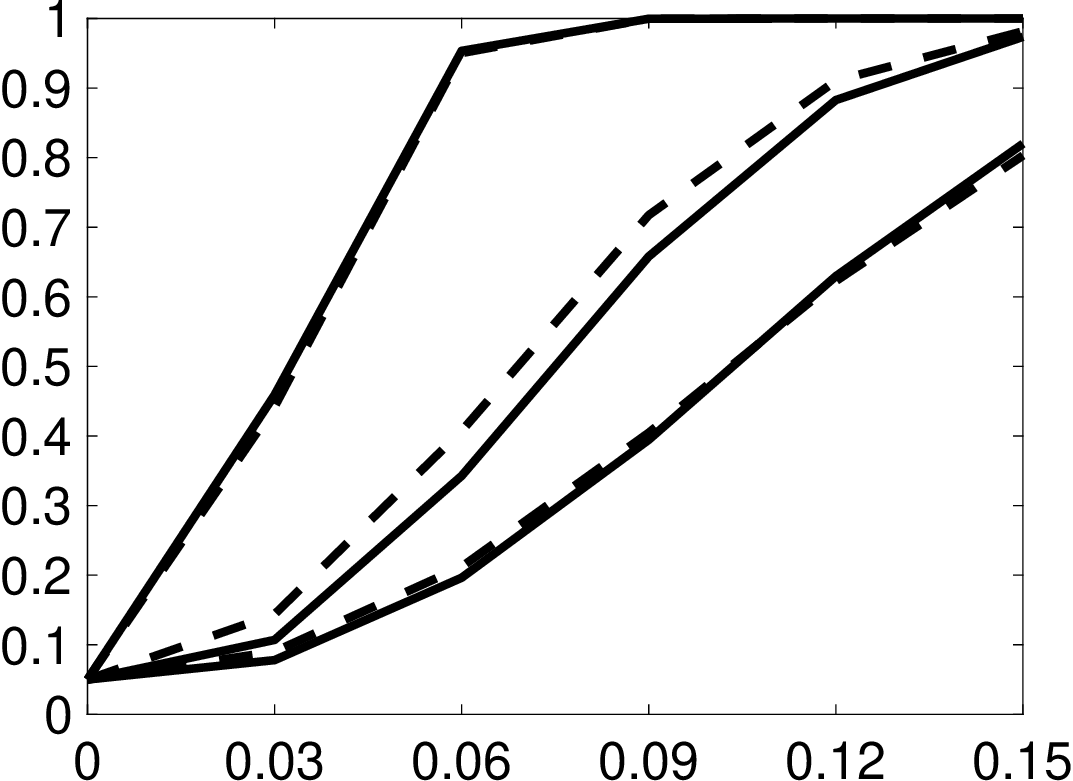}
  \end{minipage}%
  \begin{minipage}{.5 \textwidth}
    \centering
    (b) {\small $\sigma_1^2=5$, $\sigma_2^2 = 1$ \\ $\mu=1$ \\ SP(---) and $t$(- -) }

    \includegraphics[height=1.8in]{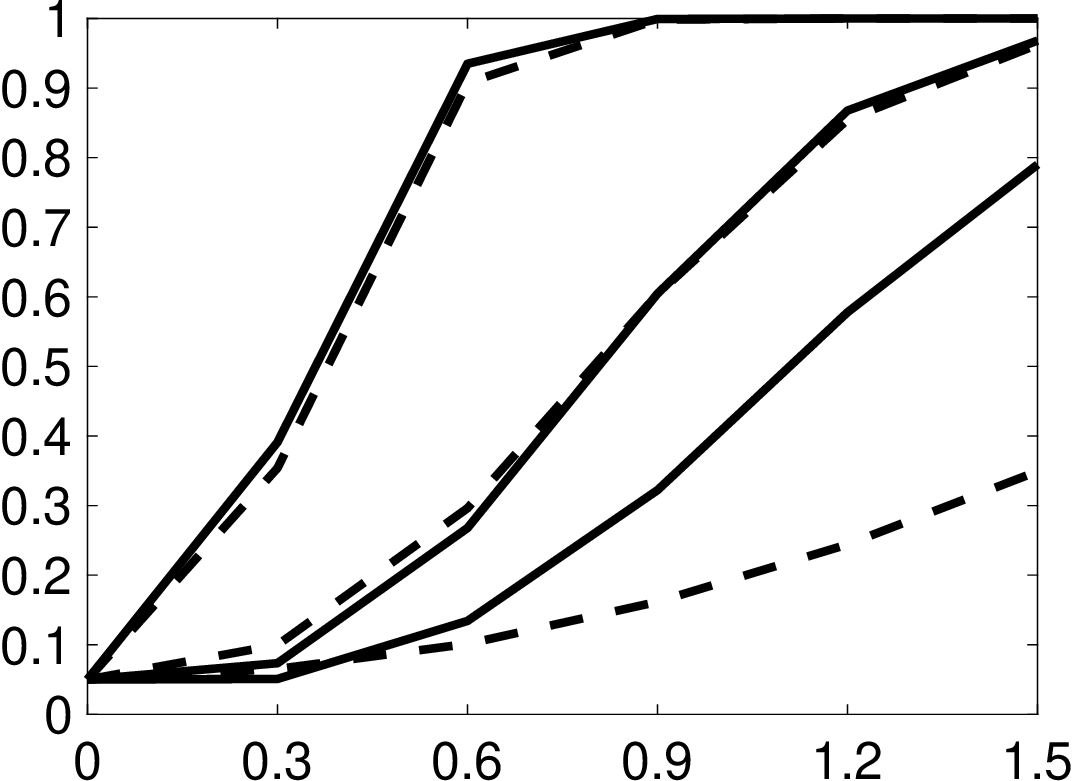}
  \end{minipage}%
\end{center}

\begin{center}
  \begin{minipage}{.5 \textwidth}
    \centering
    (c) {\small $\sigma_1^2= 5$, $\sigma_2^2 = 1$ \\ $\mu=10$ \\ SP(---), SB(- -), and SS($\cdot\hspace{-.8mm}\cdot\hspace{-.8mm}\cdot$) }

    \includegraphics[height=1.8in]{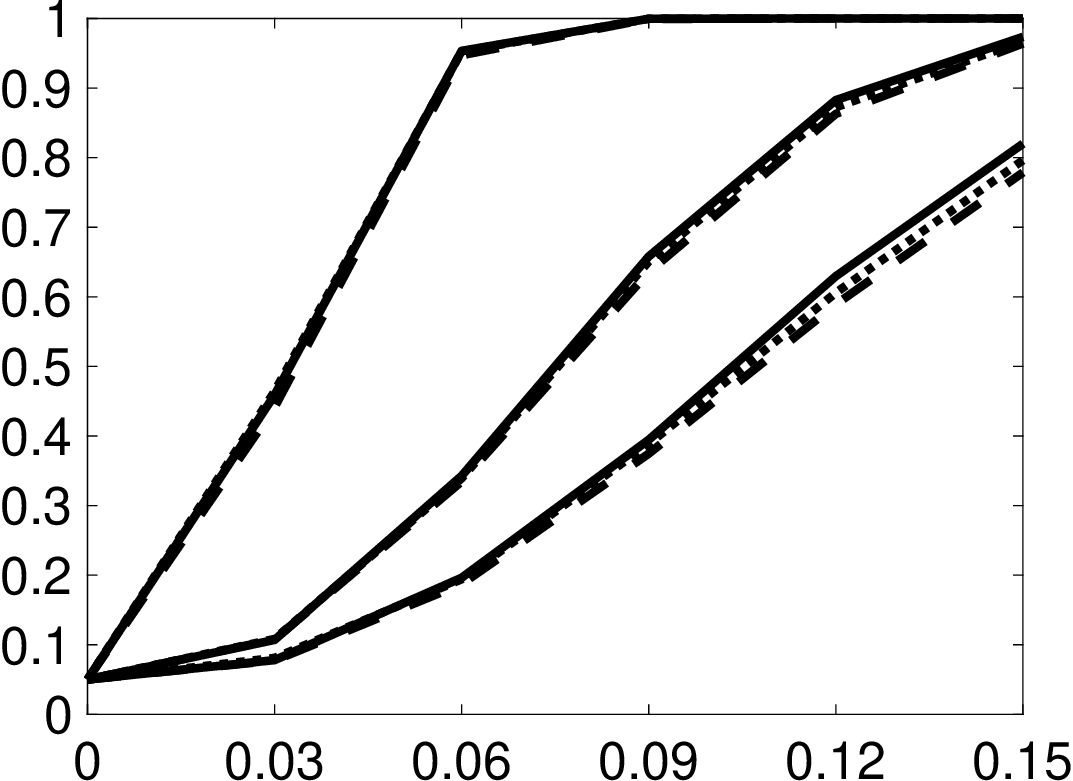}
  \end{minipage}%
  \begin{minipage}{.5 \textwidth}
    \centering
    (d) {\small $\sigma_1^2= 5$, $\sigma_2^2 = 1$ \\ $\mu=1$ \\ SP(---), SB(- -), and SS($\cdot\hspace{-.8mm}\cdot\hspace{-.8mm}\cdot$) }

    \includegraphics[height=1.8in]{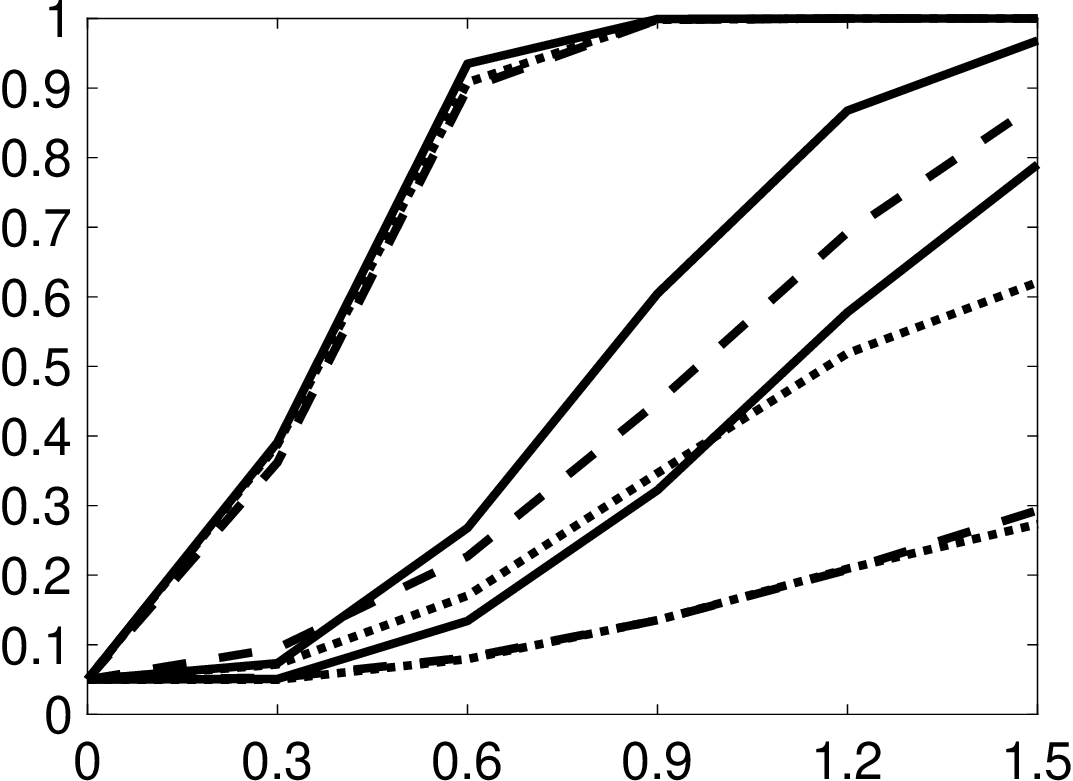}
  \end{minipage}%
\end{center}
\caption*{\footnotesize 
Notes: 
Panels (a)--(b) compare power of the studentized permutation test (SP) with the $t$-test ($t$).
Panels (c)--(d) compare SP with the studentized bootstrap (SB) and studentized subsample (SS).
Darker lines correspond to bigger sample sizes, i.e., 
black $(n_1,n_2) = (1000,1000)$,  dark gray $(n_1,n_2) = (150,150)$, and light gray $(n_1,n_2) = (75,75)$.
The x-axis represents the difference $\theta(P_1)-\theta(P_2)$ and the y-axis, the simulated probability of rejection.
Sizes of all tests 
are artificially adjusted such that the simulated rejection under the null is always equal to 5\%.
All estimates use the IK MSE-optimal bandwidth $\ha{h}_{mse}$.
\\
\vspace{3mm}
}
    \end{minipage}
\end{sidewaysfigure}

We shift the conditional mean function of $Y_2$ in both designs so that $\theta(P_1) \neq \theta(P_2)$.
We then examine the power of SP and compare it to that of the $t$, SB, and SS tests.
A direct comparison is difficult because not all tests control size in all cases.
Thus, we artificially adjust the size of the tests to make sure they have a simulated rejection rate of 5\% under the null hypothesis.\footnote{For the $t$-test, we obtain critical values from the simulated distribution under the null hypothesis, and keep those critical values to examine the simulated rejection rates under the alternative hypotheses. 
For the permutation, bootstrap, and subsampling tests, we numerically search for a nominal level $\alpha$ that gives us the simulated rejection rate of 5\% under the null hypothesis.
Once that artificial nominal level is found, we fix that nominal level and compute the simulated rejection probabilities under the various alternative hypotheses.
Rejection may be randomized in case of ties (Eq. \ref{eq:setup:testfun}) in order for the numerical search to find a solution.}
Figures \ref{figure:power_d1} and \ref{figure:power_d2} display the simulated power curves for Designs 1 and 2, respectively, for cases with $\sigma_1^2 > \sigma_2^2=1$.
Cases with $\sigma_1^2 = \sigma_2^2=1$ are found in 
Figures 5--6 of Section E 
in the supplement. 
Panels (a)--(b) compare  SP (solid line) to $t$ (dashed line), and Panels (c)--(d) compare SP (solid line) to SB (dashed line) and SS (dotted line).
Each panel displays curves associated with three different sample sizes,
 with darker colors representing larger samples.
The x-axis shows $\theta(P_1) - \theta(P_2)$, and the y-axis plots the simulated probability of rejection.
All estimates in our power analysis use the IK MSE-optimal bandwidth $\ha{h}_{mse}$.

In Design 1 (Figure \ref{figure:power_d1}), SP outperforms $t$, SB, and SS, particularly in cases with $n_1/(n_1+n_2)=0.02$; the dominance over SS is more pronounced for smaller samples.
In Design 2 (Figure \ref{figure:power_d2}), there is no clear pattern of dominance between SP and $t$, however, SP dominates SB and SS in all cases.
The discrepancies between SP and other tests  
converge to zero as $n$ increases, as predicted by our theory.
Overall, we conclude that SP has size control superior to other tests  without substantial costs, if any, in terms of power.

\vspace{-6mm}

\section{Empirical Examples}\label{sec:empirical}

\vspace{-2mm}

\indent

In this section, we revisit two classical examples in the RDD literature: \cite{lee2008} on US House elections and \cite{ludwig2007} on the Head Start (HS) funding program.
We illustrate the performance of our permutation test in practice and compare it
to that of the $t$, bootstrap, and subsample tests. 

\cite{lee2008} studies the electoral advantage of incumbent parties, using data on US House of Representatives elections from 1946 to 1998. Since districts where a party's candidate narrowly won an election are comparable to districts where the party's candidate lost by a small margin, the difference in the electoral outcomes between these two groups in the subsequent election identifies the causal effect of party incumbency. \cite{lee2008} finds that an incumbent party has a significant causal advantage of a 0.08 vote share increase in the next election (Table 2 of \cite{lee2008}).

\cite{ludwig2007} study the effects of HS on students' health and schooling. The HS program was established in 1965 to provide preschool, health, and social services to poor children, aged three to five, and their families. The program provided technical assistance to the 300 poorest counties in the US, based on the 1960 poverty rate.
This created a discontinuity in program funding between the 300th and 301st poorest counties. 
Ludwig and Miller compare child mortality rates above and below this cutoff, and estimate that HS reduces mortality by 1.198 per 100,000 children (see Table 3 of \cite{ludwig2007}, ages 5--9, HS-related causes, 1973--1983). 

We test the null hypothesis of zero discontinuity using our robust permutation test on both datasets.\footnote{We downloaded the datasets from Michal Kolesar's repository, \texttt{https://github.com/kolesarm}.}
We estimate the discontinuities with local-quadratic regressions and MSE-optimal bandwidths $\ha{h}_{mse}$ for local-linear regressions, as explained in Section \ref{sec:mc}.
The number of observations is 6,558 for \cite{lee2008} and
3,103 for \cite{ludwig2007}, respectively.

Table \ref{table:applications} reports the $p$-values of our studentized permutation test (SP) and compares them to those of the $t$, studentized bootstrap (SB), and studentized subsample (SS) tests, as implemented in Section \ref{sec:mc}. 
The party incumbency effect of \cite{lee2008} is strongly significant and robust across different tests for both choices of $\ha{h}_{mse}$, that is, the  $\ha{h}_{mse}$ of \cite{calonico2014} in the first row, and that of \cite{imbens2012} in the second row.
On the other hand, \cite{ludwig2007} find the effect of HS on child  mortality to be marginally significant, and we find the significance level ranges between 1\% to 12\%, depending on the test and choice of $\ha{h}_{mse}$.
Finally, both examples demonstrate that our permutation test is feasible to compute in mere seconds, comparing favorably with the bootstrap in terms of computation time.

\begin{table}[htbp]
\caption{
Statistical Significance and Computation Time of Four Tests
}
\label{table:applications}

\vspace{-5mm} 

\begin{center}
\small
  \begin{minipage}{.5 \textwidth}
    \centering
    (a) \cite{lee2008}
    \begin{tabular}{c|cccc}
  \hline
$\ha{h}_{mse}$ & $t$ & SP & SB & SS  \\
   \hline
   13.4400 &  0.0002 & 0.0000 &  0.0000 & 0.0000  \\
   29.3903 &  0.0000  &   0.0000 &0.0000 & 0.0000  \\ \hline
Time (sec) & 0.0113 &  9.3293 &    14.1085 &     4.1356  \\
\hline
\end{tabular}
  \end{minipage}%
  \begin{minipage}{.5 \textwidth}
    \centering
    (b) \cite{ludwig2007}

    \begin{tabular}{c|cccc}
  \hline
$\ha{h}_{mse}$ & $t$ & SP & SB & SS   \\
   \hline
    6.9510 &    0.0066 &     0.0680 &     0.0070    & 0.0931 \\
    17.0846 &    0.0357 &     0.0750 &      0.0350  & 0.1141 \\ \hline
Time (sec) & 0.0049 &  3.8703 &     4.2179 &     2.5887 \\
\hline
\end{tabular}
  \end{minipage}%
\end{center}

\vspace{-5mm}

\caption*{
\footnotesize Notes: The table reports $p$-values of studentized test statistics for the $t$, permutation (SP), bootstrap (SB), and subsample (SB) tests for two MSE-optimal bandwidths $\ha{h}_{mse}$.
The first row is based on the bandwidth selection rule proposed by \cite{calonico2014} and the second row on the rule by \cite{imbens2012}.
The last row displays the computation time in seconds for both bandwidth selections in a
2.3 GHz Quad-Core Intel Core i7 computer running Matlab in a single core. 
We use 1,000 simulated samples for SP, SB, and SS.
}
\end{table}

\vspace{-6mm}

 \section{Conclusion}
 
\vspace{-2mm}

\indent 

Classical two-sample permutation tests for the sharp null hypothesis of equal distributions are easy to implement and have exact size in finite samples.
However, for testing equality of parameters that summarize distributions, classical permutation tests fail to control size. To fix this problem, we propose robust permutation tests based on studentized test statistics. 
Our framework is general enough to cover both parametric and nonparametric models with two samples or one sample split into two subsamples. 
We also propose confidence sets with correct asymptotic coverage that have exact coverage in finite samples if population distributions are the same up to a class of transformations. 
In a simulation study, our permutation test has good size control and power curves in finite samples, outperforming the conventional $t$-test and other resampling methods.
Finally, we illustrate our permutation test with two empirical examples and show that its computation time is feasible, comparing favorably to the bootstrap. 

\ifblind
\else

\textbf{Acknowledgements:} We wish to thank Roger Koenker, Marcelo Moreira, Joe Romano, Azeem Shaikh, Xiaofeng Shao, the editor, associate editor, and two anonymous referees for many insightful comments. 
The paper benefited from feedback given by seminar participants at the U. of Sao Paulo (FEARP), 
U. of Michigan,  Boston U., U. of Essex, IAAE-Rotterdam, Bristol-ESG, KER Int'l Conference, UIUC, and Penn State. 
Bertanha acknowledges financial support received while visiting the Kenneth C. Griffin Department of Economics at the University of Chicago, where part of this work was conducted.
\fi

\vspace{-2mm}

\begin{onehalfspace}
\bibliographystyle{econ}
\bibliography{biblio}

\end{onehalfspace}

\newpage
\setcounter{page}{1}

\begin{center}
 \Large Supplement to ``Permutation Tests at Nonparametric Rates''

\ifblind
\else
\normalsize Marinho Bertanha, EunYi Chung
\fi

\end{center}

\bigskip

\textit{This supplement is organized as follows.
Section \ref{sec:publication} presents a table listing selected publications that use permutation tests. 
Section \ref{sec:app:a} contains the proofs for the lemma, theorems, and corollary  
in Section \ref{sec:theory}. 
Section \ref{sec:param} graphically 
illustrates
(non)validity of the (non)studentized permutation test using 
the
parametric example of Section \ref{sec:theory}. 
Section \ref{sec:app:b} presents all the proofs for Section \ref{sec:applications}. 
Section \ref{sec:app:montecarlo} 
plots the conditional mean functions of the simulation designs of our Monte Carlo experiments and brings additional results.
Lastly, Section \ref{sec:app:c} 
contains
auxiliary lemmas and their proofs. 
}

\begin{appendix}

\begin{singlespace}

\section{A List of Publications}
\label{sec:publication}

\indent

\begin{table}[H]
\caption{Selected Publications in Social and Natural Sciences}
\label{table:literature}

\small

\centering

\begin{tabular}{l   l}
\hline \hline 
\multicolumn{2}{l}{\textbf{Economics}} \\ \hline
\citeSM{alan2019} & Qtly. J. Economics\\
\citeSM{bick2018} & Ame. Econ. Rev. \\
\citeSM{bursztyn2019} & Rev. Econ. Stud. \\
\citeSM{cunningham2018} & Rev. Econ. Stud.\\
\citeSM{rao2019} & Ame. Econ. Rev. \\ \hline 
\multicolumn{2}{l}{\textbf{Public, Environmental \& Occupational Health}} \\ \hline
\citeSM{agha2015} & Int. J. Epidemiol.\\
\citeSM{arnup2016} & J. Clin. Epidemiol. \\
\citeSM{berk2020} & Ame. J. Public Health\\
\citeSM{burrows2017} & MMWR Morb Mortal Wkly Rep \\
\citeSM{huang2016} & Int. J. Epidemiol.\\ \hline
\multicolumn{2}{l}{\textbf{Medicine, General \& Internal}} \\ \hline
\citeSM{berk2013} & JAMA\\
\citeSM{goldfine2013} & Lancet \\
\citeSM{hansen2016} & JAMA \\
\citeSM{rajagopalan2013} & N. Engl. J. Med.\\
\citeSM{ryan2016} & Lancet \\ \hline
\multicolumn{2}{l}{\textbf{Political Science}} \\ \hline
\citeSM{bechtel2015} & Ame. J. Pol. Sci.\\
\citeSM{lax2010} & J. Politics\\
\citeSM{ramos2019} & Comp. Polit. Stud. \\
\citeSM{schafer2020} & J. Politics\\
\citeSM{wood2020} & Ame. J. Pol. Sci.\\ \hline
\multicolumn{2}{l}{\textbf{Psychology, Multidisciplinary}} \\ \hline
\citeSM{bishara2012} & Psychol. Methods\\
\citeSM{hu2014} & Psychol. Methods\\
\citeSM{jorgensen2018} & Psychol. Methods \\
\citeSM{pietschnig2015} & Perspect Psychol Sci\\
\hline \hline 
\end{tabular}

 \caption*{\footnotesize
Notes: The table lists selected publications from top journals in various disciplines
as categorized by the Journal Citation Report produced by the Web of Science.
We ranked journals by the 2020 Journal Impact Factor and 
searched for publications in the last decade from a subset of top journals.
}
\end{table}

\section{Proofs of Theorems and Corollaries}
\label{sec:app:a}

\subsection{Proof of Lemma \ref{lemma:exact} - Exact Size in Finite Samples}

\indent

Summing $\phi(\lbar{W}_n, \lbar{\bZ}_n^\pi)$ over $\pi \in \bG_n$ and taking 
the conditional expectation given $\lbar{W}_n$ yields 
\begin{align}
\alpha n! = \sum_{\pi \in \bG_n} \Exp[ \phi(\lbar{W}_n, \lbar{\bZ}_n^\pi)  |  \lbar{W}_n] = &  n!  \Exp[ \phi(\lbar{W}_n, \lbar{\bZ}_n)  |  \lbar{W}_n],
\end{align}
which implies that $\Exp[ \phi(\lbar{W}_n, \lbar{\bZ}_n)  ] = \Exp \left\{ \Exp[ \phi(\lbar{W}_n, \lbar{\bZ}_n)  |  \lbar{W}_n] \right\} = \alpha$.

$\square$

\subsection{Proof of Theorem \ref{theo1} - Asymptotic Distributions Without Studentization}

\subsubsection{Asymptotic Distribution of Test Statistic}
\label{proof:theo1:asy_dist}

\indent

Consider a sequence of $\lbar{W}_n$ 
that satisfies 
$  (n_1/n - \lambda) \to 0 $ for $\lambda \in (0,1)$.
Condition on  $\lbar{W}_n$, 
\begin{align*}
T_{n} - \sqrt{nh}\left( \theta(P_1) - \theta(P_2)\right) &= \sqrt{nh } \left[ \frac{ \sqrt{n_1 h}  }{ \sqrt{n_1 h} }  \left( \ha\theta_1 -\theta(P_1) \right)  -   \frac{ \sqrt{n_2 h} }{ \sqrt{n_2 h} } \left( \ha\theta_2 - \theta(P_2) \right) \right]
\\
  &= \frac{ \sqrt{n } }{ \sqrt{n_1 } }  \left[ \frac{1}{\sqrt{n_1}} \sum_{i=1}^{n_1} \psi_{n}(\bZ_{1,i}, P_1)  \right]  
  	-  \frac{ \sqrt{n } }{ \sqrt{n_2 } } \left[ \frac{1}{\sqrt{n_2}} \sum_{i=1}^{n_2}  \psi_{n}(\bZ_{2,i}, P_2)  \right]
\\
&+ \frac{ \sqrt{n } }{ \sqrt{n_1 } } o_{P_1}(1) 
- \frac{ \sqrt{n } }{ \sqrt{n_2 } } o_{P_2}(1),
\end{align*}
where $o_{P_k}(1)$ is a term that depends on $Z_{k,1}, \ldots, Z_{k,n_k}$ and converges in probability to zero as $n\to\infty$, $k=1,2$.

First,
\[
 \frac{ \sqrt{n } }{ \sqrt{n_1 } } o_{P_1}(1) 
- \frac{ \sqrt{n } }{ \sqrt{n_2 } } o_{P_2}(1)
\pto 0~.
\]

Second,
for each $k=1,2$,
it suffices to show that the Lindeberg condition holds.
Abbreviate $\psi_{n}(\bZ_{k,i},P_k)$ by $\psi_{n,k}$ and $\mmv[\psi_{n}(\bZ_{k,i},P_k)]$ by $\delta_{n,k}^2$.
For every $\eps>0$ 
and $\zeta$ of Assumption \ref{aspt:asy_linear}-\eqref{eq:aspt:asy_linear_lind}, note that
\begin{align*} 
\left| \frac{ \psi_{n,k} }{\delta_{n,k} } \right|^{2+\zeta}
	& \geq
	 \frac{ \psi_{n,k}^2 }{\delta_{n,k}^2 }  \left( \eps n_k \right)^{\zeta/2}
	\mmi\left\{  \frac{ \psi_{n,k}^2 }{\delta_{n,k}^2}   > \eps n_k   \right\},
\end{align*}
which implies that the Lindeberg condition holds 
and for $\xi^2(P_k) = \lim_{n\to\infty} \mmv[ \psi_{n}(\bZ_{k,i},P_k)]$,
\begin{align*}
\frac{1}{\sqrt{n_k}}  \sum_{i=1}^{n_k} \psi_{n }(\bZ_{k,i},P_k) \dto N\left(0 ; \xi^2(P_k) \right).
\end{align*}

Therefore,
\begin{align*}
  & \frac{ \sqrt{n } }{ \sqrt{n_1 } }  \left[ \frac{1}{\sqrt{n_1}} \sum_{i=1}^{n_1} \psi_{n }(\bZ_{1,i},P_1)  \right]  
  	-  \frac{ \sqrt{n } }{ \sqrt{n_2 } } \left[ \frac{1}{\sqrt{n_2}} \sum_{i=1}^{n_2}  \psi_{n}(\bZ_{2,i},P_2)  \right]
\\
&\dto N\left(0 ; \frac{ \xi^2(P_1)}{\lambda} + \frac{ \xi^2(P_2) }{ 1-\lambda }\right),
\end{align*}
which shows convergence in distribution conditional on  $\lbar{W}_n$.

Lemma \ref{lemma:cond:clt} gives convergence in distribution unconditionally.

\subsubsection{Asymptotic Linear Representation of Permuted Test Statistic}
\label{proof:theo1:asy_linear_perm}

\indent

In this and the following subsections, we make the entire analysis conditional 
on a sequence of $\lbar{W}_n$ 
that satisfies
 $  (n_1/n - \lambda) \to 0 $ for $\lambda \in (0,1)$.
Without loss of generality, re-order observations in the sample such that  
$\lbar{\bZ}_n=(\bZ_1, \ldots, \bZ_{n_1}, \bZ_{n_1+1}, \ldots, \bZ_{n})= (\bZ_{1,1}, \ldots, \bZ_{1,n_1}, \bZ_{2,1}, \ldots, \bZ_{1,n_2})$
and
$\lbar{W}_n=(W_{1,n}, \ldots, W_{n_1,n}, W_{n_1+1,n}, \ldots, W_{n,n})= (1, \ldots, 1, 2, \ldots, 2)$.

Let $\pi$ be a random permutation that is uniformly distributed over $\bG_n$ and independent of the data.
The goal of this subsection is to show that, for the test statistic $T_n$ defined in Equation \ref{eq:setup:teststat},
\begin{align*}
T_n(\lbar{W}_n, \lbar{\bZ}_n^{\pi}) = & 
 \frac{ \sqrt{n} }{ \sqrt{n_1}  }  \left[ \frac{1}{\sqrt{n_1}} \sum_{i=1}^{n_1} \psi_{n}( \bZ_{\pi(i)}, \ubar{P}_n )  \right]  
 \\
 - &  
 \frac{ \sqrt{n}  }{ \sqrt{n_2}  } \left[ \frac{1}{\sqrt{n_2}} \sum_{i=n_1+1}^{n}  \psi_{n}( \bZ_{\pi(i)}, \ubar{P}_n )  \right] 
 + o_p(1),
\end{align*}
where $\ubar{P}_n = p_{1,n} P_1 + p_{2,n} P_2$,  
$p_{k,n} = n_k/n$, $k=1,2$,
and $o_p(1)$ is a term that depends on the data and the random permutation, and it converges in probability to zero.

For each $k=1,2$, let $\bV_{1,n}, \ldots, \bV_{n_k,n}$ be an \textit{iid} sample from the distribution $\ubar{P}_n$.
As $\ubar{P}_n \in \m{P}~~ \forall n$, the uniform asymptotic linear representation of Assumption \ref{aspt:asy_linear} guarantees that
\begin{align}
R_{k,n}(\bV_{1,n}, \ldots, \bV_{n_k,n}) \doteq &
\sqrt{ n_k h} \left( \theta_{ n_k,n }(\bV_{1,n}, \ldots, \bV_{n_k,n}) - \theta(\ubar{P}_n)  \right) 
\\
& \hspace{1cm}
  - \left( \frac{1}{\sqrt{n_k}} \sum_{i=1}^{n_k} \psi_{n}   (\bV_{i,n}, \ubar{P}_n) \right) 
  \pto 0.
\end{align}

Lemma 5.3 and Remark A.2 by \citeSM{chung2013} show that
\begin{align*}
R_{1,n}(\bZ_{\pi(1)}, \ldots, \bZ_{\pi(n_1)}) = &
\sqrt{n_1h} \left( \theta_{n_1,n}(\bZ_{\pi(1)}, \ldots, \bZ_{\pi(n_1)}) - \theta(\ubar{P}_n)  \right) 
\\
& \hspace{1cm}
  - \left( \frac{1}{\sqrt{n_1}} \sum_{i=1}^{n_1} \psi_{n}   (\bZ_{\pi(i)}, \ubar{P}_n) \right) 
  \pto 0,
\\
R_{2,n}(\bZ_{\pi(n_1+1)}, \ldots, \bZ_{\pi(n)}) = &
\sqrt{n_2 h } \left( \theta_{n_2,n}(\bZ_{\pi(n_1+1)}, \ldots, \bZ_{\pi(n)}) - \theta(\ubar{P}_n)  \right) 
\\
& \hspace{1cm}
  - \left( \frac{1}{\sqrt{n_2}} \sum_{i=n_1+1}^{n} \psi_{n }   (\bZ_{\pi(i)}, \ubar{P}_n) \right) 
  \pto 0.
\end{align*}

Then, 
\begin{align*}
T_n(\lbar{W}_n, \lbar{\bZ}_n^{\pi}) 
&= 
 \frac{ \sqrt{n}  }{ \sqrt{n_1}  }  \left[ \frac{1}{\sqrt{n_1}} \sum_{i=1}^{n_1} \psi_{n}( \bZ_{\pi(i)}, \ubar{P}_n )  \right]  
 -   
 \frac{ \sqrt{n}  }{ \sqrt{n_2}  } \left[ \frac{1}{\sqrt{n_2}} \sum_{i=n_1+1}^{n}  \psi_{n }( \bZ_{\pi(i)}, \ubar{P}_n )  \right] 
\\
&\hspace{1cm} + \frac{ \sqrt{n}  }{ \sqrt{n_1}  } R_{1,n}(\bZ_{\pi(1)}, \ldots, \bZ_{\pi(n_1)})
- \frac{ \sqrt{n}  }{ \sqrt{n_2}  } R_{2,n}(\bZ_{\pi(n_1+1)}, \ldots, \bZ_{\pi(n)})
\\
&= 
 \frac{ \sqrt{n}  }{ \sqrt{n_1}  }  \left[ \frac{1}{\sqrt{n_1}} \sum_{i=1}^{n_1} \psi_{n}( \bZ_{\pi(i)}, \ubar{P}_n )  \right]  
 -   
 \frac{ \sqrt{n}  }{ \sqrt{n_2}  } \left[ \frac{1}{\sqrt{n_2}} \sum_{i=n_1+1}^{n}  \psi_{n }( \bZ_{\pi(i)}, \ubar{P}_n )  \right] 
\\
&\hspace{1cm} + o_p(1).
\end{align*}

\subsubsection{Coupling Approximation}
\label{proof:theo1:coupling}

\indent

The goal of this section is to create a data set $\lbar{\bZ}^*_n=( {\bZ}_1^*, \ldots, {\bZ}_n^*)$ that is \textit{iid} from $\ubar{P}_n $ and a permutation $\pi_0$
such that, for a random permutation $\pi$,
$T_n(\lbar{W}_n,\lbar{\bZ}^{*\pi\pi_0}_n)-T_n(\lbar{W}_n,\lbar{\bZ}^{\pi}_n) \pto 0$.

Given that 
$\lbar{\bZ}_n=(\bZ_1, \ldots, \bZ_{n_1}, \bZ_{n_1+1}, \ldots, \bZ_{n})= (\bZ_{1,1}, \ldots, \bZ_{1,n_1}, \bZ_{2,1}, \ldots, \bZ_{1,n_2})$,
the distribution of $\lbar{\bZ}_n$ is  $P_{1}^{n_1} \times P_{2}^{n_2}$, that is, the independent product of distributions 
$P_{1}$  ($n_1$ times) and $P_{2}$ (${n_2}$ times).
In what follows, we construct a dataset $\lbar{\bZ}^*_n=( {\bZ}_1^*, \ldots, {\bZ}_n^*)$ that is \textit{iid} from $\ubar{P}_n$.
For observation $i=1$, draw an index $k$ out of $\{1,  2\}$ at random with probabilities $(p_{1,n}, p_{2,n})$.
Given the resulting index $k$, set ${\bZ}_1^*={\bZ}_{k,1}$.
Move on to $i=2$ and draw an index $k'$ as before. 
If the index $k'$ drawn is the same as $k$, set ${\bZ}_2^* = {\bZ}_{k,2}$. 
Otherwise, if the index $k' \neq k$, then set ${\bZ}_2^* = {\bZ}_{k',1}$.
Keep going until you reach a point where you may run out of observations from one of the two samples.
For example, for observation $i$, if you randomly pick $k=1$, but you have already exhausted all $n_1$ observations from $P_{1}$,
then randomly draw ${\bZ}_i^*$ from $P_{1}$. 
We end up with $\lbar{\bZ}^*_n$ and $\lbar{\bZ}_n$ having many of the same observations in common.
Call $D$ the random number of observations that are different.

Reorder observations in $\lbar{\bZ}^*_n$  by a permutation $\pi_0$ so that 
$\bZ_{\pi_0(i)}^*$ equals  ${\bZ}_i$ for most $i$, except for $D$ of them.
The re-ordered sample $\lbar{\bZ}^{* \pi_0}_n$ is constructed as follows.
The sample $\lbar{\bZ}_n$ has  the $n_1$ observations from $P_{1}$ appear first, then the $n_2$ observations from $P_{2}$ appear second.
Simply take all the observations from $P_{1}$ in  $\lbar{\bZ}^*_n$ and place them first, up to $n_1$.
The observations from $P_{1}$ in  $\lbar{\bZ}^*_n$ that are equal to observations from $P_{1}$ in  $\lbar{\bZ}_n$ are placed first and in the same order as in $\lbar{\bZ}_n$.
If there are more than $n_1$ observations from $P_{1}$ in  $\lbar{\bZ}^*_n$,  put the extra observations on the side.
If there are less than $n_1$, leave the spots blank and start at spot $n_1+1$  with the observations from $P_{2}$ in $\lbar{\bZ}^*_n$.
Repeat the same procedure for those observations in  $\lbar{\bZ}^*_n$ that were drawn from $P_{2}$, and place them starting at spot $n_1+1$ up to spot $n_2$.
There are remaining spots in the newly created sample $\lbar{\bZ}^{* \pi_0}_n$, and these should be filled up with
the observations from $\lbar{\bZ}^*_n$ that were placed on the side, in any order.
Note that the distribution of $\lbar{\bZ}^{* \pi_0}_n$ is also \textit{iid} $\ubar{P}_n$.

We first need to show that the (random) number of observations $D$, that are different between $\lbar{\bZ}^*_n$ and $\lbar{\bZ}_n$, is ``small'' as a fraction of $n$.
Let $n_1^*$ denote the number of observations in $\lbar{\bZ}^*_n$ that are generated from $P_{1}$. 
Then, $n_1^*$ has the binomial distribution with $(n, p_{1,n})$ parameters.
So the mean of $n_1^*$ is $n p_{1,n} = n_1$.
We have that $D = | n_1^* - n_1|$.
\[
\mme[D] =  \mme | n_1^* - n_1 | = \mme | n_1^* - n p_{1,n} | 
\]
\[
\leq 
[\mme\{(n_1^* - n p_{1,n})^2\}]^{1/2} = [n p_{1,n}(1-p_{1,n})]^{1/2}= O(n^{1/2} ), 
\] 
where the inequality follows from the Jensen's inequality and $p_{1,n}(1-p_{1,n}) \leq 1/4$.
By a similar argument, $\mme[D^2]=O(n)$ and $\mmv[D] = O(n)$.

Let $\pi$ be a random permutation that is uniformly distributed over $\bG_n$ and independent of everything else.
Define $\Delta_i = 1$  for $i \leq n_1$, and $\Delta_i = -\frac{n_1}{n_2 }$ for $i > n_1$. 
From before, we have
\begin{align*}
T_n(\lbar{W}_n, \lbar{\bZ}_n^{\pi}) 
&= 
 \frac{ \sqrt{n}  }{ \sqrt{n_1}  }  \left[ \frac{1}{\sqrt{n_1}} \sum_{i=1}^{n_1} \psi_{n}( \bZ_{\pi(i)}, \ubar{P}_n )  \right]  
 -   
 \frac{ \sqrt{n}  }{ \sqrt{n_2}  } \left[ \frac{1}{\sqrt{n_2}} \sum_{i=n_1+1}^{n}  \psi_{n }( \bZ_{\pi(i)}, \ubar{P}_n )  \right] 
\\
&\hspace{1cm} + o_p(1)
\\
&\dd 
\underbrace{
    \frac{ \sqrt{n}  }{ n_1  }  \sum_{i=1}^{n} \Delta_{\pi(i)} \psi_{n}( \bZ_{i}, \ubar{P}_n )
}_{\doteq T_n^{\pi}}
+ o_p(1)
=
T_n^{\pi} + o_p(1),
\end{align*}
where $A \dd B$ means $A$ and $B$ have the same distribution. 
The test statistic that uses $\lbar{\bZ}^{* \pi_0}_n$ in the place of $\lbar{\bZ}_n$ as initial sample and then undergoes permutation $\pi$ is
\[
T_n(\lbar{W}_n, \lbar{\bZ}_n^{*\pi\pi_0}) \dd
\underbrace{ 
    \frac{ \sqrt{n}  }{ n_1  }  \sum_{i=1}^{n} \Delta_{\pi(i)} \psi_{n}( \bZ^*_{\pi_0(i)}, \ubar{P}_n )
}_{\doteq T_n^{* \pi \pi_0}}  + o_p(1) =  T_n^{* \pi \pi_0} + o_p(1). 
\]

The rest of this subsection shows that $T_{n}^{* \pi \pi_0} - T_n^{\pi} \pto 0$ as $n\to \infty$.

First, the expected value.
\begin{align*}
\mme \left[  T_{n}^{* \pi \pi_0} - T_n^{\pi} \right]   & = 
\frac{\sqrt{n}}{n_1} 
\sum_{i=1}^{n}
	\mme\left[ \Delta_{\pi(i)} \right] \mme\left[ \left(  \psi_{n}( \bZ^*_{\pi_0(i)}, \ubar{P}_n ) -  \psi_{n}( \bZ_{i}, \ubar{P}_n ) \right) \right] 
\\
	& = 0
\end{align*}
because $\pi$ is independent of everything 
and $\mme[\Delta_{\pi(i)}] = 0$.

Second, the variance.
\begin{align}
\Var\left( T_{n}^{* \pi \pi_0} - T_n^{\pi}  \right) & =  \Exp\left[ \Var\left(\left. T_{n}^{* \pi \pi_0} - T_n^{\pi} \right| D, \pi, \pi_0  \right) \right] 
\label{eq:coupling:var1}\\
& + 
\Var\left[ \Exp\left( \left. T_{n}^{* \pi \pi_0} - T_n^{\pi} \right| D, \pi, \pi_0 \right) \right].
\label{eq:coupling:var2}
\end{align}

\underline{Part \ref{eq:coupling:var1} :}
The elements in $\lbar{\bZ}^{* \pi_0}_n$ and $\lbar{\bZ}_n$ are the same except for $D$ of them.
This makes all the terms in the difference $T_{n}^{* \pi \pi_0} - T_n^{\pi}$ zero, except for at most $D$ of them.
Conditioning on
$D$,
$\pi_0$,
and 
$\pi$, the variance is
\begin{gather}
\Var\left[\left. T_{n}^{* \pi \pi_0} - T_n^{\pi} \right| D, \pi, \pi_0  \right] =  \frac{ n }{n_1^2} 
D
\Var \left[ \left. \Delta_{\pi(i)} \left(  \psi_{n}( \bZ_{\pi_0(i)}^*, \ubar{P}_n )  - \psi_{n}( \bZ_{i}, \ubar{P}_n ) \right) \right| D, \pi, \pi_0  \right] 
\nonumber
\\
 \leq \frac{ n }{n_1^2} 
D \max\left\{\left( \frac{n_1}{n_2} \right)^2, 1 \right\} \left\{ \mmv[ \psi_{n}( \bZ_{1}, \ubar{P}_n )   ] + \mmv[\psi_{n}( \bZ_{2}, \ubar{P}_n ) ] \right\} 
\\
=
\frac{n}{n_1^2 } D . O(1) 
\end{gather}
because
$n_1/n_2 = O(1)$ and
$\mmv[\psi_{n}( \bZ_{k}, \ubar{P}_n )] = O(1), ~  k=1,2,$ by Assumption \ref{aspt:asy_linear}-\eqref{eq:aspt:asy_linear_mom}.

Taking the expectation,
\begin{gather}
\Exp\left[ \Var\left(\left. T_{n}^{* \pi \pi_0} - T_n^{\pi} \right| D, \pi, \pi_0  \right) \right]  
\leq  
\frac{n}{n_1^2 }  
\Exp[  D ]
O(1)
=o(1).
\end{gather}

\underline{Part \ref{eq:coupling:var2} :}

Equation \ref{eq:coupling:var2} is bounded by 
\begin{gather} \label{eq:bound}
    \frac{nh}{\min \{n_1^2, n_2^2\}}[\theta(P_1)-\theta(P_2)+o(1)]^2 \mmv\{D\},
\end{gather}
which converges to 0 
under the null. 
However, when the null is not imposed, we need a new argument to bound the variance in 
\eqref{eq:coupling:var2}.

To this end, let $S$ be the number of observations among those $D$ observations that have $\Delta_{\pi(i)}=1$.
Conditioning on the random drawing of indices in the coupling construction (hence conditioning on $D$ and $\pi_0$),
the distribution of $S$ is Hypergeometric with $D$ draws out of $n$ elements, among which $n_1$ have $\Delta_{\pi(i)}=1$.
Then,
\begin{align*}
&\mme \left[ \left. T_{n}^{* \pi \pi_0} - T_n^{\pi} \right| D, \pi, \pi_0 \right]   
\\
& = 
\frac{\sqrt{n}}{n_1} 
	\left[
		S \left(\frac{n}{n_2} \right) - \left(\frac{n_1}{n_2} \right) D
	\right]
(-1)^{\mmi\{n_1^* \leq n_1 \} }
	\left[
		\mme \left[  \psi_{n}( \bZ_{1}, \ubar{P}_n ) \right]  
		-
		\mme \left[  \psi_{n}( \bZ_{2}, \ubar{P}_n ) \right]  
	\right].
\end{align*}
Take the variance conditional on $(D,\pi_0)$,
\begin{align*}
&\mmv\left[ \left.  \mme \left( \left. T_{n}^{* \pi \pi_0} - T_n^{\pi} \right| D, \pi, \pi_0 \right)  \right| D, \pi_0 \right] 
\\
& = 
\frac{n }{n_1^2} 
	\mmv \left[
		\left.  S 
		\left(\frac{n}{n_2} \right) - \left(\frac{n_1}{n_2} \right) D
		\right| D, \pi_0 
	\right]
	\left\{
		\mme \left[  \psi_{n}( \bZ_{1}, \ubar{P}_n ) \right]  
		-
		\mme \left[  \psi_{n}( \bZ_{2}, \ubar{P}_n ) \right]  
	\right\}^2
\\
& = 
\frac{n^3 }{n_1^2 n_2^2} 
	\mmv \left[
		\left.  S 
		\right| D, \pi_0 
	\right]
	O(1)
\\
& = 
\frac{n^3 }{n_1^2 n_2^2} 
	D \left( \frac{n_1}{n} \right) \left( \frac{n_2}{n} \right) \left( \frac{n-D}{n-1} \right)
	O(1)
\\
& = 
\frac{n^2  }{n_1 n_2 (n-1) } 
\left[ 
	D  
	-
	D^2 \left( \frac{1}{n} \right) 
\right]
	O(1),
\end{align*}
where the 
$\left\{
		\mme \left[  \psi_{n}( \bZ_{1}, \ubar{P}_n ) \right]  
		-
		\mme \left[  \psi_{n}( \bZ_{2}, \ubar{P}_n ) \right]  
\right\}^2=O(1)$ because of Assumption \ref{aspt:asy_linear}-\eqref{eq:aspt:asy_linear_mom}.
Take the expectation,
\begin{align*}
& \mme\left\{ \mmv\left[ \left.  \mme \left( \left. T_{n}^{* \pi \pi_0} - T_n^{\pi} \right| D, \pi, \pi_0 \right)  \right| D, \pi_0 \right] \right\}
\\
& = 
\frac{n^2  }{n_1 n_2 (n-1) } 
\left[ 
	\mme(D)  
	-
	\mme(D^2) \left( \frac{1}{n} \right) 
\right]\left\{
		\mme \left[  \psi_{n}( \bZ_{1}, \ubar{P}_n ) \right]  
		-
		\mme \left[  \psi_{n}( \bZ_{2}, \ubar{P}_n ) \right]  
\right\}^2,
\end{align*}
where $\left\{
		\mme \left[  \psi_{n}( \bZ_{1}, \ubar{P}_n ) \right]  
		-
		\mme \left[  \psi_{n}( \bZ_{2}, \ubar{P}_n ) \right]  
\right\}^2=O(1)$ because of Assumption \ref{aspt:asy_linear}-\eqref{eq:aspt:asy_linear_mom}.
Therefore, this bound converges to 0.
Furthermore, we have that \\
$\Exp\left[ \left. \Exp\left( \left. T_{n}^{* \pi \pi_0} - T_n^{\pi} \right| D, \pi, \pi_0 \right) \right| D, \pi_0 \right]=0$,
which makes\\
$\Var\left\{ \Exp\left[ \left. \Exp\left( \left. T_{n}^{* \pi \pi_0} - T_n^{\pi} \right| D, \pi, \pi_0 \right) \right| D, \pi_0 \right] \right\} = 0$,
and the law of total variance makes \ref{eq:coupling:var2} converge to $0$.

\subsubsection{Hoeffding's CLT}
\label{proof:theo1:hoeff}

\indent

Let $\pi$ and $\pi'$ be permutations that are mutually independent, uniformly distributed over $\bG_n$,
and independent of everything else. 
The 
goal of this section is to show that
\begin{align}
( T_{n}^{\pi }, T_{n}^{\pi'} )  \overset{d}{\to}(T,T'),
\end{align}
where $T,T'$ are independent normal random variables with the same distribution.
To this end, we first show that
\begin{align}
( T_{n}^{*\pi}, T_{n}^{*\pi'} )
= \left( \frac{\sqrt{n}}{n_1} \sum_{i=1}^n \Delta_{\pi(i)} \psi_{n}( \bZ_{i}^*, \ubar{P}_n ), \frac{\sqrt{n}}{n_1}  \sum_{i=1}^n \Delta_{\pi'(i)} \psi_{n}( \bZ_{i}^*, \ubar{P}_n ) \right)
\overset{d}{\to}(T,T').
\label{eq:hoeffing_mixed}
\end{align}

By the Cram\`er--Wold device, we need to verify that 
\begin{align}
 \sum_{i=1}^n \underbrace{\frac{\sqrt{n}}{n_1}(a\Delta_{\pi(i)} + b \Delta_{\pi'(i)})}_{\doteq C_{n,i}}\psi_{n}( \bZ_{i}^*, \ubar{P}_n )  =  \sum_{i=1}^n C_{n,i} \psi_{n}( \bZ_{i}^*, \ubar{P}_n ) 
\label{eq:cramer-wold}
\end{align}
is asymptotically normal for any choice of constants $a$ and $b$, where $a\neq 0$ or $b\neq 0$.
Note that $C_{n,1}, \ldots, C_{n,n}$ is a sequence of random variables that are independent of $\psi_{n}( \bZ_{i}^*, \ubar{P}_n )$, $i=1,\ldots, n$.

Call $\delta_n^2 = \mmv[\psi_{n}( \bZ_{i}^*, \ubar{P}_n )]$.
In order to apply  Lemma \ref{lemma:clt} and conclude that
\[\frac{\sum_{i=1}^n C_{n,i}  \psi_{n}( \bZ_{i}^*, \ubar{P}_n ) }{\delta_n \sqrt{\sum_{l=1}^n C_{n,l}^2}} \dto N(0,1),\]
we need to show that  there exists $\zeta>0$ for which
\begin{align}
& \left( \frac{\max_{i=1,\ldots, n}C_{n,i}^2}{\sum_{l=1}^n C_{n,l}^2} \right)^{\zeta/2} \mme\left| \frac{ \psi_{n}( \bZ_{i}^*, \ubar{P}_n ) }{\delta_{n}} \right|^{2+\zeta} \pto 0.
\label{lr1133_lind}
\end{align}

We verify \eqref{lr1133_lind}  in three steps.

First, we show that $\max_{i=1,\ldots, n}C_{n,i}^2=O_p(n^{-1})$:
\begin{align*}
& C_{n,i}^2 
= \frac{n}{n_1^2}\left(a^2\Delta_{\pi(i)} ^2 + 2ab\Delta_{\pi(i)} \Delta_{\pi'(i)} + b^2 \Delta_{\pi'(i)}^2\right)
= \frac{n}{n_1^2} O_p(1) ,
\\
&\max_{i=1, \ldots, n} C_{n,i}^2 =  \frac{n}{n_1^2}O_p(1) = O_p(n^{-1}).
\end{align*}

Second, we derive the probability limit of $\sum_{l=1}^n C_{n,l}^2$.
Note that:
\begin{align}
& \mme[\Delta_{\pi(i)}] = \mme[\Delta_{\pi'(i)}]  =
0,
\nonumber\\
&\mmv[\Delta_{\pi(i)}] = \mmv[\Delta_{\pi'(i)}]  = 
\frac{n_1}{n_2}, 
\nonumber \\
& \mmc[\Delta_{\pi(i)},\Delta_{\pi'(i)}] =\mme[\Delta_{\pi(i)}\Delta_{\pi'(i)}] = 0,\nonumber 
\\
& \mme\left[\sum_{i=1}^n C_{n,i}^2\right]  
\to 
\frac{1}{ \lambda (1- \lambda )} \left(a^2 +  b^2 \right),
\nonumber \\
& \mmv\left[ \sum_{i=1}^n \frac{n}{n_1^2}\left(a^2\Delta_{\pi(i)} ^2 + 2ab\Delta_{\pi(i)} \Delta_{\pi'(i)} + b^2 \Delta_{\pi(i)'}^2\right) \right]
= o(1).
\nonumber
\end{align}
Therefore,
\[\sum_{i=1}^n C_{n,i}^2 \pto \frac{1}{\lambda (1-\lambda)} \left(a^2 +  b^2 \right), \]
which is bounded away from zero.

Third, combining the two previous steps
\begin{align*}
\left( \frac{\max_{i=1,\ldots, n}C_{n,i}^2}{\sum_{l=1}^n C_{n,l}^2} \right) ^{\zeta/2} = O_p \left( n^{-\zeta/2}\right)
\end{align*}
which combined with Assumption \ref{aspt:asy_linear}-\eqref{eq:aspt:asy_linear_lind} yields  \eqref{lr1133_lind}.

Next, we derive the limiting distribution of $\sum_{i=1}^n C_{n,i}  \psi_{n}( \bZ_{i}^*, \ubar{P}_n )$, that is, without the standardization. 
Since we already know the probability limit of  $\sum_{l=1}^n C_{n,l}^2$, we need the limit of $\delta^2_n$.
Define $\ubar{P} = \lambda P_1 + (1-\lambda) P_2$.
\begin{align*}
\delta^2_n & = \mmv[\psi_{n}( \bZ_{i}^*, \ubar{P}_n )] - \xi^2(\ubar{P}_n) + \xi^2(\ubar{P}_n) - \xi^2(\ubar{P}) + \xi^2(\ubar{P})
\\
& \to \xi^2(\ubar{P}),
\end{align*}
where we used Assumption \ref{aspt:asy_linear}-\eqref{eq:aspt:asy_linear_var} and Assumption \ref{aspt:asy_linear}-\eqref{eq:aspt:asy_linear_xicont}.

Therefore, 
\begin{align*}
\sum_{i=1}^n C_{n,i} \psi_{n}( \bZ_{i}^*, \ubar{P}_n )  
& \dto 
N\left(
	0,
	\xi^2(\ubar{P})
	\left( 
		\frac{1}{ \lambda (1-\lambda)} \left(a^2 +  b^2 \right)
	\right) 
 \right).
\end{align*}

By the Cram\`er--Wold device, we conclude that 
$( T_{n}^{*\pi}, T_{n}^{*\pi'} ) \dto (T,T')$,
where $(T,T')$ is bivariate normal with zero means, equal variances 
\begin{gather}
\mmv[T]=\mmv[T'] = 		\left( 
		\frac{\xi^2(\ubar{P})}{\lambda(1- \lambda) }  
	 \right) \doteq \tau^2
	 \label{eq:var_perm_nonstud}
\end{gather}
and zero covariance. Thus, $T$ and $T'$ are  independent.

So far, we have shown that $( T_{n}^{*\pi}, T_{n}^{*\pi'} ) \dto (T,T')$.
 Consider the permutation $\pi_0$ from 
 Section \ref{proof:theo1:coupling}.
The conditions on $(\pi, \pi')$ imply that the permutations $(\pi \pi_0, \pi' \pi_0)$ are also mutually independent, uniformly distributed over $\bG_n$,
and independent of everything else.
This implies that  $( T_{n}^{*\pi\pi_0}, T_{n}^{*\pi'\pi_0} )  \dto (T,T')$.

Finally,
\[
( T_{n}^{\pi}, T_{n}^{\pi'} ) = (T_{n}^{\pi} - T_{n}^{*\pi\pi_0},  T_{n}^{\pi'} - T_{n}^{*\pi'\pi_0}) + ( T_{n}^{*\pi\pi_0}, T_{n}^{*\pi'\pi_0} )
\dto (T,T'),
\]
where we use the coupling argument from Section \ref{proof:theo1:coupling} to obtain $ T_n^{\pi} - T_n^{*\pi\pi_0} \pto 0$ and $ T_n^{\pi'} - T_n^{*\pi'\pi_0} \pto 0$,
and the Slutsky theorem.

\subsubsection{Unconditional Argument}
\label{proof:theo1:uncond}

\indent

In this subsection, we apply Lemma \ref{lemma:cond:clt} to show the conclusion of Section \ref{proof:theo1:hoeff} also holds unconditionally.
The analysis in sections \ref{proof:theo1:coupling} and \ref{proof:theo1:hoeff} are conditional on $\lbar{W}^\infty$
that satisfies $\left| \frac{n_1}{n} - \lambda\right|  \to 0$.

Unconditionally, $n_1$ is a random variable. 
Assumption \ref{aspt:sampling} implies that $  (n_1/n - \lambda) \pto 0 $ for $\lambda \in (0,1)$.
By Lemma \ref{lemma:cond:clt}, 
$ ( T_{n}^{\pi}, T_{n}^{\pi'} )   \dto (T,T')$ unconditionally.
By the Hoeffding's CLT (\citeSM{lehmann2005}, Theorem 15.2.3),
\[
\ha{R}_{T_n}(t) \pto \Phi\left(\frac{t}{\tau}\right),
\]
where $\Phi$ is the CDF of the standard normal distribution, and $\tau^2$ is given by \eqref{eq:var_perm_nonstud}.
Uniform consistency follows from continuity and monotonicity of $\Phi$ (Lemma \ref{lemma:unifcdf}).



\subsection{Proof of Theorem \ref{theo2} -  Asymptotic Distributions With Studentization}

\subsubsection{Asymptotic Distribution of Test Statistic}

Consider a sequence of $\lbar{W}_n$ 
that satisfies
$  (n_1/n - \lambda) \to 0 $ for $\lambda \in (0,1)$.
Condition on  $\lbar{W}_n$ for each $n$,
\begin{align*}
S_{n} - \sqrt{nh} \frac{\left( \theta(P_1) - \theta(P_2)\right)}{\ha \sigma_n} 
&= 
\frac{\sqrt{nh }}{\ha \sigma_n} 
\left[ 
	 \left( \ha\theta_1 -\theta(P_1) \right)  
	-   
	 \left( \ha\theta_2 - \theta(P_2) \right) 
\right]
\\
& \dto  N(0, 1) 
\end{align*}
because ${\sigma }/{\ha \sigma_n} \pto 1$ by Assumption \ref{aspt:var:est}, $T_n - \sqrt{nh} \left( \theta(P_1) - \theta(P_2)\right) \dto N(0,\sigma^2)$ by Theorem \ref{theo1},
and the Slutsky theorem.
The same is true unconditional on $\lbar{W}_n$ by
Lemma \ref{lemma:cond:clt}.

\subsubsection{Asymptotic Permutation Distribution}
\label{proof:thm2:contg}

Again, consider a sequence of $\lbar{W}_n$ 
that satisfies
$  (n_1/n - \lambda) \to 0 $ for $\lambda \in (0,1)$.
Condition on  $\lbar{W}_n$ for each $n$.
Without loss of generality, re-order observations in the sample such that  
$\lbar{\bZ}_n=(\bZ_1, \ldots, \bZ_{n_1}, \bZ_{n_1+1}, \ldots, \bZ_{n})= (\bZ_{1,1}, \ldots, \bZ_{1,n_1}, \bZ_{2,1}, \ldots, \bZ_{1,n_2})$
and
$\lbar{W}_n=(W_{1,n}, \ldots, W_{n_1,n}, W_{n_1+1,n}, \ldots, W_{n,n})= (1, \ldots, 1, 2, \ldots, 2)$.

Let $\pi$ be a random permutation that is uniformly distributed over $\bG_n$ and independent of the data.
For each $k=1,2$, let $\bV_{1,n}, \ldots, \bV_{n_k,n}$ be an \textit{iid} sample from the distribution $\ubar{P}_n$.
Assumption \ref{aspt:var:est} guarantees that
\begin{align}
\xi^2_{n_k,n}(\bV_{1,n}, \ldots, \bV_{n_k,n})  - \xi^2(\ubar{P}_n) \pto 0.
\end{align}
Lemma 5.3 and Remark A.2 by \citeSM{chung2013} show that
\begin{align*}
\xi^2_{n_1,n}(\bZ_{\pi(1)}, \ldots, \bZ_{\pi(n_1)}) - \xi^2(\ubar{P}_n) & \pto  0,
\\
\xi^2_{n_2,n}(\bZ_{\pi(n_1+1)}, \ldots, \bZ_{\pi(n)}) - \xi^2(\ubar{P}_n) & \pto  0.
\end{align*}
Using Assumption \ref{aspt:asy_linear}-\eqref{eq:aspt:asy_linear_xicont},
\begin{align*}
& \xi^2_{n_1,n}(\bZ_{\pi(1)}, \ldots, \bZ_{\pi(n_1)}) \pto \xi^2(\ubar{P}),
\\
& \xi^2_{n_2,n}(\bZ_{\pi(n_1+1)}, \ldots, \bZ_{\pi(n)})  \pto \xi^2(\ubar{P}),
\\
& \ha\sigma_n^{2 \pi} \doteq   
\frac{n}{n_1}  \xi^2_{n_1,n}(\bZ_{\pi(1)}, \ldots, \bZ_{\pi(n_1)})
+  
\frac{n}{n_2}    \xi^2_{n_2,n}(\bZ_{\pi(n_1+1)}, \ldots, \bZ_{\pi(n)})
\pto \frac{\xi^2(\ubar{P})}{\lambda(1-\lambda)} = \tau^2.
\end{align*}

By Lemma \ref{lemma:cond:clt}, $\ha\sigma_n^{2 \pi} \pto \tau^2$ unconditional on $\lbar{W}_n$.
Section \ref{proof:theo1:uncond} shows that $(T_n^\pi, T_n^{\pi'})\dto (T,T')$, where $(T,T')$ are independent normals with variance  $\tau^2$.
Given these two facts together with Theorem 5.2 by \citeSM{chung2013}, the asymptotic permutation distribution of
$T_n / \tau$ is the same as that of  $S_n = T_n / \ha\sigma_n^{2}$. 
Therefore, 
$\ha R_{S_n}(t) \pto \Phi(t)$.
Uniform consistency follows from the monotonicity and continuity of $\Phi$ (Lemma \ref{lemma:unifcdf}).

\subsection{Proof of Corollary \ref{coro1} - Asymptotic Size and Power}

\begin{proof}

 For $a \in (0,1)$, define $r(a) =  \inf\{t: \Phi (t) \geq a \} = \Phi^{-1}(a)$
and
$\ha r_n( a ) = \inf\{t: \ha{R}_{S_n}(t) \geq a\}$.
Lemma 11.2.1 by \citeSM{lehmann2005}  says that $\ha{R}_{S_n} \pto \Phi$
implies $\ha r_n(a) \pto r(a)$.
Rewrite the test $\phi$ as,

\begin{math}
\phi(\lbar W_n, \lbar \bZ_n) = 
\left\{\begin{array}{l l}
1 & \mbox{if} 
	~S_n  > \ha r_n(1 - \alpha/2 ) 
	~\mbox{or}~
	 S_n  < \ha r_n( \alpha/2 )  ,
\\
a & \mbox{if} 
	~S_n = \ha r_n(1 - \alpha/2 )
	~\mbox{or}~
	S_n  = \ha r_n( \alpha/2 )  ,
\\
0 & \mbox{if} 
	~\ha r_n( \alpha/2 )  < S_n
	< \ha r_n(1 - \alpha/2 ).
\end{array}
\right.
\end{math}

First, assume the null hypothesis $\theta(P_1) - \theta(P_2)$ holds.
For $n \to \infty$, $S_n  \dto S$, where $S$ is a standard normal.
\begin{align*}
& S_n  - \ha r_n( a ) +  r( a ) \dto S,
\\
&\mmp \left[ S_n   < \ha r_n(a) \right]   
\to  \mmp \left[ S < r(a)  \right]  = a,
\\
&\mmp \left[ S_n   > \ha r_n(a) \right]     \to  \mmp \left[ S >  r(a)   \right]  = 1- a,
\\
&\mmp \left[ S_n   = \ha r_n(a) \right]     \to  \mmp \left[ S =  r(a)    \right]  = 0.
\end{align*}

Then, the probability of rejection is 
\begin{align*}
\mme\left[ \phi(\lbar W_n, \lbar \bZ_n)  \right] &
\\
& \hspace{-2cm} = \mmp\left[
	 ~S_n (\lbar W_n, \lbar \bZ_n)  > \ha r_n(1 - \alpha/2 ) 
	 \right]
	 +
	 \mmp\left[
	 S_n (\lbar W_n, \lbar \bZ_n)  < \ha r_n( \alpha/2 )
\right]
+o(1)
\\
&\hspace{-2cm}   \to 1- (1-\alpha/2) +\alpha/2 =\alpha.
\end{align*}

\bigskip

Second, assume that $\theta(P_1) - \theta(P_2) = \eta$ and $\eta>0$ without loss of generality.
For $m$ fixed and $n\to \infty$,
\begin{align*}
& S_n - \sqrt{nh_n}  \eta /\ha{\sigma} + \sqrt{m h_m}  \eta /\ha{\sigma}  - \ha r_n( a ) +  r( a ) \dto S + \sqrt{m h_m}   \eta /{\sigma}.
\end{align*}

For $m$ fixed and $n$ larger than $m$,
\begin{align*}
\mmp \left[ S_n   < \ha r_n(a) \right]   
&\leq   \mmp \left[ S_n   - \sqrt{nh_n} \eta /\ha{\sigma}  + \sqrt{m h_m} \eta /\ha{\sigma} - \ha r_n(a) + r(a) < r(a)   \right] ,
\\
\lim_{n\to\infty }\mmp \left[ S_n   < \ha r_n(a) \right]   & 
\leq  \mmp \left[ S + \sqrt{m h_m}   \eta /{\sigma} < r(a)   \right] 
\\
&= \Phi\left(r(a)- \sqrt{m h_m}   \eta /{\sigma} \right) ,
\\
\mmp \left[ S_n   > \ha r_n(a) \right]   
&\geq   \mmp \left[ S_n   - \sqrt{nh_n} \eta /\ha{\sigma}  + \sqrt{m h_m} \eta /\ha{\sigma} - \ha r_n(a) + r(a) > r(a)   \right] ,
\\
\lim_{n\to\infty } \mmp \left[ S_n   > \ha r_n(a) \right]
&\geq  \mmp \left[ S + \sqrt{m h_m}   \eta /{\sigma} > r(a)   \right] \\
&= 1- \Phi\left(r(a)- \sqrt{m h_m}   \eta /{\sigma} \right). 
\end{align*}

Take limits as $m\to\infty$ from both sides,
\begin{align*}
\lim_{n\to\infty }\mmp \left[ S_n   < \ha r_n(a) \right]   & \leq \lim_{m \to\infty } \Phi\left(r(a)- \sqrt{m h_m}   \eta /{\sigma} \right) =0 ,
\\
\lim_{n\to\infty } \mmp \left[ S_n   > \ha r_n(a) \right]   & \geq  \lim_{m\to\infty } 1- \Phi\left(r(a)- \sqrt{m h_m}   \eta /{\sigma} \right) = 1. 
\end{align*}

Then,  
\begin{align*}
\mme\left[ \phi(\lbar W_n, \lbar \bZ_n)  \right] 
\to 1.
\end{align*}

Moreover, there is no loss in power in using permutation critical values. The asymptotic test rejects when $S_n > r(1-\alpha/2)$ or $S_n < r(\alpha/2)$, where $r(a) = \Phi^{-1}(a)$ is nonrandom. Suppose
\[S_n (\lbar W_n, \lbar \bZ_n) \dto \mathcal L_{\eta}\]
for some $\mathcal L_{\eta}$ under a sequence of alternatives that are contiguous to some distribution satisfying the null hypothesis
and
$\theta(P_{1n})-\theta(P_{2n})=\eta/\sqrt{n h_n}$.
Then the power of the test against local alternatives would tend to $1-\mathcal L_{\eta}(\Phi^{-1}(1-\alpha/2)) + \mathcal L_{\eta}(\Phi^{-1}(\alpha/2))$. 
For the permutation test, we have that $\hat r_{n}$ obtained from the permutation distribution satisfies $\hat r_n(a) \pto \Phi^{-1}(a)$ under the null hypothesis. 
The same results follows under the sequence of contiguous alternatives, thus implying that the permutation test has the same limiting local power as the asymptotic test which uses nonrandom critical values. 

\end{proof}

\subsection{Proof of Corollary \ref{coro2} - Confidence Set }
\label{proof:cset}

\begin{proof}

Fix $n$ and $\bQ_n$ arbitrary.
Pick any pair $(P_1,P_2)$.
Call $\delta = \Psi( P_1, P_2 )$.
By assumption, $\psi_\delta P_1  = P_2$.
Lemma \ref{lemma:exact} says that   
$\mme[\phi_{\delta}(\lbar{W}_n,\lbar{\bZ}_n)] = \mme[\phi_{\delta_0}(\lbar{W}_n,\ti{\lbar{\bZ}}_n)] = \alpha$.
Therefore,
\begin{align*}
\mmp\left[\Psi( P_1,  P_2 ) \in C_n (\lbar{W}_n,\lbar{\bZ}_n)\right] 
&= 
\mmp\left[U > \phi_{\delta }(\lbar{W}_n,\lbar{\bZ}_n) \right]
\\
&= \mme\left[ \mmp\left( U > \phi_{\delta }(\lbar{W}_n,\lbar{\bZ}_n)  |  \lbar{W}_n,\lbar{\bZ}_n \right) \right]
\\
& = 1- \mme\left[ \phi_{\delta }(\lbar{W}_n,\lbar{\bZ}_n)  \right]
=1- \alpha.
\end{align*}

Now, suppose $\psi_\delta P_1  \neq P_2$ and Assumptions \ref{aspt:asy_linear}--\ref{aspt:var:est} hold. 
Corollary \ref{coro1} says that \\
$\mme\left[ \phi_{\delta }(\lbar{W}_n,\lbar{\bZ}_n)  \right] \to \alpha$.
Take the limit as $n \to \infty$ on both sides of the equality above.
It follows that the asymptotic coverage of $C_n (\lbar{W}_n,\lbar{\bZ}_n)$ is $1-\alpha$.

\end{proof}

\section{Example \ref{example:param} (Parametric Model) - Simulation}\label{sec:param}

\indent

To illustrate the conclusions of Theorems \ref{theo1} and \ref{theo2}, we consider the following simulation design:
for $k=1,2$, $X_k \sim U[0,1], \varepsilon_k \sim N(0, v_k)$, where $X_k$ is independent of $\varepsilon_k$, and $Y_k=\theta(P_k) + (X_k-0.5) + \varepsilon_k$; the sample sizes are $n_1=20$ and $n_2 = 980$, and the variances are $v_1 = 5$ and $v_2 = 1$. 
We simulate 10,000 samples under the null hypothesis $H_0: \theta(P_1) = \theta(P_2) = 0$ and use 1,000 permutations for each sample. 
The left-hand side of Figure \ref{fig:parametric} plots the permutation distribution and the sampling distribution when the test statistic is not studentized. As predicted by Theorem \ref{theo1}, the variance of the permutation distribution is different from that of the sampling distribution because both $n_1 \neq n_2$ and $v_1 \neq v_2$. In our simulation, critical values obtained from the permutation distribution are smaller than critical values from the sampling distribution, leading to over-rejection. 
In contrast, the permutation distribution based on the studentized statistic is approximately equal to the sampling distribution as depicted on the right-hand side of Figure \ref{fig:parametric}. The permutation test now has a correct size (Theorem \ref{theo2}). 

\begin{figure}[H]
\caption{Permutation Distribution vs. Sampling Distribution}
\label{fig:parametric}
\begin{center}
\includegraphics[width=6.8in]{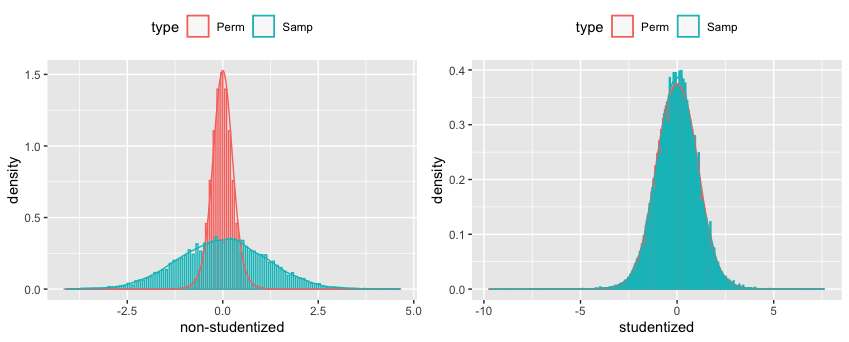}
\end{center}
\end{figure}

\section{Proof of Applications}
\label{sec:app:b}

\indent 

This appendix gathers the proofs of Propositions \ref{prop:control_mean} -- \ref{prop:bunching}. 
Key steps in these proofs consist of demonstrating uniform convergence over $\m{P}$.
By uniformly over  $\m{P}$ we mean over $P$ distribution of $\bV$ and $P$ argument of functions, e.g., $f_R(x;P)$ or $m_{S|R}(x;P)$.
For $A_n(P)$ random function of $P$, with distribution depending on $P$, we use $A_n(P) = o_{\m{P}}(1)$ to denote $\sup_{P\in \m{P}} \mmp_{P}[ |A_n(P)| > \eps ] \to 0 $;
we also use $A_n(P) = O_{\m{P}}(1)$ to denote that, for every $\delta>0$, there exists $M_{\delta}<\infty$ such that 
$\sup_{P\in \m{P}} \mmp_{P}[ |A_n(P)| > M_{\delta} ] <\delta $.
The same notation is applied when $A_n(P)$ is a deterministic function of $P$.
See Definition \ref{def:Oop} and Lemma \ref{lemma:Oop} in  Appendix \ref{sec:app:c}.

\subsection{Proof of Proposition \ref{prop:control_mean} -  Controlled Means}
\label{proof:control_mean}

\indent 

The goal of this proof is to use the assumptions listed in Proposition \ref{prop:control_mean}  to verify Assumptions \ref{aspt:asy_linear} and \ref{aspt:sampling}.
It builds on standard arguments from the literature on nonparametrics.
See, for example, Theorem 2.2 by \citeSM{liracine2007}.

Consider an \textit{iid} sample from $P\in\m{P}$ with $m$ observations, $\bV_1=(R_1,S_1), \ldots, \bV_m=(R_m,S_m)$.
The number $m$ grows with $n$ such that $m_n/n \to \gamma$, for some $\gamma \in (0,1)$.
The parameter of interest is $\theta(P) = \mme[S|R=x]$, and 
the NW estimator is 
\[
\ha{\theta}^b 
= \theta_{m,n}^b(\bV_1 , \ldots, \bV_m) 
= \frac{\sum\limits_{i=1}^{m} K \left(  \frac{ R_{i}-x }{ h }  \right)  S_{i} }
{\sum\limits_{i=1}^{m} K \left(  \frac{ R_{i}-x }{ h }  \right) }.
\]
In this section we study the asymptotic representation of
the bias-corrected estimator:
$\ha{\theta} = \ha{\theta}^b - \theta(P) - h^2 \ha{B}   $.
$\ha{B}$ is a consistent estimator for the bias term $B(P)$ (Equation \ref{eq:control_quan:BP} below).

The assumptions in Proposition \ref{prop:control_mean}  imply the following facts:
\begin{enumerate}
\item \textit{As $m\to\infty$, $h \to 0$,  $m h \to \infty$, 
$\sqrt{mh}h^2 = O(1)$, and $\sqrt{mh}h^3 = o(1)$};

\item \label{fact:control_mean:prop_kernel} 
\textit{ $\int K(u) u ~ du =0$, $\int K^r(u) u^s g(u) ~ du < \infty$  for $1 \leq r <\infty$,  $0 \leq s \leq 3$, and bounded function $g(u)$};

\item \textit{The distribution of $R$ has PDF $f_{R}(r;P)$
that is three times differentiable with respect to (henceforth wrt) $r$: 
$\nabla_r f_{R}(r;P)$, $\nabla_{r^2} f_{R}(r;P)$, and
$\nabla_{r^3} f_{R}(r;P)$;
these derivatives are bounded as functions of $(r,P)$;
$f_{R}(r;P)$ is bounded away from zero as a function of $(r,P)$;}

To see this, note that $f_{R}(r;P)$ is a convex combination of $f_{X_1}(r)$ and $f_{X_2}(r)$,
each bounded, with bounded derivatives, and bounded away from zero.

\item \label{fact:control_mean:prop_m} \textit{$m_{S|R}(r;P) = \Exp_P[S|R=r]$ has first, second, and third derivatives wrt $r$ denoted 
$\nabla_r m_{S|R}(r;P)$,
$\nabla_{r^2} m_{S|R}(r;P)$,
and $\nabla_{r^3} m_{S|R}(r;P)$, respectively;
$m_{S|R}$, $\nabla_r m_{S|R}$, $\nabla_{r^2} m_{S|R}$, and
$\nabla_{r^3} m_{S|R}$ are all  bounded as functions of $(r,P)$;}

To see this, take  $P = \alpha P_1 + (1-\alpha) P_2$ and note that, 
\[
m_{S|R}(r;P)  = \frac{\alpha f_{X_1}(r) }{ f_{R}(r;P) } m_{Y_1|X_1}(r)   + \frac{(1-\alpha) f_{X_2}(r) }{ f_{R}(r;P) } m_{Y_2|X_2}(r).
\]
The expectations $m_{Y_k|X_k}(r)$ are bounded functions of $r$.
The weights $\omega_1(r;P) \doteq \alpha \frac{f_{X_1}(r) }{ f_{R}(r;P) }$ and
$\omega_2(r;P) \doteq  (1-\alpha) \frac{f_{X_2}(r) }{ f_{R}(r;P) }$
 are bounded functions of $(r,P)$ because they are positive and sum to $1$.
The expectations $m_{Y_k|X_k}(r)$ and the weights $\omega_k(r;P)$ are three  times differentiable wrt $r$.
The derivatives of $m_{Y_k|X_k}(r)$ are bounded wrt $r$. 
The  derivatives of the weights are bounded because the derivatives of the PDFs $f_{X_k}(r)$  are bounded plus the fact that $f_{R}(r;P)$ is bounded away from zero over $(r,P)$.

\item \textit{$v_{S|R}(r;P) = \Var_P[S|R=r]$ has first derivative wrt $r$ denoted $\nabla_r v_{S|R}(r;P)$;
$v_{S|R}$, $\nabla_r v_{S|R}$ are both  bounded as functions of $(r,P)$;
$v_{S|R}(x;P)$ is bounded away from zero as a function of  $P$;}

Again, $v_{S|R}(r;P) = m_{S^2|R}(r;P) -  m_{S|R}^2(r;P)$, where
\begin{align*}
&m_{S^2|R}(r;P)  = \omega_1(r;P) m_{Y_1^2|X_1}(r)   + \omega_2(r;P) m_{Y_2^2|X_2}(r) 
\\
& \hspace{2cm} =\omega_1(r;P) \left[ v_{Y_1|X_1}(r) + m_{Y_1|X_1}^2(r) \right]  + \omega_2(r;P) \left[ v_{Y_2|X_2}(r) +  m_{Y_2|X_2}^2(r) \right],
\\
& m_{S|R}^2(r;P) = \left[ \omega_1(r;P) m_{Y_1|X_1}(r)   + \omega_2(r;P) m_{Y_2|X_2}(r)  \right]^2.
\end{align*}
A similar argument to Fact \ref{fact:control_mean:prop_m} shows that $v_{S|R}$ and $\nabla_r v_{S|R}$ are bounded functions of $(r,P)$.
Next, $v_{Y_k|X_k}(x)= m_{Y_k^2 |X_k}(x) - m_{Y_k|X_k}^2(x)$ is bounded away from zero,
so that
$m_{Y_k^2 |X_k}(x)$ is bounded away from $ m_{Y_k|X_k}^2(x)$.
It follows that $m_{S^2|R}(x;P) = \omega_1(x;P) m_{Y_1^2|X_1}(x)   + \omega_2(x;P) m_{Y_2^2|X_2}(x)$ is bounded away
from 
$\omega_1(x;P) m_{Y_1|X_1}^2(x)   + \omega_2(x;P) m_{Y_2|X_2}^2(x)$,
which is greater than or equal to
$m_{S|R}^2(x;P)$.
Thus,  $v_{S|R}(x;P)$ is bounded away from zero as a function of $P$.

\item \textit{Define $\eta(r;P) = \mme_P[|S - m_{S | R }(r;P) |^{2+\zeta} | R=r]$.
$\eta(r;P)$ is a bounded function of $(r,P)$.}
By the $c_r$-inequality,
\begin{align*}
\eta(r;P) \leq & 2^{1+\zeta} \mme_P[|S| ^{2+\zeta} | R=r]  + 2^{1+\zeta} | m_{S | R }(r;P) |^{2+\zeta} 
\\
= & \omega_1(r;P) \mme[|Y_1  |^{2+\zeta} | X_1=r] +  \omega_2(r;P) \mme[|Y_2  |^{2+\zeta} | X_2=r] 
\\
+ &  2^{1+\zeta} | m_{S | R }(r;P) |^{2+\zeta} ,
\end{align*}
which is a bounded function of $(r,P)$ because 
the weights $\omega_k(r;P)$  are bounded, 
$\mme[|Y_k  |^{2+\zeta} | X_k=r]$
are bounded,
and $m_{S | R }(r;P)$ is bounded.

\end{enumerate}

We re-write 
$\sqrt{m h}  \left( \ha\theta -\theta(P) \right)$
$= \sqrt{m h}  \left( \ha\theta^b -h^2 \ha{B} -\theta(P) \right)$
to find the asymptotic linear representation.
\begin{align}
\sqrt{m h}  \left( \ha\theta -\theta(P) \right)   
&= 
\left( 
	\frac{1}{\sqrt{m h} } \sum_{i=1}^{m} K\left( \frac{ R_{i} - x }{h} \right) 
		\left( S_{i} - m_{S|R}( R_i;P) \right) 
\right)
f_{R}(x;P)^{-1}
\label{eq:control_mean:thetahat:clt}
\\
&+ 
\left( 
	\frac{1}{\sqrt{m h} } \sum_{i=1}^{m} K\left( \frac{ R_{i} - x }{h} \right) 
		\left( S_{i} - m_{S|R}( R_i;P) \right) 
\right)
\notag
\\
&\hspace{2cm}
\left[
	\left( \frac{1}{ m h } \sum_{i=1}^{m} K\left( \frac{ R_{i} - x }{h} \right) \right)^{-1} 
	-
	f_{R}(x;P)^{-1}  
\right]
\label{eq:control_mean:thetahat:deno}
\\
&+ 
\left( 
	\frac{1}{\sqrt{m h} } \sum_{i=1}^{m} K\left( \frac{ R_{i} - x }{h} \right) 
		\left( m_{S|R}( R_i;P)  - m_{S|R}( x;P) \right) 
\right)
\notag
\\*
&\hspace{2cm}
\left( 
	\frac{1}{ m h } \sum_{i=1}^{m} K\left( \frac{ R_{i} - x }{h} \right) 
\right)^{-1} 
- \sqrt{mh} h^2 \ha{B}.
\label{eq:control_mean:thetahat:bias}
\end{align}

\begin{enumerate}

\item \underline{Assumption \ref{aspt:asy_linear} - \eqref{eq:aspt:asy_linear}:} asymptotic expansion.

 Equation \ref{eq:control_mean:thetahat:clt} above gives  the influence function $\psi_n$.

\begin{align*}
\frac{1}{\sqrt{m} } \sum_{i=1}^{m} \underbrace{ 
		K\left( \frac{ R_{i} - x }{h} \right) 
		\left( S_{i} - m_{S|R}( R_i;P) \right) h^{-1/2} f_{R}^{-1}(x;P)
}_{\doteq \psi_{n}(\bV_{i}, P)}
=
\frac{1}{\sqrt{m} } \sum_{i=1}^{m} \psi_{n}(\bV_{i}, P).
\end{align*}

We need to show that Equations \ref{eq:control_mean:thetahat:deno} and \ref{eq:control_mean:thetahat:bias} converge in probability to zero uniformly over $\m{P}$.

\noindent
\textit{Equation \ref{eq:control_mean:thetahat:deno}:} is $o_{\m{P}}(1)$. 
We show this in 3 steps.

First,
\begin{align*}
\mmv_P\left( 
	\frac{1}{\sqrt{m h} } \sum_{i=1}^{m} K\left( \frac{ R_{i} - x }{h} \right) 
		\left( S_{i} - m_{S|R}(R_i;P) \right) 
\right)
\\
 = \int K^2(u) v_{S|R}(x+uh ; P)  f_{R}(x+uh;P) ~ du,
\end{align*}
which is bounded over $\m{P}$ because of the kernel properties (Fact \ref{fact:control_mean:prop_kernel} above),  and $v_{S|R}(r;P)$ and $f_{R}(r;P)$ are bounded functions of $(r, P)$. 
Next,
\begin{align*}
\mme_P\left( 
	\frac{1}{\sqrt{m h} } \sum_{i=1}^{m} K\left( \frac{ R_{i} - x }{h} \right) 
		\left( S_{i} - m_{S|R}(R_i;P) \right) 
\right)
= 0.
\end{align*}
Use Lemma \ref{lemma:Oop}, part 2, to conclude that 
\[
	\frac{1}{\sqrt{m h} } \sum_{i=1}^{m} K\left( \frac{ R_{i} - x }{h} \right) 
		\left( S_{i} - m_{S | R}(R_i;P) \right) 
=O_{\m{P}}(1).
\]

Second, 
\begin{align*}
\mme_P\left[ \frac{1}{ m h } \sum_{i=1}^{m} K\left( \frac{ R_{i} - x }{h} \right) \right] - f_{R}(x;P) 
= \int K(u) \nabla_r f_{R}(x^*_{uh};P) uh~ du,
\end{align*}
where $x^*_{uh}$ is a point between $x$ and $x+uh$.
The expression above converges to zero uniformly over $\m{P}$ because of kernel properties (Fact \ref{fact:control_mean:prop_kernel}),  the derivative $\nabla_r f_{R}(r;P)$ is a bounded function of $(r,P)$,
and $h \to 0 $ as $m\to\infty$.
Next, the variance of the same term.
\begin{align*}
\mmv_P\left[ \frac{1}{ m h } \sum_{i=1}^{m} K\left( \frac{ R_{i} - x }{h} \right) \right]  
&=
\frac{1}{ m h^2 } \mmv_P\left[  K\left( \frac{ R_{i} - x }{h} \right) \right]  
\\
&\leq
\frac{1}{ m h^2 } \mme_P\left[ K^2\left( \frac{ R_{i} - x }{h} \right) \right]  
\\
& = \frac{1}{ m h } \int K^2(u) f_{R}(x +uh;P ) ~ du,
\end{align*}
which converges to zero uniformly over $\m{P}$ because $mh\to\infty$,  kernel properties (Fact \ref{fact:control_mean:prop_kernel}), and $f_{R}(r;P)$ is a bounded function of $(r,P)$.
Use Lemma \ref{lemma:Oop}, parts 2 and 3 to arrive at:
\begin{align}
&\frac{1}{ m h } \sum_{i=1}^{m} K\left( \frac{ R_{i} - x }{h} \right)  - f_{R}(x;P)  = o_{\m{P}}(1).
\label{eq:control_mean:fhat1}
\\
&\left( \frac{1}{ m h } \sum_{i=1}^{m} K\left( \frac{ R_{i} - x }{h} \right) \right)^{-1} 
-
f_{R}^{-1} (x;P) 
= o_{\m{P}}(1).
\label{eq:control_mean:fhat2}
\end{align}

Third, combine steps 1 and 2 and use Lemma \ref{def:Oop} - part 1:
\begin{align*}
&\left( 
	\frac{1}{\sqrt{m h} } \sum_{i=1}^{m} K\left( \frac{ R_{i} - x }{h} \right) 
		\left( S_{i} - m_{S  | R }(R_i;P) \right) 
\right)
\\
&\hspace{.5cm}\left[
	\left( \frac{1}{ m h } \sum_{i=1}^{m} K\left( \frac{ R_{i} - x }{h} \right) \right)^{-1} 
	-
	\left( f_{R}(x;P) \right)^{-1}
\right]
\\
& = O_{\m{P}}(1) o_{\m{P}}(1) = o_{\m{P}}(1).
\end{align*}

\bigskip

\noindent
\textit{Equation \ref{eq:control_mean:thetahat:bias}:} is $o_{\m{P}}(1)$.
We derive the probability limit of 
$\eqref{eq:control_mean:thetahat:bias} + \sqrt{m h} h^2 \ha{B}$
in 3 steps.

First, 
$\left( \frac{1}{ m h } \sum_{i=1}^{m} K\left( \frac{ R_{i} - x }{h} \right) \right)^{-1}
-
f_R^{-1}(x;P)= o_{\m{P}}(1) $ 
by what was shown above (Equations \ref{eq:control_mean:fhat1} and \ref{eq:control_mean:fhat2}).

Second, 
\begin{align*}
& 
\mme_P\left[ 
	\frac{1}{\sqrt{m h} } \sum_{i=1}^{m} K\left( \frac{ R_{i} - x }{h} \right) 
		\left( m_{S  | R }(R_i;P) - m_{S  | R }(x;P)  \right) 
\right]
\\
=& 
\sqrt{mh}  \int K\left( u \right) [m_{S  | R }(x+uh;P) - m_{S  | R }(x;P)] f_{R}(x+uh;P) ~ du
\\
=& 
\sqrt{mh} \int K\left( u \right) 
\left[ 
    \nabla_r m_{S  | R }(x;P) uh 
    + \nabla_{r^2} m_{S  | R }(x;P) u^2 h^2/2 
    + \nabla_{r^3} m_{S  | R }(x^{*}_{uh};P) u^3 h^3/6 
\right]   
\\
& \hspace{2.7cm} 
    f_{R}(x + uh;P) ~ du
\\
=& 
\sqrt{mh} \int K\left( u \right) 
    \nabla_r m_{S  | R }(x;P)  u h
    [f_{R}(x;P) + \nabla_r f_{R}(x;P) uh   
    + \nabla_{r^2} f_{R}(x^{**}_{uh};P)u^2 h^2/2    ] ~ du
\\
&+ 
\sqrt{mh} \int K\left( u \right) 
    \nabla_{r^2} m_{S  | R }(x;P) u^2 (h^2/2) 
       [f_{R}(x;P) + \nabla_r f_{R}(x^{***}_{uh} ; P) uh] 
       ~ du
\\
&+ 
\sqrt{mh}  \int K\left( u \right) 
    \nabla_{r^3} m_{S  | R }(x^{*}_{uh};P) u^3  (h^3/6)  
    f_{R}(x + uh;P) ~ du,
\end{align*}
where we use the existence of derivatives for the Taylor expansions and that $x^*_{uh},x^{**}_{uh},$ and  $x^{***}_{uh}$ are points between $x+uh$ and $x$.
Define $\kappa_{s,t} = \int u^s K^t(u) ~du$.
The last equation above equals to
\begin{align*}
=& 
\sqrt{mh} h \nabla_r m_{S  | R }(x;P)   
    f_{R}(x;P)
    \int u K\left( u \right) 
    ~ du
\\
&+ 
\sqrt{mh} h^2 \nabla_r m_{S  | R }(x;P) 
\nabla_r f_{R}(x;P)   
\int u^2 K\left( u \right) 
     ~ du
\\
& +
\sqrt{mh} (h^3/2) 
\nabla_r m_{S  | R }(x;P)
\int u^3 K\left( u \right) 
\nabla_{r^2} f_{R}(x^{**}_{uh};P)      ~ du
\\
&+ 
\sqrt{mh} (h^2/2) 
\nabla_{r^2} m_{S  | R }(x;P)  
f_{R}(x;P)  
\int u^2 K\left( u \right) 
       ~ du
\\
&+ 
\sqrt{mh} (h^3/2) 
 \nabla_{r^2} m_{S  | R }(x;P)   
 \int u^3 K\left( u \right) 
       \nabla_r f_{R}(x^{***}_{uh} ; P)  
       ~ du
\\
&+ 
\sqrt{mh} (h^3/6) \int u^3 K\left( u \right) 
    \nabla_{r^3} m_{S  | R }(x^{*}_{uh};P)   
    f_{R}(x + uh;P) ~ du,
\\
=& 0
\\
&+ 
\sqrt{mh} h^2 \nabla_r m_{S  | R }(x;P) 
\nabla_r f_{R}(x;P)   
\kappa_{2,1}
\\
& +
O_{\m{P}}\left( \sqrt{mh} h^3\right)
\\
&+ 
\sqrt{mh} h^2
\nabla_{r^2} m_{S  | R }(x;P)  
f_{R}(x;P)  
\frac{\kappa_{2,1}}{2}
\\
&+ 
O_{\m{P}}\left( \sqrt{mh} h^3\right)
\\
&+ 
O_{\m{P}}\left( \sqrt{mh} h^3\right)
\\
=&
\sqrt{mh} h^2 \underbrace{
\kappa_{2,1} 
\left[ 
    \nabla_r m_{S  | R }(x;P) 
\nabla_r f_{R}(x;P)
+
\nabla_{r^2} m_{S  | R }(x;P)  
f_{R}(x;P)  /2
\right]}_{\doteq B_0(P) }
+ o_{\m{P}}\left( 1\right)
\\
=& \sqrt{mh} h^2 B_0(P) + o_{\m{P}}\left( 1\right),
\end{align*}
where we use the following: (i) $\int K(u) u ~ du =0$ and other  kernel properties (Fact \ref{fact:control_mean:prop_kernel}); 
(ii) $\nabla_{r^3}  m_{S|R} (r;P)$,
$f_{R}(r;P)$,  $\nabla_r f_{R}(r;P)$, and
$\nabla_{r^2} f_{R}(r;P)$ are bounded functions of $(r,P)$;
(iii) $\sqrt{mh} h^2 =O(1)$ and $\sqrt{mh} h^3 =o(1)$.

Next, the variance. 
\begin{align*}
& 
\mmv_P\left[ 
	\frac{1}{\sqrt{m h} } \sum_{i=1}^{m} K\left( \frac{ R_{i} - x }{h} \right) 
		\left( m_{S  | R }(R_i;P)  -  m_{S  | R }(x;P)   \right) 
\right]
\\
=&
\frac{1}{ h } \mmv_P\left[ 
	 K\left( \frac{ R_{i} - x }{h} \right) 
		\left( m_{S  | R }(R_i;P)  -  m_{S  | R }(x;P)   \right) 
\right]
\\
\leq & 
\frac{1}{ h } \mme_P \left[  
		K^2\left( \frac{ R_{i} - x }{h} \right) 
		\left( m_{S  | R }(R_i;P)  -  m_{S  | R }(x;P)  \right)^2 
\right]
\\
=& 
h^2 \int K^2\left( u \right) [\nabla_r m_{S  | R }(x^*_{uh};P)]^2 u^2  f_{R}(x+uh;P) ~ du
\\
=& O_{\m{P}} \left(  h^2  \right) = o_{\m{P}} \left( 1  \right),
\end{align*}
where we use that 
(i) $\nabla_r  m_{S|R} (r;P)$ and  $f_{R}(r;P)$ are bounded functions of $(r,P)$;
(ii) $ h^2 \to 0$;
and
(iii) kernel properties (Fact \ref{fact:control_mean:prop_kernel}).

Apply  Lemma \ref{lemma:Oop}- part 2 to get
\[
	\frac{1}{\sqrt{m h} } \sum_{i=1}^{m} K\left( \frac{ R_{i} - x }{h} \right) 
		\left( m_{ S  | R }(R_i;P)  - m_{ S  | R }(x;P) \right) 
= 
\sqrt{m h}h^2 B_0(P) + o_{\m{P}}(1).
\]

Third, combine the first and second steps and apply Lemma \ref{lemma:Oop}- part 1 to arrive at
\begin{align}
& 
\left( 
	\frac{1}{\sqrt{m h} } \sum_{i=1}^{m} K\left( \frac{ R_{i} - x }{h} \right) 
		\left( m_{ S  | R }(R_i;P) - m_{ S  | R }(x;P) \right) 
\right)
\left( 
	\frac{1}{ m h } \sum_{i=1}^{m} K\left( \frac{ R_{i} - x }{h} \right) 
\right)^{-1}
\notag
\\
& =
 \sqrt{mh}h^2 
 \underbrace{
 \kappa_{2,1} 
\left[ 
    \frac{\nabla_r m_{S  | R }(x;P) 
\nabla_r f_{R}(x;P)}{f_{R}(x;P)}
+
\frac{\nabla_{r^2} m_{S  | R }(x;P)}{2}
\right]}_{\doteq B(P)}
 +o_{\m{P}}(1)
\label{eq:control_mean:BP} 
 \\
& = \sqrt{mh}h^2 B(P) + o_{\m{P}}(1),
\notag
\end{align}
which gives the bias term $B(P)$.
A consistent estimator $\ha{B}$ is,
\[
\ha{B} = B_{m,n}(\bV_1 , \ldots, \bV_m) \doteq  
\kappa_{2,1} \left[ 
        \frac{\ha{\nabla_r m_{S|R}}(x;P) \ha{\nabla_r f_{R}}(x;P) }
        { \ha{f_{R}}(x;P) }  +
        \frac{ \ha{\nabla_{r^2} m_{S|R}}(x;P) }{2}
    \right],
\]
that is, by replacing 
$f_{R}(x;P)$,
$\nabla_r f_{R}(x;P)$,
$\nabla_r m_{S|R}(x;P)$,
and
$\nabla_{r^2} m_{S|R}(x;P)$
in $B(P)$ by consistent nonparametric estimators readily available in the literature.
The tuning parameters for these additional estimators and corresponding moment conditions may be set such that the bias-correction condition (Assumption \ref{aspt:prop:bias}) 
is met.

Finally, \ref{eq:control_mean:thetahat:bias} equals to,
\begin{align*}
& 
\left( 
	\frac{1}{\sqrt{m h} } \sum_{i=1}^{m} K\left( \frac{ R_{i} - x }{h} \right) 
		\left( m_{ S  | R }(R_i;P) - m_{ S  | R }(x;P) \right) 
\right)
\left( 
	\frac{1}{ m h } \sum_{i=1}^{m} K\left( \frac{ R_{i} - x }{h} \right) 
\right)^{-1}
\\
&- \sqrt{mh}h^2 \ha{B} 
\\
& = \sqrt{mh}h^2 \left( B(P) - \ha{B} \right) + o_{\m{P}}(1)
= o_{\m{P}}(1),
\end{align*}
where we use $\sqrt{mh}h^2 = O(1)$ and 
Assumption \ref{aspt:prop:bias}
that says $ B(P) - \ha{B}  = o_{\m{P}}(1)$.

\item \underline{Assumption \ref{aspt:asy_linear} - \eqref{eq:aspt:asy_linear_zeromean}:} zero mean of influence function. 

 $\mme_P[ \psi_{n}(\bV_{i}, P) ]=0 ~~\forall P$ by construction.

\item \underline{Assumption \ref{aspt:asy_linear} - \eqref{eq:aspt:asy_linear_var}:} variance of influence function.

 Define $\xi^2(P) = \kappa_{0,2} v_{ S  | R} (x;P)  /f_{R}(x;P)$, where 
$\kappa_{0,2} = \int_{-\infty}^{\infty}K^2(u) ~ du$.
\begin{align*}
&\mmv_P\left( 
	\frac{1}{f_{R}(x;P) \sqrt{h} }  K\left( \frac{ R_{i} - x }{h} \right) 
		\left( S_{i} - m_{ S  | R} (R_i;P)  \right) 
\right) - \xi^2(P) 
\\
 =&\frac{1}{f_{R}^2(x;P) }  \int K^2(u)  \left\{ 
 	 \underbrace{v_{ S  | R }(x+uh;P)  f_{R}(x+uh;P)}_{\doteq g(uh;P)} - 
 	\underbrace{ v_{ S | R}(x;P)  f_{R}(x;P) 
 	 }_{\doteq g(x;P)} \right\}
 ~ du
 \\
=&\frac{h}{f_{R}^2(x;P) }  \int K^2(u) \nabla_r g(x^*_{uh};P ) u  ~ du = o_{\m{P}}(1).
\end{align*}
where $\nabla_r g(r;P )$ denotes the derivative of $g(r;P)$ wrt $r$.
The expression above converges to zero uniformly over $\m{P}$
because
$h\to0$,
Fact \ref{fact:control_mean:prop_kernel} on the kernel,
the derivative $\nabla_r g(r;P )$ is a bounded function of $(r,P)$,
and 
$f_{R}(x;P)$ is bounded away from zero as a function of $P$.
Therefore,
\[
\sup\limits_{P \in \m{P} } \left| \mmv_P[  \psi_n(\bV_{i}, {P} ) ] - \xi^2(P) \right| \to 0.
\]

\item \underline{Assumption \ref{aspt:asy_linear} - \eqref{eq:aspt:asy_linear_mom}:} $ \sup\limits_{P \in \m{P} } \mme [\psi_n^2(\bZ_{k,i},P)]  < \infty $ for $k=1,2$.

We have,
\begin{align*}
& \psi_n(\bZ_k,P) = K\left( \frac{ X_{k} - x }{h} \right) 
		\left( Y_{k} - m_{S|R}(X_k;P) \right) h^{-1/2} f_{R}^{-1}(x;P).
\\
& \psi_n^2(\bZ_k,P) = \frac{K^2\left( \frac{ X_{k} - x }{h} \right)}{ h f_{R}^{2}(x;P)} 
		\left[ \left(Y_{k} -  m_{Y_k|X_k}(X_k) \right) + \left(m_{Y_k|X_k}(X_k)  - m_{S|R}(X_k;P) \right)  \right]^2 
\\
 & = \frac{K^2\left( \frac{ X_{k} - x }{h} \right)}{ h f_{R}^{2}(x;P) } 
		\left[ \left( Y_{k} -  m_{Y_k|X_k}(X_k) \right)^2 + \left(m_{Y_k|X_k}(X_k)  - m_{S|R}(X_k;P)  \right)^2 \right.
\\
& \hspace{2.7cm} \left. + 2 \left(Y_{k} -  m_{Y_k|X_k}(X_k) \right) \left(m_{Y_k|X_k}(X_k)  - m_{S|R}(X_k;P) \right)   \right].
\end{align*}
\begin{align*}
&\mme\left[ \psi_n^2(\bZ_k,P) | X_k \right] 
= 
 \frac{K^2\left( \frac{ X_{k} - x }{h} \right)}{ h f_{R}^{2}(x;P) } 
		\left[ v_{Y_k|X_k}(X_k) + \left(m_{Y_k|X_k}(X_k)  - m_{S|R}(X_k;P) \right)^2 \right].
\\
&\mme\left[ \psi_n^2(\bZ_k,P) \right] 
=
\frac{1}{f_{R}^{2}(x;P)}
\int
	K^2(u)
	\left[
		v_{Y_k|X_k}(x+uh) 
	\right.
\\
& \hspace{2cm}
	\left.
		+ \left(m_{Y_k|X_k}(x+uh)  - m_{S|R}(x+uh;P) \right)^2
	\right]
	f_{X_k}(x+uh)
~du
\\
& = O_{\m{P}}(1),
\end{align*}
because 
$f_{R}(x;P)$ is bounded away from zero as a function of $P$, 
the conditional moment functions and $f_{X_k}$ inside the integral are bounded,
and
Fact \ref{fact:control_mean:prop_kernel} on the kernel.

\item \underline{Assumption \ref{aspt:asy_linear} - \eqref{eq:aspt:asy_linear_lind}:} $(2+\zeta)$-th moment condition. 

We verify it in two steps.

First,
\begin{align*}
\mmv_P \left(  \psi_n(\bV_i,P) \right) 
= &\mmv_P\left( 
	\frac{1}{f_{R}(x;P) \sqrt{h} }  K\left( \frac{ R_{i} - x }{h} \right) 
		\left( S_{i} - m_{S  | R }(R_i;P) \right) 
\right) 
\\
 =&\frac{1}{f_{R}^2(x;P) }  \int K^2(u)  \left\{ 
 		 v_{S  | R}(x+uh;P) f_{R}(x+uh;P)  
 	\right\} 
 ~ du 
\end{align*}
is bounded away from zero uniformly over $\m{P}$ and $n$ because
(i) $f_{R}(x;P)$ is bounded as a function of $P$;
and
(ii)  $v_{S  | R}(r;P)$ and $f_{R}(r;P)$  are continuous functions of $r$ and  bounded away from zero at $x=r$ and over $P$.

Second, for $\zeta$ of the moment condition in Proposition \ref{prop:control_mean}, 
call $\eta(r;P) = \mme_P[|S_i - m_{S | R }(R_i;P) |^{2+\zeta} | R_i=r]$.
\begin{align*}
n^{-\zeta/2} \mme_P \left|  \psi_n(\bV_i,P) \right| ^{2+\zeta} & 
\\
&\hspace{-2cm} =\frac{(m/n)^{\zeta/2}}{ f_{R}^{2+\zeta}(x;P) (mh)^{\zeta/2}} \int |K(u)|^{2 +\zeta} \eta(x+uh;P) f_{R}(x+uh;P) ~ du 
\\
&\hspace{-2cm} = o_{\m{P}}(1),
\end{align*}
because 
(i) $mh \to \infty$, $m/n=O(1)$;
(ii)  $f_{R}(x;P) = O_{\m{P}}(1) $ and  $f_{R}^{-1}(x;P) = O_{\m{P}}(1) $;
(iii) $\eta(x;P) = O_{\m{P}}(1)$;
and
(iv) Fact \ref{fact:control_mean:prop_kernel} on the kernel.
Combining steps 1 and 2,
\begin{align*}
n^{-\zeta/2} \sup\limits_{P \in \m{P} }  \mme_P\left| \frac{ \psi_n(\bV_i,P) }{ \sqrt{\mmv_P\left(  \psi_n(\bV_i,P) \right)} } \right|^{2 + \zeta} = o(1).
\end{align*}

\item \underline{Assumption \ref{aspt:asy_linear} - \eqref{eq:aspt:asy_linear_xicont}:} $\xi^2\left( \frac{m}{n} P_1 + \frac{n-m}{n} P_2 \right) \to \xi^2 \left( \gamma P_1 + (1-\gamma ) P_2  \right)$.

Let $\overline{P}_n = \frac{m}{n} P_1 + \frac{n-m}{n} P_2$ and $\overline{P} = \gamma  P_1 + (1-\gamma ) P_2$.
We have that \\ $\xi^2(\overline{P}_n) = \kappa_{0,2} v_{S  | R }(x; \overline{P}_n)  /f_{R}(x;\bar{P}_n)$, 
so it suffices to show that 
$v_{ S  | R }(x; \bar{P}_n ) \to v_{ S  | R }(x; \bar{P})$
and
$f_{R}(x;\bar{P}_n) \to f_{R}(x;\bar{P})$.

First, convergence of the PDF,
\begin{align*}
f_{R}(x;\bar{P}_n)  = \frac{m  }{n   } f_{X_1}(x)  + \frac{n-m  }{n  } f_{X_2}(x)  \to  \gamma f_{X_1}(x)  + (1-\gamma) f_{X_2}(x) =  f_{R}(x;\bar{P}).
\end{align*}

Second, convergence of moments. For $g(x)=x$ or $g(x)=x^2$,
\begin{align*}
\Exp_{ \bar{P}_n }[g(S)  | R =x ] = &  \frac{m f_{X_1}(x) }{ n f_{R}(x;\bar{P}_n) } \mme[g(Y_1)|X_1=x] + \frac{(n-m) f_{X_2}(x) }{n f_{R}(x;\bar{P}_n)  } \mme[g(Y_2) | X_2=x]
\\
\to & \frac{\gamma f_{X_1}(x) }{ f_{R}(x;\bar{P}) } \mme[g(Y_1)|X_1=x] + \frac{(1 -\gamma ) f_{X_2}(x) }{f_{R}(x;\bar{P})  } \mme[g(Y_2) | X_2=x]
\\
= & \Exp_{\bar{P}} [g( S  )  |  R   =x ],
\end{align*} 
which implies that $v_{ S  | R }(x; \bar{P}_n ) \to v_{ S  | R }(x; \bar{P})$.

Therefore, $\xi^2(\overline{P}_n) \to \xi^2(\overline{P})$.
 
\item \underline{Assumption \ref{aspt:sampling} :} 

We have that   $n_1$ is a deterministic sequence  and $(n_1 /n - \lambda) \to 0$ by assumption.
Assumption  \ref{aspt:asy_linear} has already been verified above for any sequence $m_n$ such that $(m/n -\gamma) \to 0$ for arbitrary $\gamma \in (0,1)$.
In particular it holds for $\gamma \in \{ \lambda, 1-\lambda\}.$

\end{enumerate}

$\square$

\subsection{Proof of Proposition \ref{prop:control_quan} - Controlled Quantiles}
\label{proof:control_quantile}

\indent

The goal of this proof is to use the assumptions listed in Proposition \ref{prop:control_quan}  to verify Assumptions \ref{aspt:asy_linear} and \ref{aspt:sampling}.
It adapts arguments from \citeSM{pollard1991}, \citeSM{chaudhuri1991}, and \citeSM{fan1994}.

Consider an \textit{iid} sample from $P\in\m{P}$ with $m$ observations, $\bV_1=(R_1,S_1), \ldots, \bV_m=(R_m,S_m)$.
The number $m$ grows with $n$ such that $m/n \to \gamma$, for some $\gamma \in (0,1)$.
The parameter of interest is $\theta(P) = \mmq_{\chi}[S|R=x]$, $\chi \in (0,1)$, and 
the NW-style estimator  is given by
\[
\hat \theta^b = \theta_{m,n}^b (\bV_1 , \ldots, \bV_m) \doteq \arg\min_{\theta} \sum_{i=1}^{m} \rho_{\chi}(S_i - \theta) K\left(\frac{R_i-x}{h}\right)~,
\]
where
$\rho_{\chi}(u) = (\chi - \mmi(u \leq 0))u$.
This section studies the asymptotic behavior of 
$\ha{\theta} = \ha{\theta}^b - h \ha{B} $,
where the expression for $\ha{B}$ is given below Equation \ref{eq:control_quan:BP}.

Define $U = S - \theta(P)$.
The assumptions in Proposition \ref{prop:control_quan}  imply the following facts:

\begin{enumerate}

\item \textit{As $m\to\infty$, $h \to 0$, $m h \to \infty$, $\sqrt{m h} h^2 =O(1)$, $\sqrt{m h} h^3 =o(1)$};

\item 
\textit{ $\int K(u) u ~ du =0$, $\int K^r(u) u^s g(u) ~ du < \infty$  for $1 \leq r <\infty$,  $0 \leq s \leq 3$, and bounded function $g(u)$};

\item \textit{The distribution of $R$ has PDF $f_{R}(r;P)$
that is three times differentiable  wrt $r$: 
$\nabla_r f_{R}(r;P)$,
$\nabla_{r^2} f_{R}(r;P)$, and
$\nabla_{r^3} f_{R}(r;P)$
respectively;
$f_{R}(r;P)$ and these derivatives are bounded as functions of $(r,P)$;
$f_{R}(r;P)$ is bounded away from zero as a function of $(r,P)$;}

To see this, note that $f_{R}(r;P)$ is a convex combination of $f_{X_1}(r)$ and $f_{X_2}(r)$,
each bounded, with bounded derivatives, and bounded away from zero.

\item  \textit{The conditional distribution of $U$ given $R$ has PDF $f_{U|R}(u|r;P)$ that is a bounded function of $(u,r,P)$, $f_{U|R}(0|x;P)$ is bounded away from zero over $P$, $f_{U|R}(u|r;P)$ is differentiable as function of $(u,r)$ and has bounded partial derivatives;}

To see this, take  $P = \alpha P_1 + (1-\alpha) P_2$ and note that, 
\[
f_{U|R}(u|r;P)  = \frac{\alpha f_{X_1}(r) }{ f_{R}(r;P) } f_{Y_1|X_1}(\theta(P) + u |r)   + \frac{(1-\alpha) f_{X_2}(r) }{ f_{R}(r;P) } f_{Y_2|X_2}(\theta(P) + u |r).
\]
The PDFs $f_{Y_k|X_k}(y_k | x_k)$ are bounded functions of $(x_k,y_k)$.
The weights $\omega_1(r;P) \doteq \alpha \frac{f_{X_1}(r) }{ f_{R}(r;P) }$ and
$\omega_2(r;P) \doteq  (1-\alpha) \frac{f_{X_2}(r) }{ f_{R}(r;P) }$
 are bounded functions of $(r,P)$ because they are positive and sum to $1$.
 The PDFs $f_{Y_k|X_k}(y_k | x_k)$ are differentiable and so are the weights.
 The partial derivatives of $f_{Y_k|X_k}(y_k | x_k)$ are bounded.
 The  derivatives of the weights wrt $r$  also are bounded because the derivatives of the PDFs $f_{X_k}(r)$  are bounded plus the fact that $f_{R}(r;P)$ is bounded away from zero over $(r,P)$.
Finally, $f_{Y_k|X_k}(y_k | x)$ is bounded away from zero over $y_k$.

\item \textit{The conditional distribution of $U$ given $R$ has CDF $F_{U|R}(u|r;P)$ that is three times partially differentiable wrt $r$ and has partial derivatives
 $\nabla_r F_{U|R}(0|r;P)$,
 $\nabla_{r^2} F_{U|R}(0|r;P)$, and
  $\nabla_{r^3} F_{U|R}(0|r;P)$ 
  that are bounded functions of $(r,P)$;}

Again, 
\[
F_{U|R}(0|r;P)  = \omega_1(r;P) F_{Y_1|X_1}(\theta(P)   |r)   + \omega_2(r;P) F_{Y_2|X_2}(\theta(P)  |r).
\]
We have that and   $F_{Y_k|X_k}(y_k   |x_k)$ are three times partially differentiable wrt $x_k$.
The weights $\omega_k(r;P)$ are three times differentiable wrt $r$ because the PDFs $f_{X_k}(x_k)$ are three times differentiable.
The first three partial derivatives of $F_{Y_k|X_k}(y_k   |x_k)$  wrt $x_k$ are bounded functions of $(x_k,y_k)$.
The first three  derivatives of $\omega_k(r;P)$ wrt $r$ are bounded functions of $(r,P)$
because the first three derivatives of $f_{X_k}(x_k)$ wrt $x_k$ are bounded and
$f_{R}(r;P)$ is bounded away from zero over $(r,P)$.

\end{enumerate}

\begin{enumerate}

\item[1.] \underline{Assumption \ref{aspt:asy_linear} - \eqref{eq:aspt:asy_linear}:} asymptotic expansion.

\end{enumerate}

Define $Z_m^b  = \sqrt{mh}\left(\hat \theta^b - \theta(P)\right) $
and
$Z_m  = \sqrt{mh}\left(\hat \theta - \theta(P)\right) =  
\sqrt{mh}\left(\hat \theta^b - \theta(P) - h^2 \ha{B} \right)$. 
We want to study the asymptotic behavior of $Z_m^b$ and $Z_m$, where $Z_m^b$ is the value that minimizes the objective function $L_m(z)$, i.e.,

\begin{align*}
Z_m^b &= \arg\min_z \underbrace{\sum_{i=1}^m \rho_{\chi}\left(S_i - \theta(P) - \frac{z}{\sqrt{mh}}\right)K\left(\frac{R_i-x}{h}\right)}_{\doteq L_m(z)} \\
&= \arg\min_z L_m(z) - L_m(0), 
\end{align*}
since $L_m(0)$ is not a function of $z$ and hence does not affect the argmin.  
 
Let $Q_m(z) = L_m(z) - L_m(0)$ and $U_i = S_i - \theta(P)$. 
We have,

\[ Q_{m}(z) =  \sum_{i=1}^{m}\left( \rho_{\chi}\left(U_{i} -\frac{1}{\sqrt{mh}}z\right) - \rho_{\chi}(U_{i})\right)K \left(  \frac{R_i-x }{ h }  \right) ,\]
which is  minimized by $Z_m^b = \sqrt{mh}\left(\hat \theta^b - \theta(P)\right)$.
Notice that since $\rho_{\chi}(\cdot)$s are convex functions of z, so is $Q_m(z)$, which is a sum of convex functions.

The derivation of the asymptotic linear representation is done in two parts.
First, we approximate $Q_m(z)$ by a quadratic function $Q_m^*(z)$ whose minimizing value $z=\eta_m$ has an asymptotic linear representation plus a bias term.
Second, we show that $Z_m^b = \sqrt{m h}(\ha{\theta}^b - \theta(P) )$ converges to $\eta_m$ in probability and therefore they share the same asymptotic behavior.
The bias-corrected version of $Z_m^b$ is $Z_m$, and $Z_m$ has an asymptotic linear representation 
with an influence function that satisfies Assumption \ref{aspt:asy_linear}.

\underline{Part I: approximating the objective function}

Let $D_i = -\chi \mmi(U_i \geq 0) + (1-\chi)\mmi(U_i <0) = \mmi(U_i<0) - \chi $ and 
\[V_{i}(z) \doteq  \rho_{\chi}\left(U_i-\frac{1}{\sqrt{mh}}z\right) -\rho_{\chi}(U_i)- \frac{1}{\sqrt{mh}}zD_i .\]
We can rewrite $Q_m(z)$ in terms of $D_i$ and $V_i(z)$ by adding and subtracting the conditional expectation of $Q_m(z)$ as follows:
\begin{align}
Q_m(z) &= \mme_P[Q_m(z)| \lbar{\bR}_m ] \label{eq:expanQ_1}
\\
&+\frac{1}{\sqrt{mh}}\sum_{i=1}^m z\left(D_i-\mme_P[D_i|R_i]\right)  K \left(  \frac{R_i-x }{ h }  \right)\label{eq:expanQ_2} 
\\
&+ \sum_{i=1}^m (V_{i}(z) - \mme_P[V_{i}(z)|R_i])K \left(  \frac{R_i-x }{ h }  \right)~,\label{eq:expanQ_3}
\end{align}
where $\lbar{\bR}_m$ is the vector $(R_1, \ldots, R_m)$.
In what follows, we show that 
\begin{align*}
&\eqref{eq:expanQ_1} = \frac{1}{2} f_{U|R}(0| x ; P)f_R(x;P) z^2 
\\ &
+ \sqrt{mh} h^2  z \kappa_{2,1} \left[ 
    \nabla_{r} F_{U|R}(0|x;P) \nabla_r f_R(x ;P)  
    +\frac{1}{2} \nabla_{r^2} F_{U|R}(0|x;P) f_R(x;P)  
\right] + o_{\mathcal P}(1).
\\
&\eqref{eq:expanQ_3} = o_{\mathcal P}(1),
\end{align*}
where $\kappa_{s,t} = \int u^s K^t(u) ~du$.

Regarding $\eqref{eq:expanQ_1}$, 
 define 
 \[
 	M(t|r;P) \doteq \mme_P\left[\rho_{\chi}(S - \theta(P) +t) |R=r \right] = \mme_P\left[\rho_{\chi}(U +t) |R=r \right].
 \]
Notice that although the check function is not differentiable, the $M$ function is differentiable.
\begin{align*}
&\nabla_t M(t|r;P) = \chi(1-F_{U|R}(-t|r;P)) -(1-\chi) F_{U|R}(-t|r;P) = \chi - F_{U|R}(-t|r;P),
\\
&\nabla_{t^2} M(t|r;P) = - \nabla_{t} \{ F_{U|R}(-t|r;P) \} =  f_{U|R}(-t|r;P),
\\
&\nabla_{t^3} M(t|r;P) = - \nabla_{u} f_{U|R}(-t|r;P),
\\
&\nabla_{t r} M(t|r;P) = - \nabla_{r} F_{U|R}(-t|r;P),
\\
&\nabla_{t r^2} M(t | r ; P ) =-\nabla_{r^2} F_{U|R}(-t|r;P),
\\
&\nabla_{t r^3} M(t | r ; P ) =-\nabla_{r^3} F_{U|R}(-t|r;P),
\end{align*}
where we use the Leibniz rule, the existence of the derivatives 
$\nabla_u f_{U|R}(u|r;P)$,  $\nabla_{r} F_{U|R}(-t|r;P)$, $\nabla_{r^2} F_{U|R}(-t|r;P)$,
and $\nabla_{r^3} F_{U|R}(-t|r;P)$.
We can write $\mme[Q_m(z)|\lbar{\bR}_m]$ in terms of $M$ as follows, 
\begin{align}
&\mme_P[Q_m(z)|\lbar{\bR}_m]  \nonumber\\*
 &= \sum_{i=1}^{m}\mme_P \bigg[ \rho_{\chi}\left(U_i -\frac{1}{\sqrt{mh}}z\Big| R_i\right)   
   - \rho_{\chi}\left(U_i \Big| R_i\right) \bigg]   K \left(  \frac{R_i-x }{ h }  \right) 
\nonumber
\\
&= \sum_{i=1}^m \left[ M\left(  - \frac{z}{\sqrt{mh}}\Big|R_i ; P\right)  
- M\left( 0\big| R_i ; P\right) \right]K\left(\frac{R_i-x}{h}\right).
\nonumber
\end{align}
Taylor expand $M$ as a function of $t$ around $0$,
\begin{align}
&\mme_P[Q_m(z)|\lbar{\bR}_m]  
\nonumber
\\*
 &= \sum_{i=1}^m \left[ \nabla_{t} M\left( 0 \big| R_i ;P \right) \frac{-z}{\sqrt{mh}}
 										+ \frac{1}{2} \nabla_{t^2} M\left(   0   \big|   R_i ; P \right)  \frac{z^2}{mh} \right. 
 \nonumber 
 \\
 & ~~~~~~~~~~~\left. - \frac{1}{6}\nabla_{t^3} M( q^*  | R_i ; P) \frac{z^3}{(mh)^{3/2}} \right]  K\left(\frac{R_i-x}{h}\right) 
 \nonumber
\\
 & =\frac{-z}{\sqrt{mh}}\sum_{i=1}^m  \nabla_{t} M\left(0 \big| R_i ; P \right) K\left(\frac{R_i-x}{h}\right) 
 \label{eq:1st}
 \\ 
& ~~~~~~~~+\frac{1}{2}\frac{z^2}{mh}\sum_{i=1}^m f_{U|R}\left( 0 | R_i ; P \right) K\left(\frac{R_i-x}{h}\right) ~
\label{eq:2nd}
\\
&~~~~~~~~ + \frac{z^3}{6\sqrt{mh}}\frac{1}{mh}\sum_{i=1}^m \nabla_u f_{U|R}\left( -q^*|R_i ;  P \right)  K\left(\frac{R_i-x}{h}\right)~
\label{eq:3rd},
\end{align}
where $q^*$ is a point between $0$ and $- z/\sqrt{mh}$.
The goal is to show that 
\begin{align*}
& \eqref{eq:1st}= \sqrt{mh} h^2  z \kappa_{2,1} \left[ 
    \nabla_{r} F_{U|R}(0|x;P) \nabla_r f_R(x ;P)  
    +\frac{1}{2} \nabla_{r^2} F_{U|R}(0|x;P) f_R(x;P)  
\right] + o_{\m{P}}(1),
\\
& \eqref{eq:2nd} = \frac{1}{2} f_{U|R}(0| x ; P)f_R(x;P) z^2 +o_{\mathcal P}(1),
\\
& \eqref{eq:3rd} =o_{\m{P}}(1).
\end{align*}

Expectation of (\ref{eq:1st}).
We Taylor expand $\nabla_{t} M\left(0 \big| R_i ;P\right)$ as a function of $R_i$ around $R_i=x$ and use the fact that
$\nabla_{t} M\left(0 \big| x ;P\right) = \chi - \mmp_P (S_i - \theta(P)   \leq  0 | R_i = x) = 0 $:
\begin{align*}
&\mme_P \left[\frac{-z}{\sqrt{mh}}\sum_{i=1}^m  \nabla_{t} M\left(0 \big| R_i ;P\right) K\left(\frac{R_i-x}{h}\right)\right]
\\
=& \mme_P \left[\frac{-z}{\sqrt{mh}}\sum_{i=1}^m  \underbrace{ \nabla_{t} M\left(0 \big| x ;P \right)}_{=0} K\left(\frac{R_i-x}{h}\right)\right]
\\
&+ \mme_P \left[\frac{-z}{\sqrt{mh}}\sum_{i=1}^m  \underbrace{ \nabla_{tr} M\left(0 \big| x ;P \right)}_{= - \nabla_{r} F_{U|R}(0|x;P) } \left(R_i-x\right)   K\left(\frac{R_i-x}{h}\right)\right]
\\
&+ \mme_P \left[\frac{-z}{\sqrt{mh}}\sum_{i=1}^m  \frac{1}{2} \underbrace{   \nabla_{tr^2}  M\left(0 \big| x ;P \right)}_{= - \nabla_{r^2} F_{U|R}(0|x;P)}\left(R_i-x\right)^2   K\left(\frac{R_i-x}{h}\right)\right]
\\
&+ \mme_P \left[\frac{-z}{\sqrt{mh}}\sum_{i=1}^m  \frac{1}{6} \underbrace{   \nabla_{tr^3}  M\left(0 \big| x^* ;P \right)}_{= - \nabla_{r^3} F_{U|R}(0|x^*;P)}\left(R_i-x\right)^3   K\left(\frac{R_i-x}{h}\right)\right]
\\
=& z \sqrt{mh}\int \left[\nabla_{r} F_{U|R}(0|x;P) uh  \right] K(u) f_R(x+uh;P)du 
\\
&+ z \sqrt{mh}\int \left[\nabla_{r^2} F_{U|R}(0|x;P) u^2 h^2 \right] K(u) f_R(x+uh;P)du
\\
&+ z \sqrt{mh}\int \left[ \nabla_{r^3} F_{U|R}(0|x^*;P) u^3 h^3 \right] K(u) f_R(x+uh;P)du 
\\
=& z \nabla_{r} F_{U|R}(0|x;P) \sqrt{mh} h \int  u   K(u) 
[f_R(x;P) + \nabla_r f_R(x;P) uh +  \frac{1}{2}\nabla_{r^2} f_R(x^{**};P) u^2 h^2]  du 
\\
&+ \frac{z}{2} \nabla_{r^2} F_{U|R}(0|x;P) \sqrt{mh} h^2 \int   u^2    K(u) [f_R(x;P) + \nabla_r f_R(x^{***};P) uh] du
\\
&+ \frac{z}{6} \sqrt{mh} h^3 \int \nabla_{r^3} F_{U|R}(0|x^*;P) u^3   K(u) f_R(x+uh;P)du,
\end{align*}
where $x^*$ is a point between $R_i$ and $x$, $x^{**}$ and  $x^{***}$ are points between $x+uh $ and $x$.
The last equation above equals to
\begin{align*}
=& z \nabla_{r} F_{U|R}(0|x;P) f_R(x;P)  \sqrt{mh} h \underbrace{\int  u   K(u) du}_{=0} 
\\
& + z \nabla_{r} F_{U|R}(0|x;P) \nabla_r f_R(x ;P) \sqrt{mh} h^2 \underbrace{\int  u^2   K(u)   du}_{= \kappa_{2,1}} 
\\
& + \frac{z}{2} \nabla_{r} F_{U|R}(0|x;P)  \sqrt{mh} h^3 \int  u^3   K(u)  \nabla_{r^2} f_R(x^{**};P)  du 
\\
&+ \frac{z}{2} \nabla_{r^2} F_{U|R}(0|x;P) f_R(x;P) \sqrt{mh} h^2 \underbrace{\int   u^2    K(u)   du}_{= \kappa_{2,1}}
\\
&+ \frac{z}{2} \nabla_{r^2} F_{U|R}(0|x;P) \sqrt{mh} h^3 \int   u^3    K(u) \nabla_r f_R(x^{***};P)  du
\\
&+ \frac{z}{6} \sqrt{mh} h^3 \int \left[ \nabla_{r^3} F_{U|R}(0|x^*;P) u^3  \right] K(u) f_R(x+uh;P)du 
\\
=& ~ 0
\\
& + \sqrt{mh} h^2  z \kappa_{2,1} \nabla_{r} F_{U|R}(0|x;P) \nabla_r f_R(x ;P) 
\\
& + O_{\m{P}} \left( \sqrt{mh} h^3 \right)  
\\
&+ \sqrt{mh} h^2  \frac{z \kappa_{2,1} }{2} \nabla_{r^2} F_{U|R}(0|x;P) f_R(x;P) 
\\
&+ O_{\m{P}} \left( \sqrt{mh} h^3 \right)  
\\
&+ O_{\m{P}} \left( \sqrt{mh} h^3 \right)  
\\
= & \sqrt{mh} h^2  z \kappa_{2,1} \left[ 
    \nabla_{r} F_{U|R}(0|x;P) \nabla_r f_R(x ;P)  
    +\frac{1}{2} \nabla_{r^2} F_{U|R}(0|x;P) f_R(x;P)  
\right] + o_{\m{P}}(1),
\end{align*}
where we use the kernel properties,
the facts that
$\nabla_{r} F_{U|R}(0|x;P)$ and $\nabla_{r^2} F_{U|R}(0|x;P)$ are bounded functions of $P$,
$\nabla_{r^3} F_{U|R}(0|r;P)$ is a bounded function of $(r,P)$,
$f_R(r ; P )$ is a bounded function of $(r,P)$, 
$\nabla_r f_R(r ; P )$ and $\nabla_{r^2} f_R(r ; P )$ are bounded functions of $(r,P)$,
and $\sqrt{mh}h^3 \to 0$.

Variance of (\ref{eq:1st}).
\begin{align*}
&\mmv_P \left[\frac{-z}{\sqrt{mh}}\sum_{i=1}^m  \nabla_{t} M\left(0 \big| R_i ;P\right) K\left(\frac{R_i-x}{h}\right)\right]
\\
=&  \frac{z^2}{h}\mmv_P \left[ \nabla_{t} M\left(0 \big| R_i ; P \right) K\left(\frac{R_i-x}{h}\right)\right]
\\
\leq &  \frac{z^2}{h} \mme_P \left[ \left\{ \nabla_{t} M \left( 0 \big| R_i ; P  \right) \right\}^2 K^2\left(\frac{R_i-x}{h}\right)\right]
\\
= &z^2 h^2\int \left\{ \nabla_{r} F_{U|R}(0|x^*;P) \right\}^2 u^2f_{R}(x+uh;P)K^2\left(u\right)du
= o_{\m{P}}(1).
\end{align*}
Therefore, 
\begin{align*}
& \eqref{eq:1st}= \sqrt{mh} h^2  z \kappa_{2,1} \left[ 
    \nabla_{r} F_{U|R}(0|x;P) \nabla_r f_R(x ;P)  
    +\frac{1}{2} \nabla_{r^2} F_{U|R}(0|x;P) f_R(x;P)  
\right] + o_{\m{P}}(1).
\end{align*}

It remains to show that the probability limits of \eqref{eq:2nd} and \eqref{eq:3rd} are zero.
Expectation of \eqref{eq:2nd}:
\begin{align*}
&\mme_P \left[\frac{1}{2}\frac{z^2}{mh}\sum_{i=1}^m f_{U|R}\left( 0 |R_i ;P \right) K\left(\frac{R_i-x}{h}\right)\right]
\\
=&\mme_P\left[\frac{1}{2}\frac{z^2}{h}\left\{f_{U|R}\left(0|x;P \right) 
		+ \nabla_r f_{U|R}\left(0|x^* \right) (R_i-x) \right\} K\left(\frac{R_i-x}{h}\right)\right]
\\
=&\frac{z^2}{2} \int \left\{f_{U|R}\left(0|x;P \right) + \nabla_r f_{U|R}\left(0|x^* \right) uh \right\} K\left(u\right) f_R(x+uh;P)du
\\
=&f_{U|R}\left(0|x;P \right) \frac{z^2}{2} \int K\left(u\right) \left\{f_R(x;P)  + \nabla_r f_R(x^{**} ; P )  uh  \right\} du
\\
&+
\frac{z^2}{2} h \int \nabla_r f_{U|R}\left(0|x^*;P \right) u  K\left(u\right) f_R(x+uh;P)du
\\
=&\frac{z^2}{2} f_{U|R}\left(0|x;P \right) f_R(x;P)  +  o_{\m{P}}(1),
\end{align*}
where we use that $f_R(r ; P )$, $\nabla_r f_R(r ; P )$, $f_{U|R}\left(0|x;P \right)$, $\nabla_r f_{U|R}\left(0|r ;P\right)$ are bounded over $(r,P)$.

Variance of \eqref{eq:2nd}:
\begin{align*}
&\mmv_P\left[\frac{1}{2}\frac{z^2}{mh}\sum_{i=1}^m f_{U|R}\left( 0 | R_i ;P \right)  K\left(\frac{R_i-x}{h}\right)\right]
\\
& = \frac{z^4}{4}\frac{1}{mh^2} \mmv_P \left[f_{U|R}\left( 0 |R_i;P \right)   K\left(\frac{R_i-x}{h}\right)   \right]
\\
& \leq \frac{z^4}{4}\frac{1}{mh^2}\mme_P\left[f_{U|R}^2\left(0 |R_i ; P \right)   K^2\left(\frac{R_i-x}{h}\right)  \right]
\\
&= \frac{z^4}{4}\frac{1}{mh}\int f_{U|R}^2\left( 0| x+uh ; P \right)  K^2\left(u\right) f_{R}(x+uh;P)du~ = o_{\m{P}}(1),
\end{align*}
because  $f_{U|R}\left( 0 | r ; P \right)$ and $f_{R}(r;P)$ are bounded functions of $(r,P)$ and $mh \to \infty$. 

Therefore, we have that 
 \eqref{eq:2nd} $=\frac{1}{2} f_{U|R}\left(0|x;P\right)f_R(x;P)z^2$ $+ o_{\m{P}}(1)$.
Moreover, $\eqref{eq:3rd} = o_{\m{P}}(1)$ because
 $mh \to \infty$, $\frac{1}{mh}\sum_{i=1}^m K\left(\frac{R_i-x}{h}\right) = O_{\mathcal P}(1)$, and $\nabla_u f_{U|R}(u|r;P)$ is a bounded function of $(u,r,P)$.
We are done in showing the probability limit of $\mme_P[Q_m(z)| \lbar{\bR}_m ]$,
\begin{align*}
& \eqref{eq:expanQ_1} = \frac{1}{2} f_{U|R}\left(0|x;P\right)f_R(x;P)z^2  
\\
& + \sqrt{mh} h^2  z \kappa_{2,1} \left[ 
    \nabla_{r} F_{U|R}(0|x;P) \nabla_r f_R(x ;P)  
    +\frac{1}{2} \nabla_{r^2} F_{U|R}(0|x;P) f_R(x;P)  
\right] + o_{\m{P}}(1).
\end{align*}
It remains to show that $\eqref{eq:expanQ_3} = o_{\m{P}}(1)$.

We show that the expectation of \eqref{eq:expanQ_3} is zero and its variance converges to zero.
The expectation of \eqref{eq:expanQ_3} is zero because it equals the expectation of
\begin{align*}
& \mme_P \left[ \sum_{i=1}^n  (V_{i}(z) - \mme_P [V_{i}(z) | R_i ] ) K\left(\frac{R_i-x}{h}\right) \Big|  \bR_m  \right] =0.
\end{align*}

Variance of \eqref{eq:expanQ_3}:
\begin{align*}
&\mmv_P \left[ \sum_{i=1}^m  (V_{i}(z) - \mme_P[V_{i}(z)|R_i]) K\left(\frac{R_i-x}{h}\right) \right]  
\\
= & \mmv_P \left[\mme_P \left[ \sum_{i=1}^m  (V_{i}(z) - \mme_P[V_{i}(z)|R_i]) K\left(\frac{R_i-x}{h} \right) \Big| \bR_m \right] \right]  
\\
& + \mme_P \left[\mmv_P \left[ \sum_{i=1}^m  (V_{i}(z) - \mme_P [V_{i}(z) | R_i ] ) K\left(\frac{R_i-x}{h}\right)\Big| \bR_m \right] \right]
\\
& \leq \sum_{i=1}^m  \mme_P \left[K^2\left(\frac{R_i-x}{h}\right)\mme_P \left[  V_{i}(z)^2\Big|R_i\right] \right]\\
&= \sum_{i=1}^m \mme_P \left[K^2\left(\frac{R_i-x}{h}\right)V_i(z)^2\right] \\
& \leq  4z^2 \int K^2(u)f_{R}(x+uh ;P) \left\{ F_{U|R}\left(\left|\frac{z}{\sqrt{mh}}\right| \bigg| x+uh;P \right) \right.
\\
&~~~~~~~~~~~~~~~~~~~~~~~~~~~~ \left. -F_{U|R}\left(-\left|\frac{z}{\sqrt{mh}}\right| \bigg|   x+uh ;P\right)\right\} ~ du
 \\
& =  \frac{8 z^3}{\sqrt{mh}} \int K^2(u)f_{R}(x+uh ;P) \nabla_u F_{U|R}\left( u^* \bigg| x+uh; P \right)    ~ du
 \\ 
&=o_{\m{P}}(1),
\end{align*}
where we use that $f_R$ and $\nabla_u F_{U|R} $ are bounded over $(u,r,P)$ and 
\begin{align*}
|V_i(z)| &= \left| \left(\rho_{\chi}\left(U_i-z/\sqrt{mh}\right) - \rho_{\chi}(U_i) - D_iz/\sqrt{mh}\right)\right|\\
&\leq 2\left|z/\sqrt{mh}\right| \mmi\left(\left|U_i\right| \leq \left|\frac{z}{\sqrt{mh}}\right|\right).
\end{align*}

Consider \eqref{eq:expanQ_1}--\eqref{eq:expanQ_3} and the probability limits of \eqref{eq:expanQ_1} and \eqref{eq:expanQ_3} that we found.
Define the following objects,
\begin{align*}
B_0(P) = & \kappa_{2,1} \left[ 
    \nabla_{r} F_{U|R}(0|x;P) \nabla_r f_R(x ;P)  
    +\frac{1}{2} \nabla_{r^2} F_{U|R}(0|x;P) f_R(x;P)  
\right],
\\
 Q_m^*(z) = & z^2 \frac{1}{2}f_{U|R}(0  | x ; P)f_R(x;P)  
\\
& + z \left[ \frac{1}{\sqrt{mh}}\sum_{i=1}^m \left(D_i-\mme_P[D_i|R_i]\right) K \left(  \frac{ R_i-x }{ h }  \right)
+  \sqrt{mh} h^2 B_0(P) \right],
\\
 r_m(z) = & Q_m(z) - Q_m^*(z),
\end{align*}
so that $Q_m(z) = Q_m^*(z) + r_m(z)$
and 
$r_m(z) = o_{\m{P}}(1)$ for fixed $z$.

Rewrite $Q_m^*(z)$ as follows.
\begin{align}
Q_m^*(z) = & z^2 \frac{1}{2}f_{U|R}(0  | x ; P)f_R(x;P)  
\nonumber
\\
& + z  \underbrace{  \left[ \frac{1}{\sqrt{mh}}\sum_{i=1}^m \left(D_i-\mme_P[D_i|R_i]\right) K \left(  \frac{ R_i-x }{ h }  \right)
+  \sqrt{mh} h^2 B_0(P) \right] }_{\doteq M_m} 
 \nonumber 
\\
& = \frac{1}{2}f_{U|R}(0|x; P)f_R(x;P)z^2  +z M_m \nonumber 
\\
&= \frac{1}{2}f_{U|R}(0|x; P)f_R(x;P) \left(z + \underbrace{\frac{1}{f_{U|R}(0|x; P)f_R(x;P)} M_m}_{\doteq -\eta_m } \right)^2 \nonumber 
\\
& ~~~~- \frac{1}{2f_{U|R}(0 | x ; P)f_{R}(x;P)}M_m^2 \nonumber
\\
& = \frac{1}{2}f_{U|R}(0 | x ; P)f_R(x;P) \left(z -\eta_m \right)^2 - \frac{1}{2}f_{U|R}(0 | x ; P)f_{R}(x;P) \eta_m^2 ,\label{eq:Qstar}
\end{align}
which is minimized at 
\begin{align*}
\eta_m & = -\frac{1}{f_{U|R}(0 | x ; P)f_R(x;P)} M_m 
\\
&=   -\frac{1}{f_{U|R}(0 | x ; P) f_R(x;P)} \frac{1}{\sqrt{mh}} \sum_{i=1}^m \left(D_i-\mme_P[D_i|R_i]\right) K \left(  \frac{ R_i-x }{ h } \right)
\\
&+ \sqrt{mh} h^2 \underbrace{\frac{-B_0(P) }{f_{U|R}(0 | x ; P) f_R(x;P)}}_{\doteq B(P)}
\\
& =   \frac{1}{\sqrt{mh}} \sum_{i=1}^m 
 \left(\frac{-1}{f_{U|R}(0 | x ; P) f_R(x;P)}\right)
\left(D_i-\mme_P[D_i|R_i]\right) K \left(  \frac{ R_i-x }{ h } \right)
\\
&+ \sqrt{mh} h^2 B(P).
\end{align*}

The bias term $B(P)$ is
\begin{align}
B(P) 
=  \frac{-\kappa_{2,1} }{f_{U|R}(0 | x ; P) f_R(x;P)}
& \bigg[
    \nabla_{r} F_{U|R}(0|x;P) \nabla_r f_R(x ;P)  
\notag
\\
&  \hspace{1cm}  +\frac{1}{2} \nabla_{r^2} F_{U|R}(0|x;P) f_R(x;P)  
\bigg]
\label{eq:control_quan:BP}
\end{align}
and is consistently estimated by $\ha{B}$,
\begin{align*}
& \ha{B} = B_{m,n}(\bV_1 , \ldots, \bV_m) 
\\
& \doteq \frac{-\kappa_{2,1} }{\ha{f_{U|R}}(0 | x ; P) \ha{f_R}(x;P)}
\left[ 
    \ha{\nabla_{r} F_{U|R}}(0|x;P) \ha{\nabla_r f_R}(x ;P)  
    +\frac{1}{2} \ha{\nabla_{r^2} F_{U|R}}(0|x;P) \ha{f_R}(x;P)  
\right],
\end{align*}
that is, by replacing 
$ f_R(x;P)$, 
$\nabla_r f_R(x ;P)$,
$f_{U|R}(0 | x ; P)$,
$\nabla_{r} F_{U|R}(0|x;P)$,
and
$\nabla_{r^2} F_{U|R}(0|x;P)$,
in $B(P)$ by consistent nonparametric estimators readily available from the literature.
The tuning parameters for these additional estimators and appropriate moment conditions may be set such that $\ha{B} - B(P) = o_{\m{P}}(1)$, as required by Assumption \ref{aspt:prop:bias}.

We have already shown above that $r_m(z) = o_{\m{P}}(1)$ for fixed $z$.
Now, we show that the convergence is also uniform over $z$ in a compact set $\m{K} \subset \mathbb R$.
To this end, consider \\
$\Lambda_m(z) \doteq  Q_m(z) 
- z \frac{1}{\sqrt{mh}}\sum_{i=1}^m  \left(D_i-\mme_P[D_i|R_i]\right)  K \left(  \frac{ R_i-x }{ h }  \right)
- z \sqrt{mh}h^2 B_0(P)$,
\\
$\Lambda(z) \doteq z^2 \frac{1}{2}f_{U|R}(0| r; P)f_R(x;P)$,
and note that $\Lambda_m(z)$ is a convex function of $z$.
Note also that $r_m(z) = Q_m(z) - Q_m^*(z) = \Lambda_m(z) - \Lambda(z)$.
By the convexity lemma (\citeSM{pollard1991}, page 187), $\sup_{z \in \m{K} } |  \Lambda_m(z) - \Lambda(z) | = o_{\m{P}}(1)$.
We have that,
for any compact subset $\m{K} \subset \mathbb R$,   
$\sup_{z \in \m{K}}\left|r_m(z)\right| = o_{\mathcal P}(1)$.

\bigskip

\underline{Part II: $Z_m^b - \eta_m =o_{\mathcal P}(1)$.}

We want to show that for each $\epsilon >0$,
\[\inf_{P \in \m{P} } \mmp_P \left(\left|Z_m^b - \eta_m\right|  \leq  \epsilon\right) \to 1.\]

Consider the closed interval $B(m)$ with center $\eta_m$ and radius $\epsilon$. Since $\eta_m$ converges in distribution, it is stochastically bounded. 
The compact set $\m{K}$ can be chosen to contain $B(m)$ with probability arbitrarily close to one, thereby implying 
$\sup_{z \in B(m) }\left|  r_m(z)  \right| = o_{\mathcal P}(1).$

For a value outside of the interval $B(m)$, suppose $z = \eta_m + \delta$ with $\delta > \epsilon$. For the boundary point $z^* = \eta_m+ \epsilon$, the convexity of $Q_m$ and (\ref{eq:Qstar}) imply
\begin{align*}
\frac{\epsilon}{\delta} Q_m(z) + &\left(1-\frac{\epsilon}{\delta}\right) Q_m(\eta_m)  \geq Q_m\left(\frac{\epsilon}{\delta} z +\left(1-\frac{\epsilon}{\delta}\right) \eta_m\right) =  Q_m(z^*)
\\
&\geq \frac{f_{U|R}(0 | x ;  P)f_R(x;P)}{2}(z^* - \eta_m)^2 - \frac{f_{U|R}(0| r ; P)f_R(x;P)}{2}\eta_m^2 - \sup_{z \in B(m)}\left|r_m(z)\right|
\\
&\geq \frac{f_{U|R}(0 | x ; P)f_R(x;P)}{2}\epsilon^2 + Q_m(\eta_m) - 2\sup_{z \in B(m)}\left|r_m(z)\right|.
\\
Q_m(z)  & \geq Q_m(\eta_m)  + \left( \frac{\delta}{\epsilon } \right) \left[ \frac{f_{U|R}(0 | x ;  P)f_R(x;P)}{2}\epsilon^2 
- 2 \sup_{z \in B(m)}\left|r_m(z)\right| \right].
\end{align*}
An analogous argument holds for $z = \eta_m - \delta$.
Define the event $A_m$ as
\begin{align*}
A_m ~ : ~  2 \sup_{z \in B(m)}\left|  r_m(z) \right|  < \frac{f_{U|R}(0 | x ;  P)f_R(x;P)}{4}\epsilon^2.
\end{align*}
We have that $\inf_{P \in \m{P}} \mmp_P [A_m] \to 1$
because  \\
$\sup_{P \in \m{P}} \mmp_P \left[ \sup_{z \in B(m)}\left|r_m(z)\right| >  f_{U|R}(0 | x ;  P)f_R(x;P) \epsilon^2 / 4  \right] \to 0$.
Conditional on $A_m$, 
\begin{align}
Q_m(z)  \geq Q_m(\eta_m)  + \frac{f_{U|R}(0 | x ; P)f_R(x;P)}{4} \epsilon^2 \text{, for any $z : |z - \eta_m| > \epsilon$}
\label{eq:ineq_qm}
\end{align}
happens with probability one.
The event in \eqref{eq:ineq_qm}   implies that $| Z_m^b - \eta_m | \leq \epsilon$
because $Z_m^b$ minimizes $Q_m(z)$ and thus $Q_m(Z_m^b) \leq Q_m(\eta_m)$.
Finally, 
\begin{align*}
&\mmp_P \left[ | Z_m^b - \eta_m | \leq \epsilon \right]
\\
&\geq
\mmp_P \left[ \inf\limits_{z : |z - \eta_m| > \epsilon} Q_m(z)  \geq Q_m(\eta_m)  + \frac{f_{U|R}(0 | x ; P)f_R(x;P)}{4} \epsilon^2 \right]
\geq
\mmp_P [ A_m ] \text{, so that }
\\
& \inf_{P \in \m{P}} \mmp_P \left[ | Z_m^b - \eta_m | \leq \epsilon \right]
\geq
\inf_{P \in \m{P}} \mmp_P [ A_m ] \to  1. 
\end{align*}
Therefore, we have the asymptotic linear representation of the estimator as follows:
\begin{align*}
& \sqrt{mh}\left(\hat \theta^b - \theta(P)\right) 
\\
& =   \frac{1}{\sqrt{mh}} \sum_{i=1}^m 
 \left(\frac{-1}{f_{U|R}(0 | x ; P) f_R(x;P)}\right)
\left(D_i-\mme_P[D_i|R_i]\right) K \left(  \frac{ R_i-x }{ h } \right)
\\
&+ \sqrt{mh} h^2 B(P) + o_{\m{P}}(1).
\end{align*}
For $\ha{B}$ such that $\ha{B} = B(P) + o_{\m{P}}(1)$,
\begin{align*}
& \sqrt{mh}\left(\hat \theta   - \theta(P)   \right)
= \sqrt{mh}\left(\hat \theta^b - \theta(P) - h^2 \ha{B} \right)
= \sqrt{mh}\left(\hat \theta^b - \theta(P) - h^2  B(P)  \right) + o_{\m{P}}(1)
\\
& =  \frac{1}{\sqrt{m}}\sum_{i=1}^m 
\left( \frac{-1}{\sqrt{h} f_{U|R}(0 | x ;  P)f_R(x;P)}\right)
K\left(\frac{R_i-x}{h}\right)\left(D_i-\mme_P[D_i |R_i ]   \right)
\\
& =  \frac{1}{\sqrt{m}}\sum_{i=1}^m 
\underbrace{
\frac{-1}{f_{U|R}(0 | x ;  P)f_R(x;P) \sqrt{h} } K\left(\frac{R_i-x}{h}\right) \left( \mmi\{S_i < \theta(P) \} - F_{S|R}(\theta(P)|R_i;P) \right)
}_{\doteq \psi_n(\bV_i, P)}
 + o_{\mathcal P}(1)
\\
& = \frac{1}{\sqrt{m}}\sum_{i=1}^m \psi_n(\bV_i, P) + o_{\mathcal P}(1),
\end{align*}
where 
we use that
$D_i - \mme_P[D_i |  R_i] =  \mmi\{S_i < \theta(P) \} - F_{S|R}(\theta(P) | R_i; P)$
and 
$\sqrt{mh}h^2 = o(1)$.

\begin{enumerate}
\item[2.] \underline{Assumption \ref{aspt:asy_linear} - \eqref{eq:aspt:asy_linear_zeromean}:} zero mean of influence function. 

 $\mme_P[ \psi_{n}(\bV_{i}, P) ]=0 ~~\forall P$ by construction.

\item[3.] \underline{Assumption \ref{aspt:asy_linear} - \eqref{eq:aspt:asy_linear_var}:} variance of influence function.

 Define
 \begin{align*}
\xi^2(P) &\doteq  \frac{1}{f_{U|R}^2(0 | x ;  P)  f_R(x;P)}  \kappa_{0,2} \mmv_P [D_i | R_i =x]
\\
& =  \frac{1}{f_{U|R}^2(0 | x ;  P)  f_R(x;P)}  \kappa_{0,2} F_{S|R}(\theta(P)| x ;P) (1-F_{S|R}(\theta(P)| x ;P))
\\
& =  \frac{\chi (1-\chi) }{f_{U|R}^2(0 | x ;  P)  f_R(x;P)}  \kappa_{0,2}.
\end{align*}
\begin{align*}
&\mmv_P \left( 
	\frac{-1}{f_{U|R}(0 | x ;  P)f_R(x;P)\sqrt{h}}K\left(\frac{R_i-x}{h}\right)\left(D_i-\mme_P[D_i|R_i]\right)\right) - \xi^2(P) 
\\
 =&\frac{1}{f_{U|R}^2(0 | x ; P)f_{R}^2(x;P) }  \int K^2(u)  \bigg\{ 
 	\underbrace{ \mmv_P[D_i | R_i=x+uh] f_{R}(x+uh;P)}_{\doteq g(uh;P)} 
\\*
& \hspace{4cm} - 
 	\underbrace{ \mmv_P [D_i|R_i=x] f_R(x;P) 
 	 }_{\doteq g(x;P)} \bigg\}
 ~ du
 \\
=&\frac{h}{f_{U|R}^2(0 | x ;  P)f_{R}^2(x;P) }  \int K^2(u)  \nabla_r g(x^*_{uh};P) u  ~ du
= o_{\m{P}}(1).
\end{align*}
where $x^*_{uh}$ is a point between $x+uh$ and $x$, and $\nabla_r g(r;P)$ denotes the derivative of $g $ wrt $r$.
The expression above is $o_{\m{P}}(1)$ 
because
$h\to0$,
the derivative $\nabla_r g(r;P)$ is a bounded function of $(r,P)$,
and 
$f_R(x;P)$ and $f_{U|R}(0| x; P)$ are bounded away from zero over $P$.
The derivative $\nabla_x g(x;P)$  is bounded because
$f_R(r;P)$, $\nabla_r f_R(r;P)$, \\ 
$\mmv_P[D_i|R_i=r] = F_{U|R}(0 | r ;P) (1-F_{U|R}( 0 | r ;P))$,
and \\
$\nabla_r \left\{F_{U|R}(0 | r ;P) (1-F_{U|R}( 0 | r ;P)) \right\} $ 
are bounded functions of $(r,P)$.

Therefore,
\[
\sup\limits_{P \in \m{P} } \left| \mmv_P [  \psi_n(\bV_{i}, {P} ) ] - \xi^2(P) \right| \to 0.
\]

\item[4.] \underline{Assumption \ref{aspt:asy_linear} - \eqref{eq:aspt:asy_linear_mom}:} $ \sup_P \mme [\psi_n^2(\bZ_k ,P)]  < \infty .$

For $\bZ_k=(X_k,Y_k) \sim P_k, ~ k=1,2$,
define
$D_{k} =  \mmi(Y_k - \theta(P_k) < 0) -\chi$,
$m_{D_k|X_k}(x_k)=\mme[  D_k | X_k=x_k  ]$,
and
$v_{D_k|X_k}(x_k)=\mmv[  D_k | X_k=x_k  ]$.
For $\bV=(R,S) \sim P \in \m{P}$,
$D =  \mmi(S - \theta(P) < 0) -\chi$
and
$m_{D | R}(r;P)=\mme_P[  D | R=r  ]$.

We have,
\begin{align*}
\psi_n(\bZ_k,P) & 
	= \frac{-1}{f_{U|R}(0 | x ; P)f_{R}(x;P)\sqrt{h}}K\left(\frac{X_k-x}{h}\right)\left(D_k-m_{D | R}(X_k ; P )\right).
 \\
\mme\left[ \psi_n^2(\bZ_k,P) \right] 
&=
\frac{1}{ f_{U|R}^2(0 | x ;  P) f_{R}^{2}(x;P)} 
\\
& ~~\int
	K^2(u)
	\bigg[
		v_{P_k}(x+uh) 
		+ \left(m_{D_k|X_k}(x+uh)  - m_{D|R}(x+uh;P) \right)^2
	\bigg]
\\
& ~~~~ ~ f_{X_k}(x+uh)
~du
\\
& = O_{\m{P}}(1),
\end{align*}
because 
$v_{D_k|X_k}(r)=F_{Y_k|X_k}(\theta(P_k)| r ) (1-F_{Y_k|X_k}(\theta(P_k)| r)) $,\\
$m_{D_k|X_k}(r)=F_{D_k|X_k}(\theta(P_k)| r ) - \chi$,
$m_{D | R}(r;P)= F_{S|R}(\theta(P) |  r) - \chi$,
and
$f_{X_k}(r)$
are bounded over $(r,P)$;
and
$f_R(x;P)$ and $f_{U|R}(0| x; P)$ are bounded away from zero over $P$.

\item[5.] \underline{Assumption \ref{aspt:asy_linear} - \eqref{eq:aspt:asy_linear_lind}:} $(2+\zeta)$-th moment condition. 

We verify it in two steps.

First, $\mmv_P\left(  \psi_n(\bZ_i,P) \right) - \xi^2(P) = o_{\m{P}}(1)$ 
and 
$\xi^2(P)$ is bounded away from zero, uniformly over $\m{P}$.
Thus,
$\mmv_P^{-1}\left(  \psi_n(\bZ_i,P) \right) = O_{\m{P}}(1)$.

Second, for any $\zeta>0$, 
call $\eta(r;P) = \mme_P[ | D_i - \mme_P [ D_i | R_i ] |^{2+\zeta} | R_i=r]$
and note that
$|\eta(r;P)| \leq 2.$
\begin{align*}
&n^{-\zeta/2} \mme_P \left|  \psi_n(\bV_i,P) \right| ^{2+\zeta} 
\\
& =\frac{(m/n)^{\zeta/2}}{ f_{U|R}^{2+\zeta}(0 | x ; P) f_{R}^{2+\zeta}(x;P) (mh)^{\zeta/2}} \int |K(u)|^{2 +\zeta} \eta(x+uh;P) f_{R}(x+uh;P) ~ du 
\\
& = o_{\m{P}}(1).
\end{align*}
Combining steps 1 and 2, 
\begin{align*}
n^{-\zeta/2} \sup\limits_{P \in \m{P} }  \mme_P \left| \frac{ \psi_n(\bV_i,P) }{ \sqrt{\mmv_P\left(  \psi_n(\bV_i,P) \right)} } \right|^{2 + \zeta} = o(1).
\end{align*}

\item[6.] \underline{Assumption \ref{aspt:asy_linear} - \eqref{eq:aspt:asy_linear_xicont}:} 
$\xi^2\left( \frac{m}{n} P_1 + \frac{n-m}{n} P_2 \right) \to \xi^2 \left( \gamma P_1 + (1-\gamma) P_2  \right)$.

Let $\overline{P}_n = \frac{m}{n} P_1 + \frac{n-m}{n} P_2$ and $\overline{P} = \gamma P_1 + (1-\gamma) P_2$.
Consider the expression for $\xi^2\left( P \right)$ given above.
It suffices to show that
$f_{R}(x;\bar{P}_n) \to f_{R}(x; \bar P)$
and
$f_{U|R}(0|x;  \bar{P_n}) \to f_{U|R }(0|x; \bar P)$.
The first is straightforward (see Section \ref{proof:control_mean}).
For the second, let $U_k =  Y_k - \theta(P_k) $ and note that
\begin{align*}
f_{U|R}(0|x ; \bar P_n) 
&= 
\frac{m}{n}\frac{f_{X_1}(x)}{f_R(x;\bar P_n)} f_{U_1|X_1}(0|x) + \frac{n-m  }{n  } \frac{f_{X_2}(x)}{f_R(x;\bar P_n)}f_{U_2|X_2}(0|x) 
\\
&\to  
\gamma \frac{f_{X_1}(x )}{f_R(x;\bar P)}f_{U_1 | X_1 }(0|x )  + (1-\gamma)\frac{f_{X_2}(x )}{f_R(x;\bar P)} f_{U_2 | X_2 }(0|x) 
\\
&= f_{U|R}(0|x;\bar P).
\end{align*}

\item[7.] \underline{Assumption \ref{aspt:sampling} :} $n_1 /n - \lambda \pto 0$.

In this case,  $n_1$ is deterministic and $(n_1 /n - \lambda) \to 0$.
Assumption  \ref{aspt:asy_linear} has already been verified above for any sequence $m_n$ such that $(m/n -\gamma) \to 0$ for arbitrary $\gamma \in (0,1)$.
In particular it holds for $\gamma \in \{ \lambda, 1-\lambda\}.$

\end{enumerate}

$\square$

\subsection{Proof of Proposition \ref{prop:rdd} - Discontinuity of Conditional Mean}
\label{proof:rdd}

\indent 

The goal of this proof is to use the assumptions listed in Proposition \ref{prop:rdd}  to verify Assumptions \ref{aspt:asy_linear} and \ref{aspt:sampling}.
It follows the general lines of the proof of Proposition \ref{prop:control_mean}, so the reader may refer to Section \ref{proof:control_mean} for the redundant details that we omit here.
It builds on arguments from the literature on asymptotic approximations of the NW estimator at a boundary point.
See, for example,  Theorem 3.2 by \citeSM{fan1996} with $p=\nu=0$.
For a generalization of this proof to the local polynomial regression estimator (LPR), see Section \ref{proof:lpr} below.

Consider an \textit{iid} sample from $P\in\m{P}$ with $m$ observations, $\bV_1=(R_1,S_1), \ldots, \bV_m=(R_m,S_m)$,
where the minimum value in the support of $R$ is 0.
The number $m$ grows with $n$ such that 
$m/n \to \gamma$, for some $\gamma \in (0,1)$.
The parameter of interest is $\theta(P) = \mme[S|R=0^+]$, and 
the NW estimator is 
\[
\ha{\theta}^b 
= \theta_{m,n}^b(\bV_1 , \ldots, \bV_m) 
\doteq \frac{\sum\limits_{i=1}^{m} K \left(  \frac{ R_{i} }{ h }  \right)  S_{i} }
{\sum\limits_{i=1}^{m} K \left(  \frac{ R_{i} }{ h }  \right) }.
\]
This section studies the asymptotic behavior of 
$\ha{\theta} = \ha{\theta}^b - h \ha{B} $,
where the expression for $\ha{B}$ is given below Equation \ref{eq:rdd:BP}.
The assumptions in Proposition \ref{prop:rdd}  imply the following facts:

\begin{enumerate}
\item \textit{As $m\to\infty$, $h \to 0$, $m h \to \infty$,
$\sqrt{m h} h =O(1),$ and $\sqrt{m h} h^2 =o(1)$};

\item \textit{The distribution of $R$ has PDF $f_{R}(r;P)$
that is twice differentiable wrt $r$ denoted 
$\nabla_r f_{R}(r;P)$ and $\nabla_{r^2} f_{R}(r;P)$;
$f_{R}(r;P)$ and the derivatives  are bounded functions of $(r,P)$;
$f_{R}(r;P)$ is bounded away from zero as a function of $(r,P)$;}

\item \textit{$m_{S|R}(r;P) = \Exp_P[S|R=r]$ is twice differentiable  wrt $r$ denoted $\nabla_r m_{S|R}(r;P)$ and $\nabla_{r^2} m_{S|R}(r;P)$;
$m_{S|R}$ and the derivatives are   bounded functions of $(r,P)$;}

\item \textit{$v_{S|R}(r;P) = \Var_P[S|R=r]$ has first derivative wrt $r$ denoted $\nabla_r v_{S|R}(r;P)$;
$v_{S|R}$, $\nabla_r v_{S|R}$ are both  bounded as functions of $(r,P)$;
$v_{S|R}(0^+;P)$ is bounded away from zero as a function of  $P$;}

\item \textit{$\eta(r;P) \doteq \mme_P[|S - m_{S | R }(R;P) |^{2+\zeta} | R=r]$ is a bounded function of $(r,P)$.}

\end{enumerate}

We re-write 
$\sqrt{m h}  \left( \ha\theta - \theta(P) \right)$
$=\sqrt{m h}  \left( \ha\theta^b - h \ha{B} - \theta(P) \right)$   
to find the asymptotic linear representation.
\begin{align}
\sqrt{m h}  \left( \ha\theta  -\theta(P) \right)   
&= 
\left( 
	\frac{1}{\sqrt{m h} } \sum_{i=1}^{m} K\left( \frac{ R_{i} }{h} \right) 
		\left( S_{i} - m_{S  | R }(R_i;P)  \right) 
\right)\left(f_{R}(0^+;P)/2\right)^{-1}
\label{eq:rdd:thetahat:clt}
\\
&+ 
\left( 
	\frac{1}{\sqrt{m h} } \sum_{i=1}^{m} K\left( \frac{ R_{i}  }{h} \right) 
		\left( S_{i} - m_{S  | R }(R_i;P)  \right) 
\right)
\notag
\\
&\hspace{2cm}
\left[
	\left( \frac{1}{ m h } \sum_{i=1}^{m} K\left( \frac{ R_{i} }{h} \right) \right)^{-1} 
	-
	\left(f_{R}(0^+;P)/2\right)^{-1}  
\right]
\label{eq:rdd:thetahat:deno}
\\
&+ 
\left( 
	\frac{1}{\sqrt{m h} } \sum_{i=1}^{m} K\left( \frac{ R_{i} }{h} \right) 
		\left( m_{S  | R }(R_i;P)   -  m_{S  | R }(0^+;P)  \right) 
\right)
\notag
\\
&\hspace{2cm}
\left( 
	\frac{1}{ m h } \sum_{i=1}^{m} K\left( \frac{ R_{i} }{h} \right) 
\right)^{-1} 
-\sqrt{mh} h  \ha{B}.
\label{eq:rdd:thetahat:bias}
\end{align}

\begin{enumerate}

\item \underline{Assumption \ref{aspt:asy_linear} - \eqref{eq:aspt:asy_linear}:} asymptotic expansion.

 Equation \ref{eq:rdd:thetahat:clt} above gives  the influence function $\psi_n$.

\begin{align*}
\frac{1}{\sqrt{m} } \sum_{i=1}^{m} \underbrace{ 
		K\left( \frac{ R_{i}  }{h} \right) 
		\left( S_{i} - m_{S  | R }(R_i;P)  \right) h^{-1/2} \left(f_{R}(0^+;P)/2\right)^{-1}
}_{\doteq \psi_{n}(\bV_{i}, P)}
=
\frac{1}{\sqrt{m} } \sum_{i=1}^{m} \psi_{n}(\bV_{i}, P).
\end{align*}

We need to show that Equations \ref{eq:rdd:thetahat:deno} and \ref{eq:rdd:thetahat:bias} converge in probability to zero uniformly over $\m{P}$.

\noindent
\textit{Equation \ref{eq:rdd:thetahat:deno}:} is $o_{\m{P}}(1)$.
We show this in 3 steps.
First,
\begin{align*}
\mmv_P\left( 
	\frac{1}{\sqrt{m h} } \sum_{i=1}^{m} K\left( \frac{ R_{i}  }{h} \right) 
		\left( S_{i} - m_{S  | R }(R_i;P)  \right) 
\right)
\\
 = \int_0^{\infty} K^2(u) v_{S|R}(uh;P) f_{R}(uh;P) ~ du = O_{\m{P}}(1).
\end{align*}
The expected value of the expression inside the variance above is zero, so we have that  
\[
	\frac{1}{\sqrt{m h} } \sum_{i=1}^{m} K\left( \frac{ R_{i} }{h} \right) 
		\left( S_{i} - m_{S  | R }(R_i;P)  \right) 
=O_{\m{P}}(1).
\]

Second, 
\begin{align*}
&\mme_P\left[ \frac{1}{ m h } \sum_{i=1}^{m} K\left( \frac{ R_{i} }{h} \right) \right] - f_{R}(0^+;P)/2 
= h \int_0^{\infty}  K(u) \nabla_r f_{R}(x^*_{uh};P) u~ du = o_{\m{P}}(1).
\\
&\mmv_P\left[ \frac{1}{ m h } \sum_{i=1}^{m} K\left( \frac{ R_{i} }{h} \right) \right]  
\leq
\frac{1}{ m h } \int K^2(u) f_{R}(uh;P ) ~ du
= o_{\m{P}}(1).
\end{align*}
Therefore,
\begin{align}
&\frac{1}{ m h } \sum_{i=1}^{m} K\left( \frac{ R_{i}  }{h} \right)  - f_{R}(0^+;P)/2  = o_{\m{P}}(1),
\label{eq:rdd:fhat1}
\\
&\left( \frac{1}{ m h } \sum_{i=1}^{m} K\left( \frac{ R_{i}  }{h} \right) \right)^{-1} 
-
\left(f_{R}(0^+;P)/2\right)^{-1}  
= o_{\m{P}}(1).
\label{eq:rdd:fhat2}
\end{align}

Third, combining steps 1 and 2 gives
\begin{align*}
&\left( 
	\frac{1}{\sqrt{m h} } \sum_{i=1}^{m} K\left( \frac{ R_{i} }{h} \right) 
		\left( S_{i} - m_{S  | R }(R_i;P)  \right) 
\right)
\\
&\hspace{1cm}
\left[
	\left( \frac{1}{ m h } \sum_{i=1}^{m} K\left( \frac{ R_{i} }{h} \right) \right)^{-1} 
	-
	\left(f_{R}(0^+;P)/2\right)^{-1}  
\right]
& = O_{\m{P}}(1) o_{\m{P}}(1) = o_{\m{P}}(1).
\end{align*}

\bigskip

\noindent
\textit{Equation \ref{eq:rdd:thetahat:bias}:} is $o_{\m{P}}(1)$.
We derive the probability limit of \eqref{eq:rdd:thetahat:bias} + 
$\sqrt{mh}h \ha{B}$ in 3 steps.\\
First, 
$\left( \frac{1}{ m h } \sum_{i=1}^{m} K\left( \frac{ R_{i} }{h} \right) \right)^{-1} = 2 f_R^{-1}(0^+;P) + o_{\m{P}}(1) $. \\
Second, 
\begin{align*}
& 
\mme_P\left[ 
	\frac{1}{\sqrt{m h} } \sum_{i=1}^{m} K\left( \frac{ R_{i}  }{h} \right) 
		\left( m_{S  | R }(R_i;P)  - m_{S  | R }(0^+;P)  \right) 
\right]
\\
=& 
\sqrt{mh}  \int_0^{\infty}  K\left( u \right) [m_{S  | R }(uh;P) - m_{S  | R }(0^+;P)] f_{R}(uh;P) ~ du
\\
=& 
\sqrt{mh}  \int_0^{\infty}  K\left( u \right) [
    \nabla_r m_{S  | R }(0^+;P)uh 
    + \nabla_{r^2} m_{S  | R }(x^*_{uh};P) u^2 h^2 
] f_{R}(uh;P) ~ du
\\
=& 
\sqrt{mh}  \int_0^{\infty}  K\left( u \right) [
    \nabla_r m_{S  | R }(0^+;P)uh 
] [
    f_{R}(0^+;P) + \nabla_r f_{R}(x^{**}_{uh};P) uh
]~ du
\\
&+ 
\sqrt{mh}  \int_0^{\infty}  K\left( u \right) [
    \nabla_{r^2} m_{S  | R }(x^*_{uh};P) u^2 h^2 
] f_{R}(uh;P) ~ du
\\
=& 
\sqrt{mh} h  \nabla_r m_{S  | R }(0^+;P)  
f_{R}(0^+;P) 
\int_0^{\infty}  u K\left( u \right) 
~ du
\\
&+ 
\sqrt{mh} h^2 \nabla_r m_{S  | R }(0^+;P)  
\int_0^{\infty}  u^2 K\left( u \right) 
 \nabla_r f_{R}(x^{**}_{uh};P) ~ du
\\
&+ 
\sqrt{mh} h^2  \int_0^{\infty}  u^2 K\left( u \right) 
    \nabla_{r^2} m_{S  | R }(x^*_{uh};P) f_{R}(uh;P) ~ du
\\
=& 
\sqrt{mh} h  \nabla_r m_{S  | R }(0^+;P)  
f_{R}(0^+;P) 
\kappa_{1,2}^{+}
\\
&+
O_{\m{P}} \left( \sqrt{mh} h^2 \right)
\\
&+ 
O_{\m{P}} \left( \sqrt{mh} h^2 \right)
\\
=& \sqrt{mh} h  \underbrace{   \nabla_r m_{S  | R }(0^+;P)  
f_{R}(0^+;P) 
\kappa_{1,2}^{+}}_{\doteq B_0(P) }
+ o_{\m{P}} \left( 1 \right)
\\
=&  \sqrt{mh} h  B_0(P) 
+ o_{\m{P}} \left( 1 \right),
\end{align*}
where $\kappa_{s,t}^{+}=\int_0^{\infty} u^s K^t(u) du$,
and we use the existence and boundedness of derivatives, kernel moments, and
the rate conditions $\sqrt{mh}h = O(1)$ and $\sqrt{mh}h^2 = o(1)$.
\begin{align*}
& \mmv_P\left[ 
	\frac{1}{\sqrt{m h} } \sum_{i=1}^{m} K\left( \frac{ R_{i}  }{h} \right) 
		\left( m_{S | R }(R_i;P) - m_{S | R}(0^+;P) \right) 
\right]
\\
=&h^2 \int_0^{\infty}  K^2\left( u \right) [\nabla_r m_{S|R}(x^*_{uh};P)]^2 u^2  f_{R}(uh;P) ~ du
= h^2  O_{\m{P}}(1) = o_{\m{P}}(1).
\end{align*}
This gives,
\begin{align*}
& \frac{1}{\sqrt{m h} } \sum_{i=1}^{m} K\left( \frac{ R_{i}  }{h} \right) 
		\left( m_{S | R }(R_i ;P) - m_{S | R }( 0^+;P)  \right) 
\\
& = \sqrt{mh} h  B_0(P)  + o_{\m{P}}(1).
\end{align*}

Third, combine the first and second steps:
\begin{align}
& 
\left( 
	\frac{1}{\sqrt{m h} } \sum_{i=1}^{m} K\left( \frac{ R_{i}  }{h} \right) 
		\left( m_{S | R }(R_i;P) - m_{S | R}(0^+;P)  \right) 
\right)
\left( 
	\frac{1}{ m h } \sum_{i=1}^{m} K\left( \frac{ R_{i} }{h} \right) 
\right)^{-1}
\notag
\\
&=
 \sqrt{mh} h  ~ 
 \underbrace{ 2 \kappa_{1,2}^{+} \nabla_r m_{S  | R }(0^+;P)}_{\doteq B(P)}  
+ o_{\m{P}}(1)
\label{eq:rdd:BP}
\\
&=
 \sqrt{mh} h  ~ 
  B(P)  
+ o_{\m{P}}(1),
\notag
\end{align}
which gives the bias term $B(P)$.
That term is consistently estimated  by $\ha{B}$,
\[
\ha{B} = B_{m,n}(\bV_1 , \ldots, \bV_m) \doteq  2 \kappa_{1,1}^{+}  
        \ha{\nabla_r m_{S|R}}(x;P),
\]
that is, by replacing 
$\nabla_r m_{S|R}(x;P)$
in $B(P)$ by a consistent nonparametric estimator.
The additional tuning parameters for such estimator  and corresponding moment conditions may be set such that the bias-correction condition in Assumption \ref{aspt:prop:bias} is met.

\item \underline{Assumption \ref{aspt:asy_linear} - \eqref{eq:aspt:asy_linear_zeromean}:} zero mean of influence function. 

 $\mme_P[ \psi_{n}(\bV_{i}, P) ]=0 ~~\forall P$ by construction.

\item \underline{Assumption \ref{aspt:asy_linear} - \eqref{eq:aspt:asy_linear_var}:} variance of influence function.

 Call $\xi^2(P) = 4 \kappa_{0,2}^{+} v_{S | R}(0^+;P) /f_{R}(0^+;P)$, where 
$\kappa_{0,2}^{+} = \int_{0}^{\infty}K^2(u) ~ du$.
\begin{align*}
&\mmv_P\left( 
	\frac{2}{f_{R}(0^+;P) \sqrt{h} }  K\left( \frac{ R_{i}  }{h} \right) 
		\left( S_{i} - m_{S | R}(R_i;P) \right) 
\right) - \xi^2(P) 
\\
 =&\frac{4}{f_{R}^2(0^+;P) }  \int_0^{\infty}  K^2(u)  \left\{ 
 	\underbrace{ v_{S | R}(uh;P) f_{R}(uh;P) }_{\doteq g(uh;P)} - 
 	\underbrace{ v_{S | R}(0^+;P) f_{R}(0^+;P) 
 	 }_{\doteq g(0^+;P) } \right\}
 ~ du
 \\
=&\frac{4h}{f_{R}^2(0^+;P) }  \int_0^{\infty}  K^2(u)  \nabla_r g(x^*_{uh};P) u  ~ du = o_{\m{P}}(1).
\end{align*}
Therefore,
\[
\sup\limits_{P \in \m{P} } \left| \mmv_P[  \psi_n(\bV_{i}, {P} ) ] - \xi^2(P) \right| \to 0.
\]

\item \underline{Assumption \ref{aspt:asy_linear} - \eqref{eq:aspt:asy_linear_mom}:} $ \sup_P \mme [\psi_n^2(\bZ_{k,i},P)]  < \infty .$

We have,
\begin{align*}
& \psi_n^2(\bZ_k,P)  = \frac{4K^2\left( \frac{ X_{k}  }{h} \right)}{ h f_{R}^{2}(0^+;P)} 
		\left[ \left(Y_{k} -  m_{Y_k|X_k}(X_k) \right) + \left(m_{Y_k|X_k}(X_k)  - m_{S|R}(X_k;P) \right)  \right]^2 ,
\\
&\mme\left[ \psi_n^2(\bZ_k,P) \right] 
\\
&=
\frac{4}{f_{R}^{2}(0^+;P)}
\int_0^{\infty}
	K^2(u)
	\left[
		v_{Y_k|X_k}(uh) 
		+ \left(m_{Y_k|X_k}(uh)  - m_{S|R}(uh;P) \right)^2
	\right]
	f_{X_k}(uh)
~du 
\\
&= O_{\m{P}}(1).
\end{align*}

\item \underline{Assumption \ref{aspt:asy_linear} - \eqref{eq:aspt:asy_linear_lind}:} $(2+\zeta)$-th moment condition. 

First, $\mmv_P^{-1}\left(  \psi_n(\bV_i,P) \right) = O_{\m{P}}(1). $

Second, $n^{-\zeta/2} \mme_P \left|  \psi_n(\bV_i,P) \right| ^{2+\zeta} = o_{\m{P}}(1)$
because
$\eta(r;P) = O_{\m{P}}(1)$.

Therefore,  
\begin{align*}
n^{-\zeta/2} \sup\limits_{P \in \m{P} }  \mme_P\left| \frac{ \psi_n(\bV_i,P) }{ \sqrt{\mmv_P\left(  \psi_n(\bV_i,P) \right)} } \right|^{2 + \zeta} = o(1).
\end{align*}

\item \underline{Assumption \ref{aspt:asy_linear} - \eqref{eq:aspt:asy_linear_xicont}:} 
$\xi^2\left( \frac{m}{n} P_1 + \frac{n-m}{n} P_2 \right) \to \xi^2 \left( \gamma P_1 + (1-\gamma) P_2  \right)$.

Let $\overline{P}_n = \frac{m}{n} P_1 + \frac{n-m}{n} P_2$ and $\overline{P} = \gamma  P_1 + (1-\gamma ) P_2$.
We have that \\ $\xi^2(\overline{P}_n) = \kappa_{0,2}^{+} v_{S  | R }(0^+; \overline{P}_n)  /f_{R}(0^+;\bar{P}_n)$, 
$v_{ S  | R }(0^+; \bar{P}_n ) \to v_{ S  | R }(0^+; \bar{P})$
and
$f_{R}(0^+;\bar{P}_n) \to f_{R}(0^+;\bar{P})$
as in Section \ref{proof:control_mean}.
Therefore, $\xi^2(\overline{P}_n) \to \xi^2(\overline{P})$.

\item \underline{Assumption \ref{aspt:sampling} :} $ (n_1 /n - \lambda) \pto 0$. 

In this setting, $n_1/n = \sum_i \mmi\{ X_i \geq 0 \} /n$ and $(n_1 /n - \lambda) \pto 0$ holds.
Assumption  \ref{aspt:asy_linear} has already been verified above for any sequence $m_n$ such that $(m/n -\gamma) \to 0$ for arbitrary $\gamma \in (0,1)$.
In particular it holds for $\gamma \in \{ \lambda, 1-\lambda\}.$

\end{enumerate}

$\square$

\subsection{Proof of Proposition \ref{prop:bunching} - Discontinuity of Density}
\label{proof:bunching}

\indent 

The goal of this proof is to use the assumptions listed in Proposition \ref{prop:bunching}  to verify Assumptions \ref{aspt:asy_linear} and \ref{aspt:sampling}.
It follows the general lines of the proof of Proposition \ref{prop:control_mean}, so the reader may refer to Section \ref{proof:control_mean} for the redundant details that we omit here.
Additional references on the kernel density estimator at boundary points include
\citeSM{marron1994} and Sections 1.9--1.10 of \citeSM{liracine2007}.

Consider an \textit{iid} sample from $P\in\m{P}$ with $m$ observations, $\bV_1=R_1, \ldots, \bV_m=R_m$,
where $R$ is a scalar random variable and $0$ is an interior point of its compact support.
The number $m$ grows with $n$ such that $m/n \to \gamma$, for some $\gamma \in (0,1)$.
The parameter of interest is $\theta(P) = (f_R(0^+;P) - f_R(0^-;P))/2$,
where $f_R(0^+;P)$ is the side limit as $x \downarrow 0$ of the PDF of $R$ under distribution $P$, and $f_R(0^-;P)$ is the negative side limit.
The kernel density estimator  for $\theta(P)$
is
\[
\ha{\theta}_{k}^b = \frac{1}{ n_k h} \sum\limits_{i=1}^{n_k} K \left(  \frac{ X_{k,i} }{ h }  \right)  \left(\mmi\{X_{k,i} \ge 0\} - \mmi\{X_{k,i} < 0\}  \right)
\text{, for }k=1,2.
\]
This section investigates the asymptotic behavior
of 
$\ha{\theta}   = \ha{\theta}^b - h \ha{B}$, where the expression for 
$\ha{B}$ is given below Equation \ref{eq:bunching:BP}.

The assumptions in Proposition \ref{prop:bunching}  imply the following facts:

\begin{enumerate}
\item \textit{As $m\to\infty$, $h \to 0$, $m h \to \infty$; 
$\sqrt{m h} h =O(1)$, and $\sqrt{m h} h^2 =o(1)$;}

\item \textit{The distribution of $R$ has PDF $f_{R}(r;P)$ that is twice differentiable wrt $r$ except at $r=0$;
$f_{R}(r;P)$ and the derivatives $\nabla_r f_{R}(r;P)$
and
$\nabla_{r^2} f_{R}(r;P)$
are bounded functions of $(r,P)$;
$f_{R}(0^+;P)$ and $f_{R}(0^-;P)$  
are bounded away from zero as functions of $P$;}

\end{enumerate}

We re-write 
$\sqrt{m h}  \left( \ha\theta -\theta(P) \right)$
$= \sqrt{m h}  \left( \ha\theta^b -h \ha{B} -\theta(P) \right)$
to find the asymptotic linear representation.
\begin{align}
&\sqrt{m h}  \left( \ha\theta -\theta(P) \right)  
\nonumber  
\\
&=  \frac{1}{\sqrt{m h} } \sum_{i=1}^{m} \left\{
	K\left( \frac{ R_i }{h} \right)\left(\mmi\{R_i \ge 0\} - \mmi\{R_i < 0\}  \right) \right.
\notag
\\
&	\hspace{3cm} \left. - \mme_P\left[ K\left( \frac{ R }{h} \right)\left(\mmi\{R \ge 0\} - \mmi\{R < 0\}  \right)\right] \right\}
\label{eq:bunching:thetahat:clt}
\\
&+ \frac{1}{\sqrt{m h} } \sum_{i=1}^{m} 
	\mme_P\left[ K\left( \frac{ R  }{h} \right)\left(\mmi\{R  \ge 0\} - \mmi\{R  < 0\}  \right)\right]
\notag
\\
&\hspace{3cm}  
	 -    \frac{ \sqrt{mh} }{2} \left(f_R(0^+;P) - f_R(0^-;P)  \right)
	 - \sqrt{mh} h \ha{B}.
\label{eq:bunching:thetahat:bias}
\end{align}

\begin{enumerate}

\item \underline{Assumption \ref{aspt:asy_linear} - \eqref{eq:aspt:asy_linear}:} asymptotic expansion.

 Equation \ref{eq:bunching:thetahat:clt} above gives  the influence function $\psi_n$, that is, 
 $\eqref{eq:bunching:thetahat:clt}=m^{-1/2} \sum_{i=1}^m \psi_n(\bV_i,P)$, where
\begin{align*}
\psi_n(\bV_i,P) 
& \doteq  \frac{1}{\sqrt{h} } \left\{
	K\left( \frac{ R_i }{h} \right)\left(\mmi\{R_i \ge 0\} - \mmi\{R_i < 0\}  \right) \right.
\notag
\\
&	\hspace{3cm} \left. - \mme_P\left[ K\left( \frac{ R }{h} \right)\left(\mmi\{R \ge 0\} - \mmi\{R < 0\}  \right)\right] \right\}.
\end{align*}

We need to show that Equation \ref{eq:bunching:thetahat:bias} converges to zero uniformly over $\m{P}$. 
\\
Define $\kappa_{s,t}^+ = \int_0^{\infty} u^s K^t(u) ~du$.
First, consider the limit of  
\eqref{eq:bunching:thetahat:bias} $+ \sqrt{mh}h \ha{B}$. 
\begin{align}
&\frac{1}{\sqrt{m h} } \sum_{i=1}^{m} 
	\mme_P\left[ K\left( \frac{ R  }{h} \right)\left(\mmi\{R  \ge 0\} - \mmi\{R < 0\}  \right)\right]  
	- \frac{\sqrt{mh}}{2}\left(f_R(0^+;P)-f_R(0^-;P)\right)
	\notag
\\
= &  \sqrt{mh}\left( \int_0^{\infty} K(u) f_{R}(uh;P) ~ du -\int_{-\infty}^0 K(u) f_{R}(uh;P) ~ du \right.
\notag
 \\
& \hspace{3cm} \left. - \frac{1}{2}\left(f_R(0^+;P) - f_R(0^-;P)\right)   \right) 
\notag
 \\
= &\sqrt{mh} \left(
    \int_0^{\infty} K(u)  [
        \nabla_r f_{R}(0^+;P)uh
        +\nabla_{r^2} f_{R}(x^*_{uh};P)u^2 h^2/2
    ] ~ du 
\right.
\notag
\\
&\hspace{3cm} \left. - 
        \int_{-\infty}^0 K(u)  [
            \nabla_r f_{R}(0^-;P) uh 
            + \nabla_{r^2} f_{R}(x^{**}_{uh};P) u^2 h^2/2 
            ] ~ du \right)~
\notag
 \\
= &\sqrt{mh} h \nabla_r f_{R}(0^+;P)
    \int_0^{\infty} u K(u)   ~ du 
    \notag
\\
 &+\sqrt{mh} (h^2)/2 
    \int_0^{\infty} u^2 K(u) \nabla_{r^2} f_{R}(x^*_{uh};P)  ~ du 
\notag
\\
& - \sqrt{mh} h
\nabla_r f_{R}(0^-;P)
\int_{-\infty}^0 u K(u) ~ du 
\notag
\\
& - \sqrt{mh} (h^2/2) \int_{-\infty}^0 u^2 K(u)  
\nabla_{r^2} f_{R}(x^{**}_{uh};P)   ~ du
\notag
\\
= &\sqrt{mh} h \nabla_r f_{R}(0^+;P)
    \kappa_{1,1}^{+}
    \notag
\\
 &+O_{\m{P}} \left( \sqrt{mh} h^2 \right) 
 \notag
\\
& + \sqrt{mh} h
\nabla_r f_{R}(0^-;P)
\kappa_{1,1}^{+}
\notag
\\
& +O_{\m{P}} \left( \sqrt{mh} h^2 \right) 
\notag
\\
= & \sqrt{mh}h ~\underbrace{\kappa_{1,1}^{+} \left(
    \nabla_r f_{R}(0^+;P)
    + \nabla_r f_{R}(0^-;P)
\right)}_{\doteq B(P)}+  o_{\m{P}}(1)
\label{eq:bunching:BP}
\\
= & \sqrt{mh}h ~ B(P) +  o_{\m{P}}(1),
\notag
\end{align}
where we used the existence and boundedness of derivatives, kernel moments, 
and the rate condition $\sqrt{mh}h^2=o(1)$.

Equation \ref{eq:bunching:BP} defines the bias term $B(P)$.
That term is consistently estimated by $\ha{B}$,
\[
\ha{B} = B_{m,n}(\bV_1 , \ldots, \bV_m) 
\doteq  \kappa_{1,1}^{+} \left[
        \ha{ \nabla_r f_{R}} (0^+;P)
        +
        \ha{ \nabla_r f_{R}} (0^-;P) 
\right],
\]
that is, by replacing 
$\nabla_r f_{R}(0^+;P)$ and 
$\nabla_r f_{R}(0^-;P)$
in $B(P)$ by consistent nonparametric estimators.
The tuning parameters for these additional estimators and corresponding moment conditions may be set such that the bias-correction condition in Assumption \ref{aspt:prop:bias}   is met.

Finally, the limit of \eqref{eq:bunching:thetahat:bias} is,
\begin{align*}
& \frac{1}{\sqrt{m h} } \sum_{i=1}^{m} 
	\mme_P\left[ K\left( \frac{ R  }{h} \right)\left(\mmi\{R  \ge 0\} - \mmi\{R  < 0\}  \right)\right]
\notag
\\
&\hspace{3cm}  
	 -    \frac{ \sqrt{mh} }{2} \left(f_R(0^+;P) - f_R(0^-;P)  \right)
	 - \sqrt{mh} h \ha{B}
\\
=& \sqrt{mh} h \left(B(P) - \ha{B} \right) +  o_{\m{P}}(1)
=  o_{\m{P}}(1),
\end{align*}
where we use $\ha{B}  - B(P) =  o_{\m{P}}(1)$ 
and
$\sqrt{mh}h=O(1)$.

\item \underline{Assumption \ref{aspt:asy_linear} - \eqref{eq:aspt:asy_linear_zeromean}:} zero mean of influence function. 

 $\mme_P[ \psi_{n}(\bV_i, P) ]=0 ~~\forall P$ by construction.

\item \underline{Assumption \ref{aspt:asy_linear} - \eqref{eq:aspt:asy_linear_var}:} variance of influence function.

 Call 
 $\xi^2(P) =  \kappa_{0,2}^{+} \left( f_{R}(0^+;P) + f_{R}(0^-;P) \right) $,
 where
 $\kappa_{0,2}^{+} = \int_{0}^{\infty}K^2(u) ~ du$.
\begin{align*}
&\mmv_P\left( 
	\frac{1}{ \sqrt{h} }  K\left( \frac{R_i  }{h} \right)\left(\mmi\{R_i \ge 0\} - \mmi\{R_i < 0\}  \right) \right) - \xi^2(P) 
 \\
 =&\frac{1}{h} \left[
 	\int_{-\infty}^{\infty}  K^2(u)f_R(uh;P)h ~du 
 \right.
 \\
 & \left. \hspace{1cm} - \left(
 		\int_0^{\infty} K(u)f_R(uh;P)h~du 
 		- \int_{-\infty}^0 K(u)f_R(uh;P)h~du
 	\right)^2
\right] - \xi^2(P)
\\
 =& \int_{0}^{\infty}  K^2(u)[f_R(0^+;P) + \nabla_r f_{R}(x^*_{uh};P) uh ]  ~du  
\\
&+ \int_{-\infty}^{0}  K^2(u)[f_R(0^-;P) + \nabla_r f_{R}(x^{**}_{uh};P)uh ]   ~du 
 \\
 & -  h \left(\int_0^{\infty} K(u) f_R(uh;P) ~du
 - \int_{-\infty}^0 K(u)f_R(uh;P) ~ du\right)^2  - \xi^2(P)
 \\
=&h \left[\int_{0}^{\infty}  K^2(u)  \nabla_r f_R(x^*_{uh};P) u  ~ du + \int_{-\infty}^{0}  K^2(u)  \nabla_r f_R(x^{**}_{uh};P) u  ~ du \right. 
\\
& \left. - \left(\int_0^{\infty} K(u) f_R(uh;P) ~du 
- \int_{-\infty}^0 K(u) f_R(uh;P) ~du\right)^2
\right]
\\
=& h O_{\m{P}}(1) = o_{\m{P}}(1).
\end{align*}
Therefore,
\[
\sup\limits_{P \in \m{P} } \left| \mmv_P [  \psi_n(\bV_i, {P} ) ] - \xi^2(P) \right| \to 0.
\]

\item \underline{Assumption \ref{aspt:asy_linear} - \eqref{eq:aspt:asy_linear_mom}:} $ \sup_P \mme [\psi_n^2(\bZ_k,P)]  < \infty$ for $k=1,2.$

We have,
\begin{align*}
& \psi_n(\bZ_k,P)  = h^{-1/2} K\left( \frac{X_{k} }{h} \right)\left(\mmi\{X_{k} \ge 0\} - \mmi\{X_{k} < 0\}  \right)  
 \\
&  - h^{-1/2}  \mme_P \left[ K\left( \frac{ V }{h} \right)\left(\mmi\{ V \geq 0\} - \mmi\{V < 0\}  \right)\right],
\\
&\psi_n^2(\bZ_k,P)  = 
		\frac{1 }{h} 
		 K^2\left( \frac{X_{k} }{h} \right) 
		 \\
		 & + 
		 \frac{1}{h} \left\{ 
		 	\mme_P\left[ K\left( \frac{ V }{h} \right) \left(\mmi\{V \ge 0\} - \mmi\{V < 0\}  \right)\right] 
		 \right\}^2 
\\
		& -\frac{2}{h} K\left( \frac{ X_k }{h} \right) \left(\mmi\{X_k \ge 0\} - \mmi\{X_k < 0\}  \right) 
		\mme_P\left[ K\left( \frac{ V }{h} \right) \left(\mmi\{V \ge 0\} - \mmi\{V < 0\}  \right)\right], 		
\\
&\mme\left[ \psi_n^2(\bZ_k,P)  \right] = 
		\int K^2\left( u \right) f_{X_k}(uh) ~ du
		 \\
		 & + 
		 h \left\{ 
		 	\int_0^{\infty} K\left( u \right) f_{R}(uh;P) ~ du 
		 	- \int_{-\infty}^{0} K\left( u \right) f_{R}(uh;P) ~ du 
		 \right\}^2 
\\
		& -2  \left\{
			\int_0^{\infty} K\left( u \right) f_{X_k}(uh) ~ du - \int_{-\infty}^{0} K\left( u \right) f_{X_k}(uh) ~ du 
		\right\}
\\
& \hspace{1cm} \times
		 h \left\{ 
		 	\int_0^{\infty} K\left( u \right) f_{R}(uh;P) ~ du 
		 	- \int_{-\infty}^{0} K\left( u \right) f_{R}(uh;P) ~ du 
		 \right\},
\\
& \mme\left[ \psi_n^2(\bZ_k,P)  \right] = 
		O_{\m{P}}(1) +  O_{\m{P}}(h) +  O_{\m{P}}(h) = O_{\m{P}}(1).
\end{align*}

\item \underline{Assumption \ref{aspt:asy_linear} - \eqref{eq:aspt:asy_linear_lind}:} $(2+\zeta)$-th moment condition. 

We verify it in two steps.

First, $\mmv_P\left(  \psi_n(\bZ_i,P) \right) - \xi^2(P) = o_{\m{P}}(1)$ 
and 
$\xi^2(P)$ is bounded away from zero, uniformly over $\m{P}$.
Thus,
$\mmv_P^{-1}\left(  \psi_n(\bZ_i,P) \right) = O_{\m{P}}(1)$.

Second, pick $\zeta>0$.
From before, $\mme_P\left[ K\left( \frac{ R }{h} \right)\left(\mmi\{R \ge 0\} - \mmi\{R< 0\}  \right)\right] = O_{\m{P}}(h)$.
\begin{align*}
& n^{-\zeta/2}  \mme_P \left|  \psi_n(\bV_{i}, {P} )    \right|^{2+\zeta} 
\\
&= \frac{1}{(nh)^{\zeta/2} h } 
\mme_P \left|  K\left( \frac{ R_i }{h} \right)\left(\mmi\{V_i \ge 0\} - \mmi\{V_i < 0\}  \right) - O_{\m{P}}(h) \right| ^{2+\zeta} 
\\
&= \frac{1}{(nh)^{\zeta/2} }
	\int_0^{\infty} \left| K(u) - O_{\m{P}}(h) \right|^{2+\zeta} f_R(uh;P) du
\\
&\hspace{1cm} + \frac{1}{(nh)^{\zeta/2} }	
	\int_{-\infty}^{0} \left| - K(u) - O_{\m{P}}(h) \right|^{2+\zeta} f_R(uh;P) du
\\
&= \frac{1}{(nh)^{\zeta/2} }
O_{\m{P}}(1) 
= o_{\m{P}}(1).
\end{align*}
Combining both steps gives, 
\begin{align*}
n^{-\zeta/2} \sup\limits_{P \in \m{P} }  \mme_P\left| \frac{ \psi_n(\bV_i,P) }{ \sqrt{\mmv_P\left(  \psi_n(\bV_i,P) \right)} } \right|^{2 + \zeta} = o(1).
\end{align*}

\item \underline{Assumption \ref{aspt:asy_linear} - \eqref{eq:aspt:asy_linear_xicont}:} 
$\xi^2\left( \frac{m}{n} P_1 + \frac{n-m}{n} P_2 \right) \to \xi^2 \left( \gamma P_1 + (1-\gamma) P_2  \right)$.

Let $\overline{P}_n = \frac{m}{n} P_1 + \frac{n-m}{n} P_2$
and  $\overline{P} = \gamma P_1 + (1-\gamma) P_2$.
Given that \\
$\xi^2(\overline{P}_n) =  \kappa_{0,2}^{+} \left( f _{R}(0^+; \overline{P}_n ) + f _{R}(0^-; \overline{P}_n ) \right)  ~ du$,
the result follows because $m/n \to \gamma$,
$f _{R}(0^+; \overline{P}_n )   \to f _{R}(0^+; \overline{P} ) $,
and
$f _{R}(0^-; \overline{P}_n )   \to f _{R}(0^-; \overline{P} ) $.

\item \underline{Assumption \ref{aspt:sampling} :} $ (n_1 /n - \lambda) \pto 0$. 

In this setting, $n_1$ is a deterministic sequence, and we have that  $n_1/n  = \lfloor n/2 \rfloor / n \to 1/2$.

Assumption  \ref{aspt:asy_linear} has already been verified above for any sequence $m_n$ such that $(m/n -\gamma) \to 0$ for arbitrary $\gamma \in (0,1)$.
In particular it holds for $\gamma = 1/2.$

\end{enumerate}

$\square$

\subsection{Extension of Proposition \ref{prop:rdd} - Local Polynomial Regression}
\label{proof:lpr}

\indent 

Sections \ref{sec:applications:rdd} and \ref{proof:rdd} gave sufficient conditions and demonstrated the asymptotic validity of our robust permutation test in the context of RDD using the NW estimator.
This section generalizes that finding to the Local Polynomial Regression (LPR) estimator of order $\rho \in \mathbb{Z}_+$.

The researcher has an \textit{iid} sample $(X_1,Y_1), \ldots, (X_n,Y_n)$ that is split into two samples
$\bZ_{k,i} = (X_{k,i},Y_{k,i})$, $k=1,2$, $i =1, \ldots, n_k$, as explained in Section \ref{sec:applications:rdd}, 
where the distribution of $\bZ_{k,i}$ conditional on $\underline{W}_n$ is $P_k$.
We have that $P_1$ is the distribution of $(X,Y)|X \geq 0$,
$P_2$ is the distribution of $(-X,Y)|X < 0$,
and
$\m{P}$ is the set of all convex combinations of $P_1$ and $P_2$.
The researcher selects a polynomial order $\rho \in \mathbb{Z}_+$ and a bandwidth $h>0$.
The LPR estimator for $\theta(P_k)$ is defined as follows:
\begin{align*}
&(\ha{a}, \ha{\boup{b}}  )
= \argmin\limits_{(a, \boup{b} )}
\sum\limits_{i=1}^{n_k}  K \left( \frac{X_{k,i}  }{h  } \right) 
 \bigg[ Y_{k,i} - a - b_1 X_{k,i}  - b_2 X_{k,i}^2  - \ldots -b_{\rho} X_{k,i}^{\rho} \bigg]^2,
\\
&\ha{\theta}^b_k 
= \theta_{n_k,n}^b(\bZ_{k,1} , \ldots, \bZ_{k,n_k}) = \ha{a},
\end{align*}
where ${a}$ is a scalar and ${\boup{b}}$ is the vector $(b_1,\ldots,b_{\rho})$.
Note that the LPR estimator becomes the NW estimator when we set $\rho=0$.

The superscript $b$ in $\ha{\theta}_{k}^b$ indicates there is bias in the asymptotic distribution of   
$\sqrt{n h} (\ha{\theta}_{k}^b - \theta(P_k))$ 
whenever the bandwidth choice converges to zero at the slowest possible rate, i.e., $h=O(n^{-1/(2\rho +3)})$
(Proposition \ref{prop:rddlpr} below).
This is the case of MSE-optimal bandwidths and inference requires bias correction in that case.
A conventional solution is to subtract  a first-order bias term $h^{\rho+1} B(P_k)$ from $\ha{\theta}_{k}^b$, 
where $B(P_k)$ is nonparametrically estimated by $\ha{B}_k = B_{n_k,n}(\bZ_{k,1}, \ldots, \bZ_{k,n_k})$.
We give the analytical formulas for $B(P_k)$ and $\ha{B}_k$ in the proof of Proposition \ref{prop:rddlpr} below (Equations \ref{eq:rddlpr:bias:bp} and \ref{eq:rddlpr:bias:bphat}).
Our permutation tests utilize the bias-corrected LPR estimator 
$\ha{\theta}_{k} = \theta_{n_k,n}(\bZ_{k,1}, \ldots, \bZ_{k,n_k}) \doteq 
 \theta_{n_k,n}^b(\bZ_{k,1}, \ldots, \bZ_{k,n_k}) - h^{\rho+1}  B_{n_k,n}(\bZ_{k,1}, \ldots, \bZ_{k,n_k})
=\ha{\theta}_{k}^b - h^{\rho+1}  \ha{B}_k$.
Note that no bias correction is needed if $h=o(n^{-1/(2\rho +3)})$.

\begin{proposition}\label{prop:rddlpr}
Assume that:   
(i) as $n\to \infty$, $h \to 0$,  $n h \to \infty$, and $\sqrt{n h} h^{\rho+1} \to c \in [0,\infty)$;
(ii) $K$ is a kernel density function that is non-negative, bounded, symmetric, and 
$\int K(u) |u|^{4(\rho+1)}~du <\infty$;
(iii) the distribution of $X$ has PDF $f_{X}$ that is bounded, bounded away from zero,
$\rho+2$ times differentiable except at $x=0$, and has bounded derivatives;
(iv) $\Exp[Y | X = x]$ is bounded, $\rho+2$ times differentiable except at $x=0$, and has bounded derivatives;
(v) $\Var[Y | X = x]$ is bounded, differentiable except at $x=0$, has bounded derivative,
$\Var[Y|X=0^+]>0$, and $\Var[Y|X=0^-]>0$, where $0^+$ and $0^-$ denote side limits;
and
(vi) $\exists \zeta>0$ such that $\Exp[|Y|^{2+\zeta} | X]$ is almost surely bounded.
Let $\bV_1=(R_1,S_1), \ldots, \bV_m=(R_m,S_m)$ be an \textit{iid} sample from a distribution $P \in \m{P}$,
where $m$ grows with $n$.
Use the definitions above to construct the the LPR estimators:
$\ha{\theta}^b = \theta_{m,n}^b(\bV_{1}, \ldots, \bV_{m})$,
$\ha{B} = B_{m,n}(\bV_{1}, \ldots, \bV_{m})$,
and 
$\ha{\theta} = \theta_{m,n}(\bV_{1}, \ldots, \bV_{m})$.
Finally, assume that 
(vii)
$\ha{B} - B(P) \pto 0$
uniformly over $P \in \mathcal{P}$ for any  sequence $m$ such that  $ m/n  \to   \lambda $
or
$ m/n  \to   1-\lambda $.
Then, 
the RDD setting with the LPR estimator defined above satisfies Assumptions \ref{aspt:asy_linear} and \ref{aspt:sampling}  with
\begin{align*}
\psi_n(\bV,P) & = 
 \frac{1}{\sqrt{h} f_R(0^+;P) }  K\left( \frac{ R  }{h } \right) 
 {\be_1}' \bGam^{-1} \bHt^{\rho} \left(S - m_{S|R}(R;P) \right),		 
\\
\xi^2(P) = & {\be_1}' \bGam^{-1} \bDel \bGam^{-1} {\be_1} \frac{ v_{S|R}(0^+;P) }{f_R(0^+;P)},
\end{align*}
where the following definitions are used,
\begin{align*}
& \be_1 \text{ is the } (\rho_1+1 \times 1) \text{ column vector }  \be_1=[1~~ 0~~ 0~~ \cdots~~ 0]', 
\\
& \bHt^{\rho}
= \left[ 1 ~~ \left( \frac{R}{h} \right) ~~ \cdots ~~ \left( \frac{R  }{h } \right)^{\rho } \right]',  ~~ (\rho+1 \times 1) \text{ column vector, }
\\
& \kappa_{s,t}^{+} = \int_{0}^{\infty} u^s K^t(u) ~ du \text{ for integers s and t, }
\\
& \bGam = \left[
\begin{array}{ccc}
\kappa_{0,1}^{+} & \ldots & \kappa_{\rho,1}^{+}
\\
\vdots & \vdots & \vdots
\\
\kappa_{\rho,1}^{+} & \ldots & \kappa_{2 \rho,1}^{+}
\end{array}
\right]
\text{  and  }
\bDel =
\left[
	\begin{array}{ccc}
		\kappa_{0,2}^{+} & \ldots & \kappa_{\rho,2}^{+}
		\\
		\vdots & \vdots & \vdots
		\\
		\kappa_{\rho,2}^{+} & \ldots & \kappa_{2 \rho,2}^{+}
	\end{array}
\right].
\end{align*}
Moreover, 
$m_{S|R}(r;P)$ is the conditional mean of $S$ given $R=r$,
$v_{S|R}(r;P)$ is the conditional variance of $S$ given $R=r$,
and 
$f_{R}(r;P)$ is the PDF of $R$ at $r$, 
all three assuming $\bV=(R,S) \sim P \in \m{P}$.
\end{proposition}

\textbf{\textit{Proof. }}
The goal of this proof is to use the assumptions listed in the proposition to verify that the RDD setting with the LPR estimator satisfies Assumptions \ref{aspt:asy_linear} and \ref{aspt:sampling}.
It follows the general lines of the proof of Proposition \ref{prop:control_mean}, so the reader may refer to Section \ref{proof:control_mean} for the redundant details that we omit here.
The proof also builds on arguments from the literature on asymptotic approximations of the LPR estimator at a boundary point: 
\cite{porter2003} (Theorem 3(a)) and \cite{fan1996} (Theorem 3.2).

Consider an \textit{iid} sample from $P\in\m{P}$ with $m$ observations, $\bV_1=(R_1,S_1), \ldots, \bV_m=(R_m,S_m)$,
where the minimum value in the support of $R$ is 0.
The number $m$ grows with $n$ such that 
$m/n \to \gamma$, for some $\gamma \in (0,1)$.
The assumptions listed in the proposition imply the following facts (see proof of Proposition \ref{prop:control_mean} in Section \ref{proof:control_mean} for details):

\begin{enumerate}
\item \textit{As $m\to\infty$, $h \to 0$, $m h \to \infty$,
$\sqrt{m h} h^{\rho+1} =O(1),$ and $\sqrt{m h} h^{\rho+2} =o(1)$};

\item \textit{The distribution of $R$ has PDF $f_{R}(r;P)$
that is $\rho+2$ times differentiable wrt $r$ denoted 
$\nabla_r f_{R}(r;P)$, \ldots,  $\nabla_{r^{\rho+2}} f_{R}(r;P)$;
$f_{R}(r;P)$ and its derivatives  are bounded functions of $(r,P)$;
$f_{R}(r;P)$ is bounded away from zero as a function of $(r,P)$;}

\item \textit{$m_{S|R}(r;P) = \Exp_P[S|R=r]$ is $\rho+2$ times differentiable  wrt $r$ 
denoted $\nabla_r m_{S|R}(r;P)$, \ldots,  $\nabla_{r^{\rho+2}} m_{S|R}(r;P)$;
$m_{S|R}$ and its derivatives are  bounded functions of $(r,P)$;}

\item \textit{$v_{S|R}(r;P) = \Var_P[S|R=r]$ has first derivative wrt $r$ denoted $\nabla_r v_{S|R}(r;P)$;
$v_{S|R}$, $\nabla_r v_{S|R}$ are both  bounded  functions of $(r,P)$;
$v_{S|R}(0^+;P)$ is bounded away from zero as a function of  $P$;}

\item \textit{$\eta(r;P) \doteq \mme_P[|S - m_{S | R }(R;P) |^{2+\zeta} | R=r]$ is a bounded function of 
$(r,P)$.}

\end{enumerate}

Use the data $\bV_1=(R_1,S_1), \ldots, \bV_m=(R_m,S_m)$ and the definitions above to construct the the LPR estimators:
$\ha{\theta}^b = \theta_{m,n}^b(\bV_{1}, \ldots, \bV_{m})$,
$\ha{B} = B_{m,n}(\bV_{1}, \ldots, \bV_{m})$,
and 
$\ha{\theta} = \theta_{m,n}(\bV_{1}, \ldots, \bV_{m})$.
The proof studies the asymptotic behavior of 
$\ha{\theta} = \ha{\theta}^b - h \ha{B} $.

Consider the following definitions,
\begin{align*}
&\be_1 \text{ is the } (\rho_1+1 \times 1) \text{ column vector }  \be_1=[1~~ 0~~ 0~~ \cdots~~ 0]', 
\\
&\bH_i^{\rho} = \left[ 1 ~~  R_i ~~ \cdots ~~ R_i^{\rho} \right]',~~ (\rho+1 \times 1) \text{ column vector, }
\\
&\bHt_i^{\rho}
= \left[ 1 ~~ \left( \frac{R_i}{h} \right) ~~ \cdots ~~ \left( \frac{R_i  }{h } \right)^{\rho } \right]',  ~~ (\rho+1 \times 1) \text{ column vector, }
\\
&\bG_m  = \left[ \frac{1}{m h } \sum_{i=1}^m K\left( \frac{R_i  }{h } \right) \bHt_{i}^{\rho}   \bHt_{i}^{' \rho}  \right]^{-1}, ~~ (\rho+1 \times \rho+1) \text{ matrix, }
\\
&\bphi_{\rho}(P) =  \left[ m_{S|R}(0^+;P)~~ \nabla_r m_{S|R}(0^+;P)/1! ~~ \cdots ~~ \nabla_{r^{\rho}} m_{S|R}(0^+;P)/\rho! \right]', 
\\
&\hspace{2cm} ~~ (\rho+1 \times 1) \text{ column vector, }
\\
& \eps_i = S_i - m_{S|R}(R_i;P), ~~ \text{scalar, }
\\
&\bar{\eps}_i^{\rho} = S_i - \bH_i^{\rho'} \bphi_{\rho}(P) = \eps_i + m_{S|R}(R_i;P)  - \bH_i^{\rho'} \bphi_{\rho}(P), ~~ \text{scalar. }
\end{align*}

It is possible to rewrite $\sqrt{m h}  \left( \ha{\theta} - \theta(P) \right)$ as follows,
\begin{align}
& \sqrt{m h}  \left( \ha{\theta} - \theta(P) \right) 
= \sqrt{m h}  \left( \ha{\theta}^b - \theta(P) \right) - \sqrt{m h} h^{\rho+1} \ha{B}
\nonumber
\\
& ={\be_1}' \bG_m 
\left[  \frac{1}{\sqrt{m h } } \sum_{i=1}^m K\left( \frac{ R_i  }{h } \right)  \bHt_{i}^{\rho}
  \bar{\eps}_i^{\rho}  \right]
- \sqrt{m h} h^{\rho+1} \ha{B}  
\nonumber
\\
& ={\be_1}' \bG_m 
\left[  \frac{1}{\sqrt{m h } } \sum_{i=1}^m K\left( \frac{ R_i  }{h } \right)  \bHt_{i}^{\rho}
\left( \bar{\eps}_i^{\rho}  - \mme_P [\bar{\eps}_i^{\rho} |R_i] \right)   \right]
\label{eq:rddlpr:clt}
\\
& +
{\be_1}' \bG_m
\left\{ \frac{1}{\sqrt{m h}} \sum_{i=1}^m K\left( \frac{ R_i  }{h } \right)  \bHt_{i}^{\rho}
\mme_P [\bar{\eps}_i^{\rho} |R_i]
-
\mme_P  \left[
\frac{1}{\sqrt{mh}} \sum_{i=1}^m K\left( \frac{ R_i  }{h } \right)  \bHt_{i}^{\rho}
\bar{\eps}_i^{\rho}
\right]
~\right\}
\label{eq:rddlpr:op1}
\\
&+
{\be_1}' \bG_m
\mme_P \left[
\frac{1}{\sqrt{m h }} \sum_{i=1}^m K\left( \frac{ R_i  }{h } \right)  \bHt_{i}^{\rho}
\bar{\eps}_i^{\rho}
\right]
- \sqrt{m h} h^{\rho+1} \ha{B}.
\label{eq:rddlpr:bias}
\end{align}

Equation \ref{eq:rddlpr:clt} gives the influence function representation and
Equations \ref{eq:rddlpr:op1}--\ref{eq:rddlpr:bias} are $o_{\m{P}}(1)$.
The limit of the first term in \eqref{eq:rddlpr:bias} characterizes the bias term $B(P)$ 
(Equation \ref{eq:rddlpr:bias:bp} below).

\bigskip

In what follows, we verify Assumptions \ref{aspt:asy_linear} and \ref{aspt:sampling} in 7 parts.

\begin{enumerate}

\item \underline{Assumption \ref{aspt:asy_linear} - \eqref{eq:aspt:asy_linear}:} asymptotic expansion.

\noindent
\underline{\textit{Equation \ref{eq:rddlpr:clt}:}} 

Let $\kappa_{s,t}^{+} = \int_{0}^{\infty} u^s K^t(u) ~ du$ and define
\begin{align*}
& \bGam = \left[
\begin{array}{ccc}
\kappa_{0,1}^{+} & \ldots & \kappa_{\rho,1}^{+}
\\
\vdots & \vdots & \vdots
\\
\kappa_{\rho,1}^{+} & \ldots & \kappa_{2 \rho,1}^{+}
\end{array}
\right],
\\
& \bG(P) = \frac{1}{f_R(0^+;P)} \bGam^{-1}.
\end{align*}

Rewrite Equation \ref{eq:rddlpr:clt} as,
\begin{align}
& {\be_1}' \bG_m 
\left[  \frac{1}{\sqrt{m h } } \sum_{i=1}^m K\left( \frac{ R_i  }{h } \right)  \bHt_{i}^{\rho}
 \eps_i    \right]
\nonumber
\\
=& {\be_1}' \left[ \bG_m - \bG(P) \right]
\left[  \frac{1}{\sqrt{m h } } \sum_{i=1}^m K\left( \frac{ R_i  }{h } \right)  \bHt_{i}^{\rho}
 \eps_i    \right]
\label{eq:rddlpr:clt:1} 
\\
&+
\left[  \frac{1}{\sqrt{m h } } \sum_{i=1}^m K\left( \frac{ R_i  }{h } \right) {\be_1}' \bG(P) \bHt_{i}^{\rho}
 \eps_i    \right].
\label{eq:rddlpr:clt:2} 
\end{align}

We show that Equation \ref{eq:rddlpr:clt:1} is $o_{\m{P}}(1)$ by showing that $\bG_m - \bG(P) = o_{\m{P}}(1)$ and \\
$\frac{1}{\sqrt{m h } } \sum_{i=1}^m K\left( \frac{ R_i  }{h } \right)  \bHt_{i}
 \eps_i$ is $O_{\m{P}}(1)$.
 
We have that $\bG_m^{-1}$ is a sample average and $\bG_m^{-1} \pto \bG^{-1}(P)$ (\cite{porter2003}, top of page 45).
The convergence is made uniform over $\m{P}$ by the fact that $\nabla_r f_R(r;P)$ is a bounded
function of $(r,P)$
and the first $2 (\rho+1)$ moments of $K(u)$ are bounded.
The absolute value of the determinant of the matrix $\bG^{-1}(P)$ is bounded away from zero because
$f_R(0^+;P)$ is bounded away from zero over $P \in \m{P}$. 
Therefore, $\bG_m - \bG(P) = o_{\m{P}}(1)$.

We have that $\frac{1}{\sqrt{m h } } \sum_{i=1}^m K\left( \frac{ R_i  }{h } \right)  \bHt_{i}
 \eps_i$ is $O_{\m{P}}(1)$ because it has zero mean and variance $O_{\m{P}}(1)$.
The variance is bounded because the kernel has bounded moments, 
and $f_R(r;P)$ and $v_{S|R}(r;P)$ are bounded functions of $(r,P)$.

Equation \ref{eq:rddlpr:clt:2} gives the influence function.
\begin{align*} 
\psi_n(\bV_i,P) & = 
 \frac{1}{\sqrt{h} }  K\left( \frac{ R_i  }{h } \right) {\be_1}' \bG(P) \bHt_{i}^{\rho} \left(S_i - m_{S|R}(R_i;P) \right)
\\
& = 
\frac{1}{\sqrt{h} f_R(0^+;P) }  
K\left( \frac{ R_i  }{h } \right) 
{\be_1}' \bGam^{-1} \bHt_{i}^{\rho} \left(S_i - m_{S|R}(R_i;P) \right).
\end{align*}

\bigskip

\noindent
\underline{\textit{Equation \ref{eq:rddlpr:op1}:}}

Equation \ref{eq:rddlpr:op1} is shown to be $o_{\m{P}}(1)$: from above $\bG_m= O_{\m{P}}(1)$ and we show that 
\[
\frac{1}{\sqrt{m h}} \sum_{i=1}^m K\left( \frac{ R_i  }{h } \right)  \bHt_{i}^{\rho}
\mme_P[\bar{\eps}_i^{\rho} |R_i]
-
\mme_P  \left[
\frac{1}{\sqrt{mh}} \sum_{i=1}^m K\left( \frac{ R_i  }{h } \right)  \bHt_{i}^{\rho}
\bar{\eps}_i^{\rho}
\right]
= o_{\m{P}}(1).
\]

To see this last equation, note that
\[
\mme_P  \left\{ 
	\frac{1}{\sqrt{m h}} \sum_{i=1}^m K\left( \frac{ R_i  }{h } \right)  \bHt_{i}^{\rho}
	\mme_P [\bar{\eps}_i^{\rho} |R_i]
\right\}
=
\mme_P  \left[
\frac{1}{\sqrt{mh}} \sum_{i=1}^m K\left( \frac{ R_i  }{h } \right)  \bHt_{i}^{\rho}
\bar{\eps}_i^{\rho}
\right].
\]
It remains to show that 
\[
\mmv_P \left\{
	\frac{1}{\sqrt{m h}  } \sum_{i=1}^m K\left( \frac{ R_i  }{h } \right)  \left( \frac{R_i}{h}\right)^l
	\mme_P [\bar{\eps}_i^{\rho} |R_i]
\right\}
= o_{\m{P}}(1),
\]
for $l=0, \ldots, \rho$.
In fact, 
\begin{align*}
&\mmv_P \left\{
	\frac{1}{\sqrt{m h}  } \sum_{i=1}^m K\left( \frac{ R_i  }{h } \right)  \left( \frac{R_i}{h}\right)^l
	\mme_P [\bar{\eps}_i^{\rho} |R_i]
\right\}
\\
&=\frac{1}{m h   } \sum_{i=1}^m \mmv_P  \left\{ 
	K\left( \frac{ R_i  }{h } \right)  \left( \frac{R_i}{h} \right)^l
	\mme_P [\bar{\eps}_i^{\rho} |R_i]
	\right\}
\\
&=\frac{1}{h } \mmv_P  \left\{ 
	K\left( \frac{ R_i  }{h } \right)  \left( \frac{R_i}{h} \right)^l
	\mme_P [\bar{\eps}_i^{\rho} |R_i]
	\right\}
\\
&\leq \frac{1}{h } \mme_P  \left\{ 
	K^2 \left( \frac{ R_i  }{h } \right)  \left( \frac{R_i}{h} \right)^{2l}
	\left( \mme_P [\bar{\eps}_i^{\rho} |R_i] \right)^2
	\right\}
\\
&= \frac{1}{h } \mme_P  \left\{ 
	K^2 \left( \frac{ R_i  }{h } \right)  \left( \frac{R_i}{h} \right)^{2l}
	\left( \mme_P [\bar{\eps}_i^{\rho+1} |R_i] + R_i^{\rho+1} \nabla_{r^{\rho+1}} m_{S|R}(0^+;P)/(\rho+1)!  \right)^2
	\right\}
\\
&\leq \frac{2}{h } \mme_P  \left\{ 
	K^2 \left( \frac{ R_i  }{h } \right)  \left( \frac{R_i}{h} \right)^{2l}
	\left[
		\left( \mme_P [\bar{\eps}_i^{\rho+1} |R_i] \right)^2 
		+ \left( R_i^{\rho+1} \nabla_{r^{\rho+1}} m_{S|R}(0^+;P)/(\rho+1)!  \right)^2
	\right]
	\right\}	
\\
&= 2 \bigintssss_0^{\infty}   
	K^2 \left( u \right)    u^{2l}
	\left( \mme_P [\bar{\eps}_i^{\rho+1} |R_i=uh] \right)^2 
	 	 f_R(uh;P) ~ du
\\
&\hspace{1cm} 
+ 2 \frac{ h^{2\rho+2} }{ [(\rho+1)!]^2 }  \left[ \nabla_{r^{\rho+1}} m_{S|R}(0^+;P)   \right]^2
	\bigintssss_0^{\infty}  
		K^2 \left( u \right)   u ^{2l+2\rho+2} f_R(uh;P)
	 	 ~ du
\\
&= 2 \bigintssss_0^{\infty}   
	K^2 \left( u \right)    u^{2l}
	\left[ \nabla_{r^{\rho+2}} m_{S|R}(r^*_u;P) (uh)^{\rho+2}/(\rho+2)! \right]^2 
	 	 f_R(uh;P) ~ du
\\
&\hspace{1cm} 
+ 2 \frac{ h^{2\rho+2} }{ [(\rho+1)!]^2 }  \left[ \nabla_{r^{\rho+1}} m_{S|R}(0^+;P)   \right]^2
	\bigintssss_0^{\infty}  
		K^2 \left( u \right)   u ^{2l+2\rho+2} f_R(uh;P)
	 	 ~ du
\\
&= 2 \frac{ h^{2\rho+4} }{ [(\rho+2)!]^2 } 
\bigintssss_0^{\infty}   
	K^2 \left( u \right)    u^{2l + 2\rho + 4}
	\left[ \nabla_{r^{\rho+2}} m_{S|R}(r^*_u;P) \right]^2 
	 	 f_R(uh;P) ~ du
\\
&\hspace{1cm} 
+ 2 \frac{ h^{2\rho+2} }{ [(\rho+1)!]^2 }  \left[ \nabla_{r^{\rho+1}} m_{S|R}(0^+;P)   \right]^2
	\bigintssss_0^{\infty}  
		K^2 \left( u \right)   u ^{2l+2\rho+2} f_R(uh;P)
	 	 ~ du	
\\
& = O_{\m{P}} \left( h^{2\rho+4} \right) + O_{\m{P}} \left( h^{2\rho+2} \right) 
= O_{\m{P}} \left( h^{2 (\rho+1) } \right)  = o_{\m{P}} \left( 1 \right),  	 		 	 		 	 	
\end{align*}
where $r^*_u$ lies between $uh$ and $0$, 
and we use the following facts:
(i) the first $4\rho+4$ moments of $K$ are bounded,
(ii) the functions $f_R(r;P)$,
$\nabla_{r^{\rho+1}} m_{S|R}(0^+;P)$,
and
$\nabla_{r^{\rho+2}} m_{S|R}(r;P)$ 
are bounded functions of $(r,P)$.

\bigskip

\noindent
\underline{\textit{Equation \ref{eq:rddlpr:bias}:}} is shown to be $o_{\m{P}}(1)$. 
Define
\begin{align}
\bgam^* &  = \left[	
	\begin{array}{c}
		\kappa_{\rho+1,1}^{+}
		\\
		\vdots 
		\\
		\kappa_{2\rho+1,1}^{+}
	\end{array}
\right],
\\
B(P) &  = \frac{1}{ (\rho+1)! } \nabla_{r^{\rho+1}} m_{S|R}(0^+;P) {\be_1}' \bGam^{-1}
\bgam^*.
\label{eq:rddlpr:bias:bp}
\end{align}
The term $B(P)$ is consistently estimated by
\begin{align}
& \ha{B} =  
B_{m,n}(\bV_1 , \ldots, \bV_m) \doteq 
\frac{1}{ (\rho+1)! } \ha{\nabla_{r^{\rho+1}} m_{S|R}}(0^+;P) {\be_1}' \bGam^{-1}
\bgam^*,
\label{eq:rddlpr:bias:bphat}
\end{align}
that is, by replacing 
$\nabla_{r^{\rho+1}} m_{S|R}(0^+;P)$
in $B(P)$ by a consistent nonparametric estimator (condition (vii) of Proposition \ref{prop:rddlpr}).

First, we show that 
\begin{align*}
& \frac{1}{h^{\rho+1}}
\mme_P \left[
	\frac{1}{ m h } \sum_{i=1}^m K\left( \frac{ R_i  }{h } \right)  \bHt_{i}^{\rho}
	\bar{\eps}_i^{\rho}
\right] 
\\
&= 
\frac{f_R(0^+;P) }{ (\rho+1)! } \nabla_{r^{\rho+1}} m_{S|R}(0^+;P) 
\bgam^*
+ o_{\m{P}}(1).
\end{align*}
It suffices to show that, for any $l=0, \ldots, \rho$, 
\begin{align*}
& \frac{1}{h^{\rho+1}}
\mme_P \left[
	\frac{1}{ m h } \sum_{i=1}^m K\left( \frac{ R_i  }{h } \right)  \left( \frac{ R_i  }{h } \right)^l
	\bar{\eps}_i^{\rho}
\right] 
\\
&= 
\frac{f_R(0^+;P) }{ (\rho+1)! } \nabla_{r^{\rho+1}} m_{S|R}(0^+;P) 
\kappa_{\rho+1+l,1}^{+}
+ o_{\m{P}}(1).
\end{align*}
To see that, 
\begin{align*}
& \frac{1}{h^{\rho+1}}
\mme_P \left[
	\frac{1}{ m h } \sum_{i=1}^m K\left( \frac{ R_i  }{h } \right)  \left( \frac{ R_i  }{h } \right)^l
	\bar{\eps}_i^{\rho}
\right] 
\\
& = \frac{1}{h^{\rho+1}}
\mme_P \left\{
	\frac{1}{ m h } \sum_{i=1}^m K\left( \frac{ R_i  }{h } \right)  \left( \frac{ R_i  }{h } \right)^l
	\mme_P\left[ \bar{\eps}_i^{\rho} | R_i \right]
\right\} 
\\
&= \frac{1}{h^{\rho+2} } \mme_P  \left\{ 
	K \left( \frac{ R_i  }{h } \right)  \left( \frac{R_i}{h} \right)^{l}
	\left( \mme_P [\bar{\eps}_i^{\rho+1} |R_i] + R_i^{\rho+1} \nabla_{r^{\rho+1}} m_{S|R}(0^+;P)/(\rho+1)!  \right)
	\right\}
\\
&= \frac{1}{h^{\rho+1} }
\bigintssss_0^{\infty}   
	K \left( u \right)    u^{l}
	\mme_P [\bar{\eps}_i^{\rho+1} |R_i=uh]
	 	 f_R(uh;P) ~ du
\\
&\hspace{1cm} 
+ \frac{ \nabla_{r^{\rho+1}} m_{S|R}(0^+;P)   }{ (\rho+1)! }  
	\bigintssss_0^{\infty}  
		K \left( u \right)   u ^{l + \rho + 1} f_R(uh;P)
	 	 ~ du
\\
&= \frac{1}{h^{\rho+1} }
\bigintssss_0^{\infty}   
	K \left( u \right)    u^{l}
	\left[ \nabla_{r^{\rho+2}} m_{S|R}(r^*_u;P) (uh)^{\rho+2}/(\rho+2)! \right]
	 	 f_R(uh;P) ~ du
\\
&\hspace{1cm} 
+ \frac{ \nabla_{r^{\rho+1}} m_{S|R}(0^+;P)   }{ (\rho+1)! }  
	\bigintssss_0^{\infty}  
		K \left( u \right)   u ^{l + \rho + 1} f_R(uh;P)
	 	 ~ du
\\
&= \frac{h}{ (\rho+2)! }
\bigintssss_0^{\infty}   
	K \left( u \right)    u^{l+\rho+2}
	\left[ \nabla_{r^{\rho+2}} m_{S|R}(r^*_u;P)  \right]
	 	 f_R(uh;P) ~ du
\\
&\hspace{1cm} 
+ \frac{ \nabla_{r^{\rho+1}} m_{S|R}(0^+;P)   }{ (\rho+1)! }  
	\bigintssss_0^{\infty}  
		K \left( u \right)   u ^{l + \rho + 1} f_R(uh;P)
	 	 ~ du	 	 
\\
&=  \frac{ \nabla_{r^{\rho+1}} m_{S|R}(0^+;P)  }{ (\rho+1)! }   
	\bigintssss_0^{\infty}  
		K \left( u \right)   u ^{l + \rho + 1} \left[ f_R(0^+;P) + \nabla_r f_R(r^*_u;P) uh \right]
	 	 ~ du	 	 
\\
& \hspace{1cm} +O_{\m{P}} \left( h \right)
\\
& 
=  \frac{ \nabla_{r^{\rho+1}} m_{S|R}(0^+;P)  f_R(0^+;P) }{ (\rho+1)! }   
	\bigintssss_0^{\infty}  
		K \left( u \right)   u ^{l + \rho + 1} 
	 	 ~ du	 	
\\
& 
\hspace{1cm} - h  \frac{ \nabla_{r^{\rho+1}} m_{S|R}(0^+;P)  }{ (\rho+1)! }   
	\bigintssss_0^{\infty}  
		K \left( u \right)   u ^{l + \rho + 2} \nabla_r f_R(r^*_u;P) 
	 	 ~ du	 	
\\
& \hspace{1cm} + O_{\m{P}} \left( h \right)	 	 
\\
& 
=  \frac{ \nabla_{r^{\rho+1}} m_{S|R}(0^+;P)  f_R(0^+;P) }{ (\rho+1)! }   
	\bigintssss_0^{\infty}  
		K \left( u \right)   u ^{l + \rho + 1} 
	 	 ~ du	 	
\\
& \hspace{1cm} +O_{\m{P}} \left( h \right) +O_{\m{P}} \left( h \right)	 	 
\\
& 
=  \frac{ \nabla_{r^{\rho+1}} m_{S|R}(0^+;P)  f_R(0^+;P) }{ (\rho+1)! }   
	\kappa_{\rho+1+l,1}^{+}	 	
+o_{\m{P}} \left( 1 \right),	 			
\end{align*}
where the first $O_{\m{P}} \left( h \right)$ term is obtained from the following facts: 
(i) the first $2 \rho+2$ moments of the kernel are bounded;
(ii) $\nabla_{r^{\rho+2}} m_{S|R}(r;P)$ is bounded as a function of $(r,P)$;
and
(iii) $f_R(r;P)$ is bounded as a function of $(r,P)$.
The Taylor expansion of the PDF $f_R(r;P)$ and the second $O_{\m{P}} \left( h \right)$ term are obtained from the following facts:
(i) the first $2 \rho+2$ moments of the kernel are bounded;
(ii) $\nabla_r f_R(r;P)$ is a well-defined and bounded function of $(r,P)$;
and
(iii) $\nabla_{r^{\rho+1}} m_{S|R}(0^+;P)$ is a well-defined and bounded function of $(r,P)$.

Rewrite Equation \ref{eq:rddlpr:bias} as,
\begin{align*}
&
{\be_1}' \bG_m
\mme_P\left[
\frac{1}{\sqrt{m h }} \sum_{i=1}^m K\left( \frac{ R_i  }{h } \right)  \bHt_{i}^{\rho}
\bar{\eps}_i^{\rho}
\right]
- \sqrt{m h} h^{\rho+1} \ha{B}
\\
&=
\sqrt{m h } h^{\rho+1} \left\{
	{\be_1}' \bG_m  \frac{1}{h^{\rho+1}}
	\mme_P\left[
		\frac{1}{ m h } \sum_{i=1}^m K\left( \frac{ R_i  }{h } \right)  \bHt_{i}^{\rho}
		\bar{\eps}_i^{\rho}
	\right]
	- \ha{B}
\right\}
\\
&=
\sqrt{m h } h^{\rho+1} \left\{
	{\be_1}' \left[ \bG(P) +  o_{\m{P}} \left( 1 \right) \right] 
	\left[
		\frac{f_R(0^+;P) }{ (\rho+1)! } \nabla_{r^{\rho+1}} m_{S|R}(0^+;P) \bgam^*
		+ o_{\m{P}}(1)
	\right]
	- \ha{B}
\right\}
\\
&=
\sqrt{m h } h^{\rho+1} \left\{
	B(P) + o_{\m{P}}(1)
	- \ha{B}
\right\}
= O(1) o_{\m{P}}(1) = o_{\m{P}}(1),
\end{align*}
where we use the facts that:
$\bG_m = \bG(P) + o_{\m{P}}(1)$, $\bG(P)=O_{\m{P}}(1)$, \\ 
${f_R(0^+;P) } \nabla_{r^{\rho+1}} m_{S|R}(0^+;P) \bgam^* /{ (\rho+1)! } = O_{\m{P}}(1)$,
$\ha{B}= B(P) + o_{\m{P}}(1)$,
and $\sqrt{mh}h^{\rho+1} = O(1)$.

\item \underline{Assumption \ref{aspt:asy_linear} - \eqref{eq:aspt:asy_linear_zeromean}:} zero mean of influence function. 

 $\mme_P[ \psi_{n}(\bV_{i}, P) ]=0 ~~\forall P$ by construction.

\item \underline{Assumption \ref{aspt:asy_linear} - \eqref{eq:aspt:asy_linear_var}:} variance of influence function.

Recall that $\kappa_{s,t}^{+} = \int_{0}^{\infty} u^s K^t(u) ~ du$ and define
\begin{align*}
\bDel = & 
\left[
	\begin{array}{ccc}
		\kappa_{0,2}^{+} & \ldots & \kappa_{\rho,2}^{+}
		\\
		\vdots & \vdots & \vdots
		\\
		\kappa_{\rho,2}^{+} & \ldots & \kappa_{2 \rho,2}^{+}
	\end{array}
\right],
\\
\xi^2(P) = & {\be_1}' \bGam^{-1} \bDel \bGam^{-1} {\be_1} \frac{ v_{S|R}(0^+;P) }{f_R(0^+;P)}
\\
= & {\be_1}' \bG(P) \bDel \bG(P) {\be_1} v_{S|R}(0^+;P)  f_R(0^+;P). 
\end{align*}

We need to show that
$ \mmv_P[  \psi_n(\bV , {P} ) ] - \xi^2(P) = o_{\m{P}}(1).$
Note that,
\begin{align*}
& \mmv_P[  \psi_n(\bV , {P} ) ] = {\be_1}' \bG(P) 
\left\{
	\frac{1}{h} \mme_P\left[
		K^2\left( \frac{R}{h} \right) \bHt_{i}^{\rho} \bHt_{i}^{\rho'} v_{S|R}(R;P),
	\right]
\right\} 
\bG(P) {\be_1}.
\end{align*}
Thus, it suffices to show that
\begin{align*}
& \frac{1}{h} \mme_P\left[
		K^2\left( \frac{R}{h} \right) \bHt_{i}^{\rho} \bHt_{i}^{\rho'} v_{S|R}(R;P),
	\right]
	= \bDel v_{S|R}(0^+;P)  f_R(0^+;P) + o_{\m{P}}(1),
\end{align*}
or equivalently to show that
\begin{align*}
& \frac{1}{h} \mme_P\left[
		K^2\left( \frac{R}{h} \right) 
		\left( \frac{R}{h} \right)^{j} \left( \frac{R}{h} \right)^{l} v_{S|R}(R;P),
	\right]
	= \kappa_{j+l,2}^{+} v_{S|R}(0^+;P)  f_R(0^+;P) + o_{\m{P}}(1),
\end{align*}
for every $j,l=0,\ldots, \rho$.
In fact,
\begin{align*}
&\frac{1}{h} \mme_P\left[
		K^2\left( \frac{R}{h} \right) 
		\left( \frac{R}{h} \right)^{j} \left( \frac{R}{h} \right)^{l} v_{S|R}(R;P)
	\right]
-	\kappa_{j+l,2}^{+} v_{S|R}(0^+;P)  f_R(0^+;P)
\\
 =& \int_0^{\infty}  K^2(u) u^{j+l} \left\{ 
 	\underbrace{ v_{S | R}(uh;P) f_{R}(uh;P) }_{\doteq g(uh;P) }  - 
 	\underbrace{ v_{S | R}(0^+;P) f_{R}(0^+;P)}_{= g(0^+;P) } \right\}
 ~ du
 \\
=& h \int_0^{\infty}  K^2(u)  u^{j+l} \nabla_r g(x^*_{uh};P)  ~ du = o_{\m{P}}(1),
\end{align*}
where we use the fact that $g(r;P)$ is differentiable wrt $r$ with derivative $\nabla_r g(r;P)$ bounded
over $(r,P)$, which is implied by bounded derivatives of $v_{S|R}(r;P)$ and $f_R(r;P)$ wrt $r$.

\item \underline{Assumption \ref{aspt:asy_linear} - \eqref{eq:aspt:asy_linear_mom}:} $ \sup_P \mme [\psi_n^2(\bZ_{k,i},P)]  < \infty .$

We have,
\begin{align*}
& \psi_n^2(\bZ_k,P)  = 
\frac{1}{ h  f_R^2(0^+;P) }  
K^2 \left( \frac{ X_k  }{h } \right) 
\left( {\be_1}' \bGam^{-1} \bHt_{k}^{\rho} \right)^2 
\left(Y_k - m_{S|R}(X_k;P) \right),
\end{align*}
where $\bHt_{k}^{\rho} 
= \left[ 1 ~~ \left( \frac{X_k}{h} \right) ~~ \cdots ~~ \left( \frac{X_k  }{h } \right)^{\rho } \right]'$,
which is a $ (\rho+1 \times 1)$  column vector.
It follows that
\begin{align*}
& \psi_n^2(\bZ_k,P)  = 
\frac{1}{ h  f_R^2(0^+;P) }  
K^2 \left( \frac{ X_k  }{h } \right) 
\left( {\be_1}' \bGam^{-1} \bHt_{k}^{\rho} \right)^2 
\\*
& \hspace{3cm} \left[ 
	\left(Y_{k} -  m_{Y_k|X_k}(X_k) \right) 
	+ \left(m_{Y_k|X_k}(X_k)  - m_{S|R}(X_k;P) \right)  
\right]^2 .
\\
&\mme\left[ \psi_n^2(\bZ_k,P) \right] = \frac{1}{ h  f_R^2(0^+;P) } \mme\left\{ 
	K^2 \left( \frac{ X_k  }{h } \right) 
	\left( {\be_1}' \bGam^{-1} \bHt_{k}^{\rho} \right)^2 	
	\right.
\\
& \hspace{5.5cm} \left. \left[
	v_{Y_k|X_k}(X_k) +  \left(m_{Y_k|X_k}(X_k)  - m_{S|R}(X_k;P) \right)^2
\right] \right\}
\\	
&= \frac{1}{ f_R^2(0^+;P) } 
	\int_0^{\infty} K^2(u)
	\left( {\be_1}' \bGam^{-1} \bHt^{\rho}(u) \right)^2 	
\\
& \hspace{2.5cm} \left[
	v_{Y_k|X_k}(uh) +  \left(m_{Y_k|X_k}(uh)  - m_{S|R}(uh;P) \right)^2
\right] 
f_{X_k}(uh)
~du
\\	
&= O_{\m{P}}(1).
\end{align*}
where $\bHt^{\rho}(u) = \left[ 1 ~~ \left( u \right) ~~ \cdots ~~  u^{\rho } \right]'$,
which is a $ (\rho+1 \times 1)$  column vector.

\item \underline{Assumption \ref{aspt:asy_linear} - \eqref{eq:aspt:asy_linear_lind}:} 
$(2+\zeta)$-th moment condition. 

First, $\mmv_P^{-1}\left(  \psi_n(\bV_i,P) \right) = O_{\m{P}}(1). $

Second, there exists $\zeta>0$ such that 
$n^{-\zeta/2} \mme_P \left|  \psi_n(\bV_i,P) \right| ^{2+\zeta} = o_{\m{P}}(1)$
because
$\eta(r;P) = O_{\m{P}}(1)$ by assumption.

Therefore,  
\begin{align*}
n^{-\zeta/2} \sup\limits_{P \in \m{P} }  
\mme_P\left| \frac{ \psi_n(\bV_i,P) }{ \sqrt{\mmv_P\left(  \psi_n(\bV_i,P) \right)} } \right|^{2 + \zeta} = o(1).
\end{align*}

\item \underline{Assumption \ref{aspt:asy_linear} - \eqref{eq:aspt:asy_linear_xicont}:} 
$\xi^2\left( \frac{m}{n} P_1 + \frac{n-m}{n} P_2 \right) \to \xi^2 \left( \gamma P_1 + (1-\gamma) P_2  \right)$.

Let $\overline{P}_n = \frac{m}{n} P_1 + \frac{n-m}{n} P_2$ and $\overline{P} = \gamma  P_1 + (1-\gamma ) P_2$.
We have that 
$v_{ S  | R }(0^+; \bar{P}_n ) \to v_{ S  | R }(0^+; \bar{P})$
and
$f_{R}(0^+;\bar{P}_n) \to f_{R}(0^+;\bar{P})$
as in Section \ref{proof:control_mean}.
This implies that 
$\xi^2(\overline{P}_n) \to \xi^2(\overline{P})$.

\item \underline{Assumption \ref{aspt:sampling} :} $ (n_1 /n - \lambda) \pto 0$. 

In this setting, $n_1/n = \sum_i \mmi\{ X_i \geq 0 \} /n$ and $(n_1 /n - \lambda) \pto 0$ holds.

Assumption  \ref{aspt:asy_linear} has already been verified above for any sequence $m_n$ such that $(m/n -\gamma) \to 0$ for arbitrary $\gamma \in (0,1)$.
In particular it holds for $\gamma \in \{ \lambda, 1-\lambda\}.$

\end{enumerate}

$\square$

\section{Monte Carlo Simulations}
\label{sec:app:montecarlo}

\indent 

This section contains additional figures and tables regarding the Monte Carlo simulations in the main text (Section \ref{sec:mc}).

\begin{figure}[H]
\caption{Conditional Mean Functions}
\label{fig:condmean}
\begin{center}
\includegraphics[width=5in]{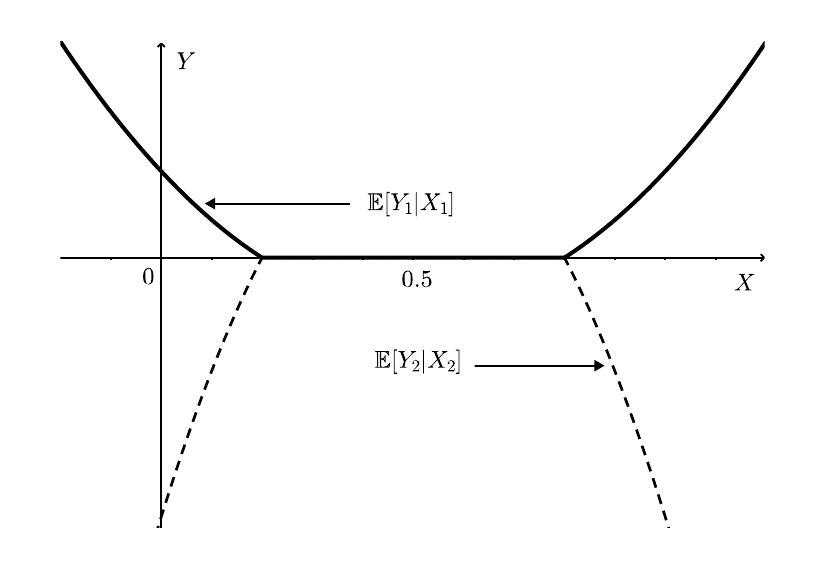}
\end{center}
\vspace{-.9cm}
\caption*{\footnotesize Notes: Conditional mean functions of Design 1.
The lines are equal and flat for  $|X-0.5| \leq  0.3$.
The conditional means  of Design 2 equal those of Design 1 shifted upwards by $1$.
}
\end{figure}

\begin{sidewaystable}[p]
\caption{Simulated Rejection Rates - 5\% Nominal Size}
\label{table:app:size}
\small
\begin{minipage}{18cm}
\centering

\begin{tabular}{c c c c|    c c c|     c c c|     c c c|     c c c  }  
\hline \hline 
\multicolumn{4}{c|}{\textbf{Design 1}}   &  \multicolumn{3}{c|}{$h=0.1$}  & \multicolumn{3}{c|}{$h=0.3$}   &  \multicolumn{3}{c|}{$h=0.5$} &  \multicolumn{3}{c}{$\ha{h}_{mse}$} 
\\ 
$\sigma_1^2$ & $\sigma_2^2$ & $n_1$ & $n_2$ & SP & SB & SS & SP & SB & SS & SP & SB & SS & SP & SB & SS 
\\ 
\hline 
 1 &  1 &  100 &  1900 &  0.049 &  0.100 &  0.137 &  0.046 &  0.061 &  0.091 &  0.147 &  0.140 &  0.244 &  0.045 &  0.067 &  0.090 
\\ 
 1 &  1 &  250 &  4750 &  0.047 &  0.067 &  0.062 &  0.048 &  0.054 &  0.062 &  0.321 &  0.277 &  0.376 &  0.046 &  0.057 &  0.056 
\\ 
 1 &  1 &  500 &  9500 &  0.048 &  0.057 &  0.043 &  0.048 &  0.053 &  0.053 &  0.569 &  0.495 &  0.606 &  0.048 &  0.055 &  0.048 
\\ 
\hline 
 1 &  1 &  40 &  1960 &  0.050 &  0.207 &  0.346 &  0.051 &  0.093 &  0.253 &  0.082 &  0.108 &  0.331 &  0.051 &  0.105 &  0.258 
\\ 
 1 &  1 &  100 &  4900 &  0.052 &  0.104 &  0.146 &  0.048 &  0.064 &  0.097 &  0.153 &  0.148 &  0.250 &  0.050 &  0.078 &  0.097 
\\ 
 1 &  1 &  200 &  9800 &  0.047 &  0.076 &  0.075 &  0.046 &  0.053 &  0.065 &  0.266 &  0.234 &  0.332 &  0.045 &  0.063 &  0.060 
\\ 
\hline 
 5 &  1 &  100 &  1900 &  0.068 &  0.104 &  0.146 &  0.056 &  0.063 &  0.094 &  0.073 &  0.073 &  0.137 &  0.056 &  0.065 &  0.095 
\\ 
 5 &  1 &  250 &  4750 &  0.059 &  0.071 &  0.056 &  0.053 &  0.053 &  0.060 &  0.104 &  0.097 &  0.134 &  0.053 &  0.052 &  0.056 
\\ 
 5 &  1 &  500 &  9500 &  0.054 &  0.058 &  0.037 &  0.051 &  0.053 &  0.051 &  0.164 &  0.147 &  0.183 &  0.052 &  0.050 &  0.046 
\\ 
\hline 
 5 &  1 &  40 &  1960 &  0.060 &  0.221 &  0.403 &  0.057 &  0.095 &  0.308 &  0.060 &  0.082 &  0.330 &  0.057 &  0.103 &  0.308 
\\ 
 5 &  1 &  100 &  4900 &  0.058 &  0.106 &  0.147 &  0.053 &  0.063 &  0.097 &  0.072 &  0.078 &  0.139 &  0.051 &  0.069 &  0.096 
\\ 
 5 &  1 &  200 &  9800 &  0.054 &  0.077 &  0.066 &  0.047 &  0.052 &  0.060 &  0.089 &  0.087 &  0.127 &  0.047 &  0.055 &  0.056 
\\ 
\hline \hline 
\end{tabular}

\end{minipage}

\vspace{0.25cm}

\begin{minipage}{18cm}
\centering

\begin{tabular}{c c c c|    c c c|     c c c|     c c c|     c c c  }  
\hline \hline 
\multicolumn{4}{c|}{\textbf{Design 2}}   &  \multicolumn{3}{c|}{$h=0.1$}  & \multicolumn{3}{c|}{$h=0.3$}   &  \multicolumn{3}{c|}{$h=0.5$} &  \multicolumn{3}{c}{$\ha{h}_{mse}$} 
\\ 
$\sigma_1^2$ & $\sigma_2^2$ & $\mu$ & $n_1=n_2$ & SP & SB & SS & SP & SB & SS & SP & SB & SS & SP & SB & SS 
\\ 
\hline 
 1 &  1 &  10 &  75 &  0.051 &  0.057 &  0.189 &  0.049 &  0.053 &  0.138 &  0.082 &  0.086 &  0.206 &  0.039 &  0.044 &  0.128 
\\ 
 1 &  1 &  10 &  150 &  0.051 &  0.048 &  0.096 &  0.050 &  0.053 &  0.088 &  0.132 &  0.131 &  0.205 &  0.045 &  0.050 &  0.084 
\\ 
 1 &  1 &  10 &  1000 &  0.049 &  0.052 &  0.043 &  0.049 &  0.051 &  0.057 &  0.606 &  0.560 &  0.636 &  0.047 &  0.049 &  0.055 
\\ 
\hline 
 1 &  1 &  1 &  75 &  0.053 &  0.483 &  0.535 &  0.050 &  0.366 &  0.451 &  0.070 &  0.363 &  0.457 &  0.045 &  0.336 &  0.429 
\\ 
 1 &  1 &  1 &  150 &  0.051 &  0.414 &  0.448 &  0.050 &  0.279 &  0.372 &  0.091 &  0.321 &  0.421 &  0.050 &  0.266 &  0.360 
\\ 
 1 &  1 &  1 &  1000 &  0.045 &  0.217 &  0.323 &  0.049 &  0.179 &  0.257 &  0.339 &  0.555 &  0.645 &  0.048 &  0.184 &  0.273 
\\ 
\hline 
 5 &  1 &  10 &  75 &  0.060 &  0.079 &  0.196 &  0.051 &  0.047 &  0.135 &  0.061 &  0.058 &  0.159 &  0.044 &  0.040 &  0.129 
\\ 
 5 &  1 &  10 &  150 &  0.058 &  0.049 &  0.098 &  0.054 &  0.051 &  0.084 &  0.076 &  0.072 &  0.123 &  0.045 &  0.044 &  0.076 
\\ 
 5 &  1 &  10 &  1000 &  0.048 &  0.048 &  0.037 &  0.048 &  0.047 &  0.054 &  0.254 &  0.222 &  0.277 &  0.049 &  0.049 &  0.053 
\\ 
\hline 
 5 &  1 &  1 &  75 &  0.066 &  0.432 &  0.467 &  0.058 &  0.274 &  0.347 &  0.070 &  0.241 &  0.331 &  0.054 &  0.228 &  0.309 
\\ 
 5 &  1 &  1 &  150 &  0.060 &  0.330 &  0.357 &  0.053 &  0.189 &  0.255 &  0.080 &  0.184 &  0.257 &  0.051 &  0.157 &  0.220 
\\ 
 5 &  1 &  1 &  1000 &  0.049 &  0.129 &  0.189 &  0.048 &  0.094 &  0.146 &  0.208 &  0.264 &  0.330 &  0.048 &  0.096 &  0.145 
\\ 
\hline \hline 
\end{tabular} 

\end{minipage} 
 \caption*{\footnotesize
Notes:  The table displays simulated rejection rates under the null hypothesis for the  studentized permutation (SP),
studentized bootstrap (SB), and studentized subsample (SS) tests.
Rows correspond to variations of Designs 1 and 2, as explained in the text.
The bandwidth choice $h$ equals one of three fixed values (0.1, 0.3, and 0.5) or the IK data-driven MSE-optimal bandwidth $\ha{h}_{mse}.$
}
\end{sidewaystable}

\begin{sidewaysfigure}[p]
    \centering
    \begin{minipage}{0.49\textwidth}
        \centering
        \caption{\label{figure:app:power_d1}Design 1 - Simulated Power Curves}
\begin{center}
  \begin{minipage}{.5 \textwidth}
    \centering
    (a) {\small $\sigma_1^2=\sigma_2^2 = 1$ \\ $n_1 / (n_1+n_2) = 0.05$ \\ SP(---) and $t$(- -)}

    \includegraphics[height=1.8in]{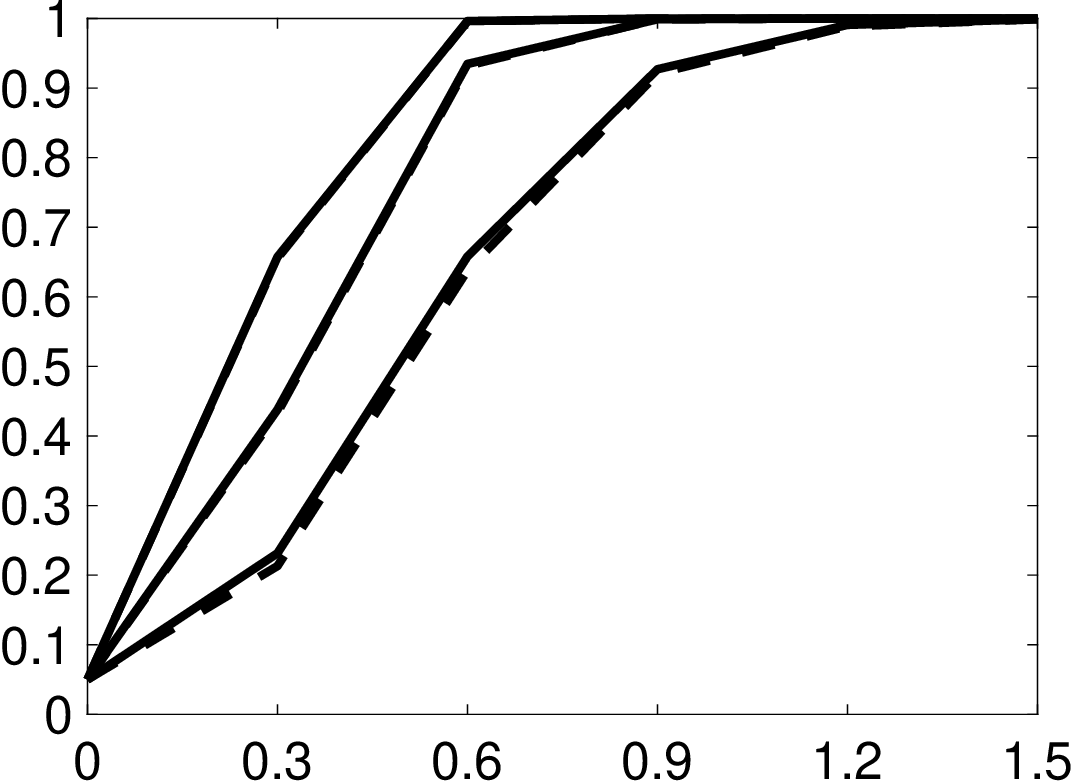}
  \end{minipage}%
  \begin{minipage}{.5 \textwidth}
    \centering
    (b) {\small $\sigma_1^2=\sigma_2^2 = 1$ \\ $n_1 / (n_1+n_2) = 0.02$ \\ SP(---) and $t$(- -)}

    \includegraphics[height=1.8in]{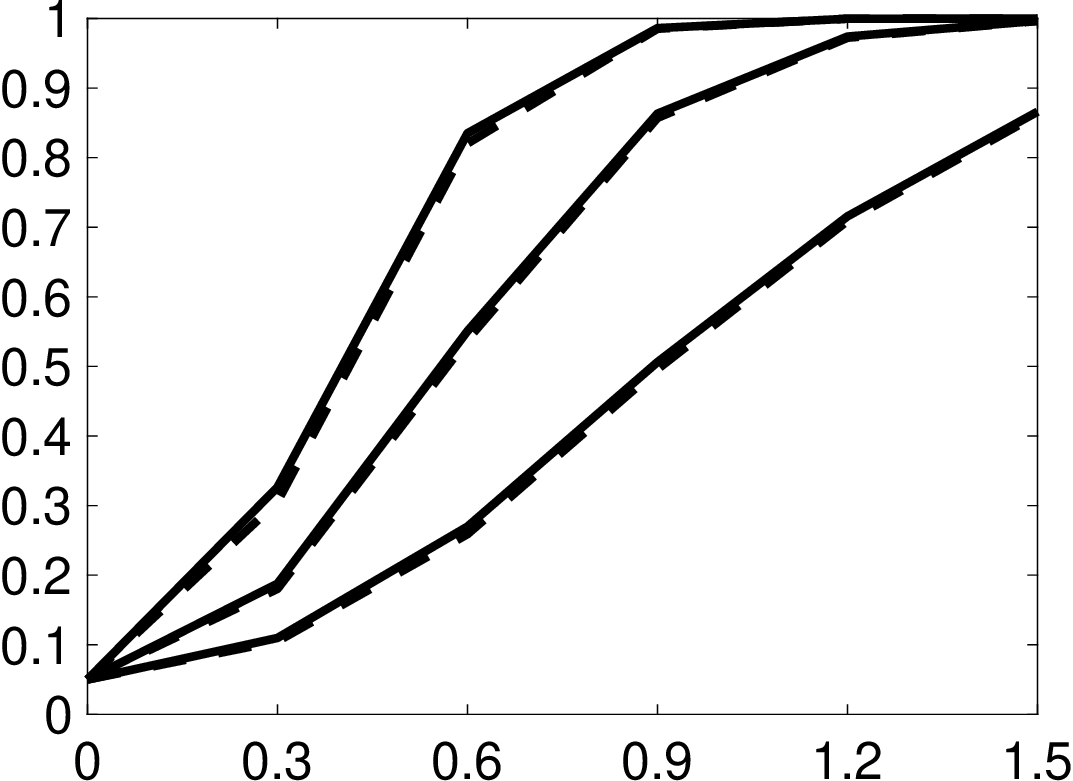}
  \end{minipage}%
\end{center}

\begin{center}
  \begin{minipage}{.5 \textwidth}
    \centering
    (c) {\small $\sigma_1^2= \sigma_2^2 = 1$ \\ $n_1 / (n_1+n_2) = 0.05$ \\ SP(---), SB(- -), and SS($\cdot\hspace{-.8mm}\cdot\hspace{-.8mm}\cdot$)}

    \includegraphics[height=1.8in]{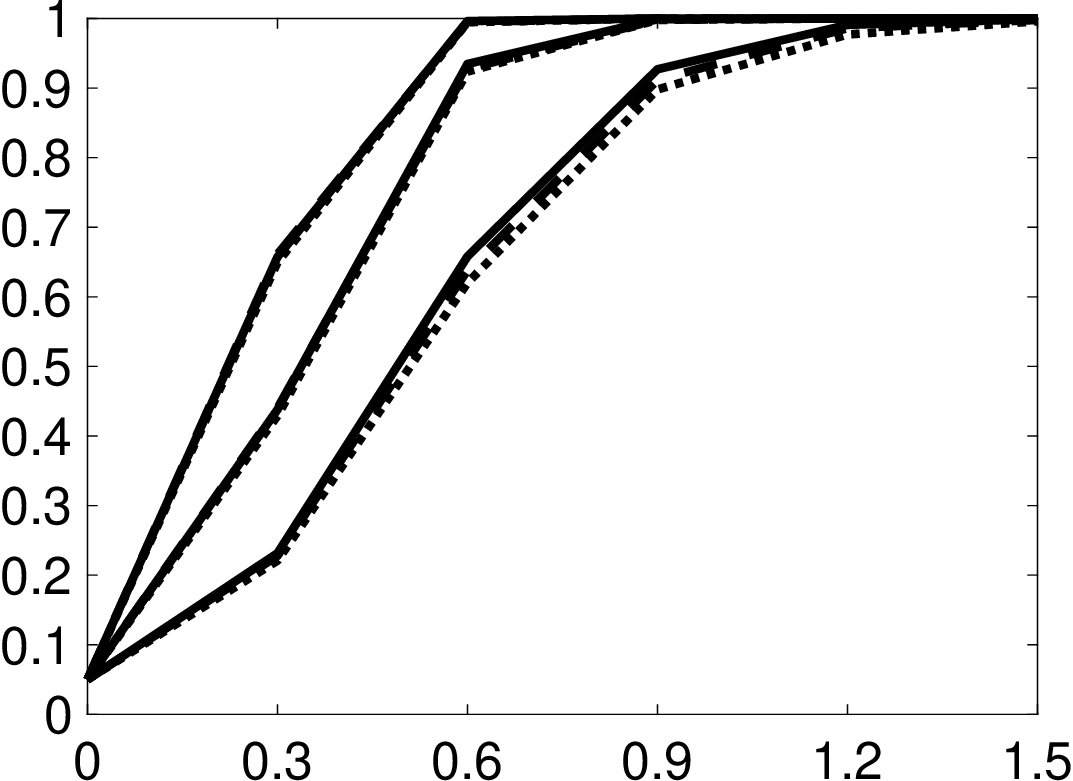}
  \end{minipage}%
  \begin{minipage}{.5 \textwidth}
    \centering
    (d) {\small $\sigma_1^2= \sigma_2^2 = 1$ \\ $n_1 / (n_1+n_2) = 0.02$ \\ SP(---), SB(- -), and SS($\cdot\hspace{-.8mm}\cdot\hspace{-.8mm}\cdot$)}

    \includegraphics[height=1.8in]{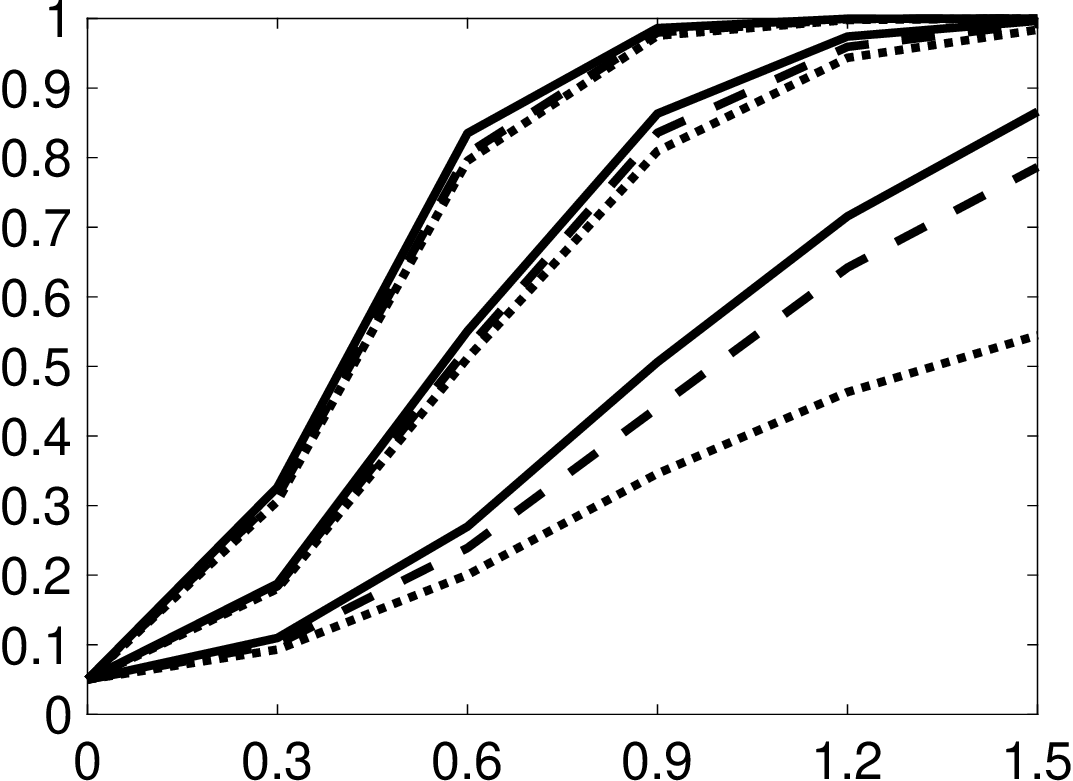}
  \end{minipage}%
\end{center}
\caption*{\footnotesize
Notes: Panels (a)--(b) compare power of the studentized permutation test (SP) with the $t$-test ($t$).
Panels (c)--(d) compare SP with the studentized bootstrap (SB) and studentized subsample (SS).
Darker lines correspond to bigger sample sizes, i.e., 
Panels (a) and (c): black $(n_1,n_2) = (500,9500)$,  dark gray $(n_1,n_2) = (250,4750)$, and light gray $(n_1,n_2) = (100,1900)$; 
Panels (b) and (d): black $(n_1,n_2) = (200,9800)$,  dark gray $(n_1,n_2) = (100,4900)$, and light gray $(n_1,n_2) = (40,1960)$.
The x-axis represents the difference $\theta(P_1)-\theta(P_2)$ and the y-axis, the simulated probability of rejection.
Sizes of all tests 
are artificially adjusted such that the simulated rejection under the null is always equal to 5\%.
All estimates use the IK MSE-optimal bandwidth $\ha{h}_{mse}$.
}
    \end{minipage} \hfill
    \begin{minipage}{0.49\textwidth}
        \centering
        \caption{\label{figure:app:power_d2}Design 2 - Simulated Power Curves}
\begin{center}
  \begin{minipage}{.5 \textwidth}
    \centering
    (a) {\small $\sigma_1^2= \sigma_2^2 = 1$ \\ $\mu=10$ \\ SP(---) and $t$(- -)}

    \includegraphics[height=1.8in]{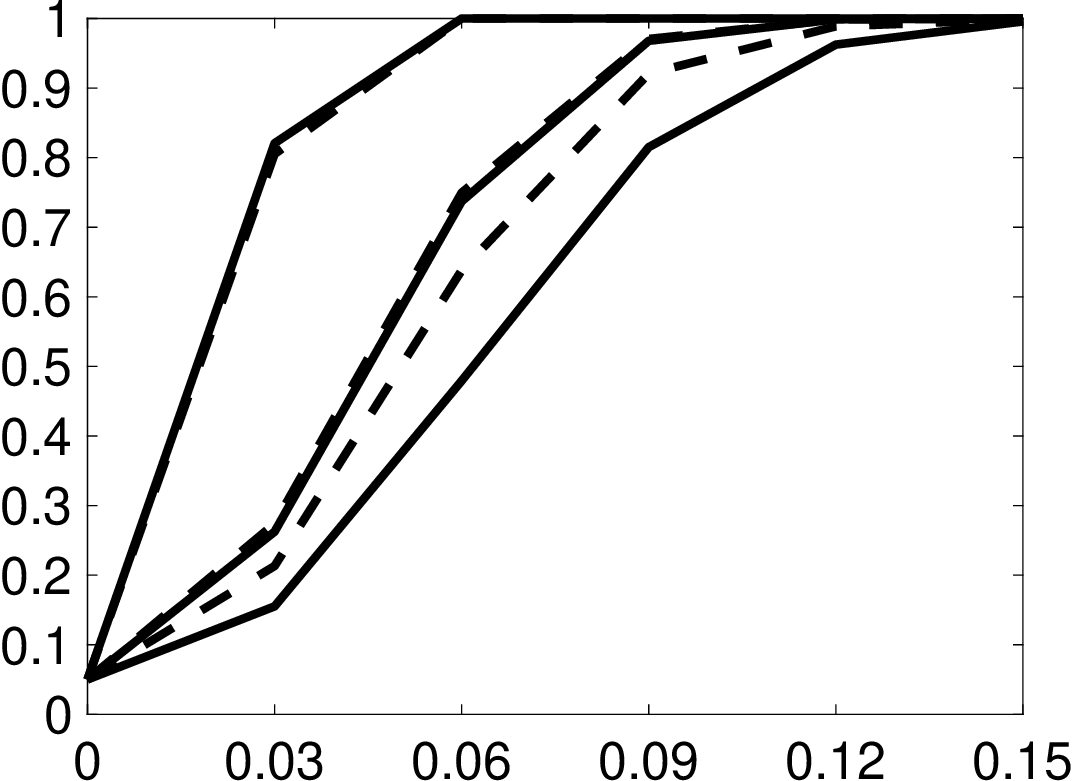}
  \end{minipage}%
  \begin{minipage}{.5 \textwidth}
    \centering
    (b) {\small $\sigma_1^2=\sigma_2^2 = 1$ \\ $\mu=1$ \\ SP(---) and $t$(- -)}

    \includegraphics[height=1.8in]{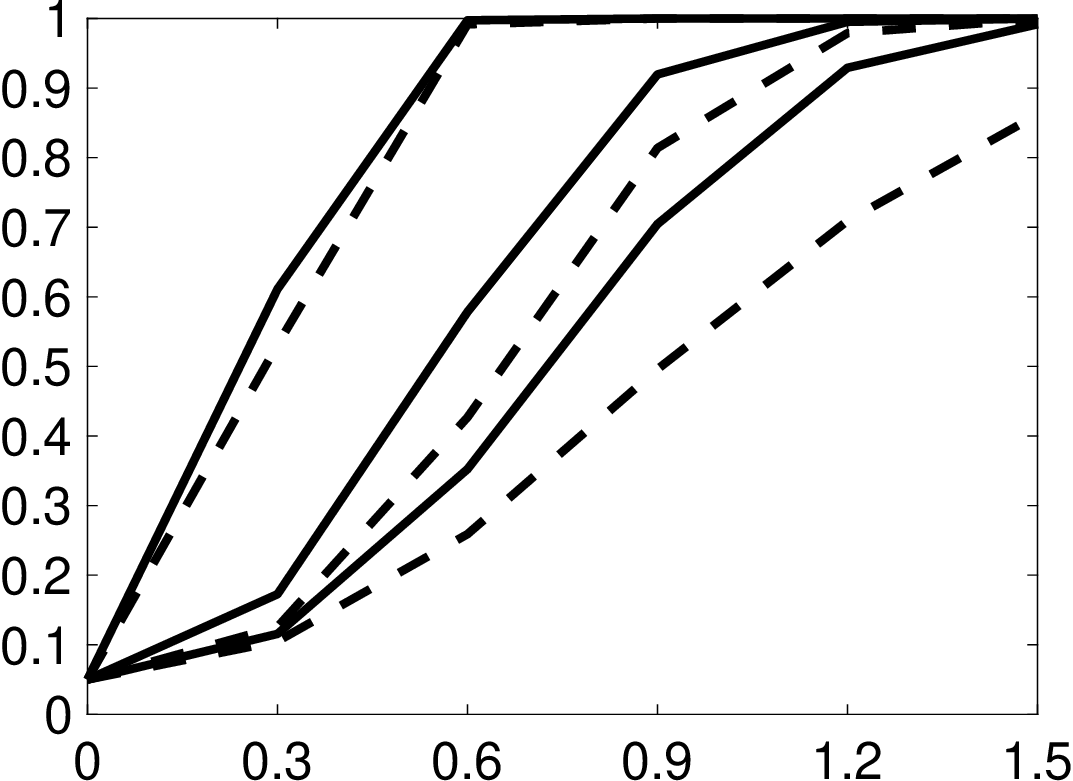}
  \end{minipage}%
\end{center}

\begin{center}
  \begin{minipage}{.5 \textwidth}
    \centering
    (c) {\small $\sigma_1^2= \sigma_2^2 = 1$ \\ $\mu=10$ \\ SP(---), SB(- -), and SS($\cdot\hspace{-.8mm}\cdot\hspace{-.8mm}\cdot$)}

    \includegraphics[height=1.8in]{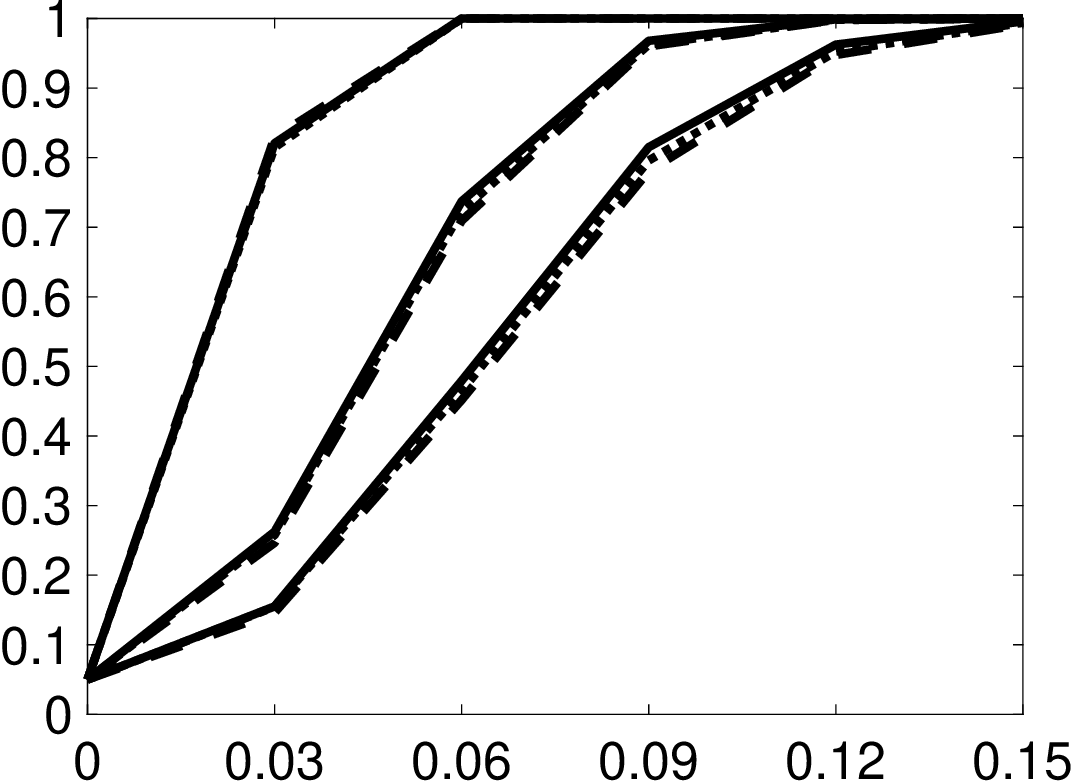}
  \end{minipage}%
  \begin{minipage}{.5 \textwidth}
    \centering
    (d) {\small $\sigma_1^2= \sigma_2^2 = 1$\\ $\mu=1$ \\ SP(---), SB(- -), and SS($\cdot\hspace{-.8mm}\cdot\hspace{-.8mm}\cdot$)}

    \includegraphics[height=1.8in]{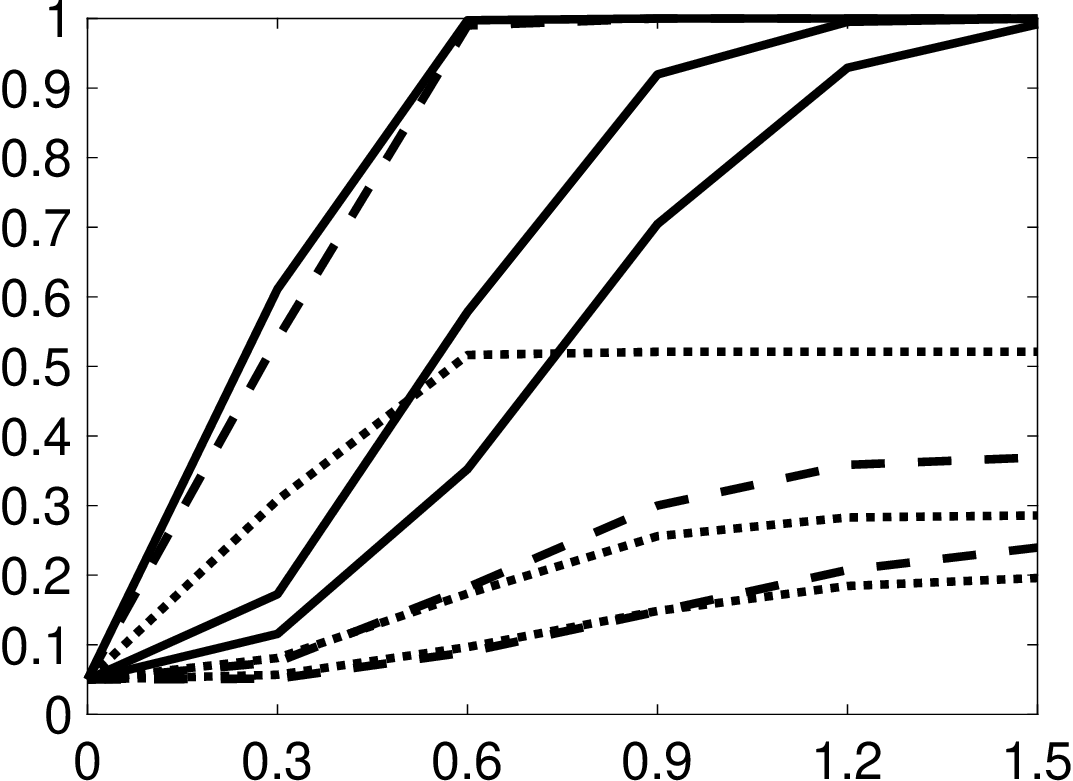}
  \end{minipage}%
\end{center}
\caption*{\footnotesize
Notes: Panels (a)--(b) compare power of the studentized permutation test (SP) with the $t$-test ($t$).
Panels (c)--(d) compare SP with the studentized bootstrap (SB) and studentized subsample (SS).
Darker lines correspond to bigger sample sizes, i.e., 
black $(n_1,n_2) = (1000,1000)$,  dark gray $(n_1,n_2) = (150,150)$, and light gray $(n_1,n_2) = (75,75)$.
The x-axis represents the difference $\theta(P_1)-\theta(P_2)$ and the y-axis, the simulated probability of rejection.
Sizes of all tests 
are artificially adjusted such that the simulated rejection under the null is always equal to 5\%.
All estimates use the IK MSE-optimal bandwidth $\ha{h}_{mse}$.
\\
\vspace{3mm}
}
    \end{minipage}
\end{sidewaysfigure}

\newpage

\section{Auxiliary Lemmas}
\label{sec:app:c}

\begin{definition}
\label{def:Oop}
Consider a sequence of measurable functions $X_n : \Omega \times \m{P} \to \mmr^q$ 
for a probability space $(\Omega, \m{B}, P)$, where $P$ belongs to a set of distributions $\m{P}$.
In other words, for any fixed $P_0 \in \m{P}$, $X_n(P_0)$ is a random variable defined on $(\Omega, \m{B}, P)$ to the Euclidean space, whose probability distribution depends on  $n$ and $P \in \m{P}$.
We say $X_n$ is uniformly bounded in probability over $\m{P}$ by a deterministic sequence $\alpha_n$ if, for every $\delta>0$, there exists a deterministic constant $M_{\delta} \in (0,\infty)$ such that
\[
\sup\limits_{P \in \m{P} } \mmp_P \left[ \left\| X_n(P) \right\| > M_{\delta} \alpha_n \right] < \delta .
\]
We denote this as $X_n = O_{\m{P}}( \alpha_n )$.
Similarly, we say $X_n = o_{\m{P}}( \alpha_n )$ if, for every $\eps,\delta>0$  $\exists n_{\eps,\delta}$ such that 
\[
\sup\limits_{P \in \m{P} } \mmp_P \left[ \left\| X_n(P) \right\| > \eps \alpha_n \right] < \delta ~~~ \forall n \geq n_{\eps,\delta}.
\]
 \end{definition}

\begin{lemma} 
\label{lemma:Oop}
Consider a deterministic sequence $\alpha_n$ and a sequence of bivariate random functions $(X_n,Y_n): \Omega \times \m{P} \to \mmr^2$ 
as in Definition \ref{def:Oop}.
\begin{enumerate}

\item If $X_n = O_{\m{P}}(\alpha_n)$  and $Y_n = o_{\m{P}}(1)$, then $X_n Y_n = o_{\m{P}}(\alpha_n)$.

\item  Assume $X_n(P)$ has expectation $\mme_P [X_n(P)]$ and variance $\mmv_P [X_n(P)]$ that are bounded over  $\m{P}$.
Then,
\[
X_n = O_{\m{P}}\left( \sup\limits_{P \in \m{P} }  \mme^2_P [X_n(P) ] +  \sup\limits_{P \in \m{P} }  \mmv_P [ X_n(P)]  \right).
\]

\item Suppose  there exists a deterministic function $y:\m{P} \to [\underline{y},\overline{y}]$ with $0<\underline{y} < \overline{y} <\infty$ such that $Y_n - y = o_{\m{P}}(1) $.
Then,
\[
\frac{1}{Y_n} - \frac{1}{y} = o_{\m{P}}(1).
\]

\end{enumerate} 
\end{lemma}
\begin{proof}
\textbf{Part 1:} 

Fix $\delta>0$ and $\eps>0$. There  exists $M_{\delta} \in (0,\infty)$ such that
\[
\sup\limits_{P \in \m{P} } \mmp_P \left[ \left| X_n(P) \right| > M_{\delta} \alpha_n \right] < \frac{\delta}{2} .
\]
There exists $\exists n_{\eps,\delta}$ such that
\[
\sup\limits_{P \in \m{P} } \mmp_P \left[ \left| Y_n(P) \right| > \frac{\eps}{M_{\delta}} \right] < \frac{\delta}{2} ~~~ \forall n \geq n_{\eps,\delta}.
\]
Given that $\left| X_n(P) Y_n(P) \right| > \eps \alpha_n$ implies $\left| X_n(P) \right| > M_{\delta} \alpha_n$ or $\left|Y_n(P) \right| > \eps / M_{\delta}$ we have
\begin{align*}
&\sup\limits_{P \in \m{P} } \mmp_P \left[ \left| X_n(P) Y_n(P) \right| > \eps \alpha_n \right] 
\\
&\leq
\sup\limits_{P \in \m{P} } \mmp_P \left[ \left| X_n(P) \right| > M_{\delta} \alpha_n	\right] 
+
\sup\limits_{P \in \m{P} } \mmp_P \left[	\left|Y_n(P) \right| > \eps / M_{\delta} \right] 
< \delta ~~~ \forall n \geq n_{\eps,\delta}.
\end{align*}

\bigskip

\textbf{Part 2:}

Call 
$A_n = \sup\limits_{P \in \m{P} } \mme^2_P [X_n(P)]  +  \sup\limits_{P \in \m{P} }  \mmv_P [ X_n(P)] $
and 
$B_n= \sup\limits_{P \in \m{P} }\mme_P\left( X_n^2(P) \right)$.
Note that
$A_n \geq B_n \geq   \mme_P\left( X_n^2(P) \right)$ for every $P \in \m{P}$.

For $\delta>0$ and $P\in \m{P}$,
\begin{align*}
&\mmp_P\left[ |X_n(P)| > \delta^{-1/2} A_n \right]
\\
&\leq
\mmp_P\left[ |X_n(P)| >  \delta^{-1/2} B_n  \right] 
\leq 
\mmp_P\left[ |X_n(P)| >  \delta^{-1/2} \mme_P\left( X_n^2(P) \right) \right] \leq \delta,
\end{align*}
where the last inequality is the Markov  inequality.
The result follows by taking the supremum of both sides.

 \bigskip
 
 \textbf{Part 3:}
 
Pick $\eta>0$ such that $\underline{y}-\eta>0$.
Define the set $A = \{y:  \underline{y}-\eta \leq y \leq \overline{y} +\eta \}$.
The function $g(y)=1/y$ is uniformly continuous over $A$.
Thus, for every $\eps>0$, there exists $\gamma_{\eps}>0$ such that $|g(y') - g(y)|>\eps \Rightarrow |y'-y|>\gamma_{\eps}$ for any $y',y \in A$. 
This implies that, for fixed $P$,
\begin{gather*}
\mmp_P\left[ \left| 1/Y_n(P) - 1/y(P) \right| > \eps, Y_n(P) \in A \right] \leq  \mmp_P\left[ \left| Y_n(P) - y(P) \right| > \gamma_\eps \right] \to 0.
\end{gather*}
This right-hand side is uniform over $\m{P}$ because $Y_n - y = o_{\m{P}}(1)$.
 
 Next, for fixed $P$, $|Y_n(P)-y(P)| < \eta \Rightarrow \underline{y} -\eta < Y_n(P)  < \overline{y} +\eta  \Rightarrow Y_n(P) \in A$.
 This implies that $\mmp_P [|Y_n(P)-y(P)| < \eta ] \leq \mmp_P[Y_n(P) \in A]  $
 and that
 $\mmp_P[Y_n(P) \in A ] \to 1 $ uniformly over $\m{P}$ because $Y_n - y = o_{\m{P}}(1)$ implies that $\mmp_P [|Y_n(P)-y(P)| < \eta ] \to 1$ uniformly over $\m{P}$.
 Then,
 \begin{align*}
 &\mmp_P\left[ \left| 1/Y_n(P) - 1/y(P) \right| > \eps \right] 
 -
 \mmp_P\left[ \left| 1/Y_n(P) - 1/y(P) \right| > \eps, Y_n(P) \in A \right]  
 \\
 &\leq 
 1 - \mmp_P\left[  Y_n(P) \in A \right]  
 \to 0
 \end{align*}
 uniformly over $\m{P}$.
 
\end{proof}

\begin{lemma}\label{lemma:lr:11.61}
Consider a sequence of random variables $X_n$ in the Euclidean space.
$X_n \pto X$ if, and only if, for every subsequence $X_{n_{k}}$ there exists a further subsequence $X_{n_{k_{j}}}$ such that $X_{n_{k_{j}}} \asto X $.
\end{lemma}
\begin{proof}  

See the proof of Theorem 2.3.2 by \citeSM{durrett2019}.

\end{proof}

\begin{lemma}\label{lemma:cond:clt}
Consider a sequence of random variables $(Z_n,X_n)$, $n=1,2, \ldots$, and a random variable $Z$, all with domain on the measure space $(\Omega, \m{B}, \mu)$.
The image of $Z_n$ and $Z$ are in a Euclidean space, and the image of $X_n $ is $\m{X}_n$.
For a measurable function $F_n : \m{X}_n \to \mmr$, assume $F_n(X_n) \pto 0$. 
Moreover, suppose that for any nonrandom sequence $x_n \in \m{X}_n$ such that $F_n(x_n) \to 0$,   $Z_n$ conditional on $X_n=x_n$ converges in distribution to  $Z$.
Then, $Z_n \dto Z$ unconditionally. 

\end{lemma}

\begin{proof}

For an arbitrary subsequence $n_k$ of the sequence $\left\{ F_n(X_n )\right\}_n$,
 Lemma \ref{lemma:lr:11.61}
says there is a further subsequence $n_{k_j}$ such that 
$ F_{n_{k_j}}(X_{n_{k_j}} )  \asto 0$ as $j\to\infty$.
Define $\mmx^{\infty}$ to be the space of all subsequences of the form 
$\left\{  x_{n_{k_{j}}}  \right\}_{j=1}^{\infty}$, for values $x_{n_{k_j}} \in \m{X}_{n_{k_j}}$ that satisfy
$F_{n_{k_j}}(x_{n_{k_j}} )  \to 0$  as $j\to\infty$.
We know that 
$\mmp\left[ \{ X_{n_{k_j}} \}^{\infty}_{j=1}  \in \mmx^\infty \right]=1$.

Let $G_{n}$ be the CDF of the distribution of 
$ Z_{n} $
conditional on ${X}_{n} = {x}_{n}$, where $x_n$ is an arbitrary sequence that satisfies $F_n(x_n) \to 0$.
By assumption, $G_n \to G$ pointwise (for every continuity point of $G$), where $G$ is the CDF of $Z$.
This implies that $G_{n_{k_j}} \to G$ for the subsequence $\{n_{k_j}\}_j$ from above.

Next, 
\begin{align*}
\mmp\left[ Z_{n_{k_j}} \leq z \right] 
=&
\int_{\Omega} \mmp\left[\left. Z_{n_{k_j}} \leq z \right| {X}_{n_{k_j}} = {X}_{n_{k_j}}(\omega)  \right]  ~ d \mu(\omega)
\\
=&
\int_{\Omega} G_{n_{k_j}}(z;\omega) 
\mmi\left\{ \omega : \{ X_{n_{k_j}} (\omega) \}^{\infty}_{j=1}  \in \mmx^\infty \right\} ~ d \mu(\omega),
\end{align*}
where $G_{n_{k_j}}(z;\omega)$ denotes the CDF of $ Z_{n_{k_j}} $
conditional on ${X}_{n_{k_j}} = {X}_{n_{k_j}}(\omega)$,
and we used that  $\mmp\left[\omega :  \{ X_{n_{k_j}} (\omega) \}^{\infty}_{j=1}  \in \mmx^\infty \right]=1$.

For each $\omega$ such that the indicator is $1$, the conditional CDF $G_{n_{k_j}}(z;\omega)$ converges to $G(z)$ as $j \to \infty$.
 So, we have an integral of a measurable function of $\omega$ that changes with $j$, is bounded by $1$, and converges pointwise in $\omega$ 
 to $G(z) \mmi\left\{ \omega : \{ X_{n_{k_j}} (\omega) \}^{\infty}_{j=1}  \in \mmx^\infty \right\} $ as $j \to \infty$.
 By the dominated convergence theorem, the integral above converges to 
 $G(z)  \mmp\left[ \{ X_{n_{k_j}} \}^{\infty}_{j=1}  \in \mmx^\infty \right] =G(z).$
 This says that $Z_{n_{k_j}} \dto Z$.

Finally, call $H_{n} $  the unconditional CDF of  $Z_{n}$.
We have just shown that, for every subsequence $\{ n_k \}_k$ there exists a further subsequence $\{ n_{k_j} \}_j$ such that 
$H_{n_{k_j}}(z) \to G(z) $.
This implies that  $H_{n}(z) \to G(z) $.
Therefore, $Z_n \dto Z$.

\end{proof}

The next lemma is a generalization of  Lemma 11.3.3 of \citeSM{lehmann2005}.

\begin{lemma}\label{lemma:clt} 
For each $n$, let $Y_{n,1}, \ldots, Y_{n,n}$ be independently identically distributed with mean zero and finite variance $\sigma^2_n$.
Let $C_{n,1}, \ldots, C_{n,n}$ be a sequence of random variables, independent of $Y_{n,1}, \ldots, Y_{n,n}$.
Assume $\exists \zeta>0$ such that 
\begin{align}
& \mme\left[ | Y_{n,i}/\sigma_n |^{2 + \zeta} \right]  \left( \frac{\max_{i=1}^n C_{n,i}^2 }{\sum_{l=1}^n C_{n,l}^2} \right)^{\zeta/2} \pto 0 \text{ as }n \to \infty.
\end{align}
Then
\begin{gather}
\frac{\sum_{i=1}^n C_{n,i} Y_{i,n} }{ \sigma_n \sqrt{ \sum_{l=1}^n C_{n,l}^2 } }\dto N(0,1).
\end{gather}

\end{lemma}

\begin{proof} 

We use Lemma \ref{lemma:cond:clt} to prove this lemma.
Call $\alpha_n = \mme\left[ | Y_{n,i}/\sigma_n |^{2 + \zeta} \right]$.
Mapping the lemma's notation to our case, we have
\begin{gather*}
Z_n = \frac{ \sum_{i=1}^{n} C_{n,i} Y_{n,i} } {\sigma_n \sqrt{ \sum_{l=1}^{n } C_{n,l}^2 } },
\\
X_n =\left(C_{n,1}, \ldots, C_{n,n} \right) ,
\\
\m{X}_n = \mmr^n,
\\
F_n(X_n) =  \alpha_n \left( \frac{\max_{1 \leq i \leq n} C_{n,i}^2 }{ \sum_{j=1}^n C_{n,j}^2} \right)^{\zeta/2}.
\end{gather*}
Then, it suffices to derive the limiting distribution of the sequence of distributions of \\
$\sum_{i=1}^{n} C_{n,i} Y_{n,i} / \sqrt{ \sum_{l=1}^{n } C_{n,l}^2 }$
conditional on $\left(C_{n,1}, \ldots, C_{n,n} \right)  = \left(c_{n,1}, \ldots, c_{n,n} \right) $, where 
the values $\left(c_{n,1}, \ldots, c_{n,n} \right)$ come from an arbitrary  triangular array with infinitely many rows that satisfies
$\alpha_n  \left( {\max_{i=1}^n c_{n,i}^2 }/{ \sum_{j=1}^n c_{n,j}^2} \right)^{\zeta/2}$ $\to 0 $ as $n \to \infty$.

For each $n$, 
we have a sum of a triangular array of random variables $C_{n,i}  Y_{n, i}$ that are independent  across $i=1, \ldots, n$ 
once we condition on $\left(C_{n,1}, \ldots, C_{n,n} \right)$.
By assumption,   $C_{n,1}, \ldots, C_{n,n}$ is  independent of $Y_{n,1}, \ldots, Y_{n,n}$ for every $n$.
Thus, 
\begin{gather}
\mme \left[ C_{n,i} Y_{n, i}   ~|~ \left(C_{n,1}, \ldots, C_{n,n} \right)  = \left(c_{n,1}, \ldots, c_{n,n} \right) \right] = 0,
\\ 
\mmv\left[  C_{n,i} Y_{n, i} ~|~ \left(C_{n,1}, \ldots, C_{n,n} \right)  = \left(c_{n,1}, \ldots, c_{n,n} \right) \right] =  c_{n,i}^2 \sigma^2_{n}
\end{gather}
for $i=1, \ldots, n$.
The sum of the variances is  $s_{n}^2 = \sigma^2_{n} \sum_{i=1}^{n} c_{n,i}^2 $.

For $n=1,2, \ldots$, the sequence of distributions of 
$ \sum_{i=1}^{n} C_{n ,i} Y_{n , i}  / \sigma_{n }  \sqrt{ \sum_{l=1}^{n} C_{n,l}^2 }$
conditional on $\left(C_{n,1}, \ldots, C_{n,n} \right)  = \left(c_{n,1}, \ldots, c_{n,n} \right) $
is the same as the sequence of distributions of

\noindent $Z_{n} = \sum_{i=1}^{n} c_{n,i} Y_{n, i}/  \sigma_{n}  \sqrt{ \sum_{l=1}^{n} c_{n,l}^2 }$.

We apply the Lindeberg CLT to derive the limiting distribution of $Z_{n}$.
We need to verify the Lindeberg condition, that is, for any $\delta>0$
\begin{gather}
\sum_{i=1}^{n } \frac{1}{s_{n }^2} \mme\left[c_{{n },i}^2 Y_{n, i}^2 \mmi\{| c_{{n},i} Y_{n, i} | > \delta s_{n} \} \right] \to 0 \text{ as } n \to \infty
\\
\Leftrightarrow \notag
\\
\sum_{i=1}^{n } \frac{ c_{{n},i}^2 }{\sigma_{n }^2 \sum_{l=1}^{n} c_{{n},l}^2} \mme\left[ Y_{n, i}^2  
\mmi\left\{ c_{{n},i}^2 Y_{n, i}^2  > \delta^2 \sigma_{n}^2  \sum_{l=1}^{n} c_{{n},l}^2 \right\} \right] \to 0
\\
\Leftarrow \notag
\\
\sum_{i=1}^{n} \frac{ c_{{n},i}^2 }{ \sum_{l=1}^{n} c_{{n},l}^2} \mme\left[ \frac{ Y_{n, i}^2 }{\sigma_{n}^2}  
\mmi\left\{  \frac{ Y_{n, i}^2 }{\sigma_{n}^2}  > \delta^2 \frac{  \sum_{l=1}^{n} c_{{n},l}^2}{ \max\limits_{1 \leq i \leq n} c_{{n},i}^2}  \right\} \right] \to 0
\\
\Leftrightarrow \notag
\\
\mme\left[ \frac{ Y_{n, i}^2 }{\sigma_{n}^2}  
\mmi\left\{  \frac{ Y_{n, i}^2 }{\sigma_{n}^2}   > \delta^2 \frac{   \sum_{l=1}^{n } c_{{n },l}^2}{ \max\limits_{1 \leq i \leq n } c_{{n },i}^2}  \right\} \right] \to 0,
\label{eq:lemma:clt:lind}
\end{gather}
where we used that
$\mmi\left\{ c_{{n },i}^2 Y_{n , i}^2  > \delta^2 \sigma_{n }^2  \sum_{l=1}^{n } c_{{n },l}^2 \right\} 
 \leq 
 \mmi\left\{  \frac{ Y_{n , i}^2 }{\sigma_{n }^2}  > \delta^2 \frac{ \sum_{l=1}^{n } c_{{n },l}^2}{ \max\limits_{1 \leq i \leq n } c_{{n },i}^2}  \right\} $
 and that

\noindent
 $\mme\left[ \frac{ Y_{n , i}^2 }{\sigma_{n }^2}
\mmi\left\{  \frac{ Y_{n , i}^2 }{\sigma_{n }^2}  > \eps^2 \frac{  \sum_{l=1}^{n } c_{{n },l}^2}{ \max\limits_{1 \leq i \leq n } c_{{n },i}^2}  \right\} \right]$ does not depend on $i$.
Thus, it suffices to verify Equation \ref{eq:lemma:clt:lind}. 

Note that,
\begin{align*} 
\left| \frac{ Y_{n, i} }{\sigma_{n} } \right|^{2+\zeta} 
	& \geq 
	\left| \frac{ Y_{n, i} }{\sigma_{n} } \right|^{2+\zeta}
	\mmi\left\{  \frac{ Y_{n, i}^2 }{\sigma_{n}^2}   > \delta^2 \frac{   \sum_{l=1}^{n } c_{{n },l}^2}{ \max\limits_{1 \leq i \leq n } c_{{n },i}^2}  \right\} 
\\
	& = 
	 \frac{ Y_{n, i}^2 }{\sigma_{n}^2 }  \left| \frac{ Y_{n, i} }{\sigma_{n} } \right|^{\zeta}
	\mmi\left\{  \frac{ Y_{n, i}^2 }{\sigma_{n}^2}   > \delta^2 \frac{   \sum_{l=1}^{n } c_{{n },l}^2}{ \max\limits_{1 \leq i \leq n } c_{{n },i}^2}  \right\} 
\\
	& \geq 
	 \frac{ Y_{n, i}^2 }{\sigma_{n}^2 }  
	 \left(
	 	\delta^2 \frac{   \sum_{l=1}^{n } c_{{n },l}^2}{ \max\limits_{1 \leq i \leq n } c_{{n },i}^2} 
	 \right)^{\zeta/2}
	\mmi\left\{  \frac{ Y_{n, i}^2 }{\sigma_{n}^2}   > \delta^2 \frac{   \sum_{l=1}^{n } c_{{n },l}^2}{ \max\limits_{1 \leq i \leq n } c_{{n },i}^2}  \right\}.
\end{align*}
Re-arranging,
\begin{align*}
\left(
	  \frac{ \max\limits_{1 \leq i \leq n } c_{{n },i}^2} { \delta^2  \sum_{l=1}^{n } c_{{n },l}^2} 
\right)^{\zeta/2}
\mme \left| \frac{ Y_{n, i} }{\sigma_{n} } \right|^{2+\zeta} 
	& \geq  \mme\left[
	 	\frac{ Y_{n, i}^2 }{\sigma_{n}^2 }  
		\mmi\left\{  \frac{ Y_{n, i}^2 }{\sigma_{n}^2}   > \delta^2 \frac{   \sum_{l=1}^{n } c_{{n },l}^2}{ \max\limits_{1 \leq i \leq n } c_{{n },i}^2}  \right\}
	\right],
\\
\delta^{-\zeta}
\left(
	  \frac{ \max\limits_{1 \leq i \leq n } c_{{n },i}^2} {   \sum_{l=1}^{n } c_{{n },l}^2} 
\right)^{\zeta/2}
\alpha_n
	& \geq  \mme\left[
	 	\frac{ Y_{n, i}^2 }{\sigma_{n}^2 }  
		\mmi\left\{  \frac{ Y_{n, i}^2 }{\sigma_{n}^2}   > \delta^2 \frac{   \sum_{l=1}^{n } c_{{n },l}^2}{ \max\limits_{1 \leq i \leq n } c_{{n },i}^2}  \right\}
	\right]	
\end{align*}
and the left-hand side converges to zero, which implies the right-hand side converges to zero and \eqref{eq:lemma:clt:lind} holds.

Therefore,
\begin{gather*}  
\frac{ \sum_{i=1}^{n } c_{n ,i} Y_{n , i} }{ \sigma_{n }  \sqrt{ \sum_{l=1}^{n } c_{n ,l}^2 } } \dto N(0,1)
\end{gather*}
and Lemma \ref{lemma:cond:clt} implies that
\begin{gather*}  
\frac{ \sum_{i=1}^{n } C_{n ,i} Y_{n , i} }{ \sigma_{n }  \sqrt{ \sum_{l=1}^{n } C_{n ,l}^2 } } \dto N(0,1).
\end{gather*}

\end{proof}

\begin{lemma}\label{lemma:unifcdf}
Consider a sequence of random CDFs $\{ F_n \}_n$ that converges pointwise in probability to a continuous CDF $F$, that is, for every $x \in \mmr$, $F_n(x) \pto F(x)$.
Then, the convergence is also uniform, namely,
\[
\sup_x | F_n(x) - F(x) | \pto 0.
\] 
\end{lemma}
\begin{proof}
Fix $\eps>0$ and choose $m: \eps>1/m$.
Use continuity of $F$ to pick points $-\infty = x_0 < x_1 < \ldots < x_{m-1}<x_m=\infty$ such that $F(x_j)-F(x_{j-1})=1/m$, for $j=1, \ldots, m$.
For any $x \in \mmr$, $\exists j$ such that $x_{j-1}\leq x \leq x_{j}$. Use monotonicity of $F$ and $F_n$ to show
\begin{align*}
F_n(x) - F(x) & \leq F_n(x_j) - F(x_{j-1}) = F_n(x_j) - F(x_{j}) + 1/m,
\\
F_n(x) - F(x) & \geq F_n(x_{j-1}) - F(x_{j}) = F_n(x_{j-1}) - F(x_{j-1}) - 1/m,
\\
|F_n(x) - F(x) | & \leq \max_j | F_n(x_{j}) - F(x_{j})|  + 1/m.
\end{align*}
Therefore,
\[ \mmp\left[ \sup_x |F_n(x) - F(x) | > \eps\right]  \leq \mmp\left[ \max_j |F_n(x_j) - F(x_j) | > \eps - 1/m \right]   \] 
\[	\leq \mmp\left[ \bigcup_j  |F_n(x_j) - F(x_j) | > \eps - 1/m \right] \]
\[ \leq \sum_j \mmp\left[ |F_n(x_j) - F(x_j) | > \eps - 1/m \right] \to 0 . \]

\end{proof}

\end{singlespace}

\begin{onehalfspace}
\bibliographystyleSM{econ}
\bibliographySM{biblio}
\end{onehalfspace}

\end{appendix}

\end{document}